\documentclass[%
  aps,%
  prx,%
  letterpaper,%
  superscriptaddress,%
  reprint,%
  longbibliography,%
  nofootinbib,%
  notitlepage,%
]{revtex4-2}

\usepackage[preset=reset,pkgset=extended,hyperrefdefs={defer,noemail}]{phfnote}
\usepackage{phfqit}
\usepackage{phfparen}
\usepackage{verbatim}

\usepackage[thmset=rich,smallproofs=false,proofref={marginbottom}]{phfthm}

\usepackage[nohelv]{newtxtext}
\usepackage[varg]{newtxmath}

\usepackage[semibold]{sourcesanspro}

\usepackage{./thermalchanneltechnicalstyle}
\usepackage{./thermalchanneldefs}

\usepackage{algorithm}
\usepackage{algpseudocode}

\begin{document}
\title{%
  Maximum channel entropy principle %
  and microcanonical channels}

\author{Philippe Faist}
\affiliation{Dahlem Center for Complex Quantum Systems, Freie Universität Berlin, 14195 Berlin, Germany}
\author{Sumeet Khatri}
\affiliation{Department of Computer Science and Center for Quantum Information Science and Engineering, Virginia Tech, Blacksburg, VA 24061, USA}
\date{August 5, 2025}
\begin{abstract}
  The thermal state plays a number of significant roles throughout physics,
  information theory, quantum computing, and machine learning. It arises from
  Jaynes' maximum-entropy principle as the maximally entropic state subject to
  linear constraints, and is also the reduced state of the microcanonical state
  on the system and a large environment.
  We formulate a maximum-channel-entropy principle, defining a thermal channel
  as one that maximizes a channel entropy measure subject to linear constraints
  on the channel.  We prove that thermal channels exhibit an exponential form
  reminiscent of thermal states.  We study examples including thermalizing
  channels that conserve a state's average energy, as well as Pauli-covariant
  and classical channels.
  We propose a quantum channel learning algorithm based on maximum channel
  entropy methods that mirrors a similar learning algorithm for quantum states.
  We then demonstrate the thermodynamic relevance of the maximum-channel-entropy
  channel by proving that it resembles the action on a single system of a
  microcanonical channel acting on many copies of the system.  Here, the
  microcanonical channel is defined by requiring that the linear constraints
  obey sharp statistics for any i.i.d.\@ input state, including for noncommuting
  constraint operators.
  Our techniques involve convex optimization methods to optimize recently
  introduced channel entropy measures, typicality techniques involving
  noncommuting operators, a custom channel postselection technique, as well as
  Schur-Weyl duality.
  As a result of potential independent interest, we prove a constrained
  postselection theorem for quantum channels.
  The widespread relevance of the thermal state throughout
  physics, information theory, machine learning, and quantum computing, inspires
  promising applications for the analogous concept for quantum channels.
\end{abstract}
\maketitle
\onecolumngrid

\tableofcontents

\section{Introduction}

Consider a quantum system $S$ and let $H_S$ be any Hermitian operator.  The
thermal state $\gamma_S$ is the state with the following Gibbs distribution of
energies:
\begin{align}
  \gamma_S(\beta)
  &= \frac1{Z_\beta}\,{e}^{-\beta H_S}\ ;
  &
    Z_\beta &= \operatorname{tr}\bigl({{e}^{-\beta H_S}}\bigr)\ .
              \label{z:9VJbRBw.}
\end{align}
The state $\gamma_S$ plays a number of significant roles throughout physics,
information theory, quantum computing, and machine learning.  In
thermodynamics, it is the state one typically attributes to a system with
Hamiltonian $H_S$ that is in equilibrium with large a heat bath at temperature
$1/\beta$~\cite{R0,R1,R2}.  In
statistical inference and information theory, this state can represent an
unknown state or probability distribution with limited prior knowledge. There,
the thermal state emerges from the maximum entropy principle, which mandates
that the inferred state should maximize the information entropy over all states
compatible with the prior
information~\cite{R3,R4,R5,R6,R7}.  Finally, the state $\gamma_S(\beta)$ has found
several uses in classical and quantum algorithms, whether in the context of the
mirror descent algorithm~\cite{R8}, the matrix
multiplicative weights algorithm~\cite{R9,R10}
or for quantum shadow tomography and quantum
learning~\cite{R11,R12,R13,R14}, quantum algorithms for semidefinite
programming~\cite{R15}, and online learning of quantum
states and processes~\cite{R16,R17}.

In this work, we extend the concept of the thermal quantum state to quantum
channels.  The present technical paper focuses on the details of our methods,
constructions, and proofs.  For a high-level overview of our work and its
significance, see our short companion paper~\cite{R18}.

The thermal state has a number of remarkable properties that lead to its broad
applicability.  Here is selection of defining properties:

\begin{enumerate}[label=(\roman*)]
\item \textbf{Thermal state from dynamical equilibration arguments~\cite{R0,R1,R2,R19,R20,R21,R22,R23,R24}.} A system evolving under open
  system dynamics in weak contact with a large bath typically relaxes to
  equilibrium by converging towards the thermal state $\gamma_S$.  In a closed
  many-body system after a long time unitary evolution, we typically expect
  local observables to reproduce the same statistics as if the entire state were
  in the thermal state $\gamma_S$.  Overall, many-body quantum systems typically
  relax towards a state that is well modeled by the thermal state, whether the
  relaxation is information-theoretically genuine (open system dynamics) or
  apparent (for a restricted set of observables).

\item \textbf{Thermal state from the maximum-entropy
    principle~\cite{R3,R4,R5,R6,R7}.} The state $\rho_S = \gamma_S$ achieves the
  maximum information-theoretic entropy ${S}_{}^{}({\rho_S})$ subject to the constraint
  $\operatorname{tr}({\rho_S H_S}) = E$, where $\beta$ is determined implicitly from $E$.

\item \textbf{Thermal state from the microcanonical
    ensemble~\cite{R25,R26}.} The \emph{microcanonical subspace
    at energy $[E,E+\Delta E]$} of a system $S'$ is defined as the subspace
  spanned by all energy eigenstates of $S'$ with energies in the interval
  $[E, E+\Delta E]$.  The \emph{microcanonical state $\pi_{S'}^{(E)}$ at energy
    $[E,E+\Delta E]$} is the maximally mixed state supported on the
  microcanonical subspace at energy $[E, E+\Delta E]$.  The microcanonical state
  models a closed, ergodic system whose energy statistics are confined in a
  small interval.
  Consider a system $S$ that is weakly interacting with a large heat bath $R$;
  here, $R$ is a system much larger than $S$ and with some suitable spectral
  properties.
  A central result in statistical mechanics states that if the joint system $SR$
  is modeled as a closed, ergodic system described by a microcanonical state,
  then the state of $S$ is the thermal state~\eqref{z:9VJbRBw.}.
  
\item \textbf{Thermal state by canonical
    typicality~\cite{R27}.} %
  The thermal state also has a much stronger property in the microcanonical
  picture: Not only does the maximally mixed state in the microcanonical
  subspace have a local reduced state that is close to the thermal state, but
  almost all individual states %
  in the subspace do, as well~\cite{R27}.

\item \textbf{Thermal state from complete passivity~\cite{R28,R29,R26}.}  Given a system $S$ with a
  Hamiltonian $H_S$, a state $\rho_S$ is \emph{energetically passive} if it is
  impossible to find a unitary operation $U_S$ that decreases the average energy
  of $\rho$, i.e., such that $\operatorname{tr}({H_S U\rho U^\dagger}) < \operatorname{tr}({H_S \rho})$.
  state $\rho_S$ is \emph{energetically completely passive} if
  $\rho^{\otimes n}$ is passive on $n$ copies of $S$, for all $n>0$.
  It turns out that the set of completely passive states of a system $S$
  coincides exactly with the set of thermal states $\gamma_S(\beta)$ for
  $\beta\geq 0$. 

\item \textbf{Thermal state from the resource theory of thermodynamics~\cite{R30,R31,R32,R33,R34,R35}.} %
  In a resource theory, we imagine an observer, or an agent, who manipulates
  quantum systems by applying operations from a set of \emph{free operations}.
  We then study what state transformations an agent is capable of carrying out;
  any state that cannot be reached using free operations can be thought of as
  being \emph{resourceful}.
  In the resource theory of thermodynamics, a common choice for the free
  operations is the set of \emph{thermal operations}: one may apply any
  energy-conserving unitary, one may include any ancillary system in its thermal
  state, and one may discard ancillary systems~\cite{R30,R31,R33}.
  We could ask, is there any other state that we could allow ancillary systems
  to be initialized in when defining the free operations?
  It turns out that allowing any other state for free renders the resource
  theory trivial---the agent can go from any state to any other state using only
  free operations.
  That is, the thermal state is singled out as the unique state (up to
  temperature) that we can allow for free in the resource theory of
  thermodynamics without trivializing the resource theory.

\end{enumerate}

The thermal state generalizes to the case where observables beyond the energy
$H$ are present.  If we maximize the entropy ${S}_{}^{}({\rho})$ over all states that
obey multiple constraints of the form $\operatorname{tr}({\rho Q_j}) = q_j$, for
$j=1, \ldots, J$, we find the generalized thermal state
\begin{align}
  \rho &= \gamma_S(\mu_1, \ldots, \mu_J)
  = \frac1{Z(\mu_1, \ldots, \mu_J)}\,{e}^{-\sum_{j=1}^J \mu_j Q_j}
         \ ;
  &
    Z(\mu_1, \ldots, \mu_J) &= \operatorname{tr}{e}^{-\sum_{j=1}^J \mu_j Q_j}
\ .
  \label{z:chtPfev.}
\end{align}
We can also write
$\gamma_S(\mu_1, \ldots, \mu_J) = {e}^{F -\sum_{j=1}^J \mu_j Q_j}$ with
$F = -\log[{Z(\mu_1, \ldots, \mu_j)}]$.  
For example, in the presence of two charges consisting of the energy $Q_1 = H$
and number of particles $Q_2 = N$ of a system, \cref{z:chtPfev.} is
simply the grand canonical ensemble of statistical mechanics.
We refer to the $\mu_j$'s in~\eqref{z:chtPfev.} as \emph{generalized
  chemical potentials} or simply \emph{chemical potentials} by extension of the
grand canonical ensemble.
In the presence of multiple charges, the thermal state is also called
\emph{generalized Gibbs ensemble (GGE)}~\cite{R36,R37,R20}.
The microcanonical, canonical typicality, and complete passivity properties also
extend to the situation where multiple charges are present, even if these
charges do not commute~\cite{R26,R35,R38}.  If the charges fail
to commute, the thermal state~\eqref{z:chtPfev.} is sometimes called the
\emph{non-Abelian thermal state}~\cite{R26,R39,R40,R41}.

Many concepts and ideas developed for quantum states have been extended to
quantum %
channels.  For instance, \emph{entropy measures} have been extended to quantum
channels~\cite{R42,R43,R44,R45,R46,R47},
and a \emph{resource theory of quantum channels} can describe the resources
required to convert one channel into another given a set of free
operations~\cite{R48,R49,R50,R47,R51,R52,R53}. 
(Recent advances in optimizing
the relative entropy for states and channels include
Refs.~\cite{R54,R55,R56,R57}.)
At
the same time, characterizing/learning noise processes in quantum systems (which
are described by quantum channels) is a critical component of developing
scalable quantum technologies. The question of inferring a quantum channel from
partial information therefore arises naturally. In particular, given partial
information about the evolution/noise of a system, specified by expectation
values resulting from sets of input states to the process and measurements at
the output, how do we determine a channel that is consistent with these known
expectation values? This question has been previously addressed using the
techniques of compressed sensing and least-squares
regression~\cite{R58,R59,R60}.
In this work, inspired by Jaynes' maximum entropy principle for quantum states, we
propose to take the channel achieving the \textit{maximum channel entropy}
subject to these constraints.
Maximum-entropy methods have furthermore been extensively studied in the
classical information theory literature in the context of maximally-entropic
stochastic processes and Markov chains~\cite{R61,R62}; see also \cite[Chapters~4, 12]{R63}.
Our work can be viewed as establishing a fully quantum counterpart of these
results.

Here, we extend the concept of a thermal state to thermal channels.  First, we
formulate and solve a \emph{maximum-channel-entropy principle for quantum
  channels}.  Its optimal solution, which we call a \emph{thermal quantum
  channel}, has an exponential form reminiscent of the thermal state.  A major
novelty in going from states to channels is that the thermal quantum channel
involves an optimization over input states, which can be understood as finding
the input state for which the channel produces the least entropic output
conditioned on the reference system.  We study multiple examples, including one
describing thermalizing dynamics while requiring a physical quantity (e.g.\@
energy) to be conserved on average on a system.  We also consider thermal
channels defined by constraints that satisfy certain symmetries, such as
covariance with respect to the Pauli group and covariance with respect to
Pauli-$Z$.

Just as the maximum-entropy principle for quantum states is used as the basis for quantum-state learning and inference~\cite{R64,R65,R66,R67,R68,R69,R70,R13,R14}, we make use of our maximum-channel-entropy principle to develop a learning algorithm for quantum channels. Our algorithm iteratively updates a guess for the unknown channel based on receiving new observable data. As a proof of concept, we apply our algorithm to single-qubit channels, and 
our numerics appear to show 
that our algorithm converges to the true,
unknown channel with an increasing number of iterations.

We then ask whether the quantum thermal channel can be derived from a
microcanonical picture, in an analogous fashion to the thermal state.  Recall
the following derivation of the generalized thermal state from a microcanonical
approach~\cite{R26}.
A microcanonical subspace associated with physical charges $\{{Q_1, \ldots Q_J }\}$
(such as energy, number of particles, etc.)  is a subspace containing all states
that are eigenstates of each $Q_j$ with an eigenvalue within a window
$[q_j, q_j + \Delta q_j]$.  If the $\{{Q_j}\}$ fail to commute, there may be no
such common eigenstates; instead, we may define an \emph{approximate
  microcanonical subspace}.  For a system $S$, and fixing real values
$\{{ q_j }\}$, we informally define an \emph{approximate microcanonical subspace}
$\mathcal{C}$ as a subspace of $S^{\otimes n}$ such that:
\begin{enumerate}[label=(\roman*)]
\item any $\rho$ with high weight in $\mathcal{C}$ has, for each $Q_j$, sharp
  statistics around $q_j$;
\item any $\rho$ with sharp statistics around $q_j$ for each $Q_j$ has high
  weight in $\mathcal{C}$.
\end{enumerate}
Here, ``sharp statistics around $q_j$'' refers to the outcome distribution of
$Q_j$ on $\rho$ having high weight within a small interval of values around
$q_j$.
It turns out that the reduced state of a maximally mixed state within an
approximate microcanonical subspace on a single copy of $S$ approaches the
generalized thermal state $\rho = \exp\{{ F- \sum \mu_j Q_j }\}$ as
$n\to\infty$~\cite{R26}.  The generalized
thermal state $\rho$, obtained initially by maximizing the entropy subject to
constraints on the charges, can therefore alternatively be derived from a
microcanonical picture.

To extend the concept of approximate microcanonical subspace to channels, we
need to make sense of `sharp statistics' in the context of a channel observable
$C^j_{BR}$ with expectation value $\operatorname{tr}[{C^j_{BR} \mathcal{N}(\Phi_{A:R})}]$.
This expectation value can be estimated over $n$ copies by preparing some input
state $\sigma_{AR}^{\otimes n}$ with
$\lvert {\sigma}\rangle _{AR} = \sigma_R^{1/2}\lvert {\Phi_{A:R}}\rangle $, applying
$\mathcal{N}^{\otimes n}$, and averaging measurements of
$\sigma_R^{-1/2} C^j_{BR} \sigma_R^{-1/2}$ on each copy.  Choosing $\sigma$
appropriately, rather than picking $\sigma_A = \mathds{1}_A/d_A$, might be important
to reliably detect how the channel $\mathcal{N}$ acts on states other than the
maximally mixed state.
The protocol can be described as measuring the
observable $\overline{H^{j;\sigma}}_{B^nR^n}$ on the state
$\mathcal{N}^{\otimes n}(\sigma_{AR}^{\otimes n})$, where
\begin{align}
  \overline{H^{j;\sigma}}_{B^nR^n}
  = \frac1n \sum_{i=1}^n \mathds{1}_{BR}^{\otimes(i-1)}
  \otimes \bigl({\sigma_{R_i}^{-1/2} C^j_{B_iR_i} \sigma_{R_i}^{-1/2}}\bigr) \otimes
  \mathds{1}_{BR}^{\otimes(n-i)}\ .
\end{align}
We then define a \emph{microcanonical channel operator} over many copies of a
system, extending the concept of an approximate microcanonical subspace.  The
microcanonical channel operator intuitively captures all quantum channels on $n$
copies of a system that produce outputs with sharp statistics of
$\overline{H^{j;\sigma}}_{B^nR^n}$, \emph{for all inputs $\sigma^{\otimes n}$}.

We then show that an associated \emph{microcanonical channel} leads to a thermal
channel on a single copy when ignoring the other copies.  This result gives an
independent characterization of the thermal quantum channel we obtained with the
maximum channel entropy principle.

Our technical proofs involve convex optimization techniques, typicality
techniques involving noncommuting
operators~\cite{R26,R71}, a postselection techniques for permutation-invariant
operators~\cite{R72,R73,R74,R75,R76}, as well as Schur-Weyl
duality~\cite{R77,R78}.  

We also prove a constrained postselection theorem for channels which might
be of potential independent interest.  We combine the features of
\R\cite{R72,R73,R74,R75,R76} to obtain an operator
upper bound on any permutation-invariant completely positive, trace-preserving
map $\mathcal{E}$ as a convex combination of i.i.d.\@ operators, with an
additional fidelity term that suppresses i.i.d.\@ operators that are far from
$\mathcal{E}$.

We discuss several aspects and consequences of our
results in \cref{z:aLWklD3j}.

\subsection{Overview of the main results}

We now provide an overview of our main technical contributions.  At this point,
the essential technical concepts required to state our main results are only
briefly introduced at a high level; we define all necessary concepts in greater
detail in \cref{z:-p8SEObu} below.

Consider systems $A$, $B$ along with a reference system $R\simeq A$.  Let
$\lvert {\Phi_{A:R}}\rangle  = \sum \lvert {j}\rangle _A\lvert {j}\rangle _R$.  Let $\{{ C^j_{BR} }\}_{j=1}^J$ be
Hermitian operators and $\{{ q_j }\}_{j=1}^J$ be real numbers.  The entropy of a
channel $\mathcal{N}_{A\to B}$ is defined
as~\cite{R45,R46}
${S}_{}^{}({\mathcal{N}}) = -{D}_{}^{}({\mathcal{N}}\mathclose{}\,\Vert\,\mathopen{}{\widetilde{\mathcal{D}}})$ with
${D}_{}^{}({\mathcal{N}}\mathclose{}\,\Vert\,\mathopen{}{\mathcal{M}}) = \max_{\lvert {\phi}\rangle _{AR}}
{D}_{}^{}({\mathcal{N}(\phi_{AR})}\mathclose{}\,\Vert\,\mathopen{}{\mathcal{M}(\phi_{AR})})$ and
$\widetilde{\mathcal{D}}(\cdot) = \operatorname{tr}(\cdot)\,\mathds{1}$, where
${D}_{}^{}({\rho}\mathclose{}\,\Vert\,\mathopen{}{\sigma}) = \operatorname{tr}\bigl({\rho[{\log(\rho) - \log(\sigma)}]}\bigr)$ is the Umegaki
quantum relative entropy.
We denote
by $\Pi^{\rho}$ the projector onto the support of $\rho$.

\paragraph{Maximum-channel-entropy principle:}
A channel $\mathcal{N}_{A\to B}$ is a \emph{thermal channel with respect to
  $\lvert {\phi}\rangle _{AR}$} if it maximizes ${S}_{\phi}^{}({\mathcal{N}})$ subject to the
constraints $\operatorname{tr}[{ C^j_{BR} \mathcal{N}(\Phi_{A:R}) }] = q_j$ for
$j=1, \ldots, J$.  It is a \emph{thermal channel} if it maximizes
${S}_{}^{}({\mathcal{N}})$ with the same constraints.
The order of the optimizations over $\mathcal{N}$ and $\phi$ is
irrelevant~\cite{R57,R79}; i.e., a thermal channel is also a thermal channel
with respect to an optimal $\phi$ in the definition of ${S}_{}^{}({\mathcal{N}})$.  We
also make a technical assumption to rule out some edge cases (cf.\@ details in
\cref{z:u4ZpmFvZ}).

\begin{overviewtheorem}[simplified]
  \label{z:dXvB1bw0}
  A channel $\mathcal{T}$ is a thermal channel if and only if its Choi matrix is
  of the form
  \begin{align}
    \mathcal{T}_{A\to B}(\Phi_{A:R})
    = \phi_R^{-1/2}\, \exp\Bigl\{{
    -\phi_R^{-1/2}\Bigl[{
    \mathds{1}_B\otimes F_R - \sum \mu_j C^j_{BR}
    - (\ldots)
    }\Bigr]\phi_R^{-1/2}
    }\Bigr\}\, \phi_R^{-1/2}
    + (\ldots)\ ,
  \end{align}
  where $F_R$ is Hermitian, where $\mu_j\in\mathbb{R}$, 
  where $(\ldots)$ represent terms that vanish unless $\phi_R$ is
  rank-deficient, and where $\lvert {\phi}\rangle _{AR} = \phi_R^{1/2}\lvert {\Phi_{A:R}}\rangle $ is
  optimal in ${S}_{}^{}({\mathcal{T}})$.  Moreover, if $\widehat{\mathcal{T}}$ is a
  complementary channel to $\mathcal{T}$, any optimal $\phi_R$ above satisfies
  \begin{align}
    \log(\phi_R) - \widehat{\mathcal{T}}^\dagger\bigl[{ \log(\widehat{\mathcal{T}}(\phi_R)) }\bigr]
    \propto \Pi_R^{\phi_R}\ .
  \end{align}
\end{overviewtheorem}
A full version of this theorem, including details of the terms $(\ldots)$, is
presented in \cref{z:u4ZpmFvZ}; see specifically
\cref{z:TLiwdJGR}.
Recall that $\rho = \exp\{{ F- \sum \mu_j Q_j }\}$, with $F,\mu_j\in\mathbb{R}$, is
the quantum state that maximizes ${S}_{}^{}({\rho})$ subject to $\operatorname{tr}({Q_j\rho}) = q_j$
for given Hermitian $Q_j$ and $q_j\in\mathbb{R}$ and for $j=1, \ldots, J$ (the
constraints fix $\mu_j, F$)~\cite{R4,R37,R26}.  In
\cref{z:dXvB1bw0}, the ``chemical potentials'' $\mu_j$ appear
in a similar fashion; the operator $F_R$ generalizes the ``free energy'' $F$.
We recover the standard thermal state if by choosing a trivial input system,
$\dim(R) = 1$, $\phi_R=1$.

We then consider the more general problem of minimizing the channel relative
entropy ${D}_{}^{}({\mathcal{N}}\mathclose{}\,\Vert\,\mathopen{}{\mathcal{M}})$ with respect to some arbitrary channel
$\mathcal{M}$.  We extend \cref{z:dXvB1bw0} to this case,
further including generalizations such as inequality constraints and a term in
the objective that is quadratic function of channel expectation values (see
\cref{z:gV43EWT.}).

\paragraph{A learning algorithm for quantum channels}

We apply the minimum channel relative entropy optimization problem to the learning of quantum channels. Specifically, we define a quantum channel generalization of the online quantum-state learning algorithms in \R\cite{R80,R69}. Our algorithm proceeds as follows. Suppose that at time step $t\in\{1,2,\dotsc\}$ in the learning procedure, our guess/estimate of the unknown channel is $\mathcal{M}^{(t)}$. We then measure an observable $E^{(t)}$ and let $s^{(t)}$ be our estimate of the expectation value of $E^{(t)}$ with respect to the unknown channel. Then, we update our guess to a new channel, $\mathcal{M}^{(t+1)}$, defined as the solution to the following optimization problem:
\begin{align}
    \begin{aligned}
        \text{minimize:}\quad & {D}_{}^{}({\mathcal{N}}\mathclose{}\,\Vert\,\mathopen{}{\mathcal{M}^{(t)}})
        + \eta \bigl({s^{(t)}-\operatorname{tr}[E^{(t)}\mathcal{N}(\Phi_{A:R})]}\bigr)^2 \\
        \text{subject to:}\quad & \mathcal{N}\text{ cp. tp.}.
    \end{aligned}
\end{align}
The quantity $\eta>0$ is a learning rate, quantifying the extent to which the error of the estimate $s^{(t)}$, namely $\bigl({s^{(t)}-\operatorname{tr}[E^{(t)}\mathcal{N}(\Phi_{A:R})]}\bigr)^2$, factors into the updated channel. We numerically solve this optimization problem for several example qubit channels, and 
our numerics appear to show that
that our algorithm converges to the true, unknown channel
as the number of iterations increases (cf.\@ \cref{z:wXQ-uunR} for details).

\paragraph{Microcanonical derivation of the thermal channel:}
Our microcanonical approach features in \cref{z:VbRRCMzN}.
Given real values $\{{ q_j }\}$, we define a \emph{approximate microcanonical
  channel operator} as an operator $P_{B^nR^n}$ with
$0\leq P_{B^nR^n} \leq \mathds{1}$ such that (informally):
\begin{enumerate}[label=(\roman*)]
\item Let $\mathcal{E}_{A^n\to B^n}$ be any channel such that
  $\operatorname{tr}[{P_{B^nR^n} \mathcal{E}(\sigma_{AR}^{\otimes n})}] \approx 1$
  for all $\sigma$ with
  eigenvalues above a small threshold.  Then the outcome probabilities of
  $\overline{H^{j;\sigma}}_{B^nR^n}$ on $\mathcal{E}(\sigma_{AR}^{\otimes n})$
  concentrates around $q_j$ for all $j$ and for all $\sigma$ with eigenvalues
  above a threshold.

\item Let $\mathcal{E}_{A^n\to B^n}$ be any channel such that the
  outcome probabilities of $\overline{H^{j;\sigma}}_{B^nR^n}$ on
  $\mathcal{E}(\sigma_{AR}^{\otimes n})$ concentrates around $q_j$ for all $j$ and
  for all $\sigma$ with eigenvalues above a threshold.  Then
  $\operatorname{tr}[{P_{B^nR^n} \mathcal{E}(\sigma_{AR}^{\otimes n})}] \approx 1$
  for all $\sigma$ with
  eigenvalues above a small threshold.
\end{enumerate}

By analogy with the microcanonical state, we define the microcanonical channel
$\Omega_n$ as the maximally entropic channel with high weight in $P_{B^nR^n}$.
Microcanonical channels lead to thermal channels
(\cref{z:1tRx46YH}):

\begin{overviewtheorem}[informal]
  Let $P_{B^nR^n}$ be an approximate microcanonical channel operator, let
  $\phi_R$ be any full-rank quantum state and let
  $\lvert {\phi}\rangle _{AR} = \phi_R^{1/2}\lvert {\Phi_{A:R}}\rangle $.  Then the quantum state
  $\operatorname{tr}_{n-1}[{ \Omega_n(\phi_{AR}^{\otimes n}) }]$ approaches
  $\mathcal{T}(\phi_{AR})$, where $\mathcal{T}$ is a thermal channel with
  respect to
  $\phi$.
\end{overviewtheorem}

One of the main technical contributions of this work is to construct an
approximate microcanonical operator.
This construction is inspired by the techniques of \R\cite{R71}.
The construction of $P_{B^nR^n}$ is designed such that the measurement of
$P_{B^nR^n}$ on any state of the form
$\mathcal{E}_{A^n\to B^n}(\sigma_{AR}^{\otimes n})$ (with
$\lvert {\sigma}\rangle _{AR} = \sigma_R^{1/2}\lvert {\Phi_{A:R}}\rangle $) equal to the probability of
the following protocol outputting ``ACCEPT'':
\begin{enumerate}[label=\arabic*.,start=0]
\item We begin with the state
  $\mathcal{E}_{A^n\to B^n}(\sigma_{AR}^{\otimes n})$ on $B^nR^n$;
\item Let $0< m< n$.  We measure the first $m$ copies of $R$ using a suitable
  POVM to obtain an estimate $\tilde\sigma$ for the input state
  $\sigma$.  (The first $m$ copies of $B$ are thrown away.)
\item On each of the remaining $\bar{n}\equiv n-m$ copies of $(BR)$, we pick
  $j\in\{{1, \ldots J}\}$ uniformly at random, measure the observable
  $\tilde\sigma_R^{-1/2} C^j_{BR} \tilde\sigma_R^{-1/2}$, and record its
  outcome.
\item We sort the outcomes by choice of $j$, and compute a quantity $\nu_j$ that
  is roughly equal to their sample average individually for each $j$.
\item We output ACCEPT if $\nu_j$ is close to $q_j$ for each $j$, and REJECT
  otherwise.
\end{enumerate}
The intention of this construction is to assert that the statistics of
measurement of each $C^j_{BR}$ (with the input state canceled out) is sharply
peaked around the prescribed values $q_j$.  We prove
(\cref{z:JjqyeW8p}):

\begin{overviewtheorem}[informal]
  \label{z:6L4E5tzm}
  The operator $P_{B^nR^n}$ constructed according to the above protocol
  satisfies the conditions of approximate microcanonical channel operator.
\end{overviewtheorem}

\paragraph{A constrained channel postselection theorem:}
As a key step in proving \cref{z:6L4E5tzm}, we derive an additional
postselection technique for quantum channels.  Postselection
techniques~\cite{R72,R73,R74,R75,R76} have found various uses
throughout quantum communication~\cite{R81},
cryptography~\cite{R72}, and
thermodynamics~\cite{R82}.  Specifically, we show that any
permutation-invariant channel $\mathcal{E}_n$ is operator-upper-bounded by an
integral over i.i.d.\@ channels $\mathcal{M}^{\otimes n}$ with an integrand that
includes a fidelity term between $\mathcal{E}_n$ and $\mathcal{M}^{\otimes n}$:
\begin{overviewtheorem}[Constrained channel postselection theorem; informal]
  There exists a measure $d\mathcal{M}$ on quantum channels such that
  for any permutation-invariant channel $\mathcal{E}_n$ and for any
  permutation-invariant operators $X,Y$,
  \begin{align}
    \mathcal{E}({Y X^\dagger \, ({\cdot}) \, X Y^\dagger})
    \leq \operatorname{poly}({n}) \int d\mathcal{M}\;
    \mathcal{M}^{\otimes n}(\cdot)\;
    F^2\Bigl({\mathcal{M}^{\otimes n}(X\lvert {\zeta }\rangle \mkern -1.8mu\relax \langle{\zeta }\rvert X^\dagger), \mathcal{E}_n(Y\lvert {\zeta }\rangle \mkern -1.8mu\relax \langle{\zeta }\rvert Y^\dagger)}\Bigr)\ ,
  \end{align}
  where $\lvert {\zeta}\rangle $ is a purification of the de Finetti state $\int d\sigma\,\sigma^{\otimes n}$
  and `$\leq$' refers to the complete positivity ordering.
\end{overviewtheorem}
See \cref{z:UEUfoDLC} for a full version. The proof
exploits Schur-Weyl duality~\cite{R77,R78}, and involves computing the de~Finetti state's
Schur-Weyl structure along with Haar-twirl integration
formulas~\cite{R83}.

As a corollary, we also derive an operator upper bound for a
permutation-invariant channel applied on any i.i.d.\@ input state:
\begin{overviewcorollary}[Constrained channel postselection theorem for i.i.d.\@
  input states; informal]
  There exists a measure $d\mathcal{M}$ on quantum channels such that for any
  permutation-invariant channel $\mathcal{E}_n$ and for any quantum state
  $\lvert {\sigma}\rangle _{AR} = \sigma_R^{1/2}\lvert {\Phi_{A:R}}\rangle $,
  \begin{align}
    \mathcal{E}({\sigma_{AR}^{\otimes n}})
    \lesssim \operatorname{poly}({n}) \int \! d\mathcal{M}\;
    \mathcal{M}^{\otimes n}(\sigma_{AR}^{\otimes n})\;
    \max_{\substack{ \tau_R:\\ \tau_R \approx \sigma_R }}
    F^2\bigl({\mathcal{M}^{\otimes n}({\tau_{AR}^{\otimes n}}),
    \mathcal{E}_n({\tau_{AR}^{\otimes n}})}\bigr)\ ,
  \end{align}
  where we write $\lvert {\tau}\rangle _{AR} \equiv \tau_R^{1/2} \lvert {\Phi_{A:R}}\rangle $, where
  `$\approx$' denotes proximity in fidelity, and where `$\lesssim$' conceals
  error terms that vanish exponentially in $n$.
\end{overviewcorollary}

\paragraph{Passivity and resource-theoretic aspects of the quantum thermal
  channel:}
We show that the thermal quantum channel is \emph{passive}, in the sense that
unitary operations on the input and the output cannot further improve the value
of any single constraint while preserving the others.  This statement extends
the corresponding passivity statement for quantum states, which states that 
no unitary can reduce the energy expectation value of the thermal state.

We furthermore discuss some challenges to understanding the role of the
quantum thermal channel in a thermodynamic resource theory of channels.

\section{Preliminaries}
\label{z:-p8SEObu}

\subsection{Quantum states and channels}
\label{z:6goiAHYb}

\paragraph{Notation for generic quantum information concepts.}
We consider quantum states on systems described by a finite-dimensional Hilbert
space. The Hilbert space associated with a system $A$ is denoted by $\mathscr{H}_A$, and
has dimension $d_A \equiv \dim(\mathscr{H}_A)$.
For any Hermitian operator $X_A$, we denote by $\Pi^{X_A}$ the projector onto
the support of $X_A$ and by $\Pi^{X_A\perp} = \mathds{1}- \Pi^{X_A}$ the projector
onto $X_A$'s kernel.
We write $X\geq 0$ for an operator $X$ if $X$ is positive semidefinite, and
$X>0$ if $X$ is positive definite.  Given two operators $X, Y$, we write
$X \geq Y$ if $(X - Y) \geq 0$ and $X > Y$ if $(X-Y) > 0$.  The \emph{operator
  norm} $\lVert {X}\rVert $ of an operator $X$ is its largest singular value; the
\emph{Schatten 1-norm} $\lVert {X}\rVert _1 = \operatorname{tr}\sqrt{X^\dagger X}$ is the sum of its
singular values.
We use the notation $\{{ A \in [a,b] }\}$ (respectively,
$\{{ A \not\in [a,b] }\}$) to denote the projector onto the
eigenspaces of a Hermitian operator $A$ with eigenvalues in  $[a,b]$
(respectively, not in $[a,b]$).  (More
generally, we can define $\{{ \chi(X) }\} \equiv \chi(X)$ for some boolean condition
function $\chi : \mathbb{R} \to \{{ 0, 1 }\}$ and Hermitian $X$, using the rule of
applying a scalar function on the eigenvalues of a Hermitian operator.)

A \emph{quantum state} (respectively, \emph{subnormalized quantum state}) on a
system $A$ is a positive semidefinite operator $\rho_A$ on $\mathscr{H}_A$ satisfying
$\operatorname{tr}(\rho_A)=1$ (respectively, $\operatorname{tr}(\rho_A)\leq1$).
A quantum measurement is specified by a \emph{positive operator-valued measure
  (POVM)}; if the measurement has a finite number of outcomes, the POVM is fully
specified by a collection of positive semidefinite operators $\{{ M_\ell }\}$ with
$\sum_\ell M_\ell = \mathds{1}$, and where the probability of obtaining $\ell$ after
measurement of $\rho$ is $\Pr[\ell] = \operatorname{tr}({M_\ell\rho})$.

Associated with each quantum system $S$ is a standard, or canonical, basis,
denoted by $\{{ \lvert {k}\rangle _S }\}$.  Given two systems $A, A'$, we write $A \simeq A'$
if their Hilbert spaces are isometric; we write $\mathds{1}_{A\to A'}$ the isometry
that maps the canonical basis of $A$ to the canonical basis of $A'$.  The
\emph{partial transpose} from $A$ to $A'$ is defined as
$t_{A\to A'}(\cdot) = \sum_{i,j} \langle {i}\mkern 1.5mu\relax \vert \mkern 1.5mu\relax {(\cdot)}\mkern 1.5mu\relax \vert \mkern 1.5mu\relax {j}\rangle _A \lvert {j}\rangle \mkern -1.8mu\relax \langle{i}\rvert _{A'}$.
For readability and/or when the systems are clear from context, we also write
$t_{A\to A'}(X_A) \equiv X_A^{t_{A\to A'}} \equiv X_A^t$.  We have the
elementary properties $t_{A'\to A}[{t_{A\to A'}(\cdot)}] = (\cdot)$ and
$t_{A\to A'}({X_A}) \, t_{A\to A'}({Y_A}) = t_{A\to A'}(Y_A X_A)$.
For $A\simeq R$, we define the \emph{nonnormalized reference maximally entangled
  ket}:
\begin{align}
  \lvert {\Phi_{A:R}}\rangle  = \sum_{k=1}^{d_A} \;  \lvert {k}\rangle _A\otimes\lvert {k}\rangle _{R}\ .
\end{align}
The latter has the following useful properties.
\begin{enumerate}[label=(\arabic*)]
\item We have
  $t_{A\to R}(\cdot) = \operatorname{tr}_{A}[{\Phi_{A:R}\,(\cdot)_A}]$ and
  $\operatorname{tr}_R[{\Phi_{A:R}\,(\cdot)^{t_{A\to R}}}] = (\cdot)_A$.
\item Any normalized or
  nonnormalized pure quantum state $\lvert {\psi_{AR}}\rangle $ can be written as
  $\lvert {\psi}\rangle _{AR} = ({\mathds{1}_A\otimes L_{R}})\lvert {\Phi_{A:R}}\rangle  = %
  ({L_R^{t_{R\to A}}\otimes\mathds{1}_R})\lvert {\Phi_{A:R}}\rangle $ where $L_{R}$ is a complex
  matrix with components
  $\langle {j}\mkern 1.5mu\relax \vert \mkern 1.5mu\relax {L_R}\mkern 1.5mu\relax \vert \mkern 1.5mu\relax {i}\rangle _R = (\langle {i}\rvert _A\otimes\langle {j}\rvert _R)\lvert {\psi}\rangle _{AR}$, where
  $L_{R}L_{R}^\dagger = \operatorname{tr}_A({\psi_{AR}}) \equiv \psi_R$,
  $(L_R^\dagger L_R)^{t_{R\to A}} = \psi_{A}$, and
  $\lVert {L}\rVert _2 = \operatorname{tr}({L^\dagger L}) = \operatorname{tr}({L L^\dagger}) = \operatorname{tr}({\psi})$.
  Furthermore, $L$ can always be made positive semidefinite by rotating
  $\lvert {\psi}\rangle _{AR}$ with a some suitable local unitary on $R$.
\end{enumerate}

For two quantum systems $A, B$, a \emph{superoperator} $\mathcal{E}_{A\to B}$ is
a linear map of operators on $\mathscr{H}_A$ to operators on $\mathscr{H}_B$.  It is
\emph{completely positive} if $\mathcal{E}_{A\to B}(\Phi_{A:R}) \geq 0$, where
$R\simeq A$.  The \emph{adjoint map} $\mathcal{E}_{A\leftarrow B}^\dagger$ of a
completely positive map $\mathcal{E}_{A\to B}$ is the unique completely positive
map satisfying
$\operatorname{tr}[{\mathcal{E}^{\dagger}_{A\leftarrow B}(X)\,Y}] =
\operatorname{tr}[{X\,\mathcal{E}_{A\to B}(Y)}]$ for all operators $X,Y$.  The map
$\mathcal{E}_{A\to B}$ is \emph{trace-preserving} if
$\mathcal{E}^\dagger(\mathds{1}_B) = \mathds{1}_A$ and \emph{trace-nonincreasing} if
$\mathcal{E}^\dagger(\mathds{1}_B) \leq \mathds{1}_A$.  A superoperator
$\mathcal{E}_{A\to B}$ that is completely positive and trace-preserving is also
called a \emph{quantum channel}. 
A \emph{Stinespring dilation} of a completely positive map
$\mathcal{E}_{A\to B}$ into an environment system $E$ is an operator
$K_{A\to BE}$ satisfying
$\mathcal{E}_{A\to B}({\cdot}) = \operatorname{tr}_E[{ K\,(\cdot)\,K^\dagger }]$.  If
$\mathcal{E}_{A\to B}$ is trace-nonincreasing, then $K_{A\to BE}$ satisfies
$K_{A\to BE}^\dagger K \leq \mathds{1}_A$; if $\mathcal{E}_{A\to B}$ is
trace-preserving, then $K_{A\to BE}$ is an isometry, meaning
$K^\dagger_{A\leftarrow BE} K_{A\to BE} = \mathds{1}_A$.
For any quantum channel $\mathcal{E}_{A\to B}$, a \emph{complementary channel}
$\widehat{\mathcal{E}}_{A\to E}$ is a quantum channel that can be written as
$\widehat{\mathcal{E}}_{A\to E}(\cdot) = \operatorname{tr}_B[{ V_{A\to BE}\,(\cdot)\,
V^\dagger }]$ where $V_{A\to BE}$ is a Stinespring dilation isometry of
$\mathcal{E}_{A\to B}$.
The \emph{Choi matrix} representation $N_{BR}$ of a channel
$\mathcal{N}_{A\to B}$ with $R\simeq A$ is defined as
$N_{BR} \equiv \mathcal{N}_{A\to B}(\Phi_{A:R})$.

We'll occasionally make use of the vectorized representation of operators and
channels.  For our purposes, the Hilbert-Schmidt space $\HS(\mathscr{H}_A)$ associated
with $\mathscr{H}_A$ is the complex linear vector space of all linear operators acting
on $\mathscr{H}_A$ with image in $\mathscr{H}_A$, and is equipped with the inner product
$(X_A, Y_A) \mapsto \operatorname{tr}({X_A^\dagger Y_A})$.  An operator $X_A$ on $\mathscr{H}_A$, viewed as
a vector in $\HS(\mathscr{H}_A)$, can be represented as
$\lvert {X_A}\rrangle  = ({X_A\otimes\mathds{1}})\lvert {\mathds{1}_A}\rrangle  = %
({\mathds{1}\otimes (X_A)^t})\lvert {\mathds{1}_A}\rrangle $ on two copies of $\mathscr{H}_A$, where
$\lvert {\mathds{1}_A}\rrangle  = \sum_{k=1}^{d_A} \lvert {k}\rangle \otimes\lvert {k}\rangle $.
The Hilbert-Schmidt inner product is then
$\operatorname{tr}({X_A^\dagger Y_A}) = \llangle {X_A}\mkern 1.5mu\relax \vert \mkern 1.5mu\relax {Y_A}\rrangle $ for
$\lvert {X_A}\rrangle ,\lvert {Y_A}\rrangle \in\HS(\mathscr{H}_A)$, where
$\llangle {X_A}\rvert  \equiv \llangle {\mathds{1}_A}\rvert  (X_A^\dagger\otimes\mathds{1})$ with
$\llangle {\mathds{1}_A}\rvert  \equiv \sum_{k=1}^{d_A} \langle {k}\rvert \otimes\langle {k}\rvert $.
Superoperators $\mathcal{E}_{A\to B}$ also act naturally in this representation,
i.e.,
$\mathcal{E}_{A\to B} \lvert {\rho_A}\rrangle  = \lvert { \mathcal{E}_{A\to B}[\rho_A] }\rrangle $.
The space of Hermitian operators, $\Herm(\mathscr{H}_A)$, is the real linear space
consisting of all Hermitian operators in $\HS(\mathscr{H}_A)$.

The \emph{trace distance} between two states $\rho,\sigma$ is defined as
$D(\rho,\sigma) = (1/2)\lVert {\rho-\sigma}\rVert _1$, and the \emph{fidelity} of
$\rho,\sigma$ is
$F(\rho,\sigma) = \lVert {\sqrt\rho \sqrt\sigma}\rVert _1 =
\operatorname{tr}\bigl[{\bigl({\rho^{1/2}\sigma \rho^{1/2}}\bigr)^{1/2}}\bigr]$.  We extend these
definitions formally for any positive semidefinite operators
$\rho,\sigma\geq 0$.  If at least one of two subnormalized states $\rho,\sigma$
is normalized, then we define the \emph{purified distance}
$P(\rho,\sigma) = \sqrt{1 - F^2(\rho,\sigma)}$ and we have
$D(\rho,\sigma) \leq P(\rho,\sigma)$~\cite{R84,R85,R86,R87}.
The proximity of two quantum channels
$\mathcal{N}_{A\to B}, \mathcal{M}_{A\to B}$ is quantified with the
\emph{diamond norm}
$(1/2)\lVert {\mathcal{N}_{A\to B} - \mathcal{M}_{A\to B}}\rVert _\diamond =
(1/2)\max_{\rho_{AR}} \lVert {\mathcal{N}_{A\to B}(\rho_{AR}) -
  \mathcal{M}_{A\to B}(\rho_{AR})}\rVert _1 = (1/2)\max_{\lvert {\phi}\rangle _{AR}}
\lVert {\mathcal{N}_{A\to B}(\phi_{AR}) - \mathcal{M}_{A\to B}(\phi_{AR})}\rVert _1$,
where the optimization ranges over states on $A$ and a reference system
$R\simeq A$ and where the maximum is always attained by a pure state.

\paragraph{Channel observables.}
We now review the notion of a channel observable~\cite{R88,R89,R90}.  Such
operators generalize the idea of quantum measurement operators for states to
operators that describe what information can be extracted from an unknown
quantum channel.  Channel observables are a key conceptual ingredient in our
construction of the thermal channel: They serve to specify partial prior
information about a channel, generalizing the constraint on the expectation
value of an observable in the maximum entropy principle for states.

Given single-copy black-box access to an unknown quantum channel
$\mathcal{E}_{A\to B}$, the most general quantum operation we may perform to
learn properties of $\mathcal{E}_{A\to B}$ is to prepare an initial state
$\psi_{AR}=\lvert {\psi}\rangle \mkern -1.8mu\relax \langle{\psi}\rvert _{AR}$ on $A$ and some additional reference system $R$, apply the
unknown channel onto $A\to B$, and perform a joint measurement on $BR$.  If the
measurement is described by a POVM $\{{ M^\ell_{BR} }\}$, the probability of
obtaining outcome $\ell$ is expressed as
\begin{align}
  \Pr\bigl[{ \ell \mathrel{\big|}  \psi_{AR}, \mathcal{E}_{A\to B}, \{{ M^j_{BR} }\} }\bigr]
  = \operatorname{tr}[{ M^\ell_{BR} \, \mathcal{E}_{A\to B}({\psi_{AR}}) }]\ .
  \label{z:yP868WWH}
\end{align}
A mixed input state $\rho_{AR}$ can be purified into an additional system that
can be included in $R$; it thus suffices to consider pure state inputs.
Furthermore, we can assume without loss of generality that $R\simeq A$.  Indeed,
all purifications of the state $\operatorname{tr}_R(\psi_{AR})$ on $A$ are equivalent via a
local partial isometry on the purifying system to one in which $R\simeq A$; the
latter can be absorbed into the POVM.  Finally, we can write
$\lvert {\psi}\rangle _{AR} = (\mathds{1}_A\otimes L_R)\lvert {\Phi_{A:R}}\rangle $ for some complex matrix $L_R$
and write
$\operatorname{tr}[{ M^\ell_{BR} \mathcal{E}_{A\to B}(\psi_{AR}) }] = \operatorname{tr}[{ L_R^\dagger
M^\ell_{BR} L_R \, \mathcal{E}_{A\to B}(\Phi_{A:R})}]$.  These outcome
probabilities therefore can be written as
\begin{align}
  \Pr[\ell] = \operatorname{tr}[{ \widetilde{M}^{\ell}_{BR} \, \mathcal{E}_{A\to B}(\Phi_{A:R}) }]\ ,
\end{align}
where $\{{ \widetilde{M}^\ell_{BR} }\}$ are now a collection of positive
semidefinite operators that satisfy
$\sum_\ell \widetilde{M}^\ell_{BR} = \mathds{1}_B\otimes({L_R L_R^\dagger}) =
\mathds{1}_B\otimes\psi_R$.  Such a collection of operators is called a
\emph{channel measurement}, \emph{channel POVM}, or \emph{process
  POVM}~\cite{R89}.

Therefore, any real-valued outcome statistics that we can obtain using quantum
operations from a single black-box access to unknown channels, being linear
combination of such outcome probabilities, can be written in the form
\begin{align}
  \operatorname{tr}[{ C_{BR} \, \mathcal{E}_{A\to B}(\Phi_{A:R}) }]\ ,
\end{align}
where $C_{BR}$ is some Hermitian operator.  We call $C_{BR}$ a \emph{channel observable}.

\subsection{Entropy measures for states and channels}

The \emph{von Neumann entropy} of a quantum state $\rho$ is
\begin{align}
  {S}_{}^{}({\rho}) = -\operatorname{tr}({\rho\log\rho})\ .
  \label{z:xkQPXbUo}
\end{align}
In this paper, $\log$ denotes the natural logarithm and entropy is quantified in
number of ``nats,'' where one bit is $\log(2)$ nats.
For any quantum state $\rho$, and for any $\Gamma\geq 0$, we define the
(Umegaki) \emph{quantum relative entropy} as~\cite{R91,R92}
\begin{align}
  {D}_{}^{}({\rho}\mathclose{}\,\Vert\,\mathopen{}{\Gamma}) = \operatorname{tr}\bigl({\rho\bigl[{\log(\rho) - \log(\Gamma)}\bigr]}\bigr)\ .
  \label{z:P7JsTs-9}
\end{align}
We conventionally set ${D}_{}^{}({\rho}\mathclose{}\,\Vert\,\mathopen{}{\Gamma}) = \infty$ if $\rho$'s support is not
contained in $\Gamma$'s.  Observe that
${S}_{}^{}({\rho}) = -{D}_{}^{}({\rho}\mathclose{}\,\Vert\,\mathopen{}{\mathds{1}})$.
For any normalized states $\rho,\sigma$, we have
${D}_{}^{}({\rho}\mathclose{}\,\Vert\,\mathopen{}{\sigma}) \geq 0$. %
Extending the definition~\eqref{z:P7JsTs-9} formally to arbitrary
positive semidefinite operators $\rho,\Gamma\geq 0$, we have the following
scaling property for any $a,b>0$:
\begin{align}
  {D}_{}^{}({a\rho}\mathclose{}\,\Vert\,\mathopen{}{b\Gamma})
  = a\Bigl[{ \operatorname{tr}(\rho) \log\Bigl({\frac{a}{b}}\Bigr) + {D}_{}^{}({\rho}\mathclose{}\,\Vert\,\mathopen{}{\Gamma}) }\Bigr]\ .
  \label{z:4CgKuT8M}
\end{align}

For a state $\rho_A$ on a system $A$, we also introduce the alternative notation
${S}_{}^{}({A})_{\rho} = {S}_{}^{}({\rho_A})$.  We also define for a bipartite state $\rho_{AB}$
the \emph{conditional von Neumann entropy}
${S}_{}^{}({A}\mathclose{}\,|\,\mathopen{}{B})_{\rho} = {S}_{}^{}({AB})_{\rho} - {S}_{}^{}({B})_{\rho}=-D(\rho_{AB}\Vert\mathds{1}_A\otimes\rho_B)$.

Let $\mathcal{N}_{A\to B}$ be a quantum channel and let $\mathcal{M}_{A\to B}$
be a completely positive map. Let $R$ be any reference system and $\rho_{AR}$
be any fixed state. The \emph{channel relative entropy with respect to
  $\rho_{AR}$} is defined as
\begin{align}
  {D}_{\rho}^{}({\mathcal{N}_{A\to B}}\mathclose{}\,\Vert\,\mathopen{}{\mathcal{M}_{A\to B}})
  &= {D}_{}^{}({\mathcal{N}_{A\to B}(\rho_{AR})}\mathclose{}\,\Vert\,\mathopen{}{\mathcal{M}_{A\to B}(\rho_{AR})})\, .
  \label{z:BKNZ0Ibj}
\end{align}
By optimizing~\eqref{z:BKNZ0Ibj} with respect to every state $\rho_{AR}$, we define the \emph{channel relative
  entropy}~\cite{R93,R50,R79} as
\begin{align}
    {D}_{}^{}({\mathcal{N}_{A\to B}}\mathclose{}\,\Vert\,\mathopen{}{\mathcal{M}_{A\to B}})
    &= \max_{\rho_{AR}} {D}_{\rho}^{}({\mathcal{N}_{A\to B}}\mathclose{}\,\Vert\,\mathopen{}{\mathcal{M}_{A\to B}})
    = \max_{\lvert {\phi}\rangle _{AR}} {D}_{\phi}^{}({\mathcal{N}_{A\to B}}\mathclose{}\,\Vert\,\mathopen{}{\mathcal{M}_{A\to B}})\ ,
    \label{z:Y.8CkTW4}
\end{align}
where the optimal value in the first optimization
in~\eqref{z:Y.8CkTW4} is always attained by some pure state
$\phi_{AR}=\lvert {\phi}\rangle \mkern -1.8mu\relax \langle{\phi}\rvert _{AR}$ with $R\simeq A$. %

We define the \emph{channel entropy with respect to $\rho_{AR}$} of the quantum channel $\mathcal{N}_{A\to B}$ as the entropy of the output system $B$ conditioned on the reference system $R$:
\begin{align}
  {S}_{\rho}^{}({\mathcal{N}}) &= S(B|R)_{\mathcal{N}(\rho)}
    \nonumber\\
  &= -{D}_{}^{}({\mathcal{N}_{A\to B}(\rho_{AR})}\mathclose{}\,\Vert\,\mathopen{}{\widetilde{\mathcal{D}}_{A\to B}(\rho_{AR})})
    \nonumber\\
    & =  \log(d_B) - {D}_{\rho}^{}({\mathcal{N}_{A\to B}}\mathclose{}\,\Vert\,\mathopen{}{\mathcal{D}_{A\to B}})
    \nonumber\\
  &= -{D}_{}^{}({\mathcal{N}_{A\to B}(\rho_{AR})}\mathclose{}\,\Vert\,\mathopen{}{\mathds{1}_B\otimes\rho_R})\ ,
  \label{z:nVlXo461}
\end{align}
where $\mathcal{D}_{A\to B}(\cdot) = \operatorname{tr}_A({\cdot})\,\mathds{1}_B/d_B$ is the
completely depolarizing channel and
$\widetilde{\mathcal{D}}_{A\to B}(\cdot) = \operatorname{tr}({\cdot})\,\mathds{1}_B$ its
nonnormalized version. The \emph{channel entropy} of a quantum channel is then the minimum conditional entropy of the output $B$, conditioned on $R$~\cite{R45,R94,R46}:
\begin{align}
  {S}_{}^{}({\mathcal{N}_{A\to B}})
  &= \min_{\lvert {\phi}\rangle _{AR}} {S}_{\phi}^{}({\mathcal{N}_{A\to B}})
    \nonumber\\
  &=\min_{\lvert {\phi}\rangle _{AR}}S(B|R)_{\mathcal{N}(\phi)}
    \nonumber\\
  &= - \max_{\lvert {\phi}\rangle _{AR}} {D}_{}^{}({\mathcal{N}_{A\to B}(\phi_{AR})}\mathclose{}\,\Vert\,\mathopen{}{\mathds{1}_B\otimes\phi_{R}})
    \nonumber\\
  &= - {D}_{}^{}({ \mathcal{N}_{A\to B} }\mathclose{}\,\Vert\,\mathopen{}{ \widetilde{\mathcal{D}}_{A\to B} })\ .
    \label{z:4jPiiO8e}
\end{align}

The entropy of a channel quantifies the minimum output entropy of a channel when
measured conditioned on a reference system $R$.  In other words, a channel with
high entropy is one that is guaranteed to output a highly entropic state
(relative to $R$) for any input state.  This interpretation makes the channel
entropy an appealing quantity to maximize in our maximum channel entropy
principle (see also discussion in our companion
paper~\cite{R18}).

A closely related entropy measure is the thermodynamic capacity of a channel.
The \emph{thermodynamic capacity} of a quantum channel $\mathcal{N}'_{A\to B}$
with respect to positive semidefinite operators $\Gamma_A, \Gamma_B'$ is
defined~\cite{R47,R82} as
\begin{align}
  T({{\mathcal{N}'_{A\to B}} \,{\Vert}\, {\Gamma_A, \Gamma'_B} })
  = \max_{\sigma_A} \bigl[{
  {D}_{}^{}({\mathcal{N}'_{A\to B}(\sigma_A)}\mathclose{}\,\Vert\,\mathopen{}{\Gamma'_B})
  - {D}_{}^{}({\sigma_A}\mathclose{}\,\Vert\,\mathopen{}{\Gamma_A})
  }\bigr]\ .
\end{align}
In the special case $\Gamma_A = \mathds{1}_A$ and $\Gamma'_B = \mathds{1}_B$, we find
\begin{align}
  T({\mathcal{N}'_{A\to B}})
  = \max_{\sigma_A} \bigl[{ {S}_{}^{}({\sigma}) - {S}_{}^{}({\mathcal{N}'(\sigma)}) }\bigr]\ .
  \label{z:NdbPdJS1}
\end{align}
In this special case, and if we further assume $d_A = d_B$,
we have that $T({\mathcal{N}'_{A\to B}})$ is always positive (via
the choice $\sigma_A=\mathds{1}_A/d_A$ in the $\max$), and it is
equal to zero for any unital channel (since
${S}_{}^{}({\mathcal{N}'(\sigma)})\geq{S}_{}^{}({\sigma})$ for any unital
channel $\mathcal{N}'$).

The channel entropy is closely related to the thermodynamic
capacity.  Let $\mathcal{N}_{A\to B}$ be a
quantum channel, let $V_{A\to BE}$ be a Stinespring dilation of $\mathcal{N}$,
and let $\widehat{\mathcal{N}}(\cdot) = \operatorname{tr}_E\bigl[{ V\,(\cdot)\,V^\dagger }\bigr]$.
Then, for any $\lvert {\phi}\rangle _{AR}$, we have
$ {D}_{}^{}({\mathcal{N}(\phi_{AR})}\mathclose{}\,\Vert\,\mathopen{}{\mathds{1}_B\otimes\phi_R}) =
\operatorname{tr}[{\mathcal{N}(\phi)\,\log(\mathcal{N}(\phi))}] - \operatorname{tr}[{\phi_R\log(\phi_R)}] $,
leading to the following alternative expressions of the channel entropy with
respect to $\lvert {\phi}\rangle _{AR}$:
\begin{align}
  {S}_{\phi}^{}({\mathcal{N}})
  &= {S}_{}^{}({\mathcal{N}(\phi_{AR})}) - {S}_{}^{}({\phi_R})
    = {S}_{}^{}({B}\mathclose{}\,|\,\mathopen{}{R})_{\mathcal{N}(\phi)}
    = -{S}_{}^{}({B}\mathclose{}\,|\,\mathopen{}{E})_{V\phi V^\dagger}
    = -{S}_{}^{}({BE})_{V\phi V^\dagger} + {S}_{}^{}({E})_{V\phi V^\dagger}
    \nonumber\\
  &= -{S}_{}^{}({R})_{\phi} + {S}_{}^{}({E})_{\widehat{\mathcal{N}}(\phi)}
    = {S}_{}^{}({ \widehat{\mathcal{N}}(\phi_{A}) }) - {S}_{}^{}({ \phi_R })\ .
    \label{z:hkVuHdrl}
\end{align}
Therefore, the channel entropy is directly related to the complementary
channel's thermodynamic capacity:
\begin{align}
  {S}_{}^{}({\mathcal{N}})
  = \min_{\lvert {\phi}\rangle _{AR}} {S}_{\phi}^{}({\mathcal{N}})
  = - T(\widehat{\mathcal{N}})\ .
  \label{z:vTzieZ6M}
\end{align}

\section{Maximum-entropy derivation of the thermal channel}
\label{z:u4ZpmFvZ}

One way to define the thermal quantum state is through Jaynes' maximum entropy
principle~\cite{R3,R4}.  Given
a collection of Hermitian observables $\{{ Q_j }\}_{j=1}^J$, along with real values
$\{{ q_j }\}_{j=1}^J$, we ask which quantum state $\rho$ maximizes the entropy
${S}_{}^{}({\rho})$ subject to the constraints $\operatorname{tr}({Q_j \rho}) = q_j$ for
$j=1, \ldots, J$.  The observables $\{{ Q_j }\}$ need not commute.  Jaynes'
calculation, presented in standard textbooks, proceeds as follows.  One
introduces Lagrange multipliers $\mu_j\in\mathbb{R}$ (for $j=1, \ldots J$) to
account for the expectation value constraints and $\lambda\in\mathbb{R}$ to
account for the constraint $\operatorname{tr}({\rho})=1$.  One then looks for the stationary
points of
$L(\rho) = {S}_{}^{}({\rho}) - \sum \mu_j [{q_j - \operatorname{tr}({Q_j\rho})}] -
\lambda[{1-\operatorname{tr}(\rho)}]$.  If we perform the variation $\rho\to\rho+\delta\rho$,
we find to first order in $\delta\rho$ that
$\delta L = -\operatorname{tr}\bigl[{(\log\rho + \mathds{1})\delta \rho}\bigr] + \sum \mu_j
\operatorname{tr}({Q_j\delta\rho}) + \lambda\delta\rho$.  For $\rho$ to be a stationary point
of $L(\rho)$, this expression must vanish for all $\delta\rho$; this happens
exactly when $\log({\rho}) + \mathds{1}+ \sum \mu_j Q_j + \lambda\mathds{1}= 0$.
Solving for $\rho$ while introducing the quantity $Z = \exp(1 + \lambda)$ yields
the familiar form for the thermal state $\rho$:
\begin{align}
  \rho = \frac{{e}^{-\sum \mu_j Q_j}}{Z}\ .
  \label{z:Jdr9-JJH}
\end{align}

Here, we formulate and solve the analogous problem for quantum channels.  Given
an input system $A$, an output system $B$, and $R\simeq A$, and given a set of
channel observables (Hermitian operators) $\{{ C^j_{BR} }\}_{j=1}^J$ along with
real values $\{{ q_j }\}_{j=1}^J$, we ask: \emph{What quantum channel
  $\mathcal{N}_{A\to B}$ maximizes the channel entropy ${S}_{}^{}({\mathcal{N}})$,
  subject to the constraints
  $\operatorname{tr}[{C^j_{BR} \mathcal{N}_{A\to B}(\Phi_{A:R})}] = q_j$?}  We call such an
optimal channel a \emph{thermal channel}.
In the following sections, we leverage a formulation of this problem as a convex
optimization problem in order to derive a general structure of thermal channels.

Furthermore, rather than maximizing the channel entropy, we can also consider
more generally minimizing the channel relative entropy with respect to any fixed
completely positive map $\mathcal{M}_{A\to B}$ subject to linear constraints.
We analyze this generalization in \cref{z:aXvqINGx} below.

\subsection{Definition of the thermal channel}
\label{z:7suOBm0G}

Let $A,B$ be quantum systems and let $R\simeq A$.  Let $\{{ C^j_{BR} }\}_{j=1}^J$
be a collection of Hermitian operators and let $\{{ q_j }\}_{j=1}^J$ with
$q_j\in\mathbb{R}$.  Consider the following optimization problem:
\begin{align}
  \label{z:sb5FfEOw}
  \begin{aligned}[t]
    \textup{maximize:} \quad
    & {S}_{}^{}({\mathcal{N}_{A\to B}})
    \\
    \textup{over:}\quad
    & \mathcal{N}_{A\to B}\ \textup{c.p., t.p.}
    \\
    \textup{such that:}\quad
    & \operatorname{tr}\bigl[{C^j_{BR}\,\mathcal{N}_{A\to B}(\Phi_{A:R})}\bigr] = q_j\quad\text{for \(j=1, \ldots, J\)}\ .
  \end{aligned}
\end{align}
The maximization is taken over all completely positive (c.p.), trace-preserving
(t.p.) superoperators $\mathcal{N}_{A\to B}$ that satisfy the linear
channel-observable constraints specified by $C^j_{BR}, q_j$.

We assume that the problem is feasible, namely that there exists a channel
$\mathcal{N}_{A\to B}$ satisfying the given constraints.  This assumption rules
out the trivial situation where the constraints are incompatible.

In fact, we henceforth make a stricter assumption which is important for our
analysis.  We assume that the problem is \emph{strictly feasbile}, namely that
there is at least one quantum channel $\mathcal{N}_{A\to B}$ that satisfies the
given constraints and whose Choi matrix $\mathcal{N}_{A\to B}(\Phi_{A:R})$ is
positive definite.  In other words, the constraints do not force
$\mathcal{N}_{A\to B}$ to lie on the boundary of the set of all completely
positive superoperators.  This assumption rules out some edge cases where the
constraints are just so finely tuned that the hyperplane of
constraint-satisfying superoperators is tangent to (and ``barely touches'') the
set of completely positive maps.

\begin{definition}[Thermal channel]
We define a \textbf{thermal quantum channel} with respect to the constraints
$(C^j_{BR}, q_j)$ as the quantum channel $\mathcal{T}_{A\to B}^{}$ that achieves the
optimal value in~\eqref{z:sb5FfEOw}.
\end{definition}

Our first main result is a general structure of the thermal channel.  Given the
optimization problem~\eqref{z:sb5FfEOw}, and with our
additional strict feasibility assumption, we have the following theorem:
\begin{theorem}[Structure of the thermal channel]
  \noproofref
  \label{z:TLiwdJGR}
  A quantum channel $\mathcal{T}_{A\to B}^{}$ is a thermal channel if and only if it satisfies all the constraints in~\eqref{z:sb5FfEOw} and it
has a
  Choi matrix of the form
  \begin{align}
    \mathcal{T}_{A\to B}^{}(\Phi_{A:R})
    = \phi_R^{-1/2}
    \exp\Bigl\{{
    - \phi_R^{-1/2}\Bigl[{
    \sum \mu_j C^j_{BR}
    -\mathds{1}_B\otimes ({F_R + \phi_R\log\phi_R}) 
    - S_{BR} }\Bigr] \phi_R^{-1/2}
    }\Bigr\}
    \phi_R^{-1/2}
    + Y_{BR}\ ,
    \label{z:s1gf97hu}
  \end{align}
  where:
  \begin{itemize}
    \item $\mu_j\in\mathbb{R}$, $j=1,\ldots, J$;
    \item $Y_{BR}$ is a
      Hermitian operator satisfying $\Pi^{\phi_R}_R Y_{BR} \Pi^{\phi_R}_R = 0$;
    \item $S_{BR}$ is a positive semidefinite operator satisfying
      $S_{BR}\,\mathcal{T}_{A\to B}^{}(\Phi_{A:R}) = 0$;
    \item it holds that
      $\Pi_R^{\phi_R\perp}\bigl({\sum \mu_j C^j_{BR}
      -\mathds{1}_B\otimes F_B  - S_{BR} }\bigr) = 0$;
    \item $F_R$ is a Hermitian operator; and
    \item $\phi_R$ is the local
  reduced state on $R$ of an optimal state
  $\lvert {\phi}\rangle _{AR} = \phi_R^{1/2}\lvert {\Phi_{A:R}}\rangle $ in the definition of the channel
  entropy
  ${S}_{}^{}({\mathcal{N}_{A\to B}}) = \min_{\lvert {\phi}\rangle _{AR}}
  {S}_{\phi}^{}({\mathcal{N}_{A\to B}})$.
  \end{itemize}
  Any optimal state $\phi_A$ (with
  $\phi_A = \operatorname{tr}_R({\phi_{AR}}) = \phi_R^{t_{R\to A}}$) must satisfy
  \begin{align}
    \log(\phi_A) - \widehat{\mathcal{T}}^\dagger\Bigl({
    \log\bigl[{ \widehat{\mathcal{T}}_{A\to E}({\phi_A}) }\bigr] }\Bigr) \ \propto\  \Pi_A\ ,
    \label{z:lWRPNfFI}
  \end{align}
  where $\widehat{\mathcal{T}}_{A\to E}$ is a complementary channel to
  $\mathcal{T}_{A\to B}^{}$.  If $\phi_A$ has full rank, then $S_{BR} = 0 = Y_{BR}$,
  and~\eqref{z:lWRPNfFI} is sufficient for
  optimality of $\phi_A$.  The channel entropy attained by $\mathcal{T}_{A\to B}^{}$ is
  \begin{align}
    {S}_{}^{}({\mathcal{T}_{A\to B}^{}}) = -\operatorname{tr}({F_R}) + \sum_{j=1}^J \mu_j q_j\ .
    \label{z:WlEdtU0A}
  \end{align}
\end{theorem}

The remainder of this section we construct a proof of the above theorem, by
analyzing the optimization~\eqref{z:sb5FfEOw} using convex
optimization techniques~\cite{R95}.

\subsection{Reduction to a fixed-input maximum channel entropy}
\label{z:0QjUzwvU}

The convex structure of the problem is not immediately obvious
from~\eqref{z:sb5FfEOw}, given that the channel entropy
${S}_{}^{}({\mathcal{N}_{AR}})$ involves a minimization over pure states $\psi_{AR}$.
Writing out the problem explicitly, we have
\begin{align}
  \text{\eqref{z:sb5FfEOw}}\quad
  &= \quad 
    \begin{aligned}[t]
      -\min_{\substack{
      \mathcal{N}:\,\mathrm{cp., t.p.}\\
      \operatorname{tr}[{ C^j_{BR}\, N_{BR} }] = q_j
      }} \max_{\lvert {\phi}\rangle _{AR}}
      \ 
      {D}_{}^{}({ \mathcal{N}_{A\to B}(\phi_{AR}) }\mathclose{}\,\Vert\,\mathopen{}{ \mathds{1}_B \otimes \phi_R })\ .
    \end{aligned}
    \label{z:ypjc2pEs}
\end{align}
The maximum-channel-entropy thermal
channel optimization is then equivalently written as
\begin{align}
  \text{\eqref{z:sb5FfEOw}}
  &= \text{\eqref{z:ypjc2pEs}}
  = \ -\min_{\substack{
      \mathcal{N}:\,\mathrm{cp., t.p.}\\
      \operatorname{tr}[{ C^j_{BR}\, N_{BR} }] = q_j
      }} \max_{\lvert {\phi}\rangle _{AR}}
      \ g_{\widetilde{\mathcal{D}}}\bigl({  \mathcal{N}_{A\to B}, \phi_{A} }\bigr)\ ,
  \label{z:jZzkURqZ}
\end{align}
where
\begin{align}
  g_{\mathcal{M}}(\mathcal{N}, \phi_{R}) \coloneqq
  {D}_{}^{}\bigl ({ \phi_R^{1/2}N_{BR}\phi_R^{1/2}}\mathclose{}\,\big \Vert\,\mathopen{}{\phi_R^{1/2}M_{BR}\phi_R^{1/2} }\bigr )\ ,
  \label{z:ZdgA6uUx}
\end{align}
where we have used the shorthand notation
$N_{BR} \equiv \mathcal{N}_{A\to B}(\Phi_{A:R})$ and $M_{BR}\equiv\mathcal{M}_{A\to B}(\Phi_{A:R})$, and we recall that $\widetilde{\mathcal{D}}_{A\to B}(\cdot) = \operatorname{tr}(\cdot)\,\mathds{1}_B$, with the
property that
$g_{\widetilde{\mathcal{D}}}(\mathcal{N}_{A\to B}, \phi_A) =
-{S}_{\phi}^{}({\mathcal{N}_{A\to B}})$. The relative entropy term only depends on the reduced state $\phi_R$, rather
than $\phi_{AR}$, since $g_{\mathcal{M}}$ remains invariant if we rotate
$\phi_{AR}$ by a local unitary on $R$. 

The function $g_{\mathcal{M}}$ is studied
in~\cite[Prop.~7.83]{R79}.  This function displays
the following useful convexity properties:
\begin{itemize}
\item $g_{\mathcal{M}}$ is jointly convex in $\mathcal{N}, \mathcal{M}$;
\item $g_{\mathcal{M}}$ is concave in $\phi_A$.
\end{itemize}
Standard minimax theorems therefore guarantee that the min and the max can be
interchanged in~\eqref{z:jZzkURqZ} (cf.\@ e.g.\@
\cite[Ex.~5.25]{R95}).
Following~\cite{R79}, we find:
\begin{align}
  \eqref{z:sb5FfEOw}
  &= - \max_{\phi_{A}} \min_{\substack{\mathcal{N}:\,\mathrm{cp., t.p.}\\
      \operatorname{tr}[{ C^j_{BR}\, N_{BR} }] = q_j}}
  g_{\widetilde{D}}\bigl({\mathcal{N}_{A\to B}, \phi_A}\bigr)
  \nonumber\\[1ex]
  &= \min_{\phi_A} \max_{\substack{\mathcal{N}:\,\mathrm{cp., t.p.}\\
      \operatorname{tr}[{ C^j_{BR}\, N_{BR} }] = q_j}} {S}_{\phi}^{}({\mathcal{N}_{A\to B}})
  \ .
  \label{z:02ppihHe}
\end{align}
We may therefore focus on the maximum-channel-entropy problem at fixed input
pure state $\phi_{AR}$:
\begin{align}
  \label{z:.wBvf2qe}
  \begin{aligned}[t]
    \textup{maximize:} \quad
    & {S}_{\phi}^{}({\mathcal{N}_{A\to B}})
    \\
    \textup{over:}\quad
    & \mathcal{N}_{A\to B}\ \textup{c.p., t.p.}
    \\
    \textup{such that:}\quad
    & \operatorname{tr}\bigl[{C^j_{BR}\,\mathcal{N}_{A\to B}(\Phi_{A:R})}\bigr] = q_j\quad\text{for \(j=1, \ldots, J\)}\ .
  \end{aligned}
\end{align}

\begin{definition}[Thermal channel with fixed input state]
The optimal quantum channel in~\eqref{z:.wBvf2qe} is
called the \textbf{thermal channel with respect to $\lvert {\phi}\rangle _{AR}$} and is
denoted by $\mathcal{T}_{A\to B}^{(\phi)}$. 
\end{definition}

The thermal channel $\mathcal{T}_{A\to B}^{}$ is then
the thermal channel with respect to the state $\phi_A$ for which
${S}_{\phi}^{}({ \mathcal{T}_{A\to B}^{(\phi)} })$ is maximal.

One of the main contributions of this paper is to give a general form of
thermal channels with respect to any fixed state $\phi_R$ (see in particular
\cref{z:33B55hFY} below).  By finally optimizing
over the input state $\phi_R$, we will obtain a characterization of a thermal
channel, proving \cref{z:TLiwdJGR}.

\subsection{Maximum channel entropy with fixed, full-rank input}

We now focus on solving the optimization
problem~\eqref{z:.wBvf2qe}.  As it turns out, the
problem becomes significantly simpler if the input state
$\lvert {\phi}\rangle _{AR} = \phi_A^{1/2}\lvert {\Phi_{A:R}}\rangle $ has a reduced state $\phi_A$ that
has full rank.  We solve this case first.

\begin{proposition}[Structure of the thermal channel with respect to full-rank $\phi_{A}$]
  \label{z:hd7.cLe.}
  Let $\phi_A$ be any full-rank quantum state and let
  $\lvert {\phi}\rangle _{AR} = \phi_A^{1/2}\lvert {\Phi_{A:R}}\rangle  = \phi_R^{1/2}\lvert {\Phi_{A:R}}\rangle $.
  There exists a Hermitian operator $F_R$ and real values $\mu_j$ such that the
  quantum superoperator $\mathcal{T}_{A\to B}^{(\phi)}$, defined through its Choi matrix
  as
  \begin{align}
    \mathcal{T}_{A\to B}^{(\phi)}(\Phi_{A:R})
    = \phi_R^{-1/2}\,\exp\Bigl\{{ - \phi_R^{-1/2}\Bigl[{
    \sum \mu_j C^j_{BR} - \mathds{1}_B\otimes \bigl({F_R + \phi_R\log\phi_R}\bigr) 
    }\Bigr] \phi_R^{-1/2} }\Bigr\} \, \phi_R^{-1/2}\ ,
    \label{z:8KUY6eJu}
  \end{align}
  is a quantum channel and is the unique optimal solution
  in~\eqref{z:.wBvf2qe}.  Furthermore, the operator
  $ \sum_j \mu_j C^j_{BR} -\mathds{1}_B\otimes \bigl({F_R + \phi_R\log\phi_R}\bigr)$ is
  positive definite, and the channel entropy with respect to $\lvert {\phi}\rangle _{AR}$
  attained by $\mathcal{T}_{A\to B}^{(\phi)}$ is
  \begin{align}
    {S}_{\phi}^{}\Bigl ({\mathcal{T}_{A\to B}^{(\phi)}}\Bigr )
    = \sum_{j=1}^J \mu_j q_j - \operatorname{tr}({F_R}) \ .
  \end{align}
\end{proposition}

The structure~\eqref{z:8KUY6eJu} can be
viewed as the generalization to thermal channels of the generic structure
$\gamma_S = {e}^{F - \sum \mu_j Q_j}$ of the thermal
state~\eqref{z:chtPfev.}.  The real values $\mu_j$ mirror the inverse
temperature and chemical potentials, while the operator $F_R$ can be viewed as a
channel equivalent of the free energy.  The term $\log\phi_R$ in the exponent
reflects the channel nature of the optimization problem.

The parameters $\{{ \mu_j }\}$ and $F_R$ must be jointly chosen such that all
original constraints are simultaneously satisfied.  The constraints include
the expectation value constraints for each $C^j_{BR}$ with $j=1, \ldots, J$,
as well as the trace-preserving constraint.
Each $\mu_j$ appears as the Lagrange dual
variable (or Lagrange multiplier) for each expectation value constraint
with $j=1, \ldots, J$.
The parameter $F_R$ appears as the Lagrange dual variable of the trace-preserving
constraint and can be interpreted as ensuring
that $\mathcal{T}_{}^{(\phi)}$ is trace preserving.
While in the case of quantum states, the partition function (or the free energy)
can be computed after fixing any temperature and/or generalized chemical potentials
simply by normalizing the state to unit trace,
it does not appear that a similarly simple
method of determining $F_R$ can be employed here.

The term $\phi_R\log(\phi_R)$ in the exponential can be loosely understood as
compensating for the $\phi_R^{-1/2}$ terms that sandwich the exponential.
Suppose indeed that $\phi_R$ commutes with
$\sum \mu_j C^j_{BR} - \mathds{1}_B\otimes F_R$.  The $\phi_R\log \phi_R$ term,
which then commutes with the remaining term in the exponential, can be taken out
of the exponential to cancel out the $\phi_R^{-1/2}$ sandwiching factors,
leaving simply
$\mathcal{T}_{A\to B}^{(\phi)}(\Phi_{A:R}) = \exp\bigl\{{ -\phi_R^{-1/2}\bigl[{
  \sum \mu_j C^j_{BR} - \mathds{1}\otimes F_R }\bigr] \phi_R^{-1/2} }\bigr\}$.  It is unclear to
us how often this situation can be expected to occur.

While the thermal channel $\mathcal{T}_{A\to B}^{(\phi)}$
in~\eqref{z:8KUY6eJu} is the unique optimal
solution to~\eqref{z:.wBvf2qe} for a fixed, full-rank
$\phi_A$, it might still happen that $F_R$ and $\mu_j$ are not uniquely
specified; this situation might arise if the real-valued constraints imposed by
the conditions $\mathcal{N}^\dagger(\mathds{1})=\mathds{1}$ and
$\operatorname{tr}[{C^j_{BR}\mathcal{N}_{A\to B}(\Phi_{A:R})}]=q_j$ are not independent.

\paragraph*{Maximally mixed input state.}
We now briefly consider the case where $\phi_R = \mathds{1}_R/d_R$ is the maximally
mixed state, meaning that $\lvert {\phi}\rangle _{AR} = \lvert {\Phi_{A:R}}\rangle /\sqrt{d_R}$ is the
maximally entangled state in the canonical basis.  In this case, the operator
$\nu_{BR} \equiv \mathcal{N}(\phi_{AR}) = N_{BR}/d_R$ is the normalized Choi
state of $\mathcal{N}$.  The objective in
problem~\eqref{z:.wBvf2qe} is equivalently written as
${S}_{}^{}({\mathcal{N}(\phi_{AR})}) - {S}_{}^{}({\phi_R})$.  Since ${S}_{}^{}({\phi_R})$ is fixed,
the problem~\eqref{z:.wBvf2qe} is equivalent to
maximizing the entropy of $\rho_{BR}$ over all quantum states $\rho_{BR}$
subject to the linear constraints $\operatorname{tr}[{({d_R C^j_{BR}}) \rho_{BR}}] = q_j$ and
$d_R \operatorname{tr}_B({\rho_{BR}}) = \mathds{1}_R$.  The latter constraint can be projected
along an orthonormal basis $\{{ P_R^k }\}$ of traceless Hermitian operators on $R$
and thereby rewritten into a finite set of scalar constraints
$\operatorname{tr}[{\rho_{BR} \, d_R ({\mathds{1}_B\otimes P_R^k})}] = 1$.  This is a standard
quantum state maximum entropy problem, for which the solution is
$\rho_{BR} = {e}^{-\sum d_R\mu_j C^j_{BR} - \sum a_k d_R ({\mathds{1}_B\otimes
  P_R^k})}/Z$, where the $a_k$ are the ``generalized chemical potentials''
associated with the $P_R^k$ constraints.  This is indeed the optimal form
provided in~\cref{z:hd7.cLe.}, with
$F=\log({Z})\,\mathds{1}_R - \sum a_k P_R^k$.

Therefore, if $\phi_R=\mathds{1}_R/d_R$, the thermal quantum channel
$\mathcal{T}_{}^{(\mathds{1}_R/d_R)}$ has a Choi state that coincides exactly with the quantum
Choi state $\rho_{BR}$ that has maximal entropy subject to the constraints
$C^j_{BR}$.

\begin{proof}[*z:hd7.cLe.]
  That the thermal channel with respect to $\lvert {\phi}\rangle _{AR}$ exists follows from
  the fact that we assumed the problem~\eqref{z:sb5FfEOw} [and
  hence~\eqref{z:.wBvf2qe}] to be strictly feasible.
  Next, we claim that any optimal $\mathcal{N}_{A\to B}$ is such that
  $\mathcal{N}_{A\to B}(\phi_{AR})$ has full rank.  Intuitively, this follows
  because the derivative of the objective function
  ${S}_{\phi}^{}({\mathcal{N}_{A\to B}})$ diverges as
  $\mathcal{N}_{A\to B}(\phi_{AR})$ approaches a non-full-rank state, and
  therefore the maximum cannot lie on that boundary.  A proof is presented as
  \cref{z:dTDQ-ZX2} in
  \cref{z:N3vMol3P}.
  In turn, this implies that any optimal $\mathcal{N}_{A\to B}$ must have a Choi
  matrix
  $N_{BR} = \phi_R^{-1/2} \, \mathcal{N}_{A\to B}({\phi_{AR}}) \, \phi_R^{-1/2}$
  that has full rank.

  Knowing that the optimum cannot lie on the boundary of the domain of the
  objective function (namely, $N_{BR}$ must be a positive semidefinite matrix),
  we may now use standard Lagrangian/convex optimization techniques to find the
  maximum-entropy channel with respect to
  $\phi_A$~\cite{R95}.  In the following, we consider
  $N_{BR}$ to be the optimization variable (which must be positive
  semidefinite), and use $\mathcal{N}_{A\to B}$ as a shorthand notation for
  $\operatorname{tr}_R[{ N_{BR} (\cdot)^{t_{A\to R}}}]$.  We minimize the objective function
  $-{S}_{\phi}^{}({\mathcal{N}})$ over the set of all positive definite matrices
  $N_{BR} > 0$, subject to the constraints
  $\mathds{1}- \mathcal{N}^\dagger(\mathds{1}) = 0$ and
  $q_j - \operatorname{tr}\bigl({C^j_{BR} N_{BR}}\bigr) = 0$.  Since the complete positivity and trace
  preserving properties are imposed by constraints, the objective function's domain
  formally extends to maps that do not have these properties.  Concretely, we use
  the expression ${S}_{}^{}({\mathcal{N}(\phi_{AR})}) - {S}_{}^{}({\phi_R})$ for the objective
  function, noting that the different expressions for ${S}_{\phi}^{}({\mathcal{N}})$ in
  \cref{z:4jPiiO8e,z:hkVuHdrl}
  are all equivalent only as long as the map $\mathcal{N}$ is completely positive
  and trace-preserving, and formally extending the function
  ${S}_{}^{}({X}) = -\operatorname{tr}(X\log(X))$ to any
  positive semidefinite operator $X$.
  We construct the following Lagrangian,
  introducing dual variables $\mu_j\in\mathbb{R}$ (for $j=1, \ldots, J$) and
  $Z_R = Z_R^\dagger$:
  \begin{align}
    \mathcal{L}[N_{BR}, Z_R, \mu_j]
    = %
    -{S}_{}^{}({\mathcal{N}(\phi_{AR})}) + {S}_{}^{}({\phi_R})
    - \sum_{j=1}^J \mu_j \bigl[{ q_j - \operatorname{tr}\bigl({ C^j_{BR} N_{BR} }\bigr) }\bigr]
    + \operatorname{tr}\bigl({ Z_R \bigl[{\mathds{1}_R - \operatorname{tr}_B({N_{BR}})}\bigr]}\bigr)\ .
  \end{align}
  We now consider a variation $N_{BR} \to N_{BR} + \delta N_{BR}$.  That is, $\delta N_{BR}$
  is any infinitesimally small perturbation of $N_{BR}$ within the space of Hermitian
  operators.
  The calculus of
  variations used here can be thought of as a way of computing the derivative of
  $\mathcal{L}$ with respect to the primal %
  variables $N_{BR}$.
  Using $\delta \operatorname{tr}[{ f(X) }] = \operatorname{tr}[{ f'(X)\,\delta X }]$ for any scalar function
  $f$, we can first compute the variation of the objective function value:
  \begin{align}
    \delta\bigl[{- {S}_{}^{}\bigl ({\mathcal{N}_{A\to B}(\phi_{AR})}\bigr ) }\bigr]
    &=
    \delta\bigl[{- {S}_{}^{}\bigl ({\phi_R^{1/2} N_{BR} \phi_R^{1/2}}\bigr ) }\bigr]
    =
    \delta \operatorname{tr}\Bigl\{{
        \phi_R^{1/2} N_{BR} \phi_R^{1/2}\,\log\bigl[{ \phi_R^{1/2} N_{BR} \phi_R^{1/2} }\bigr]
    }\Bigr\}
      \nonumber\\
    &= \operatorname{tr}\Bigl\{{
      \Bigl[{
      \log\bigl({\phi_R^{1/2} N_{BR} \phi_R^{1/2}}\bigr)
      + \mathds{1}_{BR}
      }\Bigr]\,\delta\Bigl({ \phi_R^{1/2} N_{BR} \phi_R^{1/2} }\Bigr)
      }\Bigr\}
      \nonumber\\
    &= \operatorname{tr}\Bigl\{{
      \phi_R^{1/2}\Bigl[{
      \log\bigl({\phi_R^{1/2} N_{BR} \phi_R^{1/2}}\bigr)
      + \mathds{1}_{BR}
      }\Bigr]\, \phi_R^{1/2} \, \delta N_{BR}
      }\Bigr\}\ .
  \end{align}
  Therefore,
  \begin{align}
    \delta \mathcal{L}
    &=
      \operatorname{tr}\Bigl\{{
        \Bigl[{
        \phi_R^{1/2}\,\log\bigl({\phi_R^{1/2} N_{BR} \phi_R^{1/2}}\bigr)\,\phi_R^{1/2}
        + \mathds{1}_B\otimes\phi_R
        + \sum \mu_j C^j_{BR}
        - \mathds{1}_B\otimes Z_R
        }\Bigr]\,
        \delta N_{BR}
        }\Bigr\}
      \ .
  \end{align}
  Requiring the variation $\delta \mathcal{L}$ of the Lagrangian to vanish for
  all $\delta N$, %
  we find the condition that any optimal primal and dual variables must satisfy
  (in addition to the original problem constraints):
  \begin{gather}
    \phi_R^{1/2}\,\log\bigl({\phi_R^{1/2} N_{BR} \phi_R^{1/2}}\bigr)\,\phi_R^{1/2}
    + \mathds{1}_B\otimes\phi_R
    + \sum_{j=1}^J \mu_j C^j_{BR}
    - \mathds{1}_B\otimes Z_R
    = 0\ ;
    \label{z:RMJR7P5-}
  \end{gather}
  The condition~\eqref{z:RMJR7P5-}, on the other hand,
  enables us to derive the general form of the thermal channel with respect to
  $\phi_R$.  Given that $\phi_R$ is invertible, and defining
  $F_R = Z_R - \phi_R - \phi_R\log(\phi_R)$,
  \cref{z:RMJR7P5-} can be rearranged to
  \begin{align}
    \mathcal{N}(\phi_{AR})
    &= \phi_R^{1/2} N_{BR} \phi_R^{1/2}
      = \exp\Bigl\{{ - \phi_R^{-1/2} G_{BR} \phi_R^{-1/2} }\Bigr\}\ ;
    \nonumber\\
      G_{BR}
    &= - \mathds{1}_B\otimes (F_{R} + \phi_R\log\phi_R)
      + \sum_{j=1}^J \mu_j C^j_{BR}\ ,
  \end{align}
  Applying $\phi_R^{-1/2}\,({\cdot})\,\phi_R^{-1/2}$ yields the claimed
  form~\eqref{z:8KUY6eJu}.
  Since $\mathcal{N}(\phi_{AR})$ is a full-rank quantum state (full-rank since
  $N_{BR} > 0$), it obeys $0 < \mathcal{N}(\phi_{AR}) < \mathds{1}$.  Consequently, $\phi_R^{-1/2}\,G_{BR}\,\phi_R^{-1/2} > 0$.  Since $\phi_R$ has
  full rank, this in turn implies $G_{BR} > 0$, as claimed in the
  \lcnamecref{z:hd7.cLe.} statement.
  The channel entropy with respect to $\lvert {\phi}\rangle _{AR}$ attained by
  $\mathcal{T}_{A\to B}^{(\phi)}$ is
  \begin{align}
    {S}_{\phi}^{}\bigl ({ \mathcal{T}_{A\to B}^{(\phi)} }\bigr )
    &= {S}_{}^{}\bigl ({ \mathcal{T}_{A\to B}^{(\phi)}(\phi_{AR}) }\bigr ) - {S}_{}^{}({\phi_R})
      = \operatorname{tr}\Bigl\{{
      \mathcal{T}_{A\to B}^{(\phi)}(\phi_{AR})
      \Bigl[{
      \phi_R^{-1/2}\,G_{BR}\,\phi_R^{-1/2}
      }\Bigr]
      }\Bigr\} - {S}_{}^{}({\phi_R})
      \nonumber\\
    &= \operatorname{tr}\Bigl\{{
      \mathcal{T}_{A\to B}^{(\phi)}(\Phi_{A:R})
      \Bigl[{
      -\mathds{1}_B\otimes ({F_R + \phi_R\log\phi_R})
      + \sum \mu_j C^j_{BR}
      }\Bigr]
      }\Bigr\} - {S}_{}^{}({\phi_R})
      \nonumber\\
    &= \sum_{j=1}^J \mu_j q_j  -  \operatorname{tr}({F_R})\ ,
  \end{align}
  recalling that $\mathcal{T}_{A\to B}^{(\phi)}(\Phi_{A:R})$ is trace-preserving and that
  $q_j = \operatorname{tr}\!\big[ C^j_{BR} \mathcal{T}_{A\to B}^{(\phi)}(\Phi_{A:R})\big]$.

  The constraint $\mathcal{N}^\dagger(\mathds{1}) = \mathds{1}$ imposes $d_R^2$
  independent real constraints on the variables in $\mathcal{N}$ (this value is
  the real dimension of a complex Hermitian $d_R \times d_R$ matrix).  Each
  further constraint $\operatorname{tr}[{ C^j_{BR} N_{BR} }] = q_j$ imposes one further real
  constraint, as long as each additional constraint is linearly independent
  from the previous ones.  Hence, as long as all these constraints are linearly
  independent, we have $d_R^2 + J$ real constraints on $\mathcal{N}$.  On the
  other hand, there are exactly $d_R^2 + J$ real degrees of freedom in the
  general form of $\mathcal{T}_{A\to B}^{(\phi)}$, namely $d_R^2$ for $F_B$ (through
  $Z_B$) and $J$ through $\mu_1, \ldots, \mu_J$.  In this case, the constraints
  determine these variables uniquely, meaning the solution $\mathcal{T}_{A\to B}^{(\phi)}$
  is unique.  If the constraints are not linearly independent, we may simplify
  the additional $J$ constraints into fewer constraints to arrange that they are
  all linearly independent, without changing neither the feasible set nor the
  objective function of the optimization problem.  In this simplified form it is
  clear that the solution $\mathcal{T}_{A\to B}^{(\phi)}$ is unique, even if in its
  original form it is possible that several choices of $F_R$, $\mu_j$ lead to
  the same channel $\mathcal{T}_{A\to B}^{(\phi)}$.
\end{proof}

\subsection{Maximum channel entropy with arbitrary fixed input}

We now lift our assumption that the reduced input state $\phi_A$ has full rank
and find the general structure of thermal channels for such general states.

\begin{theorem}[Structure of a thermal channel with respect to general input $\phi_A$]
  \noproofref
  \label{z:33B55hFY}
  Let $\lvert {\phi}\rangle _{AR} = \phi_A^{1/2}\lvert {\Phi_{A:R}}\rangle $, where $\phi_A$ is an
  arbitrary quantum state.  Any quantum channel $\mathcal{T}_{A\to B}^{(\phi)}$ is an
  optimal solution to~\eqref{z:.wBvf2qe} if and only if
  it satisfies all the problem constraints and it is of the form
  \begin{subequations}
    \label{z:TUY44lvP}
    \begin{gather}
      \mathcal{T}_{A\to B}^{(\phi)}(\Phi_{A:R})
      = 
      \phi_R^{-1/2}\,\exp\Bigl\{{ -\phi_R^{-1/2} G_{BR} \phi_R^{-1/2} }\Bigr\} \, \phi_R^{-1/2} + Y_{BR}\ ;
      \\
      G_{BR} = \sum \mu_j C^j_{BR} - \mathds{1}_B\otimes \bigl[{F_R + \phi_R\log(\phi_R)}\bigr] - S_{BR}\ ,
    \end{gather}
  \end{subequations}
  where $F_R$ is a Hermitian matrix, $\mu_j\in\mathbb{R}$ (for
  $j=1, \ldots, J$), $S_{BR}$ is a positive semidefinite operator satisfying
  $S_{BR}\, \mathcal{T}_{A\to B}^{(\phi)}(\Phi_{A:R}) = 0$, $G_{BR}$ satisfies
  $\Pi^{\phi_R \perp}_R G_{BR} = 0$, and $Y_{BR}$ is a Hermitian operator such
  that $\Pi^{\phi_R}_R Y_{BR} \Pi^{\phi_R}_R = 0$.
  Furthermore, for any such $\mathcal{T}_{A\to B}^{(\phi)}$, we have that $G_{BR}$ is
  positive semidefinite and $\operatorname{tr}_B({ Y_{BR} }) = \Pi_R^{\phi_R \perp}$.  The
  attained value for the channel entropy with respect to $\lvert {\phi}\rangle _{AR}$ is
  \begin{align}
    {S}_{\phi}^{}\bigl ({\mathcal{T}_{A\to B}^{(\phi)}}\bigr )
    = - \operatorname{tr}({F_R}) + \sum \mu_j q_j\ .
  \end{align}
\end{theorem}

This theorem is a special case of a more general theorem that we prove below
(\cref{z:gV43EWT.} in
\cref{z:aXvqINGx}).
The Lagrange dual version of this problem is derived as part of the more
general optimization problem studied in \cref{z:aXvqINGx};
see specifically \cref{z:t5t-uuUF}.

We now state some stability results for the thermal quantum channel and the
achieved channel entropy with respect to $\phi_R$.  These claims can be viewed
as a consequence of friendly continuity properties
of ${S}_{\phi}^{}({\mathcal{N}})$ as a function both of $\phi$ and of $\mathcal{N}$.
We define for convenience
the maximal channel entropy compatible with the given constraints, viewed as a
function of $\phi_R$:
\begin{align}
  \tilde{s}(\phi_R)
  &\equiv \max_{\substack{
    \mathcal{N}\ \textup{cp. tp.}\\
    \operatorname{tr}[{C^j_{BR} \mathcal{N}(\phi_{AR})}] = q_j
  }} {S}_{\phi}^{}({\mathcal{N}})
\end{align}

\begin{proposition}[Stability of the thermal quantum channel in $\phi_R$]
  \label{z:WYT9Hica}
  The function $\tilde{s}(\phi_R)$ is continuous in $\phi_R$ over all $\phi_R$.
  The thermal quantum channel $\mathcal{T}_{}^{(\phi)}$ is unique and a continuous
  function of $\phi_R$ for all full-rank $\phi_R>0$.
\end{proposition}
\begin{proof}[**z:WYT9Hica]
  The claim follows as a direct consequence of Berge's maximum
  theorem~\cite{R96,R97}.
\end{proof}

The following statement is equally intuitive and also follows from Berge's maximum
theorem; we provide a self-contained proof for completeness.
\begin{proposition}[Stability of the thermal quantum channel for general $\phi_R$]
  \label{z:4e07B.62}
  Let $\{{ \phi_R^z }\}_{z>0}$ be any family of states converging to some
  $\phi_R \equiv \lim_{z\to 0} \phi_R^z$.
  Let $\mathcal{T}_{}^{(\phi^z)}$ be optimizers in $\tilde{s}(\phi_R^z)$,
  and suppose that they converge
  towards some channel $\mathcal{T}_{}^{} := \lim_{z\to 0} \mathcal{T}_{}^{(\phi^z)}$.
  Then $\mathcal{T}_{}^{}$ is optimal in $\tilde{s}(\phi_R)$.
\end{proposition}
\begin{proof}[**z:4e07B.62]
  Let $\mathcal{T}_{}^{(\phi)}$ be a maximizer for ${S}_{\phi}^{}({\mathcal{N}})$.
  By continuity of ${S}_{\phi}^{}({\mathcal{N}})$ in $\phi$, there exists $\xi(z)$
  with $\lim_{z\to 0}\xi(z) = 0$ such that
  \begin{align}
    \bigl \lvert { {S}_{\phi^z}^{}({ \mathcal{T}_{}^{(\phi)} }) - {S}_{\phi}^{}({ \mathcal{T}_{}^{(\phi)} }) }\bigr \rvert 
    \leq \xi(z)\ .
  \end{align}
  Recalling that $\mathcal{T}_{}^{(\phi^z)}$ and $\mathcal{T}_{}^{(\phi)}$ maximize respectively
  ${S}_{\phi^z}^{}({\mathcal{N}})$ and ${S}_{\phi}^{}({\mathcal{N}})$,
  \begin{align}
    {S}_{\phi^z}^{}({ \mathcal{T}_{}^{(\phi)} })
    &\leq
    {S}_{\phi^z}^{}({ \mathcal{T}_{}^{(\phi^z)} })\ ;
      &
    {S}_{\phi}^{}({ \mathcal{T}_{}^{} })
    &\leq
    {S}_{\phi}^{}({ \mathcal{T}_{}^{(\phi)} })\ .
  \end{align}
  Then
  ${S}_{\phi}^{}({ \mathcal{T}_{}^{(\phi)} }) %
  \leq {S}_{\phi^z}^{}({ \mathcal{T}_{}^{(\phi)} }) + \xi(z) %
  \leq {S}_{\phi^z}^{}({ \mathcal{T}_{}^{(\phi^z)} }) + \xi(z)$, which implies
  ${S}_{\phi}^{}({ \mathcal{T}_{}^{(\phi)} }) \leq {S}_{\phi}^{}({ \mathcal{T}_{}^{} })$ in the limit
  $z\to 0$.
  Therefore, ${S}_{\phi}^{}({ \mathcal{T}_{}^{} }) = {S}_{\phi}^{}({ \mathcal{T}_{}^{(\phi)} })$ and
  $\mathcal{T}_{}^{}$ also maximizes ${S}_{\phi}^{}({ \mathcal{N} })$.
\end{proof}

For a rank-deficient state $\phi_R$, the associated thermal quantum channel
is generally not unique. Indeed, the channel entropy with respect to $\phi_R$
becomes insensitive to the channel's action outside the support of $\phi_R$.
Therefore, it might be natural to demand of a thermal channel with respect
to a rank-deficient state $\phi_R$ to be achievable as a limit of thermal
channels with respect to full-rank states that converge to $\phi_R$.  In the
examples we study below (cf.\@ \cref{z:VuBBeOAb}), the example in which energy
is preserved on average provides an illustration of a situation where such a
condition would be relevant.

We also prove the following more specific stability result. We show that if we
consider a family of commuting full-rank states $\{{ \phi_R^{z} }\}$ for $z>0$, and
if $\phi_R^z \to \phi_R$ as $z\to 0$, then under suitable conditions, the
thermal quantum channel with respect to $\phi_R^z$ converges to the thermal
quantum channel with respect to $\phi_R$.  The interest of this stability result
is to yield explicit expressions of the parameters $\mu$, $F_R$, $S_{BR}$ and
$Y_{BR}$ of the limiting channel.
This property is useful to derive thermal quantum channels with respect to a
rank-deficient state $\phi_R$.  Under suitable conditions, it suffices to
consider the thermal quantum channels for full-rank states $\phi_R^z >0$ (cf.\@
\cref{z:hd7.cLe.}) and to consider the limit
$\phi_R^z\to \phi_R$.

\begin{proposition}[Stability of the thermal quantum channel for limits of
  commuting input states]
  \label{z:NjFfeh1f}
  Let $\{{ \phi_R^z }\}_{z>0}$ be a family of pairwise commuting full-rank states
  converging to some $\phi_R \equiv \lim_{z\to 0} \phi_R^z$.  For each $z>0$,
  let $\mu_j^z$ and $F_R^z$ be the parameters of the thermal quantum channel
  $\mathcal{T}_{}^{(\phi_R^z)}$ given by \cref{z:hd7.cLe.}.
  We suppose that for all $z>0$,
  \begin{align}
    \Bigl[{ \sum \mu_j^z C^j_{BR} - \mathds{1}_B\otimes F_R^z \,,\; \mathds{1}_B\otimes\phi_R^z }\Bigr]
    = 0\ .
    \label{z:swvf-7vB}
  \end{align}
  Furthermore, assume the following limits exist:
  \begin{align}
    \mu_j &\equiv \lim_{z\to 0} \mu_j^z\ ;
    &
      F_R &\equiv \lim_{z\to 0} F_R^z\ ;
      &
      \mathcal{T}_{}^{} &\equiv \lim_{z\to 0} \mathcal{T}_{}^{(\phi_R^z)}\ ,
  \end{align}
  and assume that $T_{BR} = \mathcal{T}_{}^{}(\Phi_{A:R})$ has full rank.
  Then $\mathcal{T}_{}^{}$ is a quantum thermal channel with respect to $\phi_R$.  Its
  parameters from \cref{z:33B55hFY} are $\mu_j$,
  $F_R$, $S_{BR} = 0$, and
  $Y_{BR} = \Pi^{\phi_R\perp} T_{BR} \Pi^{\phi_R \perp}$.
\end{proposition}

\begin{proof}[**z:NjFfeh1f]
  From \cref{z:hd7.cLe.} and
  using~\eqref{z:swvf-7vB},
  the thermal quantum channel for $z>0$ is given by
  \begin{align}
    T_{BR}^z \equiv \mathcal{T}_{}^{(\phi_R^z)}(\Phi_{A:R})
    = \exp\Bigl\{{ -(\phi_R^z)^{-1/2}\,\Bigl({\sum \mu_j^z C^j_{BR} - \mathds{1}_B\otimes F_R^z}\Bigr)
    (\phi_R^z)^{-1/2} }\Bigr\}\ .
    \label{z:AeB8MVHx}
  \end{align}
  Because the limit channel $T_{BR} = \lim_{z\to 0} T_{BR}^z$ has full rank, the
  operator inside the exponential also converges to some finite operator
  \begin{align}
    -K_{BR} \equiv \log({T_{BR}})
    = \lim_{z\to 0}\Bigl\{{
    -(\phi_R^z)^{-1/2}\,\Bigl({\sum \mu_j^z C^j_{BR} - \mathds{1}_B\otimes F_R^z}\Bigr)
    (\phi_R^z)^{-1/2} 
    }\Bigr\}\ .
  \end{align}
  From~\eqref{z:AeB8MVHx} we also find that
  \begin{align}
    [{ T_{BR}^z, \phi_R^z }]
    = \Bigl[{
    \exp\Bigl\{{ -(\phi_R^z)^{-1/2}\,\Bigl({\sum \mu_j^z C^j_{BR} - \mathds{1}_B\otimes F_R^z}\Bigr)
    (\phi_R^z)^{-1/2} }\Bigr\}
    \,,\; \phi_R^z }\Bigr]
    = 0\ ,
  \end{align}
  noting that the terms in the exponential commute with $\phi_R^z$ thanks
  to~\eqref{z:swvf-7vB}.
  In the limit $z\to 0$ we find
  $[ T_{BR}, \phi_R ] = 0$ and therefore
  \begin{align}
    [{ K_{BR}, \phi_R }] = 0\ .
    \label{z:wgVkdyoK}
  \end{align}

  Henceforth we write as a shorthand $\Pi \equiv \Pi^{\phi_R}$.  We find
  \begin{align}
    \Pi K_{BR} = \Pi K_{BR} \Pi
    = \lim_{z\to 0}\Bigl\{{
    \Pi \, ({\phi_R^z})^{-1/2}
    \Bigl({ \sum \mu_j^z C^j_{BR} - \mathds{1}_B\otimes F_R^z }\Bigr)
    ({\phi_R^z})^{-1/2} \, \Pi
    }\Bigr\}\ .
  \end{align}
  Because $\{{ \phi_R^z }\}$ are pairwise commuting, they also commute with
  $\phi_R$ and because of the convergence of the individual eigenvalues in the
  common eigenbasis we find
  $\lim_{z\to 0} \Pi ({\phi_R^z})^{-1/2} = \phi_R^{-1/2}$.  Therefore all terms
  in the expression above converge individually and we find
  \begin{align}
    \Pi K_{BR} 
    = \phi_R^{-1/2} \Bigl({\sum \mu_j C^j_{BR}
    - \mathds{1}_B\otimes F_R}\Bigr) \phi_R^{-1/2}\ .
    \label{z:4l8Ju9fA}
  \end{align}
  On the other hand, taking the limit $z\to 0$
  of~\eqref{z:swvf-7vB} we
  find:
  \begin{align}
    \Bigl[{
    \sum \mu_j C^j_{BR} - \mathds{1}_B\otimes F_R
    \,,\; \phi_R
    }\Bigr] = 0\ .
    \label{z:gRWSSui6}
  \end{align}
  Now observe that
  \begin{align}
    \Bigl({\sum \mu_j^z C^j_{BR} - \mathds{1}\otimes F_R^z }\Bigr)
    &= 
    ({\phi_R^z})^{1/2}  K_{BR}^z ({\phi_R^z})^{1/2}\ ;
    &
      K_{BR}^z
      &\equiv
        ({\phi_R^z})^{-1/2} 
        \Bigl({\sum \mu_j^z C^j_{BR} - \mathds{1}\otimes F_R^z}\Bigr)
        ({\phi_R^z})^{-1/2} \ ,
  \end{align}
  with $K_{BR}^z \to K_{BR}$.  Then
  \begin{align}
    \Pi^\perp \Bigl({ \sum \mu_j C^j_{BR} - \mathds{1}_B\otimes F_R }\Bigr)
    = \Pi^\perp \lim_{z\to 0} \Bigl[{ ({\phi_R^z})^{1/2}  K_{BR}^z ({\phi_R^z})^{1/2} }\Bigr]
    =  \Pi^\perp \phi_R^{1/2} K_{BR} \phi_R^{1/2} = 0\ .
    \label{z:hMYbUZbV}
  \end{align}
  Now let
  $G_{BR} \equiv \phi_R^{1/2} K_{BR} \phi_R^{1/2} + \phi_R\log({\phi_R})$.  Using
  \cref{z:4l8Ju9fA,z:gRWSSui6,z:hMYbUZbV},
  \begin{align}
    G_{BR}
    &= \sum \mu_j C^j_{BR}  - \mathds{1}_B\otimes [{F_R + \phi_R\log({\phi_R})}] \ ,
  \end{align}
  with $\Pi^\perp G_{BR} = 0$.  Here, we set $S_{BR} = 0$.  By construction,
  $\phi_R^{-1/2} G_{BR} \phi_R^{-1/2} = \log(\phi_R) + \Pi K_{BR}$.
  Now let
  $Y_{BR} = \Pi^\perp T_{BR} = \Pi^\perp T_{BR} \Pi^\perp = \Pi^\perp
  {e}^{-\Pi^\perp K_{BR}}$.  Then
  \begin{align}
    T_{BR}
    &= {e}^{-K_{BR}}
    = \Pi {e}^{-\Pi K_{BR} \Pi} + \Pi^\perp {e}^{-\Pi^\perp K_{BR} \Pi^\perp}
    = \phi_R^{-1/2} {e}^{ -\phi_R^{-1/2} G_{BR} \phi_R^{-1/2} } \phi_R^{-1/2} + Y_{BR}\ .
  \end{align}
  We have found $\mu_j$, $F_R$, and $S_{BR}\geq 0$ that satisfy the requirements
  of \cref{z:33B55hFY}.  Furthermore
  \begin{align}
    \operatorname{tr}_B({T_{BR}})
    &= \lim_{z\to 0} \operatorname{tr}_B({T_{BR}^z}) = \mathds{1}_R \ ;
    &
      \operatorname{tr}\bigl({C^j_{BR} T_{BR}}\bigr)
      &= \lim_{z \to 0} \operatorname{tr}\bigl({ C^j_{BR} T_{BR}^z}\bigr) = q_j\ ,
  \end{align}
  so the channel $T_{BR}$ furthermore satisfies all the problem constraints in
  \cref{z:33B55hFY}.  Therefore, $T_{BR}$ is a
  thermal quantum channel with respect to $\phi_R$.
\end{proof}

We anticipate that several assumptions in
\cref{z:NjFfeh1f} might not be necessary to
achieve a similar conclusion.  In particular, the
assumptions~\eqref{z:swvf-7vB},
while convenient and necessary for our proof above, are particularly stringent;
there appears no fundamental reason why they could not, in principle, be
relaxed.

\subsection{Useful lemmas for the thermal channel and the optimal input state}

Here, we prove a handful of lemmas that provide guidance on the optimal channel
and which states $\phi_R$ to consider to achieve the optimal thermal channel.

First, we provide a necessary condition for states $\phi_A$ that is optimal in
the thermodynamic capacity.  Recall that
$\lvert {\phi}\rangle _{AR} = \phi_A^{1/2}\lvert {\Phi_{A:R}}\rangle $ is optimal for the channel
entropy of $\mathcal{N}$ if and only if $\phi_A$ is optimal for the
thermodynamic capacity of a complementary channel $\widehat{N}$.  Hence, this
lemma can be used to characterize optimal states for the channel entropy.

\begin{lemma}
  \label{z:qi4DZIqg}
  Let $\mathcal{N}'_{A\to E}$ be a quantum channel.  A quantum state $\phi_A$
  can only be optimal in the definition of the thermodynamic
  capacity~\eqref{z:NdbPdJS1} if there exists
  $\lambda\in\mathbb{R}$ such that
  \begin{align}
    \log({\phi_A}) - \mathcal{N}'^\dagger\Bigl({ \log\bigl[{ \mathcal{N}'(\phi_A) }\bigr] }\Bigr)
    - \lambda \Pi_A^{\phi_A} = 0\ .
    \label{z:FMvxXU9D}
  \end{align}
  Furthermore, if $\phi_A$
  satisfies~\eqref{z:FMvxXU9D} and
  is full rank, then it is optimal in~\eqref{z:NdbPdJS1}.
\end{lemma}
\begin{proof}[**z:qi4DZIqg]
  We seek to minimize the convex function
  \begin{align}
    f(\phi_A)
    = {S}_{}^{}({\mathcal{N}'(\phi_A)}) - {S}_{}^{}({\phi_A}) = - {S}_{}^{}({E}\mathclose{}\,|\,\mathopen{}{B})_{V\phi_A V^\dagger}\ ,
  \end{align}
  where $V_{A\to BE}$ is a Stinespring dilation of $\mathcal{N}'$.  Writing the
  function as a conditional entropy makes it obvious that $f(\phi_A)$ is convex
  in $\phi_A$.  Fix any projector $\Pi_A$. We'll look for minima of $f(\phi_A)$
  over all quantum states $\phi_A$ that are Hermitian operators supported on
  $\Pi_A$ and which have full rank within $\Pi_A$, a condition we denote by
  $\phi_A \mathrel{{>}{|}_{\Pi_A}} 0$.  Introducing the Lagrange dual variable
  $\lambda$ for the condition $\operatorname{tr}(\phi_A) = 1$, we can write the Lagrangian
  \begin{align}
    \mathcal{L}_{T;\,\Pi_A}[\phi_A, \lambda]
    = f(\phi_R) + \lambda [{1 - \operatorname{tr}(\phi_A)}]\ .
  \end{align}
  The stationary points of $\mathcal{L}_{T;\,\Pi_A}$ are determined by
  requiring the variation $\delta \mathcal{L}_{T;\,\Pi_A}$ to vanish
  when $\phi_A \to \phi_A + \delta \phi_A$.  We calculate
  \begin{align}
    \delta \mathcal{L}_{T;\,\Pi_A}
    &= -\operatorname{tr}\Bigl[{ \bigl({ \log\bigl[{\mathcal{N}'(\phi_A)}\bigr] + \mathds{1}}\bigr) \,
        \mathcal{N}'(\delta\phi_A) }\Bigr]
    + \operatorname{tr}\Bigl[{ \bigl({ \log\bigl({\phi_A}\bigr) + \mathds{1}_A }\bigr) \,\delta\phi_A }\Bigr]
    - \lambda\operatorname{tr}({\delta\phi_A})
    \nonumber\\
    &=
      \operatorname{tr}\Bigl\{{\Bigl[{
      - \mathcal{N}'^\dagger\bigl[{ \log\bigl({\mathcal{N}'(\phi_A)}\bigr) }\bigr]
      - \mathds{1}_A
      + \log({\phi_A})
      + \mathds{1}_A
      - \lambda \Pi_A
      }\Bigr] \, \delta \phi_A
      }\Bigr\}
      \ .
  \end{align}
  Requiring
  $\delta \mathcal{L}_{T;\,\Pi_A} = 0$ for all $\delta\phi_A$ within
  $\Pi_A$, we find
  \begin{align}
    -\mathcal{N'}^\dagger\bigl[{ \log \mathcal{N}'(\phi_A) }\bigr] + \log({\phi_A}) - \lambda \Pi_A
    = 0 \ .
    \label{z:Tu0r.cWs}
  \end{align}
  If $\phi_A$ is any optimal state for the thermodynamic capacity, then $\phi_A$
  must satisfy the condition~\eqref{z:Tu0r.cWs} associated to the
  initial choice of projector $\Pi_A \equiv \Pi_A^{\phi_A}$, leading to the
  stated optimality condition for $\phi_A$.

  On the other hand, if $\Pi_A \equiv \mathds{1}_A$ and a full-rank quantum state
  $\phi_A$
  satisfies~\eqref{z:FMvxXU9D},
  then $\phi_A$ is a stationary point of $\mathcal{L}_{T;\Pi_A}$ in the interior
  of this function's domain.  It is therefore optimal since the problem is
  convex.
\end{proof}

As an example application of \cref{z:qi4DZIqg},
we prove a statement specific to so-called \emph{replacer channels}.  These are
channels that trace out their input and prepare some fixed state.  The following
lemma ensures that condition~\eqref{z:lWRPNfFI}
in \cref{z:TLiwdJGR} is satisfied for such channels, for any
input state.

\begin{lemma}
  \label{z:NrJqAsN1}
  Let $\mathcal{N}_{A\to B}$ be any replacer channel with output state
  $\gamma_B$, i.e., a channel of the form
  $\mathcal{N}_{A\to B}({\cdot}) = \operatorname{tr}({\cdot})\,\gamma_B$.  Then any state
  $\sigma_R$ satisfies
  condition~\eqref{z:lWRPNfFI} for a
  complementary channel $\widehat{\mathcal{N}}$.
\end{lemma}
\begin{proof}[**z:NrJqAsN1]
  Without loss of generality, we assume that $\gamma_B$ is
  full rank. (Otherwise, decrease the dimension of $B$ with no effect on the
  channel entropy.)
  We write a Stinespring dilation of
  $\mathcal{N}_{A\to B}$ on a system $E=E_AE_B$ with $E_A\simeq A$,
  $E_B\simeq B$:
  \begin{align}
    V_{A\to BE_AE_B}
    &= \mathds{1}_{A\to E_A}\otimes \gamma_B^{1/2}\lvert {\Phi_{B:E_B}}\rangle \ ;
    &
      \mathcal{N}_{A\to B}({\cdot}) = \operatorname{tr}_{E_AE_B}\bigl\{{ V\,({\cdot})\,V^\dagger }\bigr\}\ .
  \end{align}
  A complementary channel to $\mathcal{N}_{A\to B}$ is given by
  \begin{align}
    \widehat{\mathcal{N}}_{A\to E_AE_B}({\cdot})
    &= \operatorname{tr}_B\bigl\{{ V\,({\cdot})\,V^\dagger }\bigr\}
      = ({\cdot})_{E_A}\otimes \gamma_{E_B}\ ,
  \end{align}
  namely, the identity channel which maps the input system $A$ to the output
  system $E_A$ and tensors on the fixed state $\gamma_{E_B}$.  Furthermore,
  $\widehat{\mathcal{N}}_{A\leftarrow E_AE_B}^\dagger({\cdot}) = \bigl[{
  \operatorname{tr}_{E_B}\bigl\{{ \gamma_{E_B} ({\cdot}) }\bigr\} }\bigr]_A$, where the system $E_A$ left over
  after the partial trace is relabeled to $A$. Now, let $K_A = -\log({\sigma_A})$.
  We can compute
  \begin{align}
    -\log\bigl[{\widehat{\mathcal{N}}(\sigma_A)}\bigr]
    &=
    -\log({\sigma_{E_A}\otimes\gamma_{E_B}})
    = K_{E_A}\otimes\mathds{1}_{E_B}
      - \Pi^{\sigma_{E_A}}_{E_A}\otimes\log({\gamma_{E_B}})\ ,
  \end{align}
  which implies
  \begin{align}
    \widehat{\mathcal{N}}^\dagger\bigl({ -\log[{
    \widehat{\mathcal{N}}(\sigma_A) }]}\bigr)
    &= K_{A}\operatorname{tr}({\gamma_{E_B}})
      - \Pi_A^{\sigma_A}\,\operatorname{tr}[{\gamma_{E_B} \log({\gamma_{E_B}})}]
    = K_{A} + \Pi_A^{\sigma_A}\,{S}_{}^{}({\gamma_{B}})\ .
  \end{align}
  The left hand side of
  condition~\eqref{z:lWRPNfFI} then reads
  \begin{align}
    \log({\sigma_A})
    - \mathcal{N}^\dagger\bigl({ \log\bigl[{ \widehat{\mathcal{N}}_{A\to E}(\sigma_A) }\bigr] }\bigr)
    &= -K_A + K_A + \Pi_A^{\sigma_A}\,{S}_{}^{}({\gamma_{B}})
      = \Pi_A^{\sigma_A}\,{S}_{}^{}({\gamma_{B}})\ ,
  \end{align}
  which is proportional to $\Pi_A^{\sigma_A}$ as demanded by
  condition~\eqref{z:lWRPNfFI}.
\end{proof}

We also prove a couple lemmas that provide additional guidance on the thermal
quantum channel in the general case where the constraints obey some symmetry.

\begin{lemma}[Constraints symmetric on the output system]
  \label{z:vTX.liTZ}
  Suppose that there exists a completely positive, trace preserving map
  $\mathcal{F}_{B\to B}$ that is unital (i.e.\@
  $\mathcal{F}_{B\to B}(\mathds{1}_B) = \mathds{1}_B$) and such that
  $\mathcal{F}_{B\to B}^\dagger(C^j_{BR}) = C^j_{BR}$ for all $j=1, \ldots, J$.
  If $\mathcal{T}_{}^{(\phi)}$ is a thermal quantum channel with respect to $\phi$,
  then so is $\mathcal{F}\circ\mathcal{T}_{}^{(\phi)}$.  If $\mathcal{T}_{}^{}$ is a
  thermal quantum channel, then so is $\mathcal{F}\circ\mathcal{T}_{}^{}$.
\end{lemma}

\begin{proof}[**z:vTX.liTZ]
  Fix a state $\phi_R$ and suppose that $\mathcal{T}_{}^{(\phi)}$ is a thermal
  quantum channel with respect to $\phi_R$.  Observe that the quantum channel
  $\mathcal{F}\circ\mathcal{T}_{}^{(\phi)}$ satisfies all constraints:
  \begin{align}
    \operatorname{tr}\bigl[{ C^j_{BR} \, \mathcal{F}\circ\mathcal{T}_{}^{(\phi)}(\Phi_{A:R}) }\bigr]
    = 
    \operatorname{tr}\bigl[{ \mathcal{F}^\dagger[C^j_{BR}] \, \mathcal{T}_{}^{(\phi)}(\Phi_{A:R}) }\bigr]
    = 
    \operatorname{tr}\bigl[{ C^j_{BR} \, \mathcal{T}_{}^{(\phi)}(\Phi_{A:R}) }\bigr]
    = q_j\ .
  \end{align}
  The channel entropy of $\mathcal{F}\circ\mathcal{T}_{}^{(\phi)}$ with respect to
  $\phi_R$ obeys
  \begin{align}
    {S}_{}^{}({B}\mathclose{}\,|\,\mathopen{}{R})_{{\mathcal{F}\circ\mathcal{T}_{}^{(\phi)}(\phi_{AR})}}
    =
    {S}_{}^{}\bigl ({\mathcal{F}\circ\mathcal{T}_{}^{(\phi)}(\phi_{AR})}\bigr )
    -
    {S}_{}^{}({\phi_R})
    \geq
    {S}_{}^{}\bigl ({\mathcal{T}_{}^{(\phi)}(\phi_{AR})}\bigr )
    -
    {S}_{}^{}({\phi_R})
    = 
    {S}_{}^{}({B}\mathclose{}\,|\,\mathopen{}{R})_{{\mathcal{T}_{}^{(\phi)}(\phi_{AR})}}\ ,
  \end{align}
  recalling that the unital channel $\mathcal{F}$ can only increase a state's
  von Neumann entropy.  Therefore $\mathcal{F}\circ\mathcal{T}_{}^{(\phi)}$ is also
  optimal in~\eqref{z:.wBvf2qe} and is therefore a
  thermal quantum channel with respect to $\phi_R$.  (It might, in general,
  differ from $\mathcal{T}_{}^{(\phi)}$ for rank-deficient $\phi_R$ with respect to
  which thermal quantum channels might not be unique.)

  Now suppose that $\mathcal{T}_{}^{}$ is optimal
  in~\eqref{z:sb5FfEOw}.  Again, the map
  $\mathcal{F}\circ\mathcal{T}_{}^{}$ is a quantum channel that obeys all the
  constraints of~\eqref{z:sb5FfEOw}.  Let $\phi_R$ be the
  optimal state for the channel entropy of $\mathcal{F}\circ\mathcal{T}_{}^{}$, such
  that
  ${S}_{}^{}({\mathcal{F}\circ\mathcal{T}_{}^{}}) =
  {S}_{\phi}^{}({\mathcal{F}\circ\mathcal{T}_{}^{}})$.  Then
  \begin{align}
    {S}_{}^{}({ \mathcal{F}\circ\mathcal{T}_{}^{} })
    &= 
    {S}_{}^{}({B}\mathclose{}\,|\,\mathopen{}{R})_{{\mathcal{F}\circ\mathcal{T}_{}^{}(\phi)}}
    =
    {S}_{}^{}\bigl ({\mathcal{F}[{\mathcal{T}_{}^{}(\phi_{AR})}]}\bigr )
    -
    {S}_{}^{}({\phi_R})
    \geq
    {S}_{}^{}\bigl ({\mathcal{T}_{}^{}(\phi_{AR})}\bigr )
    -
    {S}_{}^{}({\phi_R})
    = 
    {S}_{}^{}({B}\mathclose{}\,|\,\mathopen{}{R})_{{\mathcal{T}_{}^{}(\phi_{AR})}}
    \geq
    {S}_{}^{}({ \mathcal{T}_{}^{} })\ .
  \end{align}
  Therefore, $\mathcal{F}\circ\mathcal{T}_{}^{}$ is also optimal
  in~\eqref{z:sb5FfEOw}, completing the proof.
\end{proof}

We now consider constraints that are present a symmetry on the input system and
show that the corresponding symmetry is inherited by thermal quantum channels.
In order to state the following lemma, we introduce the following notation.  For
any completely positive map $\mathcal{F}_{A\to A}$, we define a corresponding
map on $R$ via:
\begin{align}
  [\mathcal{F}^t]_{R\to R}(\cdot)
  \equiv \bigl({\mathcal{F}_{A\to A}[(\cdot)^t]}\bigr)^t\ .
\end{align}
This map ensures that
$[\mathcal{F}^t]_{R\to R}(\Phi_{A:R}) = \mathcal{F}_{A\to A}(\Phi_{A:R})$, which
also shows that $\mathcal{F}^t$ is completely positive.  Given a Kraus
representation
$\mathcal{F}(\cdot) = \sum_\ell \tilde{F}_\ell(\cdot) \tilde{F}_\ell^\dagger$,
we have
$\mathcal{F}^t(\cdot) = \sum_\ell
\bigl({\tilde{F}_\ell\,(\cdot)^t\,\tilde{F}_\ell^\dagger }\bigr)^t = \sum_\ell
({\tilde{F}_\ell^t})^\dagger (\cdot) \tilde{F}_\ell^t$.  Finally, if
$\mathcal{F}$ is trace-preserving, then so is $\mathcal{F}^t$: Indeed,
$[\mathcal{F}^t]^\dagger(\mathds{1}_A) = \bigl({\mathcal{F}^\dagger[(\mathds{1})^t]}\bigr)^t =
\mathds{1}$.

\begin{lemma}[Constraints symmetric on the input system]
  \label{z:K-Ef8jou}
  Suppose that there exists a completely positive, trace preserving map
  $\mathcal{F}_{A\to A}$ such that
  $(\mathcal{F}^t)^\dagger(C^j_{BR}) = C^j_{BR}$ for all
  $j=1, \ldots, J$.  If
  $\mathcal{T}_{}^{(\phi)}$ is a thermal quantum channel with respect to $\phi$,
  then so is $\mathcal{T}_{}^{(\phi)}\circ\mathcal{F}$.  If $\mathcal{T}_{}^{}$ is a
  thermal quantum channel, then so is $\mathcal{T}_{}^{}\circ\mathcal{F}$.
\end{lemma}
\begin{proof}[**z:K-Ef8jou]
  Fix a state $\phi_R$ and suppose that $\mathcal{T}_{}^{(\phi)}$ is a thermal
  quantum channel with respect to $\phi_R$.  Observe that the quantum channel
  $\mathcal{T}_{}^{(\phi)}\circ\mathcal{F}$ satisfies all constraints:
  \begin{align}
    \operatorname{tr}\bigl({ C^j_{BR} \, \mathcal{T}_{}^{(\phi)}[{\mathcal{F}({\Phi_{A:R}})}] }\bigr)
    = 
    \operatorname{tr}\bigl({ C^j_{BR} \, [\mathcal{F}^t]_{R\to R}({\mathcal{T}_{}^{(\phi)}[{\Phi_{A:R}}]}) }\bigr)
    = 
    \operatorname{tr}\bigl({ C^j_{BR} \, \mathcal{T}_{}^{(\phi)}[{\Phi_{A:R}}] }\bigr)
    = q_j\ .
  \end{align}
  The channel entropy of $\mathcal{T}_{}^{(\phi)}\circ\mathcal{F}$ with respect to
  $\phi_R$ obeys
  \begin{align}
    {S}_{}^{}({B}\mathclose{}\,|\,\mathopen{}{R})_{{\mathcal{T}_{}^{(\phi)}[{\mathcal{F}_A({\phi_{AR}})}]}}
    &=
    {S}_{}^{}({B}\mathclose{}\,|\,\mathopen{}{R})_{{({\mathcal{F}^t})_R[{\mathcal{T}_{}^{(\phi)}({\phi_{AR}})}]}}
      \geq
    {S}_{}^{}({B}\mathclose{}\,|\,\mathopen{}{R})_{{\mathcal{T}_{}^{(\phi)}({\phi_{AR}})}}\ ,
  \end{align}
  where the inequality follows from the data processing inequality of the
  conditional entropy.  The channel $\mathcal{T}_{}^{(\phi)}\circ\mathcal{F}$ is
  therefore also optimal in~\eqref{z:.wBvf2qe}.

  Now assume that $\mathcal{T}_{}^{}$ is optimal
  in~\eqref{z:sb5FfEOw}.  Again, the map
  $\mathcal{T}_{}^{}\circ\mathcal{F}$ is a quantum channel that obeys all the
  constraints of~\eqref{z:sb5FfEOw}.  Let $\phi_R$ be an
  optimal state for the channel entropy of $\mathcal{T}_{}^{}\circ\mathcal{F}$, such
  that
  ${S}_{}^{}({\mathcal{T}_{}^{}\circ\mathcal{F}}) =
  {S}_{\phi}^{}({\mathcal{T}_{}^{}\circ\mathcal{F}})$.  Then
  \begin{align}
    {S}_{}^{}({ \mathcal{T}_{}^{} \circ \mathcal{F} })
    &= 
    {S}_{}^{}({B}\mathclose{}\,|\,\mathopen{}{R})_{{\mathcal{T}_{}^{}[{\mathcal{F}({\phi})}]}}
    =
    {S}_{}^{}({B}\mathclose{}\,|\,\mathopen{}{R})_{{(\mathcal{F}^t)_{R}[{\mathcal{T}_{}^{}({\phi})}]}}
    \geq
    {S}_{}^{}({B}\mathclose{}\,|\,\mathopen{}{R})_{{\mathcal{T}_{}^{}({\phi})}}
    \geq
    {S}_{}^{}({\mathcal{T}_{}^{}})\ .
  \end{align}
  Therefore, $\mathcal{F}\circ\mathcal{T}_{}^{}$ is also optimal
  in~\eqref{z:sb5FfEOw}, completing the proof.
\end{proof}

If the constraints present a symmetry on their input system, this information is
precious to identify optimal states $\phi_R$ that could be optimal for the
thermal quantum channel.
\begin{lemma}[Symmetry of optimal $\phi_R$ with input-symmetric constraints]
  \label{z:nCrJ.k1B}
  Suppose that there exists a completely positive, trace preserving map
  $\mathcal{F}_{A\to A}$ such that
  $(\mathcal{F}^t)^\dagger(C^j_{BR}) = C^j_{BR}$ for all
  $j=1, \ldots, J$.  
  Let $\phi_A$ be any quantum state.  If
  $\phi_A$ is optimal in~\eqref{z:02ppihHe},
  then so is $\mathcal{F}(\phi_A)$.
\end{lemma}
\begin{proof}[**z:nCrJ.k1B]
  Let $\mathcal{T}_{}^{}$ be optimal in~\eqref{z:sb5FfEOw}, or
  equivalently, in~\eqref{z:02ppihHe}.  By
  \cref{z:K-Ef8jou}, the channel
  $\mathcal{T}_{}^{}\circ\mathcal{F}$ is also optimal.  Let $\phi_A$ be optimal
  in~\eqref{z:02ppihHe}, which implies that
  $\phi_A$ is optimal for the channel entropy
  ${S}_{}^{}({\mathcal{T}_{}^{}\circ\mathcal{F}})$.  Then
  \begin{align}
    {S}_{}^{}({\mathcal{T}_{}^{}})
    &= {S}_{}^{}({\mathcal{T}_{}^{}\circ\mathcal{F}})   %
      = {S}_{\phi}^{}({\mathcal{T}_{}^{}\circ\mathcal{F}})
      = {S}_{}^{}({B}\mathclose{}\,|\,\mathopen{}{R})_{{\mathcal{T}_{}^{}[{\mathcal{F}_A({\phi})}]}}
      \ .
      \label{z:VwF9w5nG}
  \end{align}
  Let $W_{A\to A R_F}$ be a Stinespring dilation isometry
  of $\mathcal{F}_A$ with
  $\mathcal{F}_A({\cdot}) = \operatorname{tr}_{R_F}[{ W_{A\to A R_F}\,({\cdot})\,W^\dagger }]$
  with some additional environment system $R_F$.  From the data processing
  inequality of the conditional entropy,
  \begin{align}
    \text{\eqref{z:VwF9w5nG}}
    &=
      {S}_{}^{}({B}\mathclose{}\,|\,\mathopen{}{R})_{{\operatorname{tr}_{R_F}[{\mathcal{T}_{}^{}({W \phi W^\dagger})}]}}
      \geq
      {S}_{}^{}({B}\mathclose{}\,|\,\mathopen{}{R R_F})_{{\mathcal{T}_{}^{}({W \phi W^\dagger})}}
      \ .
      \label{z:jWtGwnQw}
  \end{align}
  Observe that
  $\operatorname{tr}_{R R_F}({ W_{A\to AR_F} \phi_{AR} W^\dagger}) = \mathcal{F}_A({\phi_A})$.
  As a purification of $\mathcal{F}_A({\phi_A})$, the state
  $W_{A\to AR_F} \lvert {\phi}\rangle _{AR}$ is therefore related to
  $\lvert {\phi'}\rangle _{AR} \equiv [{ \mathcal{F}_A({\phi_A}) }]^{1/2} \lvert {\Phi_{A:R}}\rangle $
  by a partial isometry on $R\to R R_F$.  Therefore,
  \begin{align}
    \text{\eqref{z:jWtGwnQw}}
    &=
      {S}_{}^{}({B}\mathclose{}\,|\,\mathopen{}{R})_{{\mathcal{T}_{}^{}({\phi'_{AR}})}}
      \geq
      {S}_{}^{}({\mathcal{T}_{}^{}})\ .
      \label{z:SHnPCWgy}
  \end{align}
  Combining \cref{z:VwF9w5nG,z:jWtGwnQw,z:SHnPCWgy} we
  find
  \begin{align}
     {S}_{}^{}({\mathcal{T}_{}^{}}) = {S}_{}^{}({B}\mathclose{}\,|\,\mathopen{}{R})_{{\mathcal{T}_{}^{}({\phi'_{AR}})}}\ ,
  \end{align}
  and therefore $\phi'_A= \mathcal{F}_A({\phi_A})$ is optimal for the channel
  entropy of $\mathcal{T}_{}^{}$.
\end{proof}

\subsection{Generalized thermal channel: Minimum channel relative entropy}
\label{z:aXvqINGx}

In Jaynes' principle, we maximize the entropy ${S}_{}^{}({\rho})$ of $\rho$ with
respect to linear constraints $\operatorname{tr}({Q_j\rho}) = q_j$ for $j=1,\ldots, J$.
Recalling that ${S}_{}^{}({\rho}) = -{D}_{}^{}({\rho}\mathclose{}\,\Vert\,\mathopen{}{\mathds{1}})$, this maximization can be
understood as finding the state $\rho$ that most resembles $\mathds{1}$, according
to the relative entropy, while being compatible with the constraints.
A slightly more general version of the problem is the \textit{minimum relative entropy problem}, which is the problem of minimizing
${D}_{}^{}({\rho}\mathclose{}\,\Vert\,\mathopen{}{\sigma})$ with respect to $\rho$, for a given state $\sigma$ and
with constraints $\operatorname{tr}({Q_j\rho})=q_j$ as before.
Here, $\sigma$ may represent prior knowledge about $\rho$, or an earlier
estimate of $\rho$ in an iterative learning algorithm.
The solution to the generalized problem is the \textit{generalized thermal state}
\begin{equation}
    \rho=\frac1Z\, e^{\log(\sigma)-\sum\mu_j Q_j}\ .
    \label{z:ZrNdNY2M}
\end{equation}
This state has an operational meaning within the context of the so-called
\textit{quantum Sanov
  theorem}~\cite{R98,R99}.  The
quantum Sanov theorem is a statement about the decay of an error parameter in a
hypothesis test involving i.i.d.\@ states.  Specifically, let $\mathcal{S}$ be a
subset of density operators, and let $\sigma$ be any quantum state.  For $n>0$,
consider the following hypothesis test: in the null hypothesis, we are handed
the state $\sigma^{\otimes n}$, and in the alternative hypothesis, we are handed
a state $\rho^{\otimes n}$ for some unknown $\rho\in\mathcal{S}$.  We seek a
POVM effect $E$ that is capable of successfully identifying any such
$\rho^{\otimes n}$ except with probability $\varepsilon>0$, while maximizing the
probability that $\mathds{1}-E$ successfully identifies $\sigma^{\otimes n}$.  The
best probability for such an $E$ successfully identifying $\sigma$ is
\begin{align}
  \beta_{\varepsilon,n}(\mathcal{S}\Vert\sigma)
  \coloneqq
  \inf_{0\leq E\leq\mathds{1}}
  \Bigl\{{ \operatorname{tr}({E\sigma^{\otimes n}}) :
  \sup_{\rho\in\mathcal{S}}\operatorname{tr}[{({\mathds{1}-E})\rho^{\otimes n}}] \leq \varepsilon }\Bigr\}\ .
\end{align}
The \emph{quantum Sanov theorem} states
that~\cite{R98,R99}
\begin{align}
  \lim_{n\to\infty} -\frac{1}{n} \log\beta_{\varepsilon,n}(\mathcal{S}\Vert\sigma)
  = \inf_{\rho\in\mathcal{S}}{D}_{}^{}({\rho}\mathclose{}\,\Vert\,\mathopen{}{\sigma})\ .
\end{align}
Now suppose that $\mathcal{S} \equiv \{{\rho : \operatorname{tr}({Q_j\rho}) = q_j\ (j=1,2,\ldots, J) }\}$.
The generalized thermal state therefore %
achieves the optimal asymptotic type-II error exponent in a hypothesis test
between $\sigma^{\otimes n}$ and any $\rho^{\otimes n}$ with
$\rho\in\mathcal{S}$.

Here, we derive a quantum channel analog of the generalized thermal
state~\eqref{z:ZrNdNY2M} by optimizing the channel
relative entropy.

Let $A, B$ be quantum systems, and let $R\simeq A$.  Let $\{{ C_{BR}^j }\}_{j=1}^{n_C}$,
$\{{ D_{BR}^\ell }\}_{\ell=1}^{n_D}$, and $\{{ E_{BR}^m }\}_{m=1}^{n_E}$ be collections
of Hermitian operators acting on $BR$, and let $\{{ q_j }\}_{j=1}^{n_C}$,
$\{{ r_\ell }\}_{\ell=1}^{n_D}$, and $\{{ s_m }\}_{m=1}^{n_E}$ be any collections of real
numbers.  Let $\tilde\eta_m \geq 0$ for $m = 1, \ldots, n_E$.
Let $\mathcal{M}_{A\to B}$ be any completely positive map, and let
$\lvert {\phi}\rangle _{AR} = \phi_A^{1/2}\lvert {\Phi_{A:R}}\rangle $, where $\phi_A$ is an arbitrary
quantum state.
Consider the following optimization problem:
\begin{align}
  \label{z:z4Hq5iYa}
  \begin{aligned}[t]
    \textup{minimize:} \quad
    & {D}_{\phi}^{}({\mathcal{N}_{A\to B}}\mathclose{}\,\Vert\,\mathopen{}{\mathcal{M}_{A\to B}})
      + \sum \tilde\eta_m \Bigl({ s_m - \operatorname{tr}\bigl[{E_{BR}^m \mathcal{N}_{A\to B}(\Phi_{A:R})}\bigr] }\Bigr)^2
    \\
    \textup{over:}\quad
    & \mathcal{N}_{A\to B}\ \textup{c.p., t.p.}
    \\
    \textup{such that:}\quad
    & \operatorname{tr}\bigl[{C^j_{BR}\,\mathcal{N}_{A\to B}(\Phi_{A:R})}\bigr] = q_j\quad\text{for \(j=1, \ldots, n_C\)}\ ;
    \\
    & \operatorname{tr}\bigl[{D^\ell_{BR}\,\mathcal{N}_{A\to B}(\Phi_{A:R})}\bigr] \leq r_\ell\quad\text{for \(\ell=1, \ldots, n_D\)}\ .
  \end{aligned}
\end{align}

\begin{theorem}[Minimum channel relative entropy with respect to fixed input $\phi_A$]
  \label{z:gV43EWT.}
  Assume that there exists a quantum channel
  $\mathcal{N}^{(\mathrm{int})}_{A\to B}$ that satisfies all the problem
  constraints and which obeys $\mathcal{N}_{A\to B}(\Phi_{A:R}) > 0$.
  Any quantum channel $\ThChGen[A\to B][\phi]$ is an optimal solution
  to~\eqref{z:z4Hq5iYa}
  if and only if it satisfies all the problem constraints and it is of the form
  \begin{subequations}
    \label{z:drLpOUuq}
    \begin{gather}
      \ThChGen[A\to B][\phi](\Phi_{A:R})
      = 
      \phi_R^{-1/2}\,\exp\Bigl\{{ -\phi_R^{-1/2} G_{BR} \phi_R^{-1/2} }\Bigr\} \, \phi_R^{-1/2} + Y_{BR}\ ;
      \\[1ex]
      G_{BR} = \sum \mu_j C^j_{BR} + \sum \nu_\ell D^\ell_{BR} + \sum w_m E_{BR}^m
        - \mathds{1}_B\otimes F_R
        - \phi_R^{1/2} \log\bigl({\phi_R^{1/2} M_{BR} \phi_R^{1/2}}\bigr) \phi_R^{1/2}
        - S_{BR}\ ,
    \end{gather}
  \end{subequations}
  where $\mu_j, w_m \in\mathbb{R}$, $\nu_\ell \geq 0$,
  where $F_R$ is a Hermitian matrix, 
  where $S_{BR}$ is a positive semidefinite operator satisfying
  $S_{BR}\, \ThChGen[A\to B][\phi](\Phi_{A:R}) = 0$,
  where $\nu_\ell \bigl({r_\ell - \operatorname{tr}\bigl[{D^\ell_{BR} \ThChGen[A\to B][\phi](\Phi_{A:R}) }\bigr]}\bigr) = 0$,
  where $w_m = 2\tilde\eta_m\bigl[{\operatorname{tr}\bigl({ E_{BR}^m \ThChGen[A\to B][\phi](\Phi_{A:R}) }\bigr) - s_m }\bigr]$,
  where $\Pi^{\phi_R\perp} G_{BR} = 0$,
  and where $Y_{BR}$ is a Hermitian
  operator such that $\Pi^{\phi_R}_R Y_{BR} \Pi^{\phi_R}_R = 0$.
  Furthermore, for any such $\ThChGen[A\to B][\phi]$, we have that $G_{BR}$ is
  positive semidefinite and that $\operatorname{tr}_B({ Y_{BR} }) = \Pi_R^{\phi_R\,\perp}$.
  The attained value for the channel relative entropy with respect to
  $\lvert {\phi}\rangle _{AR}$ is
  \begin{align}
    {D}_{\phi}^{}({\mathcal{N}_{A\to B}}\mathclose{}\,\Vert\,\mathopen{}{\mathcal{M}_{A\to B}})
    = \operatorname{tr}({F_R}) - \sum \mu_j q_j - \sum \nu_\ell r_\ell - \sum w_m s_m - \sum \frac{w_m^2}{2\tilde\eta_m}
    \ .
    \label{z:.eduKYhh}
  \end{align}
\end{theorem}

We now recast the problem~\eqref{z:z4Hq5iYa}
into a maximization, exploiting Lagrangian
duality~\cite{R95}.  The advantage of computing this
quantity as a maximization problem is that we can simultaneously maximize over
the state $\phi_R$.  This enables us to minimize the channel relative
entropy~\eqref{z:Y.8CkTW4}, without fixing the reference
state $\phi_R$.

The following theorem provides a maximization problem that based on the Lagrange
dual of~\eqref{z:z4Hq5iYa},
while retaining some elements and variables of the primal problem.  This
maximization problem is amenable to numerical computation.
\begin{theorem}[A maximization problem version of the minimum channel relative
  entropy problem]
  \label{z:t5t-uuUF}
  Consider the setting of problem~\eqref{z:z4Hq5iYa},
  and assume that there exists some quantum channel with positive definite Choi
  matrix that satisfies all problem constraints (as in
  \cref{z:gV43EWT.}).  
  Now consider the following problem:
  \begin{align}
    \label{z:MK9lgNNV}
    \textup{maximize:}
    &\quad
      \operatorname{tr}({F_R}) - \sum \mu_j q_j - \sum \nu_\ell r_\ell - \sum w_m s_m + 1 - \operatorname{tr}\bigl({N_{BR} \phi_R}\bigr)
      - \sum \frac{w_m^2}{4\tilde\eta_m}
      \\[1ex]
    \textup{over:}
    &\quad
      \mu_j\in\mathbb{R}\ (j=1,\ldots,n_C);
      \ \nu_\ell\geq 0\ (\ell=1,\ldots,n_D);
      \ w_m\in\mathbb{R}\ (m=1,\ldots, n_E); 
      \nonumber\\
    &\quad
      F_R=F_R^\dagger;
      \ N_{BR} \geq 0
      \nonumber\\[1ex]
    \textup{subject to:}
    &\quad
      \phi_R^{1/2}\log\bigl({\phi_R^{1/2} N_{BR} \phi_R^{1/2} }\bigr) \phi_R^{1/2}
      - \phi_R^{1/2}\log\bigl({\phi_R^{1/2} M_{BR} \phi_R^{1/2} }\bigr) \phi_R^{1/2}
      \nonumber\\&\qquad\qquad\qquad
      + \sum \mu_j C_{BR}^j + \sum \nu_\ell D_{BR}^\ell + \sum w_m E_{BR}^m
      - \mathds{1}\otimes F_R \geq 0\ ;
      \nonumber\\[1ex]
    &\quad \operatorname{tr}\bigl({ E^m_{BR} N_{BR} }\bigr) = s_m + \frac{w_m}{2\tilde\eta_m}  \ .
    \nonumber
  \end{align}
  The problem~\eqref{z:MK9lgNNV}
  yields the same optimal value as the problem~\eqref{z:z4Hq5iYa},
  and the variables $F_R$, $\mu_j$, $\nu_\ell$, $N_{BR}$ coincide with those for
  optimal thermal channel
  in~\cref{z:gV43EWT.}.
\end{theorem}

The derivation of \cref{z:t5t-uuUF},
including the derivation of the Lagrange dual problem of~\eqref{z:z4Hq5iYa}, is
presented in \cref{z:UmB-TaUp}.

\paragraph*{Remark on the classical minimum relative entropy problem.}
The minimum relative entropy problem has been long studied within classical information theory~\cite{R100,R101}. Given a probability distribution $Q$, we seek to minimize the relative entropy (Kullback--Leibler divergence) $D(P\Vert Q)$ with respect to distributions $P$ that satisfy linear constraints. This problem has also been referred to as the ``principle of minimum cross entropy''~\cite{R5}, and it is the following optimization problem:
\begin{align}\label{z:3gO8.YUP}
    \begin{aligned}
        \text{minimize:} \quad& D(P\Vert Q) \\ 
        \text{subject to:} \quad& P(x)\geq 0~~\forall~x\in\mathcal{X},\\
        & \sum_{x\in\mathcal{X}}P(x)=1,\\
        & \sum_{x\in\mathcal{X}}P(x)F_j(x)=f_j,\quad j\in\{1,2,\dotsc,J\},
    \end{aligned}
\end{align}
where
\begin{align}
    D(P\Vert Q)
    \coloneqq\sum_{x\in\mathcal{X}}P(x)\log_2\left(\frac{P(x)}{Q(x)}\right)\ .
\end{align}
This problem has an operational meaning in the context of the (classical) \textit{Sanov theorem}~\cite{R102,R103,R104} (see also \cite[Section~11.4]{R63}), which states that
\begin{align}
    \lim_{n\to\infty}-\frac{1}{n}\log\Pr[\hat{Q}_n\in\mathcal{S}]
    = \inf_{P\in\mathcal{S}}D(P\Vert Q)\ ,
\end{align}
where $\mathcal{S}=\{P:\sum_{x\in\mathcal{X}}P(x)F_j(x)=f_j,\,j\in\{1,2,\dotsc,J\}\}$ and $\hat{Q}_n$ is the empirical distribution corresponding to taking $n$ iid samples from $Q$. In other words, the solution to the classical generalized maximum-entropy principle corresponds to the optimal (asymptotic) error exponent for the probability that the empirial distribution is in the set $\mathcal{S}$, i.e., satisfies the required constraints.

The solution to \eqref{z:3gO8.YUP} is~\cite{R100,R101} (see also \cite[Section~11.5]{R63})
\begin{align}
  P^{\star}(x)
  = \frac{1}{Z(\vec{\lambda})}Q(x)
  \exp\mathopen{}\left[{\sum_{x'\in\mathcal{X}}\lambda_{x'}F_{x'}(x)}\right]\mathclose{}\ ,
  \label{z:2oTTHSeN}
\end{align}
where
\begin{align}
  Z(\vec{\lambda})
  =\sum_{x\in\mathcal{X}}Q(x)
  \exp\mathopen{}\left[{\sum_{x'\in\mathcal{X}}\lambda_{x'}F_{x}(x')}\right]\mathclose{}\ ,
\end{align}
and the parameters $\vec{\lambda}=(\lambda_x)_{x\in\mathcal{X}}$ are given
analogously to before via
\begin{align}
  f_x = \frac{\partial}{\partial\lambda_x}\log Z(\vec{\lambda})\ .
\end{align}    

\section{Examples of thermal channels}
\label{z:VuBBeOAb}

\subsection{Channels that discard their inputs}

\paragraph{Unconstrained thermal channel.}
The maximum channel entropy over all quantum channels is achieved by the
completely depolarizing channel~\cite{R46},
\begin{align}
  \mathcal{D}_{A\to B}({\cdot})
  = \operatorname{tr}(\cdot)\,\frac{\mathds{1}_B}{d_B}\ .
\end{align}
Its Choi matrix is proportional to the identity operator,
$\mathcal{D}_{A\to B}({\Phi_{A:R}}) = \mathds{1}_{BR}/d_B$.  This channel is
described in the structure given by \cref{z:TLiwdJGR} by
$\phi_R = \mathds{1}_R/d_R$, $F_R = -\log({d_B d_R})\,\mathds{1}_R/d_R$, $S_{BR}=0$,
$Y_{BR}=0$.

\paragraph{Single input-output constraint.}
Let $\sigma_A$ be a fixed quantum state on $A$, let $H_B$ be a Hermitian
operator on $B$, and let $q\in\mathbb{R}$.  We seek the channel
$\mathcal{N}_{A\to B}$ with maximal channel entropy subject to the constraint
$\operatorname{tr}[{\mathcal{N}({\sigma_A})\, H_B}] = q$.  Equivalently,
$\operatorname{tr}({C_{BR}\,N_{BR}}) = q$ with $C_{BR} \equiv H_B\otimes\sigma_A^{t_{A\to R}}$.
For simplicity, we assume $\sigma_R$ to be full rank.
We seek to satisfy the conditions of \cref{z:TLiwdJGR} through
a suitable choice of $F_R$, $\phi_R$, and $\mu$, verifying that the following
map satisfies all the conditions listed in \cref{z:TLiwdJGR}:
\begin{align}
  \mathcal{N}_{A\to B}(\phi_{AR})
  &= \Pi_R^{\phi_R}\,\exp\Bigl\{{ -\mu\phi_R^{-1/2} C_{BR} \phi_R^{-1/2}
    + \mathds{1}_B\otimes ({\phi_R^{-1/2} F_R \phi_R^{-1/2} + \log\phi_R}) }\Bigr\}\ .
\end{align}
First, we seek choices of $F_R$, $\phi_R$, and $\mu$ that ensure the two terms
inside the exponential commute.  Assuming such a choice exists enables us to
factorize the exponential.  Furthermore, we make the choice $\phi_R = \sigma_R$.
We obtain
\begin{align}
  \mathcal{N}_{A\to B}({\sigma_{AR}})
  &=    \exp\bigl\{{ -\mu H_B\otimes\mathds{1}_R
    + \mathds{1}_B\otimes({\phi_R^{-1/2} F_R \phi_R^{-1/2} + \log\phi_R}) }\bigr\}
    \nonumber\\
  &=
    \exp\bigl\{{ -\mu H_B }\bigr\} \otimes
    \exp\bigl\{{- ({\sigma_R^{-1/2} F_R \sigma_R^{-1/2}
    + \log\sigma_R}) }\bigr\}\ ,
\end{align}
writing $\lvert {\sigma}\rangle _{AR} \equiv \sigma_R^{1/2}\,\lvert {\Phi_{A:R}}\rangle $.  We know that
the reduced state of this expression on $R$ must be $\sigma_R$, given that
$\mathcal{N}_{A\to B}$ must be trace-preserving.  This observation motivates the
choice $F_R = -\log({Z})\,\sigma_R$ with the real number $Z=\operatorname{tr}({{e}^{-\mu H_B}})$
chosen to ensure the state is normalized.  We find:
\begin{align}
  \mathcal{N}_{A\to B}({\sigma_{AR}})
  &=
    {e}^{ -\mu H_B }\otimes \mathopen{}\left({\frac{\sigma_R}{Z}}\right)\mathclose{}
    = \gamma_B \otimes \sigma_R\ ;
  &
    \gamma_B &\equiv \frac{{e}^{-\mu H_B}}{Z} \ .
\end{align}
We recognize a channel that traces out its input, replacing it by the thermal state $\gamma_B$:
\begin{align}
  \mathcal{N}_{A\to B}({\cdot}) = \operatorname{tr}({\cdot})\,\gamma_B\ .
  \label{z:4jZiVZRe}
\end{align}
We then naturally choose $\mu$ and $Z$ such that the thermal state $\gamma_B$
satisfies both $\operatorname{tr}(H_B\gamma_B) = q$ and $\operatorname{tr}(\gamma_B) = 1$.  At this point,
our choices for $F_R$ and $\mu$ satisfy all conditions laid out in
\cref{z:hd7.cLe.}; the channel we've found is
therefore the unique thermal channel with respect to $\sigma_R$:
\begin{align}
  \mathcal{T}_{A\to B}^{(\sigma_R)}({\cdot}) = \operatorname{tr}({\cdot})\,\gamma_B\ .
\end{align}
Interestingly, this channel does not depend on $\sigma_R$.
Thanks to \cref{z:NjFfeh1f}, we find that this
quantum channel is also a thermal quantum channel with respect to any
rank-deficient $\sigma_R$.

Therefore, this channel is also the thermal channel $\mathcal{T}_{A\to B}^{}$, that is,
the thermal channel with respect to the optimal state $\phi_R$ in the definition
of the channel entropy.

\paragraph{Output-energy-constrained thermal channel.}
A natural question is, what if we impose a constraint on the output of the
channel, which should always hold regardless of the input?  For instance, we
could require that $\operatorname{tr}[{H_B \mathcal{N}(\sigma)}] = q$ for all input states
$\sigma$.

In light of the previous example, it is clear that the answer is again the
channel that traces out its input and prepares the thermal state $\gamma_B$
compatible with the constraint $\operatorname{tr}[{H_B \mathcal{N}(\sigma)}] = q$.  Indeed, the
channel obtained in the last example already satisfies all the constraints
imposed here.

\subsection{Energy-conserving channels}

\paragraph{Strictly energy-conserving thermal channel.}
Now we imagine we have some global energy conservation constraint on the
channels we consider (or some other superselection rule).  Specifically, let us
consider a setting where the input and output systems coincide, $A \simeq B$,
with an arbitrary fixed Hamiltonian $H_A = H_B$.  We now require the channel
$\mathcal{N}$ to strictly preserve energy: For any state $\lvert {\psi}\rangle _A$ supported
on an eigenspace of $H_A$ with energy $E$, we require that $\lvert {\psi}\rangle _A$ is
mapped to a state that lies in the same eigenspace on $B$.  We can formalize
this condition as follows, where $\Pi^{(E)}$ denotes the eigenspace of the
Hamiltonian for energy $E$:
\begin{align}
  \operatorname{tr}\bigl[{ \bigl({\mathds{1}- \Pi^{(E)}}\bigr) \, \mathcal{N}(\Pi^{(E)}) }\bigr] = 0
  \qquad\text{for all}\ E\ .
\end{align}
Equivalently, $\operatorname{tr}[{ \Pi^{(E)}\,\mathcal{N}(\Pi^{(E)})}] = \operatorname{tr}({\Pi^{(E)}})$, a
constraint encoded as $\operatorname{tr}[{C^{(E)}_{BR} N_{BR}}] = \operatorname{tr}({\Pi^{(E)}})$ with
$C^{(E)}_{BR} = \Pi^{(E)}\otimes\Pi^{(E)}$.  The thermal channel can no longer
be of the form~\eqref{z:4jZiVZRe}, since it must keep states within
whatever energy eigenspaces they started off in.   It is still simple
to guess the form of the thermal channel: The thermal channel completely
depolarizes the state within each energy eigenspace.  
Indeed, for any such $\mathcal{N}$, suppose that
$\lvert {\phi^{(E)}}\rangle _{AR} \equiv (\phi_R^{(E)})^{1/2}\lvert {\Phi_{A:R}}\rangle $, where
$\phi_R^{(E)}$ is supported within $\Pi^{(E)}$.  The state
$\mathcal{N}(\phi_{AR})$ must therefore lie within $\Pi^{(E)}\otimes\Pi^{(E)}$.
Applying a trace-decreasing depolarizing map with support $\Pi^{(E)}$, and using
the data processing inequality of the relative entropy, we find
${D}_{}^{}({ \mathcal{N} }\mathclose{}\,\Vert\,\mathopen{}{ \widetilde{D} }) \geq {D}_{}^{}({ \mathcal{N}(\phi^{(E)}_{AR}) }\mathclose{}\,\Vert\,\mathopen{}{
  \mathds{1}_B\otimes \phi^{(E)}_R }) \geq {D}_{}^{}\bigl ({
  (\Pi^{(E)}/\operatorname{tr}({\Pi^{(E)}})\otimes\phi^{(E)}_R }\mathclose{}\,\big \Vert\,\mathopen{}{ \mathds{1}_B\otimes \phi^{(E)}_R
}\bigr ) = -\log\operatorname{tr}({\Pi^{(E)}})$.  Therefore,
${S}_{}^{}({\mathcal{N}}) \leq \min_E \log\operatorname{tr}({\Pi^{(E)}})$.  On the other hand, this
channel entropy is achieved when $\mathcal{N}$ acts as the completely
depolarizing channel within each subspace $\Pi^{(E)}$.

\paragraph{Average-energy-conserving thermal channel.}
Another constraint we can require is average energy conservation.  If $H_A$ and
$H_B$ are the respective Hamiltonians of $A$ and $B$, we seek the map
$\mathcal{N}$ that maximizes ${S}_{}^{}({\mathcal{N}})$ while ensuring that
for all $\sigma_A$:
\begin{align}
  \operatorname{tr}[{H_B \mathcal{N}(\sigma_A)}] - \operatorname{tr}({H_A\sigma_A}) = 0\ .
  \label{z:RQBcpP4K}
\end{align}
It suffices to impose~\eqref{z:RQBcpP4K} for any finite set of
$\{{ \sigma^{(j)} }\}$ that span the space of Hermitian operators.  E.g., if $A$ is
a single qubit, we could choose for $\{{ \sigma_j }\}$ the set of density matrices
$\mathds{1}/2$, $(\mathds{1}+X)/2$, $(\mathds{1}+Y)/2$, $(\mathds{1}+Z)/2$ where $X,Y,Z$ are the
single-qubit Pauli operators.  Here, and for general $A$, it turns out that a
convenient set of states to impose this constraint for are a spanning set that
contain the energy eigenstates of $H_A$.  Let
$\{{ \lvert {e_\ell}\rangle _A }\}_{\ell=1}^{d_R}$ be an eigenbasis of $H_A$ with eigenvalues
$e_\ell$.  We pick
$\{{ \sigma^{(j)} }\} = \{{ \lvert {\ell}\rangle \mkern -1.8mu\relax \langle{\ell}\rvert _A }\}_{\ell=1}^{d_R} \cup \mathcal{S}$, where
$\mathcal{S}$ is any finite set of operators that complete the states
$\{{ \lvert {\ell}\rangle \mkern -1.8mu\relax \langle{\ell}\rvert _A }\}$ into a spanning set of all Hermitian operators on $A$.

Let $H_R = H_A^{t_{A\to R}}$.  The constraint~\eqref{z:RQBcpP4K} can
be realized by a family of constraint operators
$C^{j}_{BR} = \sigma_R^{(j) 1/2} \, \bigl({H_{B}\otimes\mathds{1}_R - \mathds{1}_B\otimes
H_R}\bigr) \, \sigma_R^{(j) 1/2}$ and with $q_j \equiv 0$ for all $j$.  By
construction, the set of states $\{{ \sigma_R^{(j)} }\}$ contains the set of states
$\{{ \lvert {\ell}\rangle \mkern -1.8mu\relax \langle{\ell}\rvert _R }\}_{\ell=1}^{d_R}$ where
$\lvert {\ell}\rangle \mkern -1.8mu\relax \langle{\ell}\rvert _R \equiv \lvert {\ell}\rangle \mkern -1.8mu\relax \langle{\ell}\rvert _A^{t_{A\to R}}$ is an eigenstate of $H_R$
associated with the eigenvalue $e_\ell$.

The thermal channel's structure is given by \cref{z:TLiwdJGR}.
We have a variable $\mu_j\in\mathbb{R}$ for each $j$.  For all $j$ corresponding
to a $\sigma_R^{(j)} = \lvert {\ell}\rangle \mkern -1.8mu\relax \langle{\ell}\rvert _R$, we write $\mu_\ell$ instead of $\mu_j$.
For all other $j$, we set $\mu_j = 0$.  With our choices of variables, the
expression~\eqref{z:s1gf97hu} takes the form
\begin{align}
  \mathcal{T}_{}^{(\phi)}(\Phi_{A:R})
  =
  \begin{aligned}[t]
    \phi_R^{-1/2}\,\exp\Bigl\{-\phi_R^{-1/2}\Bigl[
    &
    \sum \mu_\ell \lvert {e_\ell}\rangle \mkern -1.8mu\relax \langle{e_\ell}\rvert _R\bigl({H_B - H_R}\bigr)\lvert {e_\ell}\rangle \mkern -1.8mu\relax \langle{e_\ell}\rvert _R
    \\
    &\quad
      - \mathds{1}_B\otimes ({F_R + \phi_R\log\phi_R})
      \Bigr] \phi_R^{-1/2} \Bigr\} \phi_R^{-1/2} + Y_{BR}\ ,
  \end{aligned}
\end{align}
where $Y_{BR}$ satisfies $\Pi_R^{\phi_R} Y_{BR} \Pi_R^{\phi_R} = 0$.  We now
pick $\phi_R = \sum s_\ell \lvert {e_\ell}\rangle \mkern -1.8mu\relax \langle{e_\ell}\rvert $, $F_R = \sum f_\ell \lvert {e_\ell}\rangle \mkern -1.8mu\relax \langle{e_\ell}\rvert $ for
some $s_\ell, f_\ell$ to be fixed later with $s_\ell\geq 0$.  We find
\begin{align}
  \mathcal{T}_{}^{(\phi)}(\Phi_{A:R})
  &= \phi_R^{-1/2} \exp\biggl\{{\sum_{s_\ell\neq 0} \frac1{s_\ell} \Bigl[{
  ({f_\ell + s_\ell\log s_\ell}) \mathds{1}_B
  - \mu_\ell H_B + \mu_\ell e_\ell\mathds{1}_B
  }\Bigr] \otimes \lvert {e_\ell}\rangle \mkern -1.8mu\relax \langle{e_\ell}\rvert _R }\biggr\} \phi_R^{-1/2}
    + Y_{BR}
  \nonumber\\
  &= \sum_{s_\ell\neq 0}
    \phi_R^{-1/2} \Bigl({
    {e}^{\frac{f_\ell}{s_\ell} + \log(s_\ell) + \frac{\mu_\ell}{s_\ell} e_\ell}
    {e}^{-\frac{\mu_\ell}{s_\ell} H_B}
  \otimes \lvert {e_\ell}\rangle \mkern -1.8mu\relax \langle{e_\ell}\rvert _R}\Bigr) \phi_R^{-1/2}
    \ + \ Y_{BR}
    \nonumber\\
  &= \sum_{s_\ell \neq 0}
    {e}^{\frac{f_\ell}{s_\ell} + \frac{\mu_\ell}{s_\ell} e_\ell}
    {e}^{-\frac{\mu_\ell}{s_\ell} H_B}
  \otimes \lvert {e_\ell}\rangle \mkern -1.8mu\relax \langle{e_\ell}\rvert _R
    \ +\  Y_{BR}\ .
    \label{z:mkyHEKHx}
\end{align}
On the support of $\phi_R$, this channel measures its input in the energy basis
and prepares a Gibbs state $\gamma_{\beta_\ell}$ on the output with a
temperature $\beta_\ell \equiv \mu_\ell/s_\ell$ that depends on the measured
input energy.  The Gibbs state is, as usual,
$\gamma_{\beta} \equiv {e}^{-\beta H_B}/Z(\beta)$ with
$Z(\beta) = \operatorname{tr}({{e}^{-\beta H_B}})$.

For the map to conserve average energy as initially demanded, we need that
$\operatorname{tr}({H_B \gamma_{\beta_\ell}}) = e_\ell$.  This implicitly fixes $\beta_\ell$
and thereby $\mu_\ell = s_\ell \beta_\ell$.
For the map to be trace-preserving, we need the reduced state
of~\eqref{z:mkyHEKHx} on $R$ to equal the identity, leading to
\begin{align}
  {e}^{\frac{f_\ell}{s_\ell} + \beta_\ell e_\ell}
  \operatorname{tr}({{e}^{-\beta_\ell H_B}})
  &= 1\ ;
  &
    \operatorname{tr}_B Y_{BR} = \Pi^{\phi_R\,\perp}_R
    \ .
\end{align}
Solving the first equation for $f_\ell/s_\ell$ yields
\begin{align}
  \frac{f_\ell}{s_\ell} =
  \log\Bigl({\frac{1}{Z(\beta_\ell)}}\Bigr) - \beta_\ell e_\ell\ .
\end{align}
At this point, we also choose $Y_{BR}$ to complete the channel to act outside
the support of $\phi_R$ in the same way as it acts within $\phi_R$'s support,
namely by measuring the input energy and preparing a correspondingly energetic
Gibbs state.  The channel then becomes
\begin{align}
  \mathcal{T}_{}^{(\phi)}(\Phi_{A:R})
  = 
  \mathcal{T}_{}^{}(\Phi_{A:R})
  &= \sum_{\ell=1}^{d_R} \frac1{Z(\beta_\ell)} {e}^{-\beta_\ell H_B} \otimes \lvert {e_\ell}\rangle \mkern -1.8mu\relax \langle{e_\ell}\rvert _R\ ,
    \label{z:dpBa9wEu}
\end{align}
where $\beta_\ell$ is implicitly determined from
$\operatorname{tr}({\gamma_{\beta_\ell} H_B}) = e_\ell$, and noting that this map no longer
depends on $s_\ell$, i.e, on $\phi_R$.  The
map~\eqref{z:dpBa9wEu} is a valid c.p., t.p.\@ map of the
form~\eqref{z:s1gf97hu}.

The attained channel
entropy is, according to~\eqref{z:WlEdtU0A},
\begin{align}
  {S}_{}^{}({\mathcal{T}_{}^{(\phi)}})
  = -\sum_\ell f_\ell 
  = \sum_\ell s_\ell\Bigl({ \log[{Z(\beta_\ell)}] + \beta_\ell e_\ell }\Bigr)
  = \sum_\ell s_\ell\,{S}_{}^{}({\gamma_{\beta_\ell}})\ ,
  \label{z:ZbK.ytWi}
\end{align}
where we recognize the expression for the entropy of a thermal state
${S}_{}^{}({\gamma_{\beta}}) = -\operatorname{tr}[{\gamma_\beta \log (\gamma_\beta)}]
= \operatorname{tr}\bigl[{\gamma_\beta\bigl({\beta H + \log[{Z({\beta})}]\mathds{1}}\bigr)}\bigr]
= \beta \operatorname{tr}(\gamma_\beta H) + \log[{Z({\beta})}]$.
The expression is minimized by choosing $s_\ell = 0$ for all terms except the
$\ell$ (or those $\ell$) that have minimal $S(\gamma_{\beta_\ell})$.  For such a
choice of $\{{ s_\ell }\}$, we find
\begin{align}
  {S}_{}^{}({\mathcal{T}_{}^{(\phi)}})
  = \min_\ell
  {S}_{}^{}({ \gamma_{\beta_\ell} })\ ,
\end{align}
recalling that $\beta_\ell$ is determined implicitly by the condition
$\operatorname{tr}({\gamma_{\beta_\ell} H_B}) = e_\ell = \langle {\ell}\mkern 1.5mu\relax \vert \mkern 1.5mu\relax {H_R}\mkern 1.5mu\relax \vert \mkern 1.5mu\relax {\ell}\rangle $.

We can make use of \cref{z:NjFfeh1f} to conclude
that $\mathcal{T}_{}^{}$ is the thermal channel with respect to any $\phi_R$ that is
diagonal in the energy eigenbasis, including among rank-deficient states.
Actually, if $\phi_R$ is rank-deficient, then the channel entropy becomes
insensitive to the channel's action on input states outside the support of
$\phi_R$.  Indeed, the channel could prepare arbitrary, nonthermal, states for
all cases where $\ell\neq 0$, provided they have more entropy than those
$\gamma_{\beta_\ell}$'s where $s_\ell \neq 0$.  On the other hand, requiring
that the thermal channel is a limit of thermal channels with respect to
full-rank states singles out the channel~\eqref{z:dpBa9wEu}.

Furthermore, we can prove that the optimal $\phi_R$ is indeed diagonal in the
energy eigenbasis using \cref{z:nCrJ.k1B}.
Consider first the maximum channel entropy problem including only the
constraints $C^\ell_{BR}$ with $\ell=1, \ldots, d_R$.
\Cref{z:nCrJ.k1B} then states that the
optimal $\phi_R$ is, without loss of generality, diagonal in the energy
eigenbasis, given that all $C^\ell_{BR}$ obey
$C^\ell_{BR} = (\mathcal{F}^t)^\dagger(C^\ell_{BR})$ where
$\mathcal{F}(\cdot) = \sum \lvert {\ell}\rangle \mkern -1.8mu\relax \langle{\ell}\rvert ({\cdot})\lvert {\ell}\rangle \mkern -1.8mu\relax \langle{\ell}\rvert $ is a complete
dephasing operation in the energy eigenbasis.  We proved above that for these
constraints and for energy-diagonal $\phi_R$, the thermal channel takes the
form~\eqref{z:dpBa9wEu}.  Now, this channel automatically
satisfies all remaining constraints with $C^j_{BR}$ for $j > d_R$; therefore
$\mathcal{T}_{}^{}$ in~\eqref{z:dpBa9wEu} is automatically a thermal
channel for the wider, redundant set of constraints, as well.

All in all, we proved that the quantum channel~\eqref{z:dpBa9wEu}
is indeed a quantum thermal channel for the constraints of average energy
conservation for all input states.

\subsection{Channel with Pauli-covariant constraints}
\label{z:JM5--Sw6}

Here, we suppose that $B\simeq A$ with $d_B = d_A \equiv d\in\{2,3,\dotsc\}$.  The discrete Weyl
operators $W^{z,x}$ on a $d$-dimensional system are defined as:
\begin{align}
  W^{z,x}
  &=Z(z)X(x)\ ;
  &
    Z(z) &=\sum_{k=0}^{d-1}e^{\frac{2\pi i kz}{d}}\lvert {k}\rangle \mkern -1.8mu\relax \langle{k}\rvert \ ;
  &
    X(x) &=\sum_{k=0}^{d-1}\lvert {k+x}\rangle \mkern -1.8mu\relax \langle{k}\rvert \ ,
\end{align}
where the addition in the definition of $X(x)$ is performed modulo $d$. These operators generalize the single-qubit Pauli operators to qudits and are sometimes
called qudit Pauli operators.

A map $\mathcal{N}$ is called Pauli-covariant if for all $z,x\in\{0,1,\dotsc,d-1\}$,
\begin{align}
  \mathcal{N}\bigl({W^{z,x}(\cdot)W^{z,x\dagger}}\bigr)
  = W^{z,x}\mathcal{N}(\cdot)W^{z,x\dagger}\ .
\end{align}
If
$\mathcal{N}$ is Pauli-covariant, then
\begin{align}
    \frac{1}{d^2}\sum_{z,x=0}^{d-1}W^{z,x\dagger}
    \mathcal{N}\bigl({W^{z,x}(\cdot)W^{z,x\dagger}}\bigr)W^{z,x} = \mathcal{N}(\cdot)\ .
\end{align}
Letting $N_{BR}$ denote the Choi representation of $\mathcal{N}$, we can write the above equation as follows:
\begin{align}
    \mathcal{B}(N_{BR})
  \coloneqq
  \frac{1}{d^2} \sum_{z,x=0}^{d-1}
  (W_B^{z,x}\otimes W_R^{z,x\ast})^{\dagger}
  N_{BR}
  (W_B^{z,x}\otimes W_R^{z,x\ast})
  = N_{BR}\ .
\end{align}
A convenient way to describe the map $\mathcal{B}$ is using the Bell states. We define the (unnormalized) two-qudit Bell states as follows:
\begin{align}
  \lvert {\Phi_{z,x}}\rangle 
  &\coloneqq (W^{z,x}\otimes\mathds{1})\lvert {\Phi}\rangle \ ;
  &
    \lvert {\Phi}\rangle 
  &= \sum_{k=1}^{d}\lvert {k,k}\rangle \ ,
\end{align}
for $z,x\in\{0,1,\dotsc,d-1\}$. It is straightforward to show that $\mathcal{B}$
is the Bell-basis pinching channel (see, e.g.,
\cite[Appendix~C]{R17}), i.e.,
\begin{align}
  \mathcal{B}(N_{BR})
  = \frac{1}{d^2}\sum_{z,x=0}^{d-1}
  \lvert {\Phi_{z,x}}\rangle \mkern -1.8mu\relax \langle{\Phi_{z,x}}\rvert  N_{BR} \lvert {\Phi_{z,x}}\rangle \mkern -1.8mu\relax \langle{\Phi_{z,x}}\rvert \ .
\end{align}
Therefore, if $\mathcal{N}$ is Pauli-covariant, then its Choi representation is
diagonal in the Bell basis.

A \emph{Pauli channel} is a quantum channel of the form
\begin{align}
    \mathcal{P}(\cdot)
  = \sum_{z,x=0}^{d-1}p_{z,x} W^{z,x} (\cdot) W^{z,x\dagger}\ ,
\end{align}
where $p_{z,x}\geq 0$ and $\sum_{z,x} p_{z,x} = 1$. %
Every Pauli channel is manifestly
Pauli-covariant. Note also that
$\widetilde{\mathcal{D}}(\cdot) = \operatorname{tr}({\cdot})\mathds{1}$ is Pauli-covariant.  It
follows that the entropy of a Pauli channel is simply the entropy of the
probability distribution that defines it~\cite{R79}, i.e.,
\begin{align}
  \label{z:3vMyEQwA}
  {S}_{}^{}({\mathcal{P}})
  = -\sum_{z,x=0}^{d-1}p_{z,x}\log p_{z,x}\  .
\end{align}

Now, consider the maximum entropy problem \eqref{z:sb5FfEOw}, and suppose that the constraint operators $C_{BR}^{j}$ are Pauli-covariant, in the sense that $\mathcal{B}(C_{BR}^j)=C_{BR}^j$ for all $j$. An example of this is Bell sampling, in which $C_{BR}^{z,x}=\frac{1}{d^2}\Phi_{BR}^{z,x}$. This observable corresponds to a channel measurement that consists of preparing the maximally-entangled state $\frac{1}{d}\Phi$, sending one-half of it through the channel, and then performing a Bell measurement, i.e., measuring both systems with respect to the POVM $\{\frac{1}{d}\Phi^{z,x}\}_{z,x}$. 
Because the $C_{BR}^j$ are Pauli-covariant, they are diagonal in the Bell basis, i.e.,
\begin{align}
    C_{BR}^j
    = \frac{1}{d^2}\sum_{z,x=0}^{d-1}c_{z,x}^j\Phi_{BR}^{z,x} \ ,
\end{align}
where $c_{z,x}^j=\frac{1}{d}\operatorname{tr}\bigl[{C_{BR}^j\Phi_{BR}^{z,x}}\bigr]$. The constraint
$\operatorname{tr}[{C_{BR}^jN_{BR}}]=q_j$ is then equivalent to
\begin{align}
    \sum_{z,x=0}^{d-1} c_{z,x}^j
  \underbrace{\frac{1}{d^2}\operatorname{tr}[N_{BR}\Phi_{BR}^{z,x}]}_{\equiv p_{z,x}}
  = q_j\ ,
\end{align}
where $p_{z,x}=\frac{1}{d^2}\operatorname{tr}[N_{BR}\Phi_{BR}^{z,x}]$ %
satisfy $0\leq p_{z,x}\leq 1$ and $\sum_{x,z} p_{z,x} = 1$.
Indeed, because $N_{BR}$ is a Choi matrix, it
holds that $p_{z,x}\geq 0$ for all $z,x\in\{0,1,\dotsc,d-1\}$, and
$\sum_{z,x=0}^{d-1}p_{z,x}=\operatorname{tr}[{N_{BR}\frac{1}{d^2}\sum_{z,x=0}^{d-1}\Phi_{BR}^{z,x}}]=\frac{1}{d}\operatorname{tr}[N_{BR}\mathds{1}_{BR}]=\frac{1}{d}\operatorname{tr}[\mathds{1}_R]=1$,
where we used the fact that
$\frac{1}{d}\sum_{z,x=0}^{d-1}\Phi_{BR}^{z,x}=\mathds{1}_{BR}$.

Now, note that 
\begin{align}
    {S}_{}^{}({\mathcal{N}})
    = -{D}_{}^{}({\mathcal{N}}\mathclose{}\,\Vert\,\mathopen{}{\widetilde{\mathcal{D}}})
    \leq -{D}_{}^{}\bigl ({\Theta_{\mathcal{B}}(\mathcal{N})}\mathclose{}\,\big \Vert\,\mathopen{}{
           \Theta_{\mathcal{B}}(\widetilde{\mathcal{D}})}\bigr )
    =-{D}_{}^{}\bigl ({\Theta_{\mathcal{B}}(\mathcal{N})}\mathclose{}\,\big \Vert\,\mathopen{}{\widetilde{\mathcal{D}}}\bigr )
    =-\sum_{z,x=0}^{d-1}p_{z,x}\log p_{z,x}\ ,
\end{align}
where we have used the data-processing inequality, the fact that $\widetilde{\mathcal{D}}$ is Pauli-covariant, and the expression in \eqref{z:3vMyEQwA} for the entropy of a Pauli channel, noting that $\Theta_{\mathcal{B}}(\mathcal{N})$ is a Pauli channel. Here, $\Theta_{\mathcal{B}}$ is a super channel such that the Choi representation of $\Theta_{\mathcal{B}}(\mathcal{N})$ is $\mathcal{B}(N)$, with $N$ being the Choi representation of $\mathcal{N}$. Explicitly, the channel $\Theta_{\mathcal{B}}(\mathcal{N})$ has the following action:
\begin{align}
  \Theta_{\mathcal{B}}(\mathcal{N}_{A\to B})(\cdot)
  = \operatorname{tr}_R\bigl[{(\cdot)_A^{t_{A\to R}}\,\mathcal{B}(N_{BR})}\bigr].
\end{align}
Combining the above inequality with Pauli-covariance of the constraints, it follows that problem \eqref{z:sb5FfEOw} reduces to the following:
\begin{align}
  \label{z:B8hGV.Fe}
  \begin{aligned}[t] 
    \textup{maximize:} \quad
    & -\sum_{z,x=0}^{d-1} p_{z,x}\log p_{z,x}
    \\
    \textup{over:}\quad
    & p_{z,x}\geq 0,\, \sum_{z,x=0}^{d-1} p_{z,x}=1
    \\
    \textup{such that:}\quad
    & \sum_{z,x=0}^{d-1}c_{z,x}^jp_{z,x}= q_j\quad\text{for \(j=1, \ldots, J\)}\ ,
  \end{aligned}
\end{align}
which is nothing but the usual (classical) maximum-entropy problem. The optimal
channel is therefore a Pauli channel for which the associated probability
distribution has the form of a Gibbs/thermal distribution, i.e.,
\begin{align}
    N_{BR}^{\star}
  = \sum_{z,x=0}^{d-1}
  \frac{1}{Z} e^{-\sum_{j=1}^J\sum_{z,x=0}^{d-1}c_{z,x}^j\mu_{z,x}^j} \Phi_{BR}^{z,x}
  \ .
\end{align}

\subsection{Classical thermal channel}
\label{z:9P7ZvFQO}

We now study the classical version of thermal quantum channels and connect our
results to known concepts from classical information theory. Specifically, we consider the special case of \eqref{z:sb5FfEOw} in which all constraints are diagonal in the joint computational basis of $B$ and $R$. In this case, we show that the problem reduces to a classical version of the maximum channel entropy problem.

Let us start by computing the quantum channel entropy ${S}_{}^{}({\mathcal{N}})$ for a quantum
channel implementing a classical stochastic map. 
A classical stochastic mapping in $d\in\{2,3,\dotsc\}$ dimensions is defined by a $d\times d$ matrix $T$ of conditional probabilities, i.e., $T=\sum_{j,k=1}^d T_{k|j}\lvert {k}\rangle \mkern -1.8mu\relax \langle{j}\rvert $, such that $T_{k|j}$ represents the probability that a system transitions to the state $k$ from the state $j$. As such, the columns of $T$ sum to one, i.e., $\sum_{k=1}^d T_{k|j}=1$ for all $j\in\{1,2,\dotsc,d\}$. We can write this as a quantum channel in the following way:
\begin{equation}\label{z:pFhIXD-e}
    \mathcal{N}_{A\to B}(\lvert {i}\rangle \mkern -1.8mu\relax \langle{j}\rvert )=\delta_{i,j}\sum_{k=1}^dT_{k|j}\lvert {k}\rangle \mkern -1.8mu\relax \langle{k}\rvert ,
\end{equation}
for all $i,j\in\{1,2,\dotsc,d\}$. The Choi representation of $\mathcal{N}$ is then
\begin{equation}
    N_{BR}=\sum_{j,k=1}^d T_{k|j}\lvert {k,j}\rangle \mkern -1.8mu\relax \langle{k,j}\rvert .
\end{equation}

\begin{proposition}[Entropy of a classical stochastic mapping]
  \label{z:A9iM-feK}
  Let $\mathcal{N}$ be the quantum channel corresponding to a classical
  stochastic mapping $T$, as in \eqref{z:pFhIXD-e}. Its entropy is
  \begin{align}
    S(\mathcal{N})
    = \min_{j\in\{1,2,\dotsc,d\}}\left\{-\sum_{k=1}^d T_{k|j}\log T_{k|j}\right\}\ .
  \end{align}
\end{proposition}

\begin{proof}[**z:A9iM-feK]
  Invoking \cref{z:nCrJ.k1B} with
  $\mathcal{F}$ the completely dephasing channel in the canonical basis, the
  optimal input state in the definition of the channel entropy of $\mathcal{N}$
  can be chosen without loss of generality to be diagonal in the canonical
  basis.
  Therefore, let $\rho_R=\sum_{j=1}^d p_j\lvert {j}\rangle \mkern -1.8mu\relax \langle{j}\rvert $ be an input distribution, such that
  $p_j\geq 0$, $\sum_{j=1}^d p_j=1$. The corresponding output distribution
  is
  \begin{align}
      \omega_{BR}
      = \rho_R^{1/2}N_{BR}\rho_R^{1/2}
      = \sum_{k,j=1}^d T_{k|j}p_j\lvert {k,j}\rangle \mkern -1.8mu\relax \langle{k,j}\rvert .
  \end{align}
  The conditional entropy ${S}_{}^{}({B}\mathclose{}\,|\,\mathopen{}{R})_{\omega}$ is then
    \begin{align}
      {S}_{}^{}({B}\mathclose{}\,|\,\mathopen{}{R})_{\omega}
      &={S}_{}^{}({BR})_{\omega} - {S}_{}^{}({R})_{\omega}
        \nonumber\\
      &=-\sum_{k,j=1}^d T_{k|j}p_j\log(T_{k|j}p_j)+\sum_{j=1}^dp_j\log p_j
        \nonumber\\
      &=-\sum_{k,j=1}^dT_{k|j}p_j\left(\log T_{k|j}+\log p_j\right)+\sum_{j=1}^d p_j\log p_j
        \nonumber\\
      &=-\sum_{k,j=1}^d p_jT_{k|j}\log T_{k|j}-\sum_{j=1}^d\underbrace{\left(\sum_{k=1}^d T_{k|j}\right)}_{=1\,\,\forall j}p_j\log p_j+\sum_{j=1}^d p_j\log p_j
        \nonumber\\
      &=-\sum_{k,j=1}^d p_jT_{k|j}\log T_{k|j}\ .
    \end{align}
    Therefore,
    \begin{align}
      {S}_{}^{}({\mathcal{N}})
      = \min\left\{-\sum_{k,j=1}^d p_jT_{k|j}\log T_{k|j}:p_j\in[0,1],\,\sum_j p_j=1\right\}.
    \end{align}
    Using the fact that $\sum_{j=1}^d p_j=1$, we get
    \begin{align}
      \sum_{k,j=1}^d p_jT_{k|j}\log T_{k|j}
      &= \sum_{j=1}^d p_j\underbrace{\sum_{k=1}^dT_{k|j}\log T_{k|j}}_{\leq\max_j \sum_{k=1}^d T_{k|j}\log T_{k|j}}
        \nonumber\\
      &\leq\left(\sum_{j=1} p_j\right)\max_{j\in\{1,2,\dotsc,d\}}\sum_{k=1}^d T_{k|j}\log T_{k|j}=\max_{j\in\{1,2,\dotsc,d\}}\sum_{k=1}^dT_{k|j}\log T_{k|j}\ ,
    \end{align}
    which implies that
    ${S}_{}^{}({\mathcal{N}})\geq\min_{j\in\{1,2,\dotsc,d\}}\left\{-\sum_{k=1}^dT_{k|j}\log
      T_{k|j}\right\}$. Furthermore, by picking the distribution $\rho_R$ such
    that $p_{j}=\delta_{j,j^{\star}}$, with
    $j^{\star}=\argmax_{j\in\{1,2,\dotsc,d\}}\sum_{k=1}^d T_{k|j}\log T_{k|j}$,
    we obtain
    ${S}_{}^{}({\mathcal{N}})\leq\min_{j\in\{1,2,\dotsc,d\}}\left\{-\sum_{k=1}^d
      T_{k|j}\log T_{k|j}\right\}$. This concludes the proof.
\end{proof}

We now show that the problem \eqref{z:sb5FfEOw}, in the
presence of constraint operators that are diagonal in the joint computational
basis of $B$ and $R$, reduces to a classical maximum channel entropy problem. 
Let $\mathcal{Z}(\cdot) = \sum_{k=1}^{d}\lvert {k}\rangle \mkern -1.8mu\relax \langle{k}\rvert (\cdot)\lvert {k}\rangle \mkern -1.8mu\relax \langle{k}\rvert $ be the dephasing channel with respect to the orthonormal basis $\{\lvert {k}\rangle \}_{k=1}^d$.
Suppose that the constraint operators $C_{BR}^j$ satisfy
\begin{align}
    C_{BR}^j = (\mathcal{Z}_B\otimes\mathcal{Z}_R)(C_{BR}^j)\ ,
\end{align}
for all $j$. This implies that
\begin{align}
  C_{BR}^j = \sum_{k,\ell=1}^d c_{k,\ell}^j\lvert {k,\ell}\rangle \mkern -1.8mu\relax \langle{k,\ell}\rvert \ ,
\end{align}
where $c_{k,\ell}^j=\langle {k,\ell}\rvert C_{BR}^j\lvert {k,\ell}\rangle $. The constraint
$\operatorname{tr}[C_{BR}^jN_{BR}]=q_j$ is then equivalent to
\begin{align}
    \operatorname{tr}[C_{BR}^jN_{BR}]
  = \sum_{k,\ell=1}^d c_{k,\ell}^j\langle {k,\ell}\rvert N_{BR}\lvert {k,\ell}\rangle  = q_j.
\end{align}
Now, observe that $\langle {k,\ell}\rvert N_{BR}\lvert {k,\ell}\rangle \geq 0$ for all
$k,\ell\in\{1,2,\dotsc,d\}$, and
$\sum_{k=1}^d\langle {k,\ell}\rvert N_{BR}\lvert {k,\ell}\rangle =\langle {\ell}\rvert _R\operatorname{tr}_B[N_{BR}]\lvert {\ell}\rangle _R=\langle {\ell}\mkern 1.5mu\relax \vert\mkern 1.5mu\relax {\ell}\rangle =1$
for all $\ell\in\{1,2,\dotsc,d\}$, where we used the fact that
$\operatorname{tr}_B[N_{BR}]=\mathds{1}_R$. This means that
$\langle {k,\ell}\rvert N_{BR}\lvert {k,\ell}\rangle \equiv T_{k|\ell}$ defines a stochastic
matrix. In other words, the constraint is equivalent to
\begin{align}
    \sum_{k,\ell=1}^d c_{k,\ell}^jT_{k|\ell} = q_j.
\end{align}
Furthermore, let $\Theta_{\mathcal{Z}}$ be the superchannel such that the Choi representation of $\Theta_{\mathcal{Z}}(\mathcal{N})$ is $(\mathcal{Z}_B\otimes\mathcal{Z}_R)(N_{BR})$, with $N_{BR}$ being the Choi representation of $\mathcal{N}$. Observe that $\Theta_{\mathcal{Z}}(\widetilde{\mathcal{D}})=\widetilde{\mathcal{D}}$. This fact, along with the data-processing inequality, implies that
\begin{align}
   {S}_{}^{}({\mathcal{N}})
  = -{D}_{}^{}({\mathcal{N}}\mathclose{}\,\Vert\,\mathopen{}{\widetilde{\mathcal{D}}})
  &\leq -{D}_{}^{}({\Theta_{\mathcal{Z}}(\mathcal{N})}\mathclose{}\,\Vert\,\mathopen{}{
        \Theta_{\mathcal{Z}}(\widetilde{\mathcal{D}})})
  = -D(\Theta_{\mathcal{Z}}(\mathcal{N})\Vert\widetilde{\mathcal{D}})
    \nonumber\\
  &=\min_{\ell\in\{1,2,\dotsc,d\}}\left\{-\sum_{k=1}^dT_{k|\ell}\log T_{k|\ell}\right\},
\end{align}
where for the final equality we used the fact that
$(\mathcal{Z}_B\otimes\mathcal{Z}_R)(N_{BR})=\sum_{k,\ell=1}^d
T_{k|\ell}\lvert {k,\ell}\rangle \mkern -1.8mu\relax \langle{k,\ell}\rvert $ along with
Proposition~\ref{z:A9iM-feK}. Our maximum classical channel
entropy problem is then equivalent to the following problem:
\begin{align}
  \label{z:HTRrXESa}
  \begin{aligned}[t] 
    \textup{maximize:} \quad
    & \min_{\ell\in\{1,2,\dotsc,d\}}\left\{-\sum_{k=1}^{d} T_{k|\ell}\log T_{k|\ell}\right\}
    \\
    \textup{over:}\quad
    & T_{k|\ell}\geq 0,\,\, \sum_{k=1}^{d} T_{k|\ell}=1\,\,\forall\,\ell
    \\
    \textup{such that:}\quad
    & \sum_{k,\ell=0}^{d}c_{k,\ell}^jT_{k|\ell}= q_j\quad\text{for \(j=1, \ldots, J\)}\ ,
  \end{aligned}
\end{align}
which is the classical analogue of our maximum channel entropy problem. The solution to this problem is a special case of \cref{z:TLiwdJGR}.

Problems similar to \eqref{z:HTRrXESa} have been considered before.
In \R\cite{R61,R62},
the entropy of a stochastic mapping (equivalently, a Markov chain transition matrix) is defined with respect to a fixed input distribution, or it is required that the input distribution is a stationary distribution of the stochastic mapping being optimized; see also \cite[Chapter~4]{R63}. The problem with a fixed input distribution is:
\begin{align}
  \label{z:5QSwGAcy}
  \begin{aligned}[t] 
    \textup{maximize:} \quad
    & -\sum_{k=1}^{d} p_kT_{k|\ell}\log T_{k|\ell}
    \\
    \textup{over:}\quad
    & T_{k|\ell}\geq 0,\,\, \sum_{k=1}^{d} T_{k|\ell}=1\,\,\forall\,\ell
    \\
    \textup{such that:}\quad
    & \sum_{k,\ell=1}^{d}c_{k,\ell}^jT_{k|\ell}= q_j\quad\text{for \(j=1, \ldots, J\)}\ ,
  \end{aligned}
\end{align}
where $\{p_k\}_{k=1}^d$ is the fixed input distribution. 
This problem is the classical analog of
problem~\eqref{z:.wBvf2qe}, and $T_{k|\ell}$
is a classical analog of the thermal
quantum channel with respect to a fixed $\phi_R$.
The solution to such a problem can be obtained as a special case of either \cref{z:hd7.cLe.} or \cref{z:33B55hFY}. Assume for simplicity that the initial distribution has full rank. We have also shown above that it suffices to optimize with respect to Choi matrices $N_{BR}$ such that $N_{BR}=\sum_{k,\ell=1}^d T_{k|\ell}\lvert {k,\ell}\rangle \mkern -1.8mu\relax \langle{k,\ell}\rvert $. Therefore, by examining the proof of \cref{z:hd7.cLe.}, we conclude that the operator $F_R$ in \eqref{z:8KUY6eJu} can be taken diagonal in the computational basis, i.e., $F_R=\sum_{k=1}^df_k\lvert {k}\rangle \mkern -1.8mu\relax \langle{k}\rvert $ for $f_k\in\mathbb{R}$. We also have $C_{BR}^j=\sum_{k,\ell=1}^d c_{k,\ell}^j\lvert {k,\ell}\rangle \mkern -1.8mu\relax \langle{k,\ell}\rvert $. Therefore, the Choi matrix $T_{BR}^{(\phi)}$ in \eqref{z:drLpOUuq} has the form
\begin{equation}
    T_{BR}^{(\phi)}=\sum_{k,\ell=1}^d p_{\ell}^{-2}\exp\left(p_{\ell}^{-1}\sum_{j=1}^d \mu_jc_{k,\ell}^j\right)\exp(-f_{\ell}p_{\ell}^{-1})\lvert {k,\ell}\rangle \mkern -1.8mu\relax \langle{k,\ell}\rvert ,
\end{equation}
where the $\mu_j$ are coefficients corresponding to the constraints $\sum_{k,\ell=1}^d c_{k,\ell}^j T_{k|\ell}=q_j$. Let
\begin{equation}
    Z_{\ell}\equiv\sum_{k=1}^d\exp\left(p_{\ell}^{-1}\sum_{j=1}^d\mu_j c_{k,\ell}^j\right).
\end{equation}
Then, the requirement $\operatorname{tr}_B[T_{BR}^{(\phi)}]=\mathds{1}_R$ implies that
\begin{equation}
    Z_{\ell}p_{\ell}^{-2}\exp(-f_{\ell}p_{\ell}^{-1})=1\quad\forall~\ell\in\{1,2,\dotsc,d\}.
\end{equation}
Imposing this requirement immediately leads to
\begin{equation}
    T_{BR}^{(\phi)}=\sum_{k,\ell=1}^d\frac{1}{Z_{\ell}}\exp\left(p_{\ell}^{-1}\sum_{j=1}^d\mu_j c_{k,\ell}^j\right)\lvert {k,\ell}\rangle \mkern -1.8mu\relax \langle{k,\ell}\rvert .
\end{equation}

\section{Learning algorithm for quantum channels}
\label{z:wXQ-uunR}

A prominent application of the maximum-entropy principle for quantum states is tomography --- in particular, the reconstruction of quantum states using ``incomplete'' knowledge, in the form of expectation-value estimates for a given set of observables~\cite{R64,R65,R66,R67,R68}. The maximum-entropy approach to state tomography mandates that our estimate of the quantum state should be the one that maximizes the entropy subject to the constraints corresponding to our expectation-value estimates.

Recent years have seen a resurgence in the idea of learning using incomplete
knowledge, with it being referred to as ``shadow tomography'', i.e., learning a
state in terms of its expectation values on a given set of observables, often
provided randomly from a known
ensemble~\cite{R11,R105}. This concept has been
combined with the maximum-entropy principle to obtain quantum state learning algorithms~\cite{R69,R70,R13,R14}.
These learning algorithms are based on an online procedure, in which a guess of the true quantum state is updated iteratively as more observable data becomes available. Suppose $\rho^{(t)}$ is a guess of the true state at time step $t\in\{1,2,\dotsc\}$ of the algorithm. Given a number of uses of the true state, a measurement of an observable $E^{(t)}$ is then made, and an estimate $s^{(t)}$ of the expectation value of this observable with respect to the true state is provided.
Using this estimate, 
an updated guess $\rho^{(t+1)}$ of the true state is obtained as a solution to the following optimization problem~\cite{R80,R69}:
\begin{align}
  \label{z:ihMgQosJ}
    \begin{aligned}
        \text{minimize:}\quad & {D}_{}^{}({\rho}\mathclose{}\,\Vert\,\mathopen{}{\rho^{(t)}})
        + \eta L_t(\rho) \\
        \text{subject to:}\quad & \rho\geq 0,\,\operatorname{tr}[\rho]=1,
    \end{aligned}
\end{align}
where $L_t(\rho)=(\operatorname{tr}[\rho E^{(t)}]-s^{(t)})^2$ is a loss function, which quantifies the error in the estimate $s^{(t)}$ compared to the expectation value of $E^{(t)}$ with respect to $\rho$. The ``learning rate'' $\eta>0$ models the tradeoff between keeping the new estimate close to the old one, represented by the first term in the objective function, and minimizing the loss in the second term. The optimization problem \eqref{z:ihMgQosJ} can be solved to obtain the following explicit update rule~\cite{R80,R69}:
\begin{equation}
    \rho^{(t+1)}=\frac{\exp(G^{(t)})}{\operatorname{tr}[\exp(G^{(t)})]},\quad G^{(t)}=\log(\rho^{(t)})-2\eta(\operatorname{tr}[\rho^{(t)}E^{(t)}]-s^{(t)})E^{(t)}.
\end{equation}
Under certain conditions on the learning rate $\eta$, this algorithm is guaranteed to converge to the true state as $t\to\infty$~\cite{R80,R69}.

Here, we consider the analogous learning problem for quantum channels. A prior work~\cite{R106} has applied the quantum state maximum-entropy principle to the Choi states of quantum channels. We go beyond this here by using the quantum channel relative entropy, which %
involves an optimization over all input states. We consider an online learning setting in which we are tasked with learning a quantum channel in a sequential manner. 
Specifically, given an arbitrary sequence of channel observables, our
algorithm iteratively updates a current guess, or estimate, of the unknown
channel as more observable data is made available.  At each iteration, our
learning algorithm estimates the expectation value of a given channel observable by making use the unknown channel a
fixed number of times. The estimate incurs a loss, depending how close it is to the true expectation value, and this loss is used to compute an
updated estimate of the unknown channel. The update rule is chosen such that
over many iterations, the estimate hopefully approaches the channel with maximal
entropy that is compatible with the measured data.

Concretely, our algorithm is as follows. It is a direct generalization of the
learning algorithm considered in \R\cite{R80,R69}
in the context of quantum state learning. 

\begin{algorithm}[H]
  \begin{algorithmic}[1]
    \Require {$\eta \in (0, 1)$; $\mathcal{M}^{(0)} = \mathcal{D}$.}
    \For{\texttt{$t = 1, 2, \ldots, T$}}
    \State Receive the observable $E_{RB}^{(t)}$.
    \State Obtain an estimate $s^{(t)}$ of the true expectation value.
    \State Update: $\mathcal{M}^{(t)}=\argmin\bigl\{{
      {D}_{}^{}({\mathcal{N}}\mathclose{}\,\Vert\,\mathopen{}{\mathcal{M}^{(t-1)}})
      + \eta L_t(\mathcal{N}):\mathcal{N}\textup{ cp. tp.}
    }\bigr\}$.
    \EndFor
    \Ensure { $\mathcal{M}^{(T)}$} 
  \end{algorithmic}
  \caption{Minimum relative entropy channel learning}
  \label{z:igCtH9MB}
\end{algorithm}

\Cref{z:igCtH9MB} consists of the following elements.

\begin{itemize}
    \item An initial guess of $\mathcal{M}^{(0)}=\mathcal{D}$, the completely depolarizing channel. This is the channel the maximizes the channel entropy in the absence of any prior knowledge, i.e., expectation-value estimates.

    \item Given an observable $E_{RB}^{(t)}$ at time step $t\in\{1,2,\dotsc\}$, the estimate $s^{(t)}$ is obtained via a running average, similar to \R\cite{R69}. Specifically,
        \begin{equation}\label{z:g841q3VJ}
            s^{(t)}=\frac{(n_{E^{(t)}}-1)s^{(t-1)}+\hat{s}^{(t)}}{n_{E^{(t)}}},
        \end{equation}
        where $n_{E^{(t)}}$ is the number of times $E_{BR}^{(t)}$ has appeared up to time $t$ and $\hat{s}^{(t)}$ is the empirical estimate of the expectation value at time step $t$, obtained using a given number of channel uses.

    \item In order to obtain an updated estimate of the unknown channel, our algorithm solves a special case of the general minimum channel relative entropy problem in \eqref{z:z4Hq5iYa}, namely:
        \begin{align}
          \label{z:8R8JZliG}
            \begin{aligned}
                \text{minimize:}\quad & {D}_{}^{}({\mathcal{N}_{A\to B}}\mathclose{}\,\Vert\,\mathopen{}{\mathcal{M}_{A\to B}^{(t)}})
                + \eta L_t(\mathcal{N}_{A\to B}) \\
                \text{subject to:}\quad & \mathcal{N}_{A\to B} \text{ cp. tp.,}
            \end{aligned}
        \end{align}
        where 
        \begin{equation}
            L_t(\mathcal{N}_{A\to B})\coloneqq\bigl({s^{(t)}-\operatorname{tr}[E_{BR}^{(t)}\mathcal{N}_{A\to B}(\Phi_{A:R})]}\bigr)^2
        \end{equation}
        is the loss function at time step $t\in\{1,2,\dotsc\}$, which simply computes the squared error of the estimate $s^{(t)}$
compared to the expectation value with respect to $\mathcal{N}_{A\to B}$. Note that this optimization problem is a direction generalization of the one in \eqref{z:ihMgQosJ}. Note also that we consider the quadratic term in \eqref{z:8R8JZliG} to model the loss, rather than an equality constraint, in order to account for the statistical fluctuations in the estimates $s^{(t)}$. 

\end{itemize}

As a proof-of-principle example, we apply
\cref{z:igCtH9MB} to the learning of single-qubit
channels. For this, we let
$\mathcal{S}=\{\lvert {0}\rangle \mkern -1.8mu\relax \langle{0}\rvert ,\lvert {1}\rangle \mkern -1.8mu\relax \langle{1}\rvert ,\lvert {\pm}\rangle \mkern -1.8mu\relax \langle{\pm}\rvert ,\lvert {\pm\I}\rangle \mkern -1.8mu\relax \langle{\pm\I}\rvert \}$
be the set of single-qubit stabilizer states and $\mathcal{P}=\{X,Y,Z\}$ be the
set of non-identity Pauli operators.  The channel observables are chosen %
of the
form $E_{BR}=P_B\otimes\rho_R$, where $P_B\in\mathcal{P}$ and
$\rho_R\in\mathcal{S}$. In every iteration of \Cref{z:igCtH9MB}, we make a uniformly random choice of
$P\in\mathcal{P}$ and $\rho\in\mathcal{S}$, take the learning rate to be $\eta=0.15$, obtain the empirical estimates $\hat{s}^{(t)}$ in \eqref{z:g841q3VJ} with $10\,000$ uses of the unknown channel, and solve the problem~\eqref{z:8R8JZliG}
numerically using the semidefinite programming techniques put forward in
\R\cite{R107,R57}. Our code makes use of
the QuTip~\cite{R108},
SciPy~\cite{R109} and CVXPY~\cite{R110} software
frameworks.

\begin{figure}
    \centering
    \includegraphics[width=0.48\textwidth]{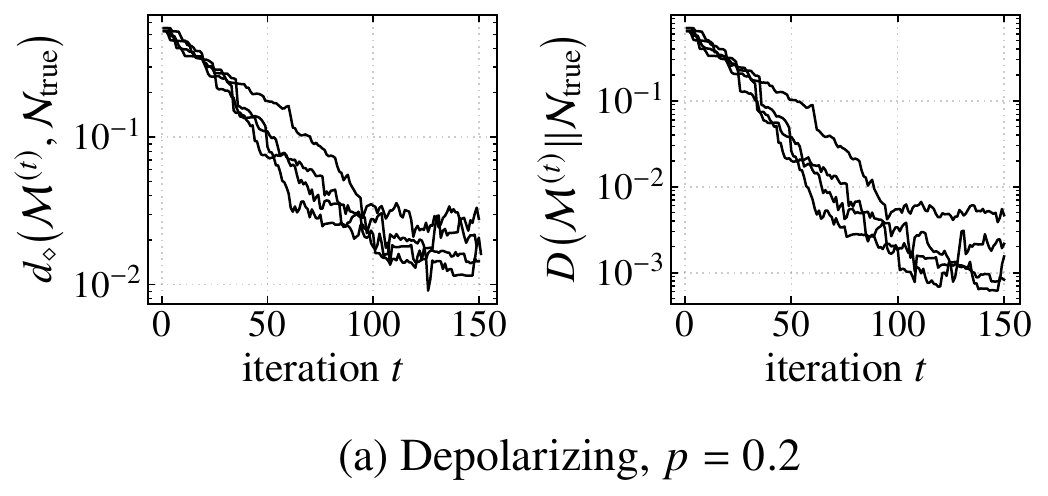}\quad\includegraphics[width=0.48\textwidth]{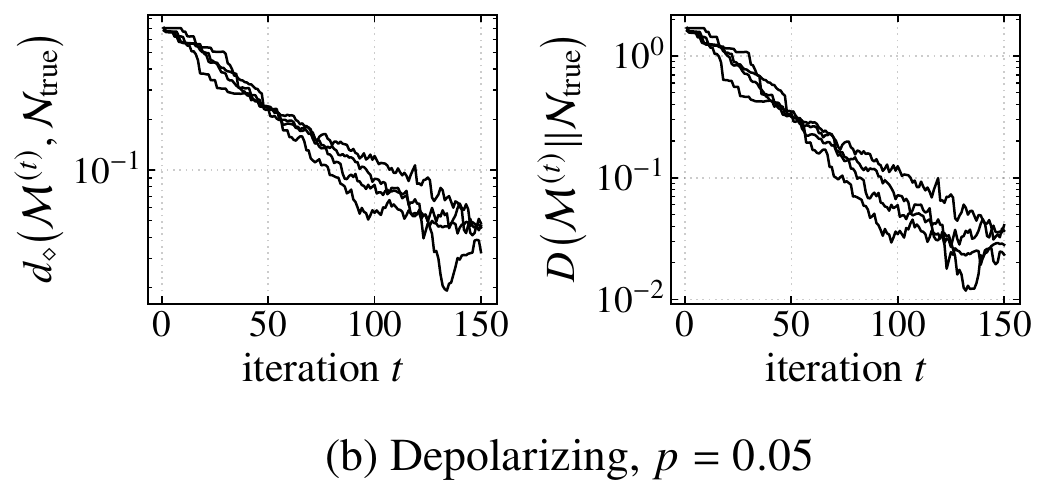}\\[3ex]\includegraphics[width=0.48\textwidth]{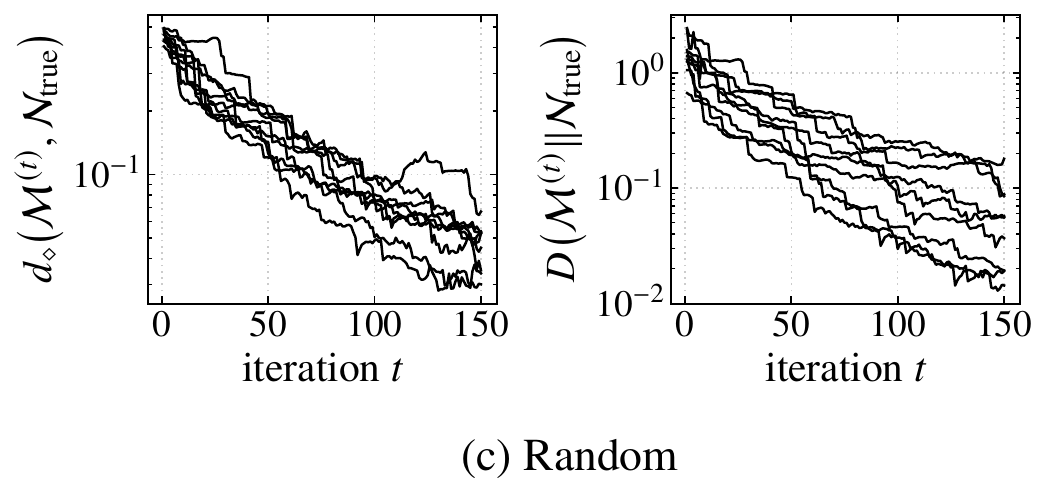}\quad\includegraphics[width=0.48\textwidth]{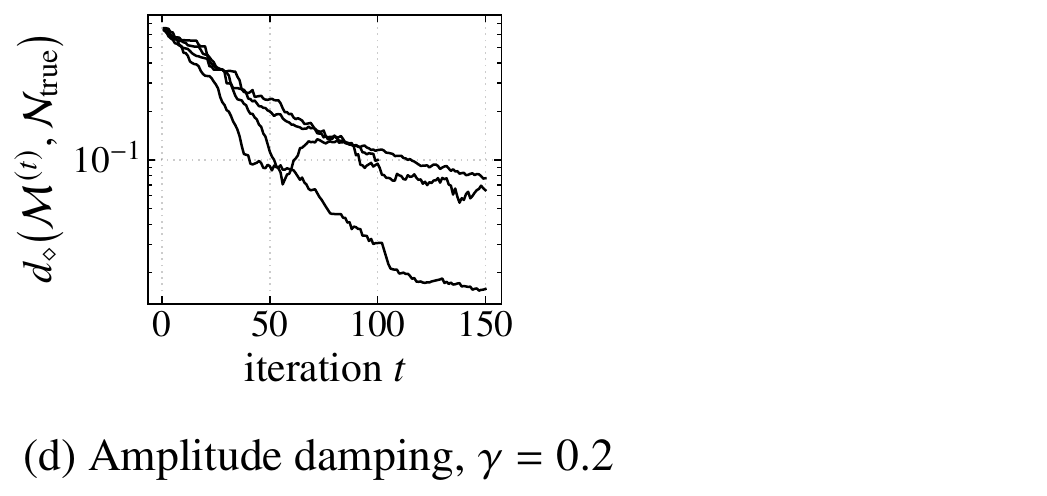}
    \caption{Learning of quantum channels using \cref{z:igCtH9MB}. In all cases, we take $\eta=0.15$ as the learning rate. We plot the diamond-norm distance $d_{\diamond}(\mathcal{M}^{(t)},\mathcal{N}_{\text{true}})=\frac{1}{2}\|\mathcal{M}^{(t)}-\mathcal{N}_{\text{true}}\|_{\diamond}$ between the channel $\mathcal{M}^{(t)}$ at every iteration of the algorithm and the true channel $\mathcal{N}_{\text{true}}$. We also plot the channel relative entropy, $D(\mathcal{M}^{(t)}\Vert\mathcal{N}_{\text{true}})$, of the same channels. We consider the cases that $\mathcal{N}_{\text{true}}$ is: (a) the depolarizing channel $\mathcal{D}_p$, defined in \eqref{z:K9m1mp5S}, with $p=0.2$; (b) the depolarizing channel with $p=0.05$; (c) randomly-generated channels; and (d) the amplitude-damping channel $\mathcal{A}_{\gamma}$, defined in \eqref{z:MHRlVog6}, with $\gamma=0.2$. For this channel, we have omitted the relative entropy plot because of the fact that it does not have full Kraus rank, and therefore the support condition required for a finite value of the relative entropy is not necessarily satisfied.}
    \label{z:Sss.OiJv}
\end{figure}

Our results are shown in \cref{z:Sss.OiJv}. We took as the true, unknown channel the depolarizing channel, amplitude-damping channel, and randomly-generated channels. The depolarizing and amplitude-damping channels are defined as 
\begin{align}
    \mathcal{D}_p(\cdot)&=(1-p)(\cdot)+\frac{p}{3}(X(\cdot)X+Y(\cdot)Y+Z(\cdot)Z),\label{z:K9m1mp5S}\\
    \mathcal{A}_{\gamma}(\cdot)&=K_1(\cdot)K_1^{\dagger}+K_2(\cdot)K_2^{\dagger},\quad K_1=\begin{pmatrix} 1 & 0\\ 0 & \sqrt{1-\gamma}\end{pmatrix},\,K_2=\begin{pmatrix} 0 & \sqrt{\gamma} \\ 0 & 0 \end{pmatrix}.\label{z:MHRlVog6}
\end{align}
The random qubit-to-qubit channels are defined by random Choi matrices, which we generate as follows~\cite{R111}. We generate a $4\times 4$ random complex matrix $G$ by sampling the real and imaginary parts of every matrix element of $G$ from the standard normal distribution. We then let $P_{BR}\equiv GG^{\dagger}$ and $Q_R=\operatorname{tr}_B[P_{BR}]$, such that our desired Choi matrix is $N_{BR}=Q_R^{-1/2}P_{BR}Q_R^{-1/2}$. 
We calculated both the diamond-norm distance between the guess and the true channel at every iteration, and also their channel relative entropy.
Our results
seem to indicate %
that \Cref{z:igCtH9MB} defines a sequence of guesses that
converges to the true channel in the limit of a large number of iterations.
Our study here is meant to serve as an initial proof
of concept, while a rigorous analysis of the algorithm's convergence rate, error
bounds, and other algorithmic guarantees goes beyond the scope of the present
paper.
The convergence guarantees of the state learning
algorithms~\cite{R80,R69,R13} rely on the fact that the relative entropy is a
so-called \emph{Bregman divergence}~\cite{R112}; it remains
unclear whether the channel entropy has the same property.
Therefore, it may be the case that proving the %
 convergence of
\Cref{z:igCtH9MB} could require %
 a different, or entirely new technique.

\section{Microcanonical derivation of the thermal channel}
\label{z:VbRRCMzN}

Here, we present an alternative derivation of the thermal channel: We generalize
to quantum channels the argument that a thermal state of a system $S$ can be
expressed as the reduced state on $S$ of a joint microcanonical state on $S$ and
a large heat bath.
We show that the thermal channels we obtain in this way coincide with the
thermal channels that we defined in
\cref{z:u4ZpmFvZ}.

An adaptation of the standard argument in statistical mechanics to derive the
thermal state on $S$ from the microcanonical state on a larger system proceeds
as follows.  Consider the heat bath to be an additional $n-1$ copies of $S$, for
some large $n$.  (Perhaps $S$ is a single particle of a large, $n$-particle gas
evolving as a closed system.)  For simplicity, suppose that the particles are
completely noninteracting, leading to the total Hamiltonian
$H_{\mathrm{tot}} = H_1 + H_2 + \cdots H_n$ with $H_i$ the system Hamiltonian
applied to the $i$-th particle.  We define the microcanonical subspace at energy
$[E,E+\Delta E]$ as the subspace spanned by all energy eigenstates of
$H_{\mathrm{tot}}$ whose energies lie in $[E, E+\Delta E]$.  Now assume that the
global state is a microcanonical state $\pi(E, \Delta E)$ at energy
$[E,E+\Delta E]$, which assigns an equal probability weight to all states in the
corresponding microcanonical subspace.  Using standard typicality arguments, one
can show that the reduced state $\rho_1$ on a single copy of $S$ obeys
\begin{align}
  \rho_1 \approx \frac{{e}^{-\beta H_1}}{Z}\ ,
  \label{z:vdMGaB8Q}
\end{align}
where $\beta$ can be determined from $E$ from the known energy density
$\operatorname{tr}(\rho_1 H_1) = E/n$.  It is one of the keystone results of statistical
mechanics and information theory that the forms of the canonical
state~\eqref{z:vdMGaB8Q} and the constrained maximum entropy
state~\eqref{z:Jdr9-JJH} coincide.

The argument above carries over to the case where multiple conserved charges
$Q^{(1)}, \ldots Q^{(J)}$ are present, rather than the energy $H$ alone.  If the
charges all commute, then the microcanonical subspace can be defined as the one
spanned by the simultaneous eigenvectors of all the charge operators whose
eigenvalue associated with charge $Q^{(j)}$ lies in a fixed interval
$[ q_j, q_j + \Delta q_j ]$.  This case typically arises when we construct the
grand canonical ensemble in statistical mechanics, where the charges are energy
and particle number.

On the other hand, more involved proof techniques are required in the case where
the conserved charges fail to
commute~\cite{R26}.  In this case, we cannot
define a microcanonical subspace from simultaneous eigenspaces of the charge
operators, as these do not necessarily exist.  Instead, one can resort to an
\emph{approximate microcanonical subspace}, constructed as
follows~\cite{R26}.  Given noncommuting charges
$Q^{(1)}$, \ldots, $Q^{(j)}$, we can construct their $n$-copy versions
$\bar{Q}^{(j)} = (1/n)\sum_{i=1}^n Q^{(j)}_i$, where $Q^{(j)}_i$ represents the
$j$-th charge operator applied on the $i$-th copy of the system.  It turns out
that the $\{{ \bar{Q}^{(j)} }\}$ approximately
commute~\cite{R113}.  Furthermore, it is possible to find
a subspace of the $n$-copy system with the following properties: (i)~Any quantum
state $\rho$ with high weight in the subspace has sharp statistics for
$\bar{Q}^{(j)}$ around $q_j$, for all $j$; (ii)~Any quantum state $\rho$ with
sharp statistics for charge $\bar{Q}^{(j)}$ around $q_j$ (for all $j$) has high
weight in the subspace.  Here, we say that $\rho$ has sharp statistics for a
charge $\bar{Q}^{(j)}$ around $q_j$ if the measurement outcome distribution of
$\bar{Q}^{(j)}$ on $\rho$ has weight at least $1-\delta$ in the window
$[q_j - \eta, q_j + \eta]$, for suitable tolerance parameters $\delta, \eta$.
Such a subspace is called an \emph{approximate microcanonical subspace}.  It
captures approximately all quantum states that have sharp statistics
simultaneously for all the charges, providing an approximate version of the
microcanonical subspace in the case of commuting observables.
The maximally mixed state in this subspace is referred to as an
\emph{approximate microcanonical state}.
In \R\cite{R26}, it was shown that the reduced
state on a single system of the approximate microcanonical state is close to the
thermal state~\eqref{z:Jdr9-JJH}.

Here, we adapt this argument to the context of quantum channels.  Consider $n$
input systems $A^n$, $n$ output systems $B^n$, and let $R^n \simeq A^n$.  Let
$\{{ C^j_{BR} }\}_{j=1}^J$ be a collection of channel observables, and let
$q_j \in \mathbb{R}$ for $j=1, \ldots J$.  The channel observables represent
``conserved channel charges.''  Recall a channel observable is meant to be
measured against the Choi matrix
$N_{BR} \equiv \mathcal{N}_{A\to B}(\Phi_{A:B})$ of a channel
$\mathcal{N}_{A\to B}$, yielding the expectation value
$\operatorname{tr}\bigl[{ C^j_{BR} N_{BR} }\bigr]$.  Loosely speaking, the $R$ system of a channel
observable may be interpreted as being fed into the channel's input, and the
channel's output is measured against the $B$ part of the charge; this
interpretation is accurate if the channel observable is of the form
$Q^j_B \otimes \rho^j_R$.

We need to identify a physical measurement that can test whether or not a given
channel satisfies the desired constraints.  For any full-rank $\sigma_R$, a
given constraint operator $C^j_{BR}$ can always be written in the form
$C^j_{BR} = \sigma_R^{1/2} H^{j,\sigma}_{BR} \sigma_R^{1/2}$ with
$H^{j,\sigma}_{BR} = \sigma_R^{-1/2} C^j_{BR} \sigma_R^{-1/2}$.
A physical experiment whose expectation value reveals the constraint operator
$C^j_{BR}$'s expectation value consists in preparing
$\lvert {\sigma}\rangle _{AR} \equiv \sigma_R^{1/2}\lvert {\Phi_{A:R}}\rangle $, applying the channel on
$A\to B$, and measuring $H^{j,\sigma}_{BR}$.  Indeed:
\begin{align}
  \operatorname{tr}\bigl[{ H^{j,\sigma}_{BR} \mathcal{N}({\sigma_{AR}})}\bigr]
  = 
  \operatorname{tr}\bigl[{ \sigma_R^{-1/2} C^{j}_{BR} \sigma_R^{-1/2}\,
  \mathcal{N}({\sigma_{R}^{1/2} \Phi_{A:R} \sigma_R^{1/2}})}\bigr]
  = \operatorname{tr}\bigl[{ C^{j}_{BR}\, \mathcal{N}({\Phi_{A:R}})}\bigr]\ .
\end{align}

Importantly, the constraint operator $C^j_{BR}$ alone provides no guidance as to
which input state $\sigma_R$ was meant to be used to test the constraint.  Any
full-rank input state $\sigma_R$ can be used in the construction above.

Even more importantly, a general quantum channel $\mathcal{E}_{A^n \to B^n}$ can
distinguish i.i.d.\@ states arbitrarily well in the limit $n\to\infty$ and its
action can therefore differ significantly on different i.i.d.\@ inputs.  Testing
the constraint solely on a fixed i.i.d.\@ input would, therefore, allow the
channel to act freely on all other i.i.d.\@ states.  Therefore, we need to
ensure the constraints are tested \emph{for all input states}, at least in the
limit $n\to\infty$.

We now construct an $n$-copy measurement with respect to an arbitrary input
state $\sigma_R$ to test the constraint $C^j_{BR}$ on
$\mathcal{E}_{A^n\to B^n}$.
Let $\sigma_R$ be any full-rank quantum state and let
$\lvert {\sigma}\rangle _{AR} \equiv \sigma_R^{1/2} \lvert {\Phi_{A:R}}\rangle $.  We prepare the state
$\lvert {\sigma}\rangle _{AR}^{\otimes n}$ and we send the copies of $A$ through
$\mathcal{E}_{A^n\to B^n}$, resulting in the state
$\mathcal{E}_{A^n\to B^n}(\sigma_{AR}^{\otimes n})$.  We now measure the
operator $H^{j,\sigma}_{BR}$ on each copy and compute the sample average of the
outcomes.  This procedure is equivalent to measuring a global observable
$\overline{H^{j,\sigma}}_{B^nR^n}$ on the resulting state
$\mathcal{E}_{A^n\to B^n}(\sigma_{AR}^{\otimes n})$, where
\begin{align}
  \overline{H^{j,\sigma}}_{B^nR^n}
  &=
    \frac1n\sum_{i=1}^n 
    ({\mathds{1}_{BR}})^{\otimes(i-1)}
    \otimes ({\sigma_{R_i}^{-1/2} C_{B_iR_i}^j \sigma_{R_i}^{-1/2}})
    \otimes ({\mathds{1}_{BR}})^{\otimes(n-i)}\ .
  \label{z:RVr9.4VA}
\end{align}

We use an overline notation to represent the $n$-sample average observable; specifically,
for an observable $O_A$, we write
$\overline{O}_{A^n} \equiv (1/n)\sum_{i=1}^n \mathds{1}_A^{\otimes (i-1)} \otimes O_{A_i} \otimes \mathds{1}_A^{\otimes (n-i)}$ as the sample average observable associated with $O$ and which
is an operator on $A^n$.

We may now sketch our generalization of the approximate microcanonical subspace
to quantum channels.
We identify a POVM effect $P_{B^nR^n}$, which we term \emph{approximate
  microcanonical channel operator}, with the following properties
($\eta,\delta,\epsilon,\eta',\delta',\epsilon'>0$ are tolerance parameters):
\begin{enumerate}[label=(\alph*)]
\item Suppose a quantum channel $\mathcal{E}_{A^n\to B^n}$ satisfies
  $\operatorname{tr}\bigl[{P_{B^nR^n} \, \mathcal{E}_{A^n\to B^n}(\sigma_{AR}^{\otimes n}) }\bigr] \geq
  1-\epsilon$ for ``most'' states $\sigma$, where
  $\lvert {\sigma}\rangle _{AR} \equiv \sigma_R^{1/2}\lvert {\Phi_{A:R}}\rangle $; then for all $j$ and
  for ``most'' $\sigma$,
  \begin{align}
    \operatorname{tr}\Bigl[{\Bigl\{{ \overline{H^{j,\sigma}}_{BR} \in [q_j - \eta, q_j + \eta] }\Bigr\} \;
    \mathcal{E}_{A^n\to B^n}(\sigma_{AR}^{\otimes n}) }\Bigr]
    \geq 1 - \delta\ ,
    \label{z:gjO6FWMB}
  \end{align}
  where $\bigl\{{ C^j_{BR} \in [q_j - \eta, q_j + \eta] }\bigr\}$ denotes the projector
  onto the eigenspaces of $C^j_{BR}$ associated with eigenvalues within $\eta$
  of $q_j$.

\item Suppose a quantum channel $\mathcal{E}_{A^n\to B^n}$ satisfies
  $\operatorname{tr}\Bigl[{\Bigl\{{ \overline{H^{j,\sigma}}_{BR} \in [q_j - \eta, q_j + \eta] }\Bigr\}
  \; \mathcal{E}_{A^n\to B^n}(\sigma_{AR}^{\otimes n}) }\Bigr] \geq 1 - \delta'$ for
  all $j$ and for ``most'' states $\sigma$.  Then
  \begin{align}
    \operatorname{tr}\bigl[{P_{B^nR^n} \, \mathcal{E}_{A^n\to B^n}(\sigma_{AR}^{\otimes n}) }\bigr]
    \geq 1-\epsilon\ ,
    \label{z:sYU6AmKw}
  \end{align}
  also for  ``most'' states $\sigma$.
\end{enumerate}
The conditions above do not hold for states $\sigma$ that have very small
eigenvalues.  Specifically, the sets of states designated vaguely above as
``most states'' are defined as sets of all quantum states whose eigenvalues are
above a suitable threshold.  The threshold can be made arbitrarily small at the
cost of loosening the other tolerance parameters.  All these parameters along
with the thresholds can be taken to go to zero for large $n$.

The construction of an approximate microcanonical channel operator is a first
result presented in this section:
\begin{theorem*}[Construction of an approximate microcanonical channel operator; informal]
  There exists an explicit construction of an approximate microcanonical channel
  operator $P_{B^nR^n}$, which is furthermore permutation invariant.
\end{theorem*}
A formal statement appears as \cref{z:JjqyeW8p} below.

The reason that we should not consider $\sigma$ with minuscule eigenvalues is
the following.  The observable $\sigma_R^{-1/2} C^j_{BR} \sigma_R^{-1/2}$ that
appears in \eqref{z:RVr9.4VA} has a norm that can diverge as the
smallest nonzero eigenvalue of $\sigma_R$ goes to zero.  The statistics of such
an observable can fluctuate wildly: When estimating the expectation value of
this observable over a finite number of samples, a single low-probability
outcome with a very large measurement result can significantly influence the
sample average.  This poses an issue for conditions of the
form~\eqref{z:gjO6FWMB}, which state that
the measurement statistics of such observables are sharp.  This issue does not
arise if we are guaranteed an upper bound on the norm of
$\sigma_R^{-1/2} C^j_{BR} \sigma_R^{-1/2}$; such a guarantee can be enforced by
ensuring that all eigenvalues of $\sigma_R$ are above some threshold.

Armed with an approximate microcanonical channel operator $P_{B^nR^n}$, we can
define a \emph{microcanonical channel}.  We define the microcanonical channel as
the channel with maximal channel entropy that has high weight with respect to
$P_{B^nR^n}$.  This definition mirrors the property of a microcanonical state
being the most entropic among all states supported on the microcanonical
subspace.  We show that a microcanonical channel leads to thermal channels in
the following sense: If we apply the microcanonical channel on $n$ copies of a
fixed state $\phi_{AR}$, then the reduced state on the first system pair $AR$ is
close to the state obtained by applying a thermal channel with respect to
$\phi$, denoted $\mathcal{T}_{A\to B}^{(\phi)}$, onto $\phi$ [cf.\@
\cref{z:.wBvf2qe}]:

\begin{theorem*}[Thermal channels from a microcanonical channel; informal]
  Let $\Omega_{A^n\to B^n}$ be a permutation-invariant microcanonical channel.
  For any full-rank state $\phi_R$, let
  $\lvert {\phi}\rangle _{AR} = \phi_R^{1/2}\lvert {\Phi_{A:R}}\rangle $.  Then
  \begin{align}
    \operatorname{tr}_{2, \ldots, n}\Bigl\{{ \Omega_{A^n\to B^n}[\phi_{AR}^{\otimes n} ] }\Bigr\}
    \approx
    \mathcal{T}_{A\to B}^{(\phi)}(\phi_{AR})\ .
  \end{align}
\end{theorem*}
A formal statement appears as \cref{z:1tRx46YH}
below.

The remainder of this section is devoted to a precise formulation and careful
proof of both the above theorems.  Our proofs are inspired by an alternative
construction of the approximate microcanonical subspace presented in
\R\cite{R71}.

As an intermediate step, we present a custom, ``constrained,'' postselection
theorem for channels that is likely of independent interest.  Namely, we extend
standard postselection techniques~\cite{R72,R73,R74,R75,R76} to a channel version in
which a permutation-invariant channel is operator-upper-bounded by an integral
over i.i.d.\@ channels, where the integrand further includes a fidelity term of
the i.i.d.\@ channel with the original channel.

First, we present in \cref{z:05ekSHxY}
our custom postselection theorem.
We then detail in \cref{z:-1cmqcJ7} the definition of an
approximate microcanonical channel operator.  As a first warm-up result, we show
in \cref{z:Z6.vONxH} that an approximate
microcanonical channel operator acts as a channel analog of a typical projector
for a thermal channel: It always assigns high weight to the $n$-fold tensor
product of a thermal channel associated with the same charge values $q_j$.
In \cref{z:RH9GpkuM}, we show how to recover the
thermal channels derived in \cref{z:u4ZpmFvZ} from
the microcanonical channel.
We finally dive in \cref{z:STrCXktf} into the details of
our construction of an approximate microcanonical channel operator.

\subsection{A constrained channel postselection theorem}
\label{z:05ekSHxY}

An intermediate result in this section can be of independent interest in the
context of the theory of i.i.d.\@ channels in quantum information theory.
Specifically, we prove a tighter (``constrained'') version of a postselection
theorem~\cite{R72,R74,R73,R75,R76} for quantum channels, in
which the integrand of the upper bound in the postselection operator inequality
includes a fidelity term, generalizing the constrained state postselection
theorems in~\cite[Appendix~B]{R74} and
\R\cite{R75} as well as the channel postselection
theorem in~\cite[Corollary~3.3]{R73}.

To state our postselection theorem, we introduce the following \emph{de Finetti
  state}:
\begin{align}
  \zeta_{R^n} = \operatorname{tr}_{A^n}\biggl[{ \int d\psi_{AR}\, \lvert {\psi}\rangle \mkern -1.8mu\relax \langle{\psi}\rvert _{AR} }\biggr]\ ,
  \label{z:4bxW5WdC}
\end{align}
where the integration is carried out of the measure on the pure states
$\lvert {\psi}\rangle _{AR}$ of $AR$ induced by the Haar measure on $\mathrm{U}({d_Ad_R})$, and where
the measure is normalized in such a way that $\operatorname{tr}\bigl({\zeta_{R^n}}\bigr) = 1$.
The de Finetti state appears in quantum versions of de Finetti's
theorem~\cite{R114,R115,R116} and in the postselection
technique~\cite{R72}.

\begin{theorem}[Constrained channel postselection theorem]
  \label{z:UEUfoDLC}
  Let $A,B$ be quantum systems and let $R\simeq A$.  Let $n>0$.  There exists a
  universal measure $d\mathcal{M}_{A\to B}$ on quantum channels $A\to B$ such
  that for any permutation-invariant quantum channel $\mathcal{E}_{A^n\to B^n}$,
  and for any permutation-invariant operators $X_{R^n}, Y_{R^n}$,
  \begin{align}
    \hspace*{3em}&\hspace*{-3em}
    X_{R^n}^\dagger Y_{R^n}\,
    E_{B^nR^n}\,
    Y_{R^n}^\dagger X_{R^n}
    \nonumber\\
    &\leq \operatorname{poly}(n)
    \int dM_{BR}
    \,
      M_{BR}^{\otimes n}
    \,
    F^2\Bigl({
    \mathcal{M}_{A\to B}^{\otimes n}\bigl({ X_{R^n} \zeta_{A^nR^n} X^\dagger_{R^n} }\bigr)
      \, ,\,
    \mathcal{E}_{A^n\to B^n}\bigl({ Y_{R^n} \zeta_{A^nR^n} Y_{R^n}^\dagger }\bigr)
    }\Bigr)\ ,
  \end{align}
  where $M_{BR} \equiv \mathcal{M}_{A\to B}(\Phi_{A:R})$ is the Choi matrix of
  $\mathcal{M}_{A\to B}$, where $dM_{BR}$ is the measure on Choi matrices
  corresponding to the channel measure $d\mathcal{M}_{A\to B}$, where
  $E_{B^nR^n} \equiv \mathcal{E}_{A^n\to B^n}({\Phi_{A^n:R^n}})$ is the Choi
  matrix of $\mathcal{E}_{A^n\to B^n}$, and where
  $\lvert {\zeta}\rangle _{A^nR^n} \equiv \zeta_{R^n}^{1/2}\,\lvert {\Phi_{A^n:R^n}}\rangle $ with
  $\zeta_{R^n}$ defined in~\eqref{z:4bxW5WdC}.
\end{theorem}

The arguments of the fidelity term can be also be reformulated in terms of the
Choi matrices $M_{BR}$ and $E_{B^nR^n}$ as
$X_{R^n}\zeta_{R^n}^{1/2}\,M_{BR}^{\otimes n}\,\zeta_{R^n}^{1/2} X^\dagger_{R^n}$
and $Y_{R^n}\zeta_{R^n}^{1/2}\,E_{B^nR^n}\,\zeta_{R^n}^{1/2} Y^\dagger_{R^n}$,
respectively.

A suitable choice of the operators $X_{R^n}$, $Y_{R^n}$ can help derive upper
bounds on the fidelity term by influencing the inputs to
$\mathcal{M}^{\otimes n}$ and $\mathcal{E}$.  We can choose, for instance,
$X_{R^n}$, $Y_{R^n}$ to be typical projectors with respect to some state of
interest or projectors onto selected Schur-Weyl blocks.  A suitable choice for
these operators enables us to derive the following corollary, suitable for upper
bounding the application of a permutation-invariant channel on an arbitrary
i.i.d.\@ input state:

\begin{corollary}
  \label{z:pnP5Wiy4}
  Let $\mathcal{E}_{A^n\to B^n}$ be any permutation-invariant quantum
  channel. Let $\sigma_R$ be any state and let
  $\lvert {\sigma}\rangle _{AR} = \sigma_R^{1/2}\lvert {\Phi_{A:R}}\rangle $.  Let $w>0$.  Then there
  exists $\Delta_{B^nR^n}\geq 0$ with
  $\operatorname{tr}({\Delta_{B^nR^n}}) \leq \operatorname{poly}({n}) {e}^{-n w/2}$ such that
  \begin{align}
    \mathcal{E}_{A^n\to B^n}\bigl({\sigma_{AR}^{\otimes n}}\bigr)
    \leq
    \operatorname{poly}(n) \biggl[{\int \! dM_{BR} \, \mathcal{M}^{\otimes n}\bigl({\sigma_{AR}^{\otimes n}}\bigr)
    \max_{\substack{\tau_R:\\ F(\tau_R, \sigma_R)\geq e^{-w}}}
    F^2\bigl({ \mathcal{M}^{\otimes n}\bigl({\tau_{AR}^{\otimes n}}\bigr) ,
        \mathcal{E}\bigl({\tau_{AR}^{\otimes n}}\bigr) }\bigr) }\biggr]
    \ +\     \Delta_{B^nR^n}
    \ ,
  \end{align}
  where $M_{BR} \equiv \mathcal{M}_{A\to B}(\Phi_{A:R})$ and where
  $\lvert {\tau}\rangle _{AR} \equiv \tau_R^{1/2} \lvert {\Phi_{A:R}}\rangle $.
\end{corollary}

We prove \cref{z:UEUfoDLC} and
\cref{z:pnP5Wiy4} in
\cref{z:WfwOwBXs,z:6QVwEcrI}.

We also provide proofs of two statements that are used in the proof of
\cref{z:UEUfoDLC}, but which can be of independent
interest and which we state for reference.  To a large extent, they are part of
the field's folklore and follow directly from other well-known results; cf.\@ in
particular \R\cite{R117,R83,R118}.
A first lemma simply determines the block-diagonal structure of the de Finetti
state~\eqref{z:4bxW5WdC} in the Hilbert space structure
imposed by Schur-Weyl duality.  A brief introduction to Schur-Weyl duality,
along with relevant definitions and notation conventions, appear in
\cref{z:pJbOynb9}.  To understand the following lemma at this stage, it
suffices to know that $\{{ \Pi^\lambda_{R^n} }\}_\lambda$ are a set of orthogonal
projectors with $\sum_{\lambda} \Pi^\lambda_{R^n} = \mathds{1}_{R^n}$, where
$\lambda$ ranges over an index set that we denote by $\Young(d_R, n)$;
furthermore, $d_{\mathcal{Q}_\lambda}, d_{\mathcal{P}_\lambda}$ are positive
integers with
$\operatorname{tr}({\Pi^\lambda_{R^n}}) = d_{\mathcal{Q}_\lambda} d_{\mathcal{P}_\lambda}$ and
$d_{\mathcal{Q}_\lambda} \leq \operatorname{poly}(n)$.
\begin{lemma}[Schur-Weyl structure of the de~Finetti state]
  \label{z:TvA2E9KA}
  The de Finetti state has the following decomposition in Schur-Weyl blocks:
  \begin{align}
    \zeta_{R^n} &=
    \frac1{ d_{\mathrm{Sym}(n,d_R^2)} }
    \sum_{\lambda\in\Young(d_R, n)}
    \frac{d_{\mathcal{Q}_\lambda}}{d_{\mathcal{P}_\lambda}} \Pi^\lambda_{R^n}\ ,
  \end{align}
  where $d_{\mathrm{Sym}(n,d_R^2)}$ is the dimension of the symmetric subspace
  of $n$ copies of $\mathbb{C}^{d_R^2}$.
\end{lemma}

A second intermediate claim used in the proof of
\cref{z:UEUfoDLC} concerns a specific average over random
unitaries.  Specifically, we consider a nonnormalized pure state
$\lvert {\Psi^0}\rangle _{SR}$ over two systems $S,R$, such that
$\operatorname{tr}_S[{\Psi^0_{SR}}] = \mathds{1}_R$.  Such an operator could be the Choi matrix of
an isometric quantum channel.  We compute the average, over all unitaries $W_S$
according to the Haar measure, of the $n$-fold tensor product of the rotated
state $W_S\Psi^0_{SR} W_S^\dagger$.  This average can be viewed as a channel
version of the average in~\eqref{z:4bxW5WdC} that defines the
de~Finetti state.  In the following proposition, $\Pi^{\mathrm{Sym}}_{(SR)^n}$
denotes the symmetric subspace of $({\mathscr{H}_S\otimes\mathscr{H}_R})^{\otimes n}$, i.e., the
subspace spanned by all states $\lvert {\psi}\rangle _{(SR)^n}$ that are invariant under any
permutation of the copies of the system $(SR)$.
\begin{proposition}[Haar twirl of an isometric channel's Choi matrix]
  \label{z:7r5n0QDD}
  Let $S$, $R$ be any quantum systems with $d_S\geq d_R$, and let $n>0$.  Let
  $\lvert {\Psi^0}\rangle _{SR}$ be any ket such that $\operatorname{tr}_S[{ \Psi^0_{SR} }] = \mathds{1}_R$.
  Then
  \begin{align}
    \int dW_{S} \, W_{S}^{\otimes n}\,
    [{\Psi^0_{SR}}]^{\otimes n} \, W_{S}^{\otimes n\,\dagger}
    =
    \Pi^{\mathrm{Sym}}_{(SR)^n}
    \sum_{\lambda\in\Young(d_R,n)}
    \frac{d_{\mathcal{P}_\lambda}}{d_{\mathcal{Q}_\lambda}}
    \Pi_{R^n}^\lambda
    = d_{\mathrm{Sym}(n,d_R^2)}^{-1}  \zeta_{R^n}^{-1} \,
    \Pi_{(SR)^n}^{\mathrm{Sym}} \ ,
    \label{z:axznnWkb}
  \end{align}
  further noting that $\bigl[{ \zeta_{R^n}, \Pi_{(SR)^n}^{\mathrm{Sym}} }\bigr] = 0$.
\end{proposition}

\subsection{Definition of an approximate microcanonical channel operator}
\label{z:-1cmqcJ7}

We aim to define an approximate microcanonical channel operator in such a way
that it can identify channels $\mathcal{E}_{A^n\to B^n}$ displaying suitably
sharp statistics with respect to the constraint operators $C^j_{BR}$.
Specifically, we might demand that the observable
$\overline{H^{j,\sigma}}_{B^nR^n}$ defined in~\eqref{z:RVr9.4VA} has
sharp statistics around $q_j$ on the state
$\mathcal{E}_{A^n\to B^n}(\sigma_{AR}^{\otimes n})$, for any
$\lvert {\sigma}\rangle _{AR} = \sigma_R^{1/2}\lvert {\Phi_{A:R}}\rangle $ and for any $j$.  This
condition cannot hold, however, for all $\sigma_R$:  If $\sigma_R$ has nearly
vanishing eigenvalues, the norm of the observable
$\sigma_R^{-1/2} C^j_{BR} \sigma_R^{-1/2}$ can diverge to infinity, which in
turn can prevent the concentration of the outcomes of
$\sigma_R^{-1/2} C^j_{BR} \sigma_R^{-1/2}$ at large $n$.  (This can be seen, for
instance, in Hoeffding's bound: The exponent in the upper bound on the tail
probability depends on the inverse square of the range of values a random
variable can take.)  To remedy this issue, we ask that the observable
$\overline{H^{j,\sigma}}_{B^nR^n}$ has sharp statistics on
$\mathcal{E}_{A^n\to B^n}(\sigma_{AR}^{\otimes n})$ for any state $\sigma_R$
that satisfies $\sigma_R\geq y\mathds{1}$ for some fixed threshold value $y$, i.e.,
all eigenvalues of $\sigma_R$ are greater than or equal to $y$.  This assumption
ensures that $\overline{H^{j,\sigma}}_{B^nR^n}$ has bounded norm: For any
$\sigma_R\geq y\mathds{1}$, we find
\begin{align}
  \bigl \lVert { \overline{H^{j,\sigma}}_{B^nR^n} }\bigr \rVert 
  \leq
  \bigl \lVert { \sigma_R^{-1/2}\,C^{j}_{BR}\,\sigma_R^{-1/2} }\bigr \rVert 
  \leq \bigl \lVert {C^j_{BR}}\bigr \rVert \,\bigl \lVert {\sigma_R^{-1/2}}\bigr \rVert ^2
  = \frac{ \lVert {C^j_{BR}}\rVert  }{ \lambda_{\mathrm{min}}(\sigma) }
  \leq y^{-1}\lVert {C^j_{BR}}\rVert \ .
  \label{z:iA0NNPDu}
\end{align}
The lower the threshold value $y$ is chosen, the more states $\sigma_R$ the
condition holds for; yet the slower $\overline{H^{j,\sigma}}_{B^nR^n}$
concentrates in $n$.  In the limit $n\to \infty$, we can take $y\to 0$, meaning
that the condition includes all full-rank states $\sigma_R$.

\begin{definition}[Approximate microcanonical channel operator]
  \label{z:Caf8GVY-}
  An operator $P_{B^nR^n}$ satisfying $0\leq P \leq \mathds{1}$ is called an
  \emph{$(\eta,\epsilon,\delta, y, \nu, \eta',\epsilon',\delta', y',
    \nu')$-approximate microcanonical channel operator} with respect to
  $\{{ (C^j_{BR}, q_j) }\}$ if the following two conditions hold.  The conditions
  are formulated in terms of
  $P_{B^nR^n}^\perp \equiv \mathds{1}_{B^nR^n} - P_{B^nR^n}$ and use the shorthand
  $\lvert {\sigma_{AR}}\rangle  \equiv \sigma_R^{1/2}\lvert {\Phi_{A:R}}\rangle $ for any $\sigma_R$:
  \begin{enumerate}[label=(\alph*)]
  \item For any channel $\mathcal{E}_{A^n\to B^n}$ such that
    \begin{align}
      \max_{\sigma_{R}\geq y\mathds{1}} \operatorname{tr}\bigl[{
      P_{B^n R^n}^\perp \, \mathcal{E}_{A^n\to B^n}\bigl({\sigma_{AR}^{\otimes n}}\bigr)
      }\bigr] &\leq \epsilon\ ,
    \end{align}
    then for all $j=1, \ldots, J$,
    \begin{align}
      \max_{\sigma_{R}\geq \nu y\mathds{1}} \operatorname{tr}\Bigl[{
        \Bigl\{{ \overline{H^{j,\sigma}}_{B^nR^n} \notin [q_j\pm \eta ] }\Bigr\} \,
        \mathcal{E}_{A^n\to B^n}(\sigma_{AR}^{\otimes n})
        }\Bigr] &\leq \delta\ ,
        \label{z:APDcifUu}
    \end{align}
    where $\bigl\{{ X \notin I }\bigr\}$ denotes the projector onto the eigenspaces of a
    Hermitian operator $X$ associated with eigenvalues not in a set
    $I\subset\mathbb{R}$.

  \item %
    For any channel $\mathcal{E}_{A^n\to B^n}$ such that
    \begin{align}
      \max_{\sigma_{R}\geq y'\mathds{1}} \operatorname{tr}\Bigl[{
        \Bigl\{{ \overline{H^{j,\sigma}}_{B^nR^n} \notin [q_j\pm \eta' ] }\Bigr\} \,
        \mathcal{E}_{A^n\to B^n}(\sigma_{AR}^{\otimes n})
        }\Bigr] &\leq \delta' \qquad\text{for all}\ j=1,\ldots, J\ ,
    \end{align}
    then
    \begin{align}
      \max_{\sigma_{R}\geq \nu' y'\mathds{1}} \operatorname{tr}\bigl[{
      P_{B^nR^n}^\perp \mathcal{E}_{A^n\to B^n}(\sigma_{AR}^{\otimes n})
      }\bigr] &\leq \epsilon'\ .
          \label{z:msGKg5m1}
    \end{align}
  \end{enumerate}
\end{definition}

In order for this definition to make sense, the parameters of the approximate
microcanonical channel operator should satisfy
\begin{align}
  \begin{aligned}
  0 &< \eta \leq \frac{2}{y}\lVert {C^j_{BR}}\rVert \ ;
  &
    0 &< \epsilon \leq 1\ ;
  &
    0 &< \delta \leq 1\ ;
  &
    0 &< y < 1/(\nu d_R)\ ;
  &
    \nu &> 0\ ;
  \\
  0 &< \eta' \leq \frac{2}{y}\lVert {C^j_{BR}}\rVert \ ;
  &
    0 &< \epsilon' \leq 1\ ;
  &
    0 &< \delta' \leq 1\ ;
  &
    0 &< y' < 1/(\nu' d_R)\ ;
  &
    \nu' &> 0\ .
  \end{aligned}
\end{align}

\subsection{The approximate microcanonical channel operator identifies i.i.d.\@
  channels with correct constraints}
\label{z:Z6.vONxH}

As a first warm-up lemma, we show that our notion of approximate microcanonical
channel operator attributes high weight to the $n$-fold tensor product of a
channel that satisfies all the constraints specified by $\{{C^j_{BR}, q_j}\}$.  We
can think of an approximate microcanonical channel operator as a test that
accepts any i.i.d.\@ channel that is feasible
in~\eqref{z:sb5FfEOw}.  This property holds in particular for
the thermal channels defined via maximum-channel-entropy principles in
\cref{z:7suOBm0G}.

\begin{lemma}[Approximate microcanonical channel operators capture i.i.d.\@
  channels with compatible constraints]
  \label{z:LEyJcHwi}
  Let $P_{B^nR^n}$ be an
  $(\eta, \epsilon, \delta, y, \nu, \eta', \epsilon', \delta', y', \nu')$-approximate
  microcanonical channel operator. 
  Let $\mathcal{N}_{A\to B}$ be any channel such that
  $\operatorname{tr}\bigl[{ C^j_{BR} \mathcal{N}_{A\to B}(\Phi_{A:R}) }\bigr] = q_j$ for
  $j=1, \ldots, J$.  Assuming that
  $2\lVert {C^j_{BR}}\rVert ^2\,\log({2/\delta'}) \leq n \eta'^2 y'^2$ for all
  $j=1,\ldots, J$, then
  \begin{align}
    \max_{\sigma_{R}\geq \nu' y'\mathds{1}}
    \operatorname{tr}\bigl[{ P_{B^nR^n}^\perp \,
        \mathcal{N}_{A\to B}^{\otimes n} (\sigma_{AR}^{\otimes n})
    }\bigr]
    \leq \epsilon'\ .
    \label{z:ZnVyvQV1}
  \end{align}
\end{lemma}

\begin{proof}[**z:LEyJcHwi]
  Let $\sigma_R \geq y'\mathds{1}$ and write
  $\lvert {\sigma}\rangle _{AR} = \sigma_R^{1/2}\lvert {\Phi_{A:R}}\rangle $.  Measuring
  $\overline{H^{j,\sigma}}_{B^nR^n}$ on the state $ \mathcal{N}_{A\to B}
  ^{\otimes n}(\sigma_{AR}^{\otimes n})$ corresponds
  to measuring $\sigma_R^{-1/2} C^{j}_{BR} \sigma_R^{-1/2}$ on each individual
  copy of $\mathcal{N}_{A\to B}(\sigma_{AR})$
  and computing the sample average of the outcomes.  The average of the
  single-copy outcome random variable is simply
  $\operatorname{tr}\bigl[{\sigma_R^{-1/2} C^{j}_{BR} \sigma_R^{-1/2} \,
    \mathcal{N}_{A\to B}(\sigma_{AR})
  }\bigr] = \operatorname{tr}\bigl[{ C^{j}_{BR}\,
    \mathcal{N}_{A\to B}(\Phi_{A:R})
  }\bigr] = q_j$.  From Hoeffding's inequality,
  \begin{align}
    \operatorname{tr}\Bigl({
       \bigl\{{ \overline{H^{j,\sigma}}_{B^nR^n} \notin [q_j \pm \eta'] }\bigr\} \,
        \bigl[{\mathcal{N}_{A\to B}(\sigma_{AR})}\bigr]^{\otimes n}
    }\Bigr)
    &\leq 2\exp\Biggl({-\frac{2 \,\eta'^2\, n}{
      4\, \bigl \lVert {\sigma_R^{-1/2} C^j_{BR} \sigma_R^{-1/2}}\bigr \rVert ^2}}\Biggr)
    \\
    &\leq 2\exp\Biggl({-\frac{\eta'^2\, y'^2\, n}{ 2\, \bigl \lVert {C^j_{BR}}\bigr \rVert ^2}}\Biggr)
    \leq \delta'\ ,
  \end{align}
  using~\eqref{z:iA0NNPDu} and where the last inequality follows
  from the additional assumption in the proposition statement.
  The defining properties of the approximate microcanonical channel operator
  finally guarantees
  that~\eqref{z:ZnVyvQV1} holds.
\end{proof}

Our construction for $P_{B^nR^n}$, detailed in
\cref{z:STrCXktf} below, has an even stronger property:
Not only does it correctly identify any i.i.d.\@ channel with the correct
constraints, but it also correctly rejects any i.i.d.\@ channel with a
constraint that is violated.

\subsection{Thermal channel from a microcanonical channel}
\label{z:RH9GpkuM}

Given an approximate microcanonical channel operator, we can define an channel
analogue of the microcanonical state.  Recall that given a microcanonical
subspace, we define the microcanonical state as the maximally mixed state within
that subspace.  Equivalently, it is the maximally entropic state that is
supported within the microcanonical subspace.  We extend this definition to
channels:

\begin{definition}
  \label{z:v1wYNQz5}
  Let $P_{B^nR^n}$ be a
  $(\eta, \epsilon, \delta, y, \nu, \eta', \epsilon', \delta', y',
  \nu')$-approximate microcanonical channel operator with respect to
  $\{{ (C^j_{BR}, q_j) }\}$.  Then the associated \emph{approximate microcanonical
    channel} is defined as the channel $\Omega_{A^n\to B^n}$ that maximizes the
  channel entropy ${S}_{}^{}({\Omega_{A^n\to B^n}})$ subject to the constraint
  \begin{align}
    \max_{\sigma_{R}\geq y\mathds{1}} \operatorname{tr}\bigl[{
        P^\perp_{B^nR^n} \, \Omega_n\bigl({
            \sigma_{AR}^{\otimes n}
        }\bigr)
    }\bigr]
    \leq \epsilon\ .
    \label{z:Y7YnuhFN}
  \end{align}
\end{definition}

The following theorem statement makes reference to the thermal channel
$\mathcal{T}_{A\to B}^{(\phi)}$ with respect to a state $\phi$, defined in
\cref{z:0QjUzwvU}.

\begin{theorem}[The microcanonical channel resembles the thermal channel on a single copy]
  \label{z:1tRx46YH}
  Let $\Omega_n$ be a approximate microcanonical channel associated with a
  $(\eta,\epsilon,\delta,y,\nu, \eta',\epsilon',\delta',y',\nu')$-approximate
  microcanonical channel operator $P_{B^nR^n}$, and let
  \begin{align}
    \omega_{BR}
    = \frac1n \sum_{i=1}^n\operatorname{tr}_{n\setminus i}\bigl[{\Omega_n\bigl({\phi_{AR}^{\otimes n}}\bigr)}\bigr]\ ,
  \end{align}
  where $\operatorname{tr}_{n\setminus i}$ denotes the partial trace over all copies of $(BR)$
  except $(BR)_i$.  Let $\phi_R>0$ be any full-rank state with
  $\lambda_{\mathrm{min}}(\phi_R) \geq \nu y$ and
  $\lambda_{\mathrm{min}}(\phi_R) \geq y'$ and let $\mathcal{T}_{A\to B}^{(\phi)}$ be the
  thermal channel with respect to $\phi$.
  Assume that $2\lVert {C^j_{BR}}\rVert ^2\,\log({2/\delta'}) \leq n \eta'^2 y'^2$ for
  all $j=1,\ldots,J$.  Additionally, we assume that $\epsilon' \leq \epsilon$.
  Then
  \begin{align}
    {D}_{}^{}\bigl ({ \omega_{BR} }\mathclose{}\,\big \Vert\,\mathopen{}{
    \mathcal{N}_{\mathrm{th}}(\phi_{AR}) }\bigr )
    \leq
    \sum \mu_j \bigl({\eta + 2y^{-1}\bigl \lVert {C^j_{BR}}\bigr \rVert \,\epsilon}\bigr)\ .
  \end{align}
\end{theorem}

If the approximate microcanonical channel operator $P_{B^nR^n}$ is
permutation-invariant, then the approximate microcanonical channel $\Omega_n$
can also be chosen to be permutation-invariant.  In this case, $\omega_{BR}$ is
simply the reduced state of $\Omega_n(\phi_{AR}^{\otimes n})$ on any of the $n$
copies of $BR$,
\begin{align}
  \omega_{BR} = \operatorname{tr}_{n-1}\bigl[{ \Omega_n(\phi_{AR}^{\otimes n}) }\bigr]\ .
\end{align}
Our construction of an approximate microcanonical channel operator, which we
detail further below, has this property.

\begin{proof}[**z:1tRx46YH]
  This proof is inspired by an analogous statement for quantum states in
  \R\cite{R26}.
  From the definition of the relative entropy,
  \begin{align}
    {D}_{}^{}\bigl ({ \omega_{BR} }\mathclose{}\,\big \Vert\,\mathopen{}{
         \mathcal{T}_{A\to B}^{(\phi)}(\phi_{AR})
    }\bigr )
    &= -S\bigl({ \omega_{BR} }\bigr)
    - \operatorname{tr}\Bigl({ \omega_{BR} \,
    \log\bigl[{ \mathcal{T}_{A\to B}^{(\phi)}(\phi_{AR}) }\bigr] }\Bigr)\ .
      \label{z:xSVFMXXz}
  \end{align}
  \Cref{z:hd7.cLe.} asserts that the maximum-entropy
  thermal channel $\mathcal{T}_{A\to B}^{(\phi)}$ with respect to a full-rank state
  $\phi_R$ obeys, for some Hermitian operator $F_R$ and real values $\mu_j$,
  \begin{subequations}
    \begin{align}
      \mathcal{T}_{A\to B}^{(\phi)}(\phi_{AR})
      &= \exp\Bigl\{{ \phi_R^{-1/2}\Bigl[{\mathds{1}_B\otimes F_R - \sum\mu_j C^j_{BR} }\Bigr]\phi_R^{-1/2} }\Bigr\}\ ;
        \label{z:LbB.WJ3B}
      \\[1ex]
      {S}_{\phi}^{}\bigl ({\mathcal{T}_{A\to B}^{(\phi)}}\bigr )
      &= \sum \mu_j q_j -\operatorname{tr}\bigl({F_R}\bigr) - {S}_{}^{}({\phi_R})\ .
        \label{z:4qQyS3JJ}
    \end{align}
  \end{subequations}
  Consider the second term in~\eqref{z:xSVFMXXz}.  We find
  \begin{align}
    \operatorname{tr}\Bigl({ \omega_{BR}\,
    \log\bigl[{ \mathcal{T}_{A\to B}^{(\phi)}({\phi_{AR}}) }\bigr] }\Bigr)
    &= \frac1n \operatorname{tr}\Bigl[{ \Omega_n\bigl({\phi_{AR}^{\otimes n}}\bigr)\,
      \sum_{i=1}^n \phi_{R_i}^{-1/2}
      \Bigl({\mathds{1}_{B_i}\otimes F_{R_i} - \sum_{j=1}^J \mu_j C^j_{(BR)_i}}\Bigr) \phi_{R_i}^{-1/2}
      }\Bigr]
      \nonumber\\
    &= \frac1n\sum_{i=1}^n
      \operatorname{tr}({ F_{R_i} })
      - \frac1n\sum_{i=1}^n \sum_{j=1}^J \mu_j \operatorname{tr}\Bigl[{ \Omega_n\bigl({\phi_{AR}^{\otimes n}}\bigr)
      \,\phi_{R_i}^{-1/2} C^j_{(BR)_i}\phi_{R_i}^{-1/2} }\Bigr]
    \nonumber\\
    &=
      \operatorname{tr}({ F_{R} })
      - \sum_{j=1}^J \mu_j \operatorname{tr}\Bigl[{ \Omega_n\bigl({\phi_{AR}^{\otimes n}}\bigr)
      \, \overline{H^{j,\phi_R}}_{B^nR^n} }\Bigr]\ ,
      \label{z:4dxypuAD}
  \end{align}
  where we used $\Omega_n^\dagger(\mathds{1}_{B^n}) = \mathds{1}_{A^n}$ in the second
  equality.  Using our assumption that $P_{B^nR^n}$ is an approximate
  microcanonical channel operator along
  with~\eqref{z:Y7YnuhFN}, we have that
  $\operatorname{tr}\bigl[{\Omega_n\bigl({\phi_{AR}^{\otimes n}}\bigr) \,
  \overline{H^{j,\phi_R}}_{B^nR^n} }\bigr]$ must concentrate around $q_j$ for each
  $j$ [cf.~\eqref{z:APDcifUu}].
  Specifically, let
  \begin{align}
    h_j = \operatorname{tr}\bigl[{\Omega_n\bigl({\phi_{AR}^{\otimes n}}\bigr) \, \overline{H^{j,\phi_R}}_{B^nR^n} }\bigr] - q_j\ ;
  \end{align}
  now with $R = \bigl\{{ \overline{H^{j,\phi_R}}_{B^nR^n} \in [{ q_j \pm \eta }] }\bigr\}$
  and
  $R^\perp = \bigl\{{ \overline{H^{j,\phi_R}}_{B^nR^n} \notin [{ q_j \pm \eta
    }] }\bigr\} = \mathds{1}- R$, we have
  $\bigl \lVert { \bigl({\overline{H^{j,\phi_R}}_{B^nR^n} - q_j\mathds{1}}\bigr) R }\bigr \rVert  \leq \eta$ and
  \begin{align}
    \bigl \lvert { h_j }\bigr \rvert 
    &= 
    \Bigl \lvert {
      \operatorname{tr}\bigl[{\Omega_n\bigl({\phi_{AR}^{\otimes n}}\bigr) \,
      \bigl({\overline{H^{j,\phi_R}}_{B^nR^n} - q_j}\bigr)R}\bigr]
    + \operatorname{tr}\bigl[{\Omega_n\bigl({\phi_{AR}^{\otimes n}}\bigr) \,
      \bigl({\overline{H^{j,\phi_R}}_{B^nR^n} - q_j}\bigr) R^\perp}\bigr]
    }\Bigr \rvert 
      \nonumber\\
    &\leq
      \bigl \lVert {\bigl({\overline{H^{j,\phi_R}}_{B^nR^n} - q_j}\bigr) R}\bigr \rVert \,
      + \bigl \lVert {\overline{H^{j,\phi_R}}_{B^nR^n} - q_j}\bigr \rVert \,
      \operatorname{tr}\bigl[{\Omega_n\bigl({\phi_{AR}^{\otimes n}}\bigr) R^\perp}\bigr]
      \nonumber\\
    &\leq
      \eta + \bigl \lVert {\overline{H^{j,\phi_R}}_{B^nR^n}}\bigr \rVert \,\epsilon
      + \bigl \lVert {C^{j}_{BR}}\bigr \rVert \,\epsilon
      \leq
      \eta + 2 y^{-1}\bigl \lVert {C^{j}_{BR}}\bigr \rVert \,\epsilon
      \ ,
  \end{align}
  where in arriving at the third line we used the crude inequality
  $\bigl \lVert {\overline{H^{j,\phi_R}}_{B^nR^n} - q_j}\bigr \rVert  \leq 
  \bigl \lVert {{H^{j,\phi_R}}_{B^nR^n}}\bigr \rVert  + \bigl \lVert {C^{j}_{BR}}\bigr \rVert $.

  Consider now the first term in~\eqref{z:xSVFMXXz}.  Using the
  concavity and the subadditivity of the von Neumann entropy, and recalling the
  expression for the channel entropy in terms of the state von Neumann entropy
  ${S}_{}^{}({\mathcal{N}}) = \min_{\lvert {\phi'}\rangle _{AR}} \bigl[{
  {S}_{}^{}({\mathcal{N}(\phi'_{AR})}) - {S}_{}^{}({\phi'_R}) }\bigr]$, we find
  \begin{align}
    {S}_{}^{}\bigl ({\omega_{BR}}\bigr )
    \geq
    \frac1n\, \sum_{i=1}^n {S}_{}^{}\Bigl ({\operatorname{tr}_{n\setminus i}\bigl[{ \Omega_n\bigl({\phi_{AR}^{\otimes n}}\bigr) }\bigr]}\Bigr )
    \geq
    \frac1n\,{S}_{}^{}\Bigl ({\Omega_n\bigl[{\phi_{AR}^{\otimes n}}\bigr]}\Bigr )
    \geq
    \frac1n {S}_{}^{}\bigl ({\Omega_n}\bigr ) + {S}_{}^{}\bigl ({\phi_{AR}}\bigr )\ .
  \end{align}
  Now recall that $\Omega_n$ maximizes ${S}_{}^{}\bigl ({\Omega_n}\bigr )$ subject to the
  condition~\eqref{z:Y7YnuhFN}.  Another channel that
  satisfies condition~\eqref{z:Y7YnuhFN} is
  $\bigl[{\mathcal{T}_{A\to B}^{(\phi)}}\bigr]^{\otimes n}$, thanks to
  \cref{z:LEyJcHwi} as well as our additional
  assumption that $\epsilon'\leq\epsilon$.  Therefore,
  \begin{align}
    \frac1n {S}_{}^{}\bigl ({\Omega_n}\bigr )
    \geq
    \frac1n {S}_{}^{}\bigl ({\mathcal{T}_{A\to B}^{(\phi)}}\bigr )
    = {S}_{}^{}\bigl ({\mathcal{T}_{A\to B}^{(\phi)}}\bigr ) \ ,
  \end{align}
  using the additivity of the channel entropy under tensor products.

  Combining the above, we find
  \begin{align}
    {D}_{}^{}\bigl ({ \omega_{BR} }\mathclose{}\,\big \Vert\,\mathopen{}{ \mathcal{T}_{A\to B}^{(\phi)}(\phi_{AR}) }\bigr )
    \leq - {S}_{}^{}\bigl ({\mathcal{T}_{A\to B}^{(\phi)}}\bigr )  - {S}_{}^{}\bigl ({\phi_{AR}}\bigr )
    - \operatorname{tr}\bigl({ F_R }\bigr) + \sum \mu_j \bigl({q_j + h_j}\bigr)\ .
  \end{align}
  Plugging in~\eqref{z:4qQyS3JJ} yields
  \begin{align}
    {D}_{}^{}\bigl ({ \operatorname{tr}_{n-1}\bigl[{\phi_{AR}^{\otimes n}}\bigr] }\mathclose{}\,\big \Vert\,\mathopen{}{ \mathcal{T}_{A\to B}^{(\phi)}(\phi_{AR}) }\bigr )
    \leq
    \sum \mu_j h_j \leq
    \sum \mu_j \bigl({ \eta + 2y^{-1}\bigl \lVert {C^{j}_{BR}}\bigr \rVert \,\epsilon }\bigr)\ ,
  \end{align}
  as claimed.
\end{proof}

\subsection{Construction of an approximate microcanonical channel operator}
\label{z:STrCXktf}

We now present an explicit construction of an approximate microcanonical channel
operator.  This construction can be viewed as an extension to quantum channels
of the alternative construction in \R\cite{R71} of an
approximate microcanonical subspace for quantum states.
We define the operator $P_{B^nR^n}$ as the effective POVM outcome associated
with a specific protocol producing the output ``SUCCESS.''  The protocol
additionally depends on a parameter $m$ (with $0<m<n$) and on a condition
function $\chi(\tilde\sigma,\boldsymbol j,\boldsymbol z)$ (which takes values in
$\{{0,1}\}$) that we define and specify later.
The protocol proceeds as follows:
\begin{enumerate}[label={\arabic*.}]
\item[0.] For a better intuitive understanding of this protocol, we imagine the
  input state is $\mathcal{E}_{A^n\to B^n}(\sigma_{AR}^{\otimes n})$ with
  $\lvert {\sigma}\rangle _{AR}=\sigma_R^{1/2}\lvert {\Phi_{A:R}}\rangle $.  This input state is, however,
  not a part of the protocol that technically defines $P_{B^nR^n}$;

\item We randomly permute all copies of $(BR)$, with the effect of symmetrizing
  the input state;

\item \label{z:PrLOj3lZ} %
  We use $m$ out of the $n$ copies of $(BR)$ to run a suitable state estimation
  procedure on the input registers $R^m$, arriving at an approximation
  $\tilde\sigma_R$ of the actual input state $\sigma_R$;

\item For each of the remaining $\bar{n} \equiv n-m$ copies,
  $i = 1, \ldots, \bar{n}$, we pick $j_i \in \{{ 1, \ldots J}\}$ independently and
  uniformly at random (these correspond to random choices of measurement
  settings).  %

\item On each copy of those remaining $\bar{n}$ copies of $(BR)$ labeled by
  $i=1, \ldots, \bar{n}$, we measure the Hermitian observable
  $\tilde\sigma_R^{-1/2} C_{BR}^{j_i} \tilde\sigma_R^{-1/2}$, obtaining the
  outcome $z_i \in \mathbb{R}$.  %

\item A condition
  $\chi(\tilde\sigma, \boldsymbol j, \boldsymbol z) \in \{{ 0, 1 }\}$ is tested on
  the measurement outcomes $\boldsymbol z$, the estimated input state
  $\tilde\sigma$, the randomly sampled measurement settings $\boldsymbol j$, and
  parameters such as $\eta$, $q_j$.  If
  $\chi(\tilde\sigma, \boldsymbol j, \boldsymbol z) = 1$, we output ``SUCCESS.''
  Otherwise, we output ``FAILURE.''
\end{enumerate}

The approximate microcanonical channel operator $P_{B^nR^n}$ we construct is
obtained by a specific choice of a condition function
$\chi(\tilde\sigma, \boldsymbol j, \boldsymbol z)$ to be defined soon below. 
However, we first need to prove some properties of protocols of the above form
for other condition functions, to form important building blocks for our proofs.

For any final condition $\chi(\tilde\sigma, \boldsymbol j, \boldsymbol z)$, we
can write the operator $P^{\chi}_{B^nR^n}$ resulting from the above protocol as:
\begin{align}
  P^\chi_{B^nR^n}
  &= \mathcal{S}_{(BR)^n}\Biggl\{{
    \int d\tilde\sigma \, R_{R^m}^{(\tilde\sigma)\dagger}R_{R^m}^{(\tilde\sigma)}
    \otimes
    \frac1{J^{\bar n}} \sum_{\boldsymbol j} \int d\boldsymbol z\,
    \chi(\tilde\sigma,\boldsymbol j, \boldsymbol z)\,
    \biggl[{\bigotimes_{i=1}^n
    \bigl\{{ \tilde\sigma_{R_i}^{-\frac12} C^{j_i}_{B_iR_i} \tilde\sigma_{R_i}^{-\frac12} = z_i }\bigr\}}\biggr]
    }\Biggr\}
    \ ,
    \label{z:qVquPoYP}
\end{align}
where:
\begin{itemize}
\item $\boldsymbol j \equiv (j_i)_{i=1}^{\bar n}$ with $j_i \in \{{1, \ldots, J}\}$
  for $i = 1, \ldots \bar{n}$;
\item $\boldsymbol z \equiv (z_i)_{i=1}^{\bar n}$ with $z_i\in\mathbb{R}$ for
  $i = 1, \ldots, \bar{n}$;
\item
  $R_{R^m}^{(\tilde\sigma)} = (\tilde\sigma_R^{\otimes m})^{1/2} \zeta_{R^m}$ is
  such that
  $\bigl\{{ R_{R^m}^{(\tilde\sigma)\dagger } R_{R^m}^{(\tilde\sigma)}
  }\bigr\}_{\tilde\sigma}$ is a pretty good measurement on $R^m$
  associated with the family of states $\{{ \tilde\sigma^{\otimes m} }\}$ (cf.\@
  \cref{z:Ehi5guwR});
\item $\mathcal{S}_{(BR)^n}\bigl\{{ \,\cdot\, }\bigr\}$ is the quantum channel that
  randomly permutes the copies of $(BR)^n$.
\end{itemize}

Let us first prove some elementary properties of $P^{\chi}_{B^nR^n}$ for general
unspecified condition functions
$\chi(\tilde\sigma, \boldsymbol j, \boldsymbol z)$:

\begin{lemma}[Elementary properties of $P^{\chi}_{B^nR^n}$]
  \label{z:-xGoRZBj}
  The following properties hold:
  \begin{enumerate}[label=(\roman*)]
    \item \label{z:pLOBrNeb} We have
    $0 \leq P^\chi_{B^nR^n} \leq \mathds{1}_{B^nR^n}$.  Furthermore, for
    $\chi_{\mathrm{yesss}}(\tilde\sigma,\boldsymbol j,\boldsymbol z) = 1$ then
    $P^{\chi_{\mathrm{yesss}}}_{B^nR^n} = \mathds{1}_{B^nR^n}$;
  \item \label{z:NxY3iV0.} The operator
    $P^\chi_{B^nR^n}$ is linear in $\chi$: If
    $\chi(\tilde\sigma,\boldsymbol j, \boldsymbol z) = a \chi_1(\tilde\sigma,\boldsymbol j,
    \boldsymbol z) + b \chi_2(\tilde\sigma,\boldsymbol j, \boldsymbol z)$ for
    $a,b\in\mathbb{R}$, then
    $P^\chi_{B^nR^n} = a P^{\chi_1}_{B^nR^n} + b P^{\chi_2}_{B^nR^n}$;
  \item \label{z:W4OXReVA} The operator
    $P^\chi_{B^nR^n}$ obeys a monotonicity property in $\chi$: If
    $\chi_1(\tilde\sigma,\boldsymbol j, \boldsymbol z) \leq \chi_2(\tilde\sigma,\boldsymbol
    j, \boldsymbol z)$, then $P^{\chi_1}_{B^nR^n} \leq P^{\chi_2}_{B^nR^n}$;
  \end{enumerate}
\end{lemma}
\begin{proof}[**z:-xGoRZBj]
    For any $\chi$, the operator $P^\chi_{B^nR^n}$ is positive semidefinite by
  definition.  The linearity and monotonicity of $P^\chi_{B^nR^n}$ in $\chi$ are
  also immediate from its definition.  With
  $\chi_{\mathrm{boring}}(\tilde\sigma,\boldsymbol j, \boldsymbol z) \equiv 1$, we
  find
  \begin{align}
      P^{\chi_{\mathrm{boring}}}_{B^nR^n}
    &=
      \int d\tilde\sigma\,
      R^{(\tilde\sigma) \dagger} R^{(\tilde\sigma)}
      \otimes
      \frac1{J^{\bar n}}\sum_{\boldsymbol j}
      \Biggl[{\bigotimes_{i=1}^{\bar n}
      \underbrace{
      \int dz\,
      \bigl\{{ \tilde\sigma_{R_i}^{-\frac12} C^{j_i}_{B_iR_i} \tilde\sigma_{R_i}^{-\frac12} = z }\bigr\}
      }_{=\mathds{1}_{B_iR_i}}}\Biggr]
      \nonumber\\
    &= \int d\tilde\sigma\,
      R^{(\tilde\sigma) \dagger} R^{(\tilde\sigma)}
      \otimes \mathds{1}_{B^{\bar n}R^{\bar n}}
      = \mathds{1}_{B^nR^n}\ .
  \end{align}
  Since for any $\chi$, we have
  $\chi(\tilde\sigma,\boldsymbol j, \boldsymbol z) \leq 1 =
  \chi_{\mathrm{boring}}(\tilde\sigma,\boldsymbol j, \boldsymbol z)$, the above
  facts imply that $0 \leq P^\chi_{B^nR^n} \leq \mathds{1}_{B^nR^n}$.  We have
  established~\ref{z:pLOBrNeb},
  \ref{z:NxY3iV0.},
  and~\ref{z:W4OXReVA}.
\end{proof}

The main remaining ingredient is to determine the condition function $\chi$ in
order to define our approximate microcanonical channel operator, and to prove
that all the desired properties laid out in \cref{z:v1wYNQz5} are
satisfied.  We proceed through some intermediate results that involve operators
$P^{\chi}_{B^nR^n}$ with different useful condition functions $\chi$.

The condition functions we consider make use of the following quantities, which
are functions of $(\boldsymbol j, \boldsymbol z)$:
\begin{align}
  \tilde{z}_i^j & = \begin{cases}
                      J z_i &\text{if \(j_i = j\),}\\
                      0 &\text{if \(j_i \neq j\);}
                    \end{cases}
  &
    \nu_j(\boldsymbol j, \boldsymbol z)
  &= \frac1n \sum_{i=1}^n \tilde{z}_i^j\ .
    \label{z:x9yjGcyJ}
\end{align}
The quantities $\tilde z_i^j$ and $\nu_j(\boldsymbol j, \boldsymbol z)$ can be
thought of as random variables depending on the estimate $\tilde\sigma$,
$\boldsymbol j$ along with the random outcomes $z_i$ of the protocol outlined
above.  The variable $\tilde z_i^j$ takes the value of the sample $z_i$ scaled
by $J$ if we happened to measure $j$ on the $i$-th copy, otherwise it takes the
value zero.  Roughly speaking, we can imagine that we sort all measurement
outcomes $z_i$ by the choices $j_i$, i.e., collecting all measurement outcomes
associated with $j_i=1$ separately from those with $j_i=2$, etc.; the
$\nu_j(\boldsymbol j, \boldsymbol z)$ can roughly be thought of as taking the
sample averages of each of those outcomes per choice of $j$.  (This rough
explanation would be accurate if we had exactly $n/J$ samples for each choice of
measurement setting.  But since each $j_i$ is chosen independently at random,
the number of samples per choice of measurement setting fluctuates around $n/J$
by $O(\sqrt{n})$.)

We construct our approximate microcanonical channel operator in two steps.  As a
first step, we construct an operator $P^{\chi}_{B^nR^n}$ that can discriminate
between i.i.d.\@ channels based on their expectation values with respect to the
observables $C^j_{BR}$, by identifying a suitable condition function $\chi$.
Then, we use this construction to build our approximate microcanonical channel
operator.

\subsubsection{Construction of a tester that discriminates i.i.d.\@ channels
  based on their expectation values}

First, we investigate the following condition function.  For any positive
semidefinite operator $M_{BR}$ with $\operatorname{tr}_B(M_{BR}) = \mathds{1}_R$, for any $h>0$,
and for any $j \in \{{1, \ldots, J}\}$, we define:
\begin{align}
  \chi_{j;M;>h}(\tilde\sigma,\boldsymbol j, \boldsymbol z)
  &= \chi\Bigl\{{
    \bigl \lvert { \nu_j(\boldsymbol j, \boldsymbol z) - \operatorname{tr}\bigl({C^j_{BR} M_{BR}}\bigr)
    }\bigr \rvert  > h }\Bigr\}\ ,
    \label{z:ASsLiTCZ}
\end{align}
where $\chi\{{ \cdots }\}$ is the characteristic function equal to one whenever
the condition $(\cdots)$ is true and zero if it is false.  The condition
function $\chi_{j;M;>h}$ tests whether the variable
$\nu_j(\boldsymbol j, \boldsymbol z)$, computed based on the estimated
$\tilde\sigma$, the sampled $\boldsymbol j$, and the measured $\boldsymbol z$,
deviates from the expectation value $\operatorname{tr}\bigl({C^j_{BR} M_{BR}}\bigr)$ by more than
$h$.  Recalling that $\nu_j(\boldsymbol j, \boldsymbol z)$ is meant to represent
an estimation of the average of the outcome of measurement setting $j$, we
expect this sample average to concentrate around the ideal expectation value
$\operatorname{tr}\bigl({C^j_{BR} M_{BR}}\bigr)$ for large $n$.  The following lemma establishes this
fact:
\begin{lemma}
  \label{z:HY7bJv7-}
  Let $M_{BR}$ be the Choi matrix of a quantum channel $\mathcal{M}_{A\to B}$.
  For any $j=1, \ldots, J$, for any $0< y' < 1/d_R$, for all
  $0< h < \lVert {C^j_{BR}}\rVert $, and for all $\sigma_R \geq y' \mathds{1}_R$ with
  corresponding $\lvert {\sigma}\rangle _{AR} \equiv \sigma_R^{1/2}\,\lvert {\Phi_{A:R}}\rangle $, we
  have
  \begin{align}
    \operatorname{tr}\bigl[{
    P^{\chi_{j;M;>h}}_{B^nR^n}
    \,\mathcal{M}^{\otimes n}(\sigma_{AR}^{\otimes n})
    }\bigr]
    \leq \operatorname{poly}(n)\,\exp\biggl\{{
    -n \min\Bigl({ \frac{m}{n}, \frac{\bar{n}}{n} }\Bigr) \, \frac{h^8\, y'^8}{5^8\, \lVert {C^j_{BR}}\rVert ^8}
    }\biggr\}\ .
  \end{align}
\end{lemma}

We prove this lemma in \cref{z:iBCoMrzo}.  The
core part of the proof is an application of Hoeffding's bound.  Some challenges
include the fact that while the true input state is $\sigma_R$, the measurement
that is carried out by the protocol is
$\tilde\sigma_R^{-1/2} C^j_{BR} \tilde\sigma_R^{-1/2}$ where
$\tilde\sigma_R \approx \sigma_R$.  Properties of the pretty good measurement
combined with a suitable application of continuity bounds enable us to show that
the true expectation value of the measurement outcomes does not deviate too far
from the ideal expectation value $\operatorname{tr}\bigl({C^j_{BR} M_{BR}}\bigr)$ in order to apply
Hoeffding's bound.

The above lemma enables the construction of a test that can discriminate
channels based on their expectation values with respect to the charges
$C^j_{BR}$.  Specifically, fix some real values
$\boldsymbol q = (q_j)_{j=1}^{J} \in \mathbb{R}^J$, let
$0<h'< \min_j\lVert {C^j_{BR}}\rVert $, and define
\begin{subequations}
  \label{z:ZEhFJ.kh}
  \begin{align}
    \chi_{\boldsymbol q;\leq h'}(\tilde\sigma,\boldsymbol j, \boldsymbol z)
    &=
      \chi\Bigl\{{ 
      \forall\ j\in \{{1,\ldots, J}\} :\ 
      \bigl \lvert { \nu_j(\boldsymbol j, \boldsymbol z) - q_j }\bigr \rvert  \leq h'
      }\Bigr\}\ ;
      \label{z:usabQCei}
    \\
    \chi_{\boldsymbol q;\not \leq h'}(\tilde\sigma,\boldsymbol j, \boldsymbol z)
    &=
      \chi\Bigl\{{ \exists\ j \in \{{1, \ldots J}\} :\ 
      \bigl \lvert { \nu_j(\boldsymbol j, \boldsymbol z) - q_j }\bigr \rvert  > h' }\Bigr\}
      = 1 - \chi_{\boldsymbol q;\leq h'}(\tilde\sigma,\boldsymbol j, \boldsymbol z)
      \ .
      \label{z:FdxPFyNe}
  \end{align}
\end{subequations}

The POVM $\bigl\{{ P^{\chi_{\boldsymbol q;\leq q}}_{B^nR^n}, %
  P^{\chi_{\boldsymbol q;\not \leq q}}_{B^nR^n} }\bigr\}$ defined via these condition
functions via~\eqref{z:qVquPoYP} behaves as a
test that determines whether an i.i.d.\@ channel
$\mathcal{M}_{A\to B}^{\otimes n}$ with Choi matrix $M_{BR}$ has expectation
values $\operatorname{tr}({C^j_{BR}\,M_{BR}})$ that are all close to the $q_j$'s, or if there
is at least one value that deviates far from the corresponding $q_j$.

\begin{proposition}[General i.i.d.\@ channel discriminator]
  \label{z:lmCUsTOF}
  The following statements hold:
  \begin{enumerate}[label=(\roman*)]
  \item \label{z:wdi9pumP}
    Let $0 < a < h'$, let $0<y'<1/d_R$, and
    let $\sigma_R \geq y'\mathds{1}$ with corresponding
    $\lvert {\sigma}\rangle _{AR} \equiv \sigma_R^{1/2}\,\lvert {\Phi_{A:R}}\rangle $.  Let
    $\mathcal{M}_{A\to B}$ be any quantum channel such that
    \begin{align}
      \bigl \lvert { \operatorname{tr}\bigl[{C^j_{BR} \mathcal{M}\bigl({\Phi_{A:R}}\bigr) }\bigr]  - q_j }\bigr \rvert 
      < a\quad\forall\ j=1, \ldots, J.
    \end{align}
    Then
    \begin{align}
      \operatorname{tr}\Bigl[{
      P^{\chi_{\boldsymbol{q};\not \leq h'}}_{B^nR^n}
      \,\mathcal{M}^{\otimes n}(\sigma_{AR}^{\otimes n})
      }\Bigr]
    \leq
    \operatorname{poly}({n}) \exp\biggl\{{
    -n\,\min\Bigl({\frac{m}{n}, \frac{\bar n}{n}}\Bigr)\,
      \frac{(h' - a)^8\,y'^8}{5^8 \max_j\lVert {C^j_{BR}}\rVert ^8}
    }\biggr\}\ .
    \end{align}
  \item \label{z:9tjp9JsO} 
    Let $b > 0$ such that $h' < b \leq \min_j \lVert {C^j_{BR}}\rVert $, let
    $0<y'<1/d_R$, and let $\sigma_R \geq y'\mathds{1}$ with corresponding
    $\lvert {\sigma}\rangle _{AR} \equiv \sigma_R^{1/2}\,\lvert {\Phi_{A:R}}\rangle $.  Let
    $\mathcal{M}_{A\to B}$ be any quantum channel such that there exists
    $j_0 \in \{{1, \ldots, J}\}$ with
    \begin{align}
      \bigl \lvert { \operatorname{tr}\bigl[{C^{j_0}_{BR} \mathcal{M}\bigl({\Phi_{A:R}}\bigr) }\bigr]  - q_{j_0} }\bigr \rvert 
      > b\ .
      \label{z:.24EDF0r}
    \end{align}
    Then
    \begin{align}
      \operatorname{tr}\Bigl[{
      P^{\chi_{\boldsymbol{q};\leq h'}}_{B^nR^n}
      \,\mathcal{M}^{\otimes n}(\sigma_{AR}^{\otimes n})
      }\Bigr]
      &\leq
        \operatorname{poly}({n}) \exp\biggl\{{
        -n\,\min\Bigl({\frac{m}{n}, \frac{\bar n}{n}}\Bigr)\,
        \frac{(b - h')^8\,y'^8}{5^8\,\lVert {C^{j_0}_{BR}}\rVert ^8}
        }\biggr\}
        \nonumber\\
        &\leq \operatorname{poly}({n}) \exp\biggl\{{
    -n\,\min\Bigl({\frac{m}{n}, \frac{\bar n}{n}}\Bigr)\,
      \frac{(b - h')^8\,y'^8}{5^8 \max_j\lVert {C^{j}_{BR}}\rVert ^8}
    }\biggr\}\ .
    \end{align}
  \end{enumerate}
\end{proposition}

We present the full proof of this
\lcnamecref{z:lmCUsTOF} in
\cref{z:zCc78QJk}.  The main proof
strategy is to reduce the
conditions~\eqref{z:ZEhFJ.kh} to
conditions of the
type~\eqref{z:ASsLiTCZ}.
In~\ref{z:wdi9pumP}: If a channel
$\mathcal{M}_{A\to B}$ has expectation values $a$-close to the $q_j$'s, then the
event $\bigl \lvert {\nu_j(\boldsymbol j, \boldsymbol z) - q_j}\bigr \rvert  > h'$ can only
happen if
$\bigl \lvert {\nu_j(\boldsymbol j, \boldsymbol z) - \operatorname{tr}(C^j_{BR} M_{BR})}\bigr \rvert  > h' -
a$, whose probability we can upper bound
using~\cref{z:HY7bJv7-}.  A similar argument holds
for~\ref{z:9tjp9JsO}.

We left the dependency of $m$ on $n$ general in the statement of
\cref{z:HY7bJv7-} and
\cref{z:lmCUsTOF}; a suitable choice might be to set
$m$ to a constant fraction of $n$, say, $m = c n$ with $0<c<1$.  If furthermore
$c\leq 1/2$, then we can simplify the terms in the bounds using
\begin{align}
  \min\Bigl({ \frac{m}{n} \,,\; \frac{\bar{n}}{n} }\Bigr) = c\ .
\end{align}
In the following, we can simply pick $c=1/2$.

We observe that the optimal decay rates in the expressions above might not
happen at $c=1/2$, as opposed to the bounds we proved.  Indeed, the proof of our
bounds proceeded with crude inequalities of the type $y'^8 < y'$ for $y'<1$ in
order to obtain a simple expression as the decay rate.  It is possible that a
more careful analysis in the proof of
\cref{z:lmCUsTOF} would reveal a better choice of $m$
as a function of $n$ other than $m=n/2$.

\subsubsection{Construction of an approximate microcanonical operator}

We construct an approximate microcanonical operator $P_{B^nR^n}$, along with its
complement $P^\perp_{B^nR^n} \equiv \mathds{1}_{B^nR^n} - P_{B^nR^n}$, by choosing
the operator $P^\chi_{B^nR^n}$ associated with a condition function $\chi$ of
the form~\eqref{z:ZEhFJ.kh}.
The following theorem establishes that the operator constructed in this way can
satisfy the requirements of an approximate microcanonical channel operator
(\cref{z:Caf8GVY-}).
\begin{theorem}[Construction of an approximate microcanonical channel operator]
  \label{z:JjqyeW8p}
  Let $\boldsymbol q = \{{ q_j }\}_{j=1}^J$, let
  $0<\eta'<\eta < \min_j\lVert {C^j_{BR}}\rVert $, and write $\bar\eta=(\eta'+\eta)/2$.
  Let
  \begin{align}
    P_{B^nR^n}
    &\equiv P^{\chi_{\boldsymbol q;\leq\bar\eta}}_{B^nR^n} \ ;
    &
      P^\perp_{B^nR^n}
    &\equiv  P^{\chi_{\boldsymbol q;\not\leq\bar\eta}}_{B^nR^n} \ ,
  \end{align}
  where $ P^{\chi_{\boldsymbol q;\leq\bar\eta}}_{B^nR^n}, %
  P^{\chi_{\boldsymbol q;\not\leq\bar\eta}}_{B^nR^n} $ are defined
  in~\eqref{z:qVquPoYP} with $m = n/2$ and
  using~\eqref{z:ZEhFJ.kh}.  The
  following statements hold:
  \begin{enumerate}[(\roman*)]
  \item \label{z:W3JG7egy}%
    For any $\epsilon>0$, $\nu>1$, and for any $0<y<1/(\nu d_R)$, let
    $\mathcal{E}_{A^n\to B^n}$ be any quantum channel such that
    \begin{align}
      \max_{\sigma_R\geq y\mathds{1}}
      \operatorname{tr}\bigl[{ P^\perp_{B^nR^n} \mathcal{E}\bigl({\sigma_{AR}^{\otimes n}}\bigr) }\bigr]
      \leq \epsilon\ ,
    \end{align}
    using the shorthand $\lvert {\sigma}\rangle _{AR} \equiv \sigma_R^{1/2}\lvert {\Phi_{A:R}}\rangle $.
    Assume furthermore that
    $\nu \geq 1+ ({\eta-\eta'})/(4\max_j \lVert {C^j_{BR}}\rVert )$.  Then, for any
    $j=1, \ldots, J$,
    \begin{multline}
      \max_{\sigma_R\geq \nu y\mathds{1}} \operatorname{tr}\Bigl[{
      \bigl\{{ \overline{H^{j,\sigma}}_{B^nR^n} \notin [{q_j \pm \eta}] }\bigr\}
      \,
      \mathcal{E}_{A^n\to B^n}\bigl({\sigma_{AR}^{\otimes n}}\bigr)
      }\Bigr]
      \\
      \leq \operatorname{poly}({n}) \exp\Biggl\{{
        -n y^8 \min\biggl({
            -\frac{\log(\epsilon)}{n y^8}
            \,,\;
            \frac{c'({\eta-\eta'})^8}{\max_j \lVert {C^j_{BR}}\rVert ^8}
        }\biggr)
      }\Biggr\}\ ,
      \label{z:jFwlnxgT}
    \end{multline}
    with $c'=1/(2\times 5^8)$.

  \item \label{z:wEDQuyLi}%
    For any $\delta'>0$, $\nu' > 1$, and for any
    $0<y'<1/(\nu' d_R)$, let $\mathcal{E}_{A^n\to B^n}$ be any quantum channel
    such that for all $j=1, \ldots, J$,
    \begin{align}
      \max_{\sigma_R\geq  y'\mathds{1}} \operatorname{tr}\Bigl[{
      \bigl\{{ \overline{H^{j,\sigma}}_{B^nR^n} \notin [{q_j \pm \eta'}] }\bigr\}
      \,
      \mathcal{E}_{A^n\to B^n}\bigl({\sigma_{AR}^{\otimes n}}\bigr)
      }\Bigr]
      \leq \delta'\ ,
    \end{align}
    using the shorthand $\lvert {\sigma}\rangle _{AR} \equiv \sigma_R^{1/2}\lvert {\Phi_{A:R}}\rangle $.
    Assume furthermore that
    $\nu' \geq 1+ ({\eta - \eta'})/(4 \max_j\lVert {C^j_{BR}}\rVert )$.  Then
    \begin{align}
      \max_{\sigma_R\geq \nu' y'\mathds{1}}
      \operatorname{tr}\bigl[{ P^\perp_{B^nR^n} \mathcal{E}\bigl({\sigma_{AR}^{\otimes n}}\bigr) }\bigr]
      &
      \leq
      \operatorname{poly}({n}) \exp\Biggl\{{
      -n y'^8 \min\Biggl({
      -\frac{\log\bigl({\delta'}\bigr)}{ n y'^8 } ,
      \frac{c' ({\eta-\eta'})^8}{\,\max_j \lVert {C^j_{BR}}\rVert ^8 }
      }\Biggr)
      }\Biggr\}
      \ ,
      \label{z:aM9IqJnT}
    \end{align}
    with $c'=1/(2\times 5^8)$.
  \end{enumerate}
\end{theorem}

We prove \cref{z:JjqyeW8p} in
\cref{z:p2qVGb1M}.  The overarching proof strategy
is to use our constrained postselection theorem
(\cref{z:UEUfoDLC}) to reduce the global channel
$\mathcal{E}_{A^n\to B^n}$ to i.i.d.\@ channels
$\mathcal{M}_{A\to B}^{\otimes n}$.

\subsubsection{Parameter regimes for the construction of an
  approximate microcanonical operator}

\Cref{z:JjqyeW8p} implies that there exist approximate
microcanonical channel operators with the following parameters.  Let
$c_{\mathrm{min}} \equiv \min_j \lVert {C^j_{BR}}\rVert $,
$c_{\mathrm{max}} \equiv \max_j \lVert {C^j_{BR}}\rVert $,
and $c'' = (c_{\mathrm{min}}/c_{\mathrm{max}})^8 \cdot (2\times 10^8)^{-1}$.
Let $\alpha_1>0$, $\alpha_2>0$, $\beta_1>0$, $\beta_2>0$, $\gamma>0$, with
$\gamma+\beta_1 < 1/8$ and $\gamma+\beta_2 < 1/8$, and set
\begin{align}
  y &= n^{-\beta_1}\ ;
  &
    y' &= n^{-\beta_2}\ ;
  &
    \eta &= c_{\mathrm{min}} n^{-\gamma}\ ;
  &
    \eta'&= \eta/2\ ;
  &
    \nu &= \nu' = 3/2\ .
\end{align}
Observe that $\eta,\eta' < c_{\mathrm{min}}$ and
$(\eta-\eta') \leq 2c_{\mathrm{min}} \leq 2 c_{\mathrm{max}}$ so the
$\eta,\eta',\nu,\nu'$ parameters satisfy the constraints of
\cref{z:JjqyeW8p}.  
Then, by \cref{z:JjqyeW8p}, there exists an
$(\eta,\epsilon,\delta, y, \nu, \eta',\epsilon',\delta', y',\nu')$-approximate
microcanonical channel operator, with
\begin{align}
  \begin{aligned}
  \epsilon &= \exp\bigl({-n^{\alpha_1}}\bigr)\ ;
  &\qquad\qquad
  \delta &=
           \operatorname{poly}({n}) \exp\biggl\{{
               - n^{\min\bigl({ \alpha_1, 1-8\beta_1-8\gamma + \frac{\log(c'')}{\log(n)} }\bigr) }
           }\biggr\}\ ;
  \\
    \delta' &= \exp\bigl({-n^{\alpha_2}}\bigr)\ ;
  &
    \epsilon' &=
           \operatorname{poly}({n}) \exp\biggl\{{- n^{
                \min\bigl({ \alpha_2, 1-8\beta_2-8\gamma + \frac{\log(c'')}{\log(n)} }\bigr)
           }}\biggr\}
           \ .
  \end{aligned}
\end{align}
To be within the scope of \cref{z:1tRx46YH} (to
show that we recover a thermal channel from a microcanonical channel), we need
some further restrictions on the parameters.  Namely, the conditions
$\epsilon'\leq\epsilon$ and
$2c_{\mathrm{max}}^2 \log(2/\delta') \leq n\eta'^2 y'^2$ are satisfied for large
enough $n$ if
\begin{align}
  \min\Bigl({ \alpha_2, 1-8\beta_2-8\gamma}\Bigr)
  &>
  \alpha_1\ ;
  &
    1 - 2\beta_2 - 2\gamma
  &> \alpha_2\ .
\end{align}
Concretely, we can choose $\alpha_1 = 1 - 17\gamma$, $\alpha_2 = 1 - 5\gamma$,
$0 < \gamma = \beta_1 = \beta_2 < 1/16$, in which case
\begin{align}
  \delta
  &= \operatorname{poly}({n}) \exp\bigl({-n^{1-17\gamma}}\bigr)\ ;
  &
    \epsilon'
  &= \operatorname{poly}({n}) \exp\bigl({-n^{1-17\gamma}}\bigr)\ .
\end{align}
Choosing arbitrarily small $\gamma>0$ will have
$\epsilon,\delta,\delta',\epsilon'$ all decay almost as $\sim\exp(-n)$.

We'll keep the degree of the $\operatorname{poly}(n)$ polynomial term a closely-guarded state
secret communicated solely via Signal messenger.

\section{Passivity and resource-theoretic considerations for the thermal quantum channel}
\label{z:MO8enbcL}

\subsection{Thermal quantum channels are passive}

An important property obeyed by the thermal state is its energy passivity.
Given a Hamiltonian $H$, a state $\rho$ is \emph{energetically passive} if for
any unitary operation $U$ we have $\operatorname{tr}({ U \rho U^\dagger H}) \geq \operatorname{tr}({\rho H})$.
I.e., a unitary operation can only increase the energy of the state.  A state
$\rho$ is \emph{energetically completely passive} if $\rho^{\otimes n}$ is
passive for all $n$ with respect to the $n$-copy Hamiltonian
$H^{(n)} = H_1 + H_2 + \cdots$.

Energy passivity refers to the following property of the thermal state: It is
impossible to lower the thermal state's average energy by applying a unitary.
This property is reversed for negative temperatures.  In a spin system at nearly
maximal energy, where the thermal state has negative temperature, it is
impossible to \emph{increase} the energy of the state by applying a unitary.
The sign of the temperature indicates the direction in which it is impossible to
change the energy.

In the presence of multiple conserved quantities, a thermal state can act as a
``converter'' between different charges, lowering one charge at the expense of
increasing another one. %
For instance, in a grand canonical setting, there might be a unitary that lowers
the energy at the expense of increasing the number of particles.  To formulate
passivity in the presence of multiple charges, we ask here that the unitary
lowers the state's energy (at positive temperature) without increasing any of
the other charges (at positive chemical potentials).

It is worth phrasing this version of the passivity property for states more
generally in the context of multiple conserved charges.  From Lagrange duality,
the ``generalized chemical potentials'' $\mu_j$ in \cref{z:chtPfev.}
provide information about the ``direction'' in which the constraint
$\operatorname{tr}(\rho Q_j)=q_j$ is active~\cite{R95}, in the
following sense.  If $\mu_j=0$, then the constraint is not active; it can be
removed without changing the optimal state $\gamma$.  If $\mu_j > 0$, the
constraint is active in the positive direction: it can be replaced by an
inequality $\operatorname{tr}({\rho Q_j})\leq q_j$ without changing the optimal state $\gamma$.
Finally if $\mu_j < 0$, the constraint is active in the other direction and can
be replaced by $\operatorname{tr}({\rho Q_j})\geq q_j$ without changing the optimal state
$\gamma$.
A passivity property for the thermal state with respect to one of the charges,
say $Q_1$, can be proven as follows.
We first assume that $\mu_j \geq 0$ for all $j=1,\ldots,J$ (or else we flip the
corresponding $Q_j, q_j$ to $-Q_j, -q_j$).  Then all equality constraints
$\operatorname{tr}({\rho Q_j}) = q_j$ can be replaced by inequalities $\operatorname{tr}({\rho Q_j})\leq q_j$
without changing the optimal state $\gamma$.
We ask whether there exists a unitary $U$ such that
$\operatorname{tr}({Q_1 U \gamma U^\dagger}) < \operatorname{tr}({Q_1 \gamma})$ and such that for all
$j=2, \ldots, J$, we have $\operatorname{tr}({Q_j U \gamma U^\dagger}) \leq \operatorname{tr}({Q_j \gamma})$.
Suppose such a $U$ existed, with $\bar\gamma= U \gamma U^\dagger$ satisfying
$\operatorname{tr}({Q_1 \bar\gamma}) = \operatorname{tr}({Q_1 U \gamma U^\dagger}) \equiv \bar{q}_1
< \operatorname{tr}({Q_1 \gamma}) = q_1$.
By sensitivity analysis and Lagrange
duality~\cite{R95}, and since $\mu_1>0$, we must have
${S}_{}^{}({\bar\gamma}) < {S}_{}^{}({\gamma})$.  (The dual variable associated with a
constraint determines the variation of the objective function if the constraint
is perturbed.)  But this statement contradicts the fact that $\bar\gamma$ and
$\gamma$ are related by a unitary and must therefore have the same entropy.

As it turns out, the above argument can be extended to a passivity passivity for
the thermal quantum channel.  Consider the
problem~\eqref{z:sb5FfEOw} and let $\mathcal{T}$ be the
corresponding thermal quantum channel of the
form~\eqref{z:s1gf97hu} with generalized chemical
potentials $\{{\mu_j}\}$.  Assume that $\mu_j\geq0$ for all $j=1, \ldots, J$ and
that $\mu_1 > 0$.  We ask whether there exist unitary operations $U_A, U_B'$
such that for the unitarily rotated channel
$\mathcal{T}'({\cdot}) \equiv U_B' \mathcal{T}\bigl({ U_A ({\cdot}) U_A^\dagger }\bigr)
U_B'^\dagger$ we have $\operatorname{tr}[{C^1_{BR}\mathcal{T}'(\Phi_{A:R})}] < q_1$ while still
obeying all remaining inequality constraints.

Suppose such $U_A, U_B'$ existed.  By sensitivity analysis of convex problems,
it must hold that ${S}_{}^{}({\mathcal{T}'}) < {S}_{}^{}({\mathcal{T}})$.
Let $U_R \equiv (U_A)^{t_{A\to R}}$.
Exploiting unitary invariance of the relative entropy for the unitary
$U_B^\dagger \otimes U_R^\dagger$, we find
\begin{align}
  {S}_{}^{}({\mathcal{T}'})
  &= - \max_{\phi_{R}}
  {D}_{}^{}\Bigl ({ U_B\mathcal{T}\bigl({
  U_A \phi_R^{1/2} \Phi_{A:R} \phi_R^{1/2} U_A^\dagger
  }\bigr) U_B^\dagger
  }\mathclose{}\,\Big \Vert\,\mathopen{}{\mathds{1}_B\otimes \phi_R}\Bigr )
  \nonumber\\
  &= - \max_{\phi_{R}}
  {D}_{}^{}\Bigl ({ \mathcal{T}\bigl({
  U_R^\dagger \phi_R^{1/2} (U_A)^t \Phi_{A:R}
    (U_A^\dagger)^t \phi_R^{1/2} U_R^\dagger
  }\bigr) 
  }\mathclose{}\,\Big \Vert\,\mathopen{}{\mathds{1}_B\otimes U_R^\dagger \phi_R U_R}\Bigr )
    \nonumber\\
  &= -\max_{\phi_R'} {D}_{}^{}\Bigl ({ \mathcal{T}\bigl({
    \phi_R'^{1/2} \Phi_{A:R} \phi_R'^{1/2}
  }\bigr) 
  }\mathclose{}\,\Big \Vert\,\mathopen{}{\mathds{1}_B\otimes \phi'_R}\Bigr )
    \quad = {S}_{}^{}({\mathcal{T}})\ ,
\end{align}
letting $\phi'_R = U_R^\dagger \phi_R U_R$.  This statement contradicts our
earlier conclusion that ${S}_{}^{}({\mathcal{T}'}) < {S}_{}^{}({\mathcal{T}})$.
In conclusion, there can exist no unitaries $U_A, U_B'$ such that
$\operatorname{tr}[{C^1_{BR}\mathcal{T}'(\Phi_{A:R})}] < q_1$ and
$\operatorname{tr}[{C^j_{BR}\mathcal{T}'(\Phi_{A:R})}] \leq q_j$ for all $j>1$.

We expect it is possible to continue along this approach and generalize the idea
of complete passivity to channels.  The anticipation is that, for a given set of
constraints and generalized chemical potentials, the unique completely passive
channel should be the thermal quantum channel.  We discuss some challenges in
extending this argument from states to channels in the discussion section below.

\subsection{Challenges for a thermodynamic resource theory of channels}

A resource theory studies possible transformations that an agent can perform on
an abstract set of objects.  The objects considered here can be quantum states
or quantum channels.  (The term `dynamical resource theory' is sometimes used
when the objects are quantum channels.)
The agent is allowed to perform any sequence of operations
from a fixed set (the \emph{free operations}).  They are allowed to tensor in
any additional object from another fixed set (the \emph{free states} or
\emph{free channels}).  
The resource theory of thermodynamics for states
provides a solid basis to refine the laws of thermodynamics in the quantum,
single-shot regime (cf.\@ e.g.\@~\cite{R30,R33,R119,R120,R26,R121,R122,R123,R34,R124} and references therein).

The state resource theory of work and heat~\cite{R122,R125,R71} considers a resource theory in
which both purity and energy are individual resources.  Specifically, free
operations in this resource theory are defined as unitary operations that are
strictly energy-conserving, and there are no free states.  Purity is a resource:
Outputting a state with low entropy requires an initial state that itself is
sufficiently pure.  Energy is also a resource: Producing an output state at a
given energy requires an input state with that energy, and changing the energy
of a state requires an opposite energy change of some ancillary system.
In this resource theory, an ancillary system $A$ in the thermal state
$\gamma_\beta \propto {e}^{-\beta H_A}$ has the property of enabling conversion
of purity to energy.  Given an input state with high purity but low energy, and
given access to $\gamma_\beta$, we can produce an output state with high energy.
Asymptotically reversible, an amount of negative entropy $-dS$ is converted into
energy $dE$ at a proportion determined by the temperature of the thermal state,
$\beta dE = dS$.  This relation is a manifestation of the first law of
thermodynamics.  We might view the thermal state as a ``bank,'' converting an
amount of one ``currency'' (energy) into another ``currency'' (purity), at a
fixed rate (determined by the thermal state's temperature).

In a more general setting, we consider the merging of two individual arbitrary
resource theories~\cite{R125}.  The corresponding
multi-resource theory is identified as the resource theory whose free operations
lie in the intersection of both resource theories.  If both resource theories
are individually asymptotically reversible with corresponding monotones
$E_1(\rho)$ and $E_2(\rho)$, and under certain additional assumptions, then
there is a special type of state (``bank state'') that enables the conversion of
one type of resource into another.  A state $\tau$ is a bank state if and only
if for all $\sigma$~\cite{R125},
\begin{align}
  E_1(\sigma) > E_1(\tau)
  \quad\text{or}\quad
  E_2(\sigma) > E_2(\tau)
  \quad\text{or}\quad
  [\ E_1(\sigma) = E_1(\tau) ~~\text{and}~~ E_2(\sigma) = E_2(\tau)\ ]\ .
  \label{z:fZv-ZoXu}
\end{align}
It is also clear from the resource diagrams of~\cite{R125}
that a bank state is a state that minimizes one resource if the other resource
is kept constant.  In the case of energy and purity, this minimization
corresponds to Jaynes' principle.

Recently, several quantum resource theories have been extended from states to
channels~\cite{R50,R47,R126,R53,R52}.
It would be natural to assume that in a channel version of the resource theory
of thermodynamics, the thermal quantum channel plays a role that is analogous to
the thermal state in the quantum state resource theory of thermodynamics.
In particular, one might expect that a thermal quantum channel would enable the
conversion between two putative resources of channel purity and channel energy.

Here, we point to missing foundations to establish a thermodynamic resource theory
of channels that would have such a property.

We outline a challenge in identifying a channel version of noisy
operations~\cite{R127}, a degenerate version of thermodynamics
where the system Hamiltonian is trivial~\cite{R31}.
In the resource theory of noisy operations,
a state $\rho$ with high entropy ${S}_{}^{}({\rho})$ is less useful than a state with low
entropy.
Anticipating that the channel's entropy ${S}_{}^{}({\mathcal{N}})$ should play an analogous
role to the state's entropy, we find that the identity channel would be the most
resourceful channel given that it has minimal entropy.
This observation is in tension with most
common channel resource theories, in which the identity channel is considered a no-op
allowed for free (cf.\@ e.g.\@~\cite{R50}).
We anticipate that to construct a thermodynamic resource theory of channels, it is
useful to consider a scenario in which the identity channel is resourceful.
Such a scenario occurs in the context of quantum communication, where
the identity channel describes
perfect communication between two parties.
One typically aims to distill such a highly resourceful channel
using any available lower quality noisy channels.
A scenario in which the reversible conversion rate is the channel
entropy is detailed in \R\cite{R46}.  One considers
a three-party setting in which Alice communicates to Bob and Eve via a pure broadcast
channel modeled by an isometry $V_{A\to BE}$.  The optimal rate at which Bob can perform
quantum state merging~\cite{R128,R129} of
his state with Eve coincides with the entropy of the channel
$\mathcal{N}_{A\to B}(\cdot) = \operatorname{tr}_E[{ V_{A\to B}\,({\cdot})\,V^\dagger }]$.

Let us now suppose that we constructed a resource
theory of channels in which the resource is channel purity, as measured by
$-{S}_{}^{}({\mathcal{N}})$; we assume this resource theory provides
some satisfactory (even if rough) channel analog of the resource
theory of noisy operations.
Mimicking the state approach to the resource theory of work and
heat~\cite{R122,R125,R71}, one would consider a multi-resource
theory combining the channel purity resource theory with a channel energy resource theory.
The latter might be defined, for instance,
by considering channel superoperations that strictly conserve
both the input and output energy of any channel.
To establish the thermal quantum channel as being able to convert between resources,
the full analysis of ref.~\cite{R125} would have to be carried out
again in the channel setting.  In particular, one would have to ensure that both
individual channel resource theories are asymptotically reversible with a single
monotone.
One might anticipate, in such a case, that the ``bank channel'' defined analogously
to~\eqref{z:fZv-ZoXu}, is the quantum thermal channel.  This would follow
from the fact that the quantum thermal channel would optimize one resource monotone
(the channel entropy) under a constraint fixing the other monotone (an energy monotone,
which one would consider as a constraint in the definition of the thermal quantum channel).

\section{Discussion}
\label{z:aLWklD3j}

We establish the concept of a \emph{thermal channel} as an extension to quantum
channels of the thermal state.
We present two independent constructions of the thermal channel, extending
different equivalent constructions of the thermal state, and we
show that they lead to the same channels.
The widespread relevance of the thermal state throughout
physics, information theory, machine learning, and quantum computing,
inspires promising applications for the analogous concept for quantum channels.

We extend Jaynes' fundamental maximum entropy
principle~\cite{R3,R4,R6} to quantum channels,
exploiting recent extensions of the concept of information-theoretic entropy to
channels~\cite{R79,R45,R130,R93,R50,R46}.  Specifically, we determine which quantum channel
$\mathcal{T}$ has maximal \emph{channel entropy} subject to a set of linear
constraints.
The channel $\mathcal{T}$ has a form that extends the exponential form
of the Gibbs distribution of the thermal state, in a way that accounts for
the optimal input state in the definition of the channel entropy.
We find an explicit form for thermal channels resulting from the maximum channel
entropy principle.  Such channels have a Choi matrix with an exponential form
reminiscent of the thermal state.  The form also involves a state $\phi_R$,
interpreted as a hypothetical input state to the channel, and identified as the
state that is optimal in the definition of the channel entropy.

A second independent approach, which extends the microcanonical ensemble for
quantum states to quantum channels, reinforces the maximum channel-entropy
principle approach by leading to the same concept of a thermal channel.
Specifically, we identify a set of channels that act on $n$ copies of the
input system and for which measurement of the constraint operators give
suitably sharp statistics for almost all input states.  We define the
\emph{microcanonical channel} as the channel that is most ``mixed'' (according to
its channel entropy) in this set.
If we act on any i.i.d.\@ state $\phi^{\otimes n}$, the microcanonical channel's
reduced action on a single pair of input and output systems reduces to the
thermal channel with respect to $\phi$.

The general mathematical structure of the thermal quantum channel
(\cref{z:TLiwdJGR}) involves a state $\phi_R$, defined
implicitly as the input for which the corresponding channel produces the least
entropy relative to $R$.  If the constraints obey some symmetry on their input
system, the $\phi_R$ inherits the same symmetry (cf.\@
\cref{z:nCrJ.k1B,z:JM5--Sw6,z:9P7ZvFQO}).  This property significantly
narrows down the possible optimal $\phi_R$ in cases, for example, where the
constraint operators are Pauli-covariant, are classical, or all commute with a
fixed operator on $R$.
Yet the optimal state $\phi_R$ might be difficult to determine in general from
the constraint operators directly.  In such cases, it is convenient to fix
$\phi_R$ and to compute the \emph{thermal quantum channel with respect to
  $\phi_R$}, defined as a channel maximizing
${S}_{}^{}({B}\mathclose{}\,|\,\mathopen{}{R})_{{\mathcal{N}(\phi_{AR})}}$ subject to the given constraints but for
fixed $\lvert {\phi}\rangle _{AR} \equiv \phi_R^{1/2}\lvert {\Phi_{A:R}}\rangle $.  For full-rank
$\phi_R$, the maximizer is unique and has the form given in
\cref{z:hd7.cLe.}.
\Cref{z:33B55hFY} gives the mathematical form of
the thermal quantum channel with respect to a general $\phi_R$.  The
interpretation of fixing $\phi_R$ is to quantify the channel's \emph{average}
output entropy (relative to $R$) over input states, weighted by $\phi_R$; in
contrast, ${S}_{}^{}({\mathcal{N}})$ computes the minimum of the output entropy
(relative to $R$) over all inputs.
The channel entropy with respect to $\phi_R$ can vary significantly as a
function of $\phi_R$.  Consider a channel
$\mathcal{T}(\cdot) = \langle {0}\mkern 1.5mu\relax \vert \mkern 1.5mu\relax {\cdot}\mkern 1.5mu\relax \vert \mkern 1.5mu\relax {0}\rangle _A\lvert {0}\rangle \mkern -1.8mu\relax \langle{0}\rvert _B +
(1-\langle {0}\mkern 1.5mu\relax \vert \mkern 1.5mu\relax {\cdot}\mkern 1.5mu\relax \vert \mkern 1.5mu\relax {0}\rangle _A)\,\mathds{1}_B/d_B$, which outputs the maximally mixed state
for nearly all inputs.  (Such a channel may arise as a thermal quantum channel
through a particular type of constraint, such as strict energy conservation with
respect to a Hamiltonian $H = \lvert {0}\rangle \mkern -1.8mu\relax \langle{0}\rvert $.)  In such a case, the channel's entropy
with respect to the maximally mixed state is high, $\sim ({1-1/d_A})\log(d_B)$,
whereas the channel's entropy is zero as attained by $\phi_R = \lvert {0}\rangle \mkern -1.8mu\relax \langle{0}\rvert _R$.

A possible alternative approach to define the thermal channel might have been to
maximize the entropy of a channel's normalized Choi state subject to the
constraints.  (The requirement that the state be maximally mixed on the
reference system could be imposed by further linear constraints.)  From the
state maximum-entropy principle, the solution is a Choi state with the
exponential form of a thermal state.
In fact, this approach coincides with the thermal channel with respect to the
input maximally mixed state $\phi_R = \mathds{1}_R/d_R$.  However, this approach
neglects the fact that the channel can act very differently on distinct input
states.  The channel's entropy, for instance, can vary significantly if it is
computed with respect to a different input state.
Such a behavior can appear naturally for large $n$, a regime in which all
i.i.d.\@ states are nearly perfectly distinguishable; in this regime, a $n$-copy
channel can choose to act as it pleases on different i.i.d.\@ inputs.
The concept of thermal channel defined in this work avoids designating \emph{a
  priori} a preferred input state.  This property is evident in the
microcanonical approach: There exist channels acting on $n$ copies of the inputs
with sharp constraint-measurement statistics for the maximally mixed input state
but where those measurements can fluctuate significantly for other i.i.d.\@
inputs.

Our constructions reduce to the standard thermal state simply by considering the
input system to be a trivial system (a one-dimensional system spanned by a
single state $\lvert {0}\rangle $).  In this case, the channel entropy is the output state's
entropy, and the constraints we consider translate to linear constraints on the
output state.  Therefore, the maximum-channel-entropy principle coincides with
the state maximum-entropy principle.  Furthermore, our microcanonical approach
reduces to the concept of an approximate microcanonical subspace
(cf.~\R\cite{R26}) on $n$ copies of the system,
whose reduced state on a single copy is close to the thermal state.

Our approach works for arbitrary linear constraints on the channel, including
inequality constraints as well as constraints associated with charges that do
not commute.
Inequality constraints are useful, for example, should we wish to constrain an
expectation value to an interval
$\operatorname{tr}[{ C^j_{BR} \mathcal{N}({\Phi_{A:R}}) }] \in [q_j - \epsilon, q_j +
\epsilon]$, as well as for passivity arguments (cf.\@
\cref{z:MO8enbcL}).
Noncommuting constraints appear already in the case of quantum states.
A microcanonical derivation of the thermal states with noncommuting charges
presented a number of challenges owing to the fact that there are generally no
common eigenspaces to noncommuting
observables~\cite{R26}.  
Recently, a number of platforms and settings were investigated where
noncommuting conserved charges can lead to the so-called \emph{non-Abelian
  thermal state}~\cite{R41,R131,R40}.
We anticipate similar exciting applications for thermal quantum channels with
respect to noncommuting constraints.

Recently, \R\cite{R57} considered the problem of
optimizing the relative entropy between quantum channels using semidefinite
programming, by discretizing an integral representation of the relative
entropy~\cite{R132}, and the techniques of
\R\cite{R56}.
Their optimization is well-suited for computing resource measures in a resource
theory of channels, which involves minimizing the channel relative entropy with
its second argument ranging over a convex set of free operations.
Their representation can further be leveraged to numerically compute
approximations of the thermal quantum channel, by optimizing over the first
argument of the channel relative entropy rather than the second.
We employ their techniques for computing the updates in our proof of concept
learning algorithm runs in \cref{z:wXQ-uunR}.
While the optimization in the maximum channel entropy principle has favorable
convexity properties, it appears difficult to obtained closed form expressions
of the ``chemical potentials'' $\mu_j$, the ``operator free energy'' $F_R$, and
of $\phi_R$ in the thermal channel, beyond the conditions stated in
\cref{z:TLiwdJGR}.
However, a similar issue already arises for quantum states: While finding
$\gamma_S(\beta)$ is a convex optimization problem, determining the partition
function $Z(\beta)$ (from which we can compute physical properties of the
system, including a relation between $\beta$ and the constraint energy $E$) can
be hard (cf.\@ e.g.\@ \cite{R133}).

What channel would one find if we minimized the thermodynamic
capacity $T(\mathcal{N})$ rather than maximizing the channel's entropy
${S}_{}^{}({\mathcal{N}})$?  After all, these quantities are equivalent up to a sign
and up to exchanging the output and environment systems
[cf. \cref{z:vTzieZ6M}]; the two optimizations only
differ in whether the channel or its complement is subject to the constraints.
The optimization of the channel entropy is ultimately justified by our
microcanonical channel arguments.  Also, optimizing $T(\mathcal{N})$ appears
poorly motivated for singling out a unique thermal channel in most cases.  In
the absence of constraints, the unique channel that maximizes the channel
entropy is the fully depolarizing channel.  On the other hand, any unital
channel minimizes the thermodynamic capacity if the input and output system
dimensions coincide; the unital channels form a large set that includes
depolarizing channels, the identity channel, as well as measurement/dephasing
channels.
(It can appear counterintuitive that the optimization of the channel entropy and
that of the thermodynamic capacity are qualitatively so different, in the light
of the equivalence of these measures
in~\eqref{z:vTzieZ6M}.  The difference lies in the
dimensionalities of the output and environment systems.  Specifically,
maximizing the channel entropy $A\to B$ is equivalent to minimizing the
thermodynamic capacity of a channel $A\to E$, but whose Stinespring dilation
environment is constrained to be of dimension at most $d_B$ with $d_E = d_Ad_B$.
The latter constraint severely restricts the channels considered in this
optimization.)

Our microcanonical approach to define the thermal channel introduces an
additional form of typicality for quantum channels and multipartite or relative
quantum states~\cite{R134,R135,R136,R126,R137,R82,R138}.
A distinct feature of our approximate microcanonical operator, as opposed to
typical projectors for states, is that relevant concentration
properties hold \emph{for (almost) all input states to the channel}.  Indeed,
the operator $P_{B^nR^n}$ we construct selects a set of quantum channels
$\{{ \mathcal{E}_{A^n\to B^n} }\}$ with some desired concentration properties
by giving high weight to all states of the form
$\mathcal{E}_{A^n\to B^n}(\sigma_{AR}^{\otimes n})$ for $\mathcal{E}$ in this
set (with
$\lvert {\sigma}\rangle _{AR}=\sigma_R^{1/2}\lvert {\Phi_{A:R}}\rangle $; as long as $\sigma_R$ avoids
nearly vanishing eigenvalues), while leaving low weight to all such states
for channels $\mathcal{E}$ that fail to satisfy the desired concentration
properties.
A naive usage of a state typical projector fails to capture this property.  Using
a projector onto suitable charge eigenspaces (or an approximate microcanonical
projector~\cite{R26,R71})
for a state of the form
$\mathcal{E}_{A^n\to B^n}(\sigma_{AR}^{\otimes n})$
depends on a choice of $\sigma_{AR}$, and rejects states of the form
$\mathcal{E}_{A^n\to B^n}(\sigma_{AR}'^{\otimes n})$ because of the different
reduced state on $R^n$.
Rather, the operator must not reject states based on their reduced state on
$R^n$, but rather only select states with specific correlations between $R^n$
and $B^n$.

We furthermore anticipate that our construction can be leveraged to define a
channel analog of a state's \emph{typical projector}.  A quantum channel can be
uniquely singled out by $d_A^2 (d_B^2 - 1)$ independent linear constraints
The microcanonical operator associated with such constraints can be thought of
as a generalized typical subspace for that channel, as it would select only
global channels compatible with the statistics of the $n$-copy i.i.d.\@ channel.
(Again, the typical projector for a channel's Choi state
$[\mathcal{M}(d_A^{-1}\Phi_{A:R})]^{\otimes n}$ would fail to attribute high
weight to operators of the type $[\mathcal{M}(\sigma_{AR})]^{\otimes n}$ where
$\sigma_R$ is not maximally mixed.) %

Defining the \emph{microcanonical channel} from an associated approximate
microcanonical channel operator presents challenges that do
not appear in the case of quantum states.  For quantum states, once a
microcanonical subspace (approximate or not) is identified, it suffices
to normalize the projector onto the subspace to unit trace to find the
most equiprobable state in that subspace.  This state is simultaneously
the most entropic state in that subspace, the unique state
that is invariant under all unitaries within the subspace, as well as
the average state under the measure induced by the Haar measure on those
unitaries.  These properties leave little ambiguity in defining the
microcanonical state.
In the case of quantum channels, however, defining the microcanonical channel
from an approximate microcanonical channel operator $P_{B^nR^n}$ presents
new challenges.  First, it is unclear if the operator $P_{B^nR^n}$ has a reduced
state on $R^n$ that is proportional to the identity $\mathds{1}_{R^n}$, meaning
we might not obtain a valid quantum channel if we simply normalize
$P_{B^nR^n}$ by a suitable constant.  We could attempt to compute the
reduced operator $P_{R^n} = \operatorname{tr}_{B^n}({P_{B^nR^n}})$, and define the now
valid quantum channel
$\Omega'_{B^nR^n} \equiv P_{R^n}^{-1/2} P_{B^nR^n} P_{R^n}^{-1/2}$.
But because of the factors
$P_{R^n}^{-1/2}$, it is unclear if the channel $\Omega'_{B^nR^n}$
inherits the concentration properties captured by $P_{B^nR^n}$ in the
first place---how might we prove that
$\operatorname{tr}[{P_{B^nR^n} \Omega'_{A^n\to B^n}(\sigma_{AR}^{\otimes n})}] \approx 1$?
Alternatively, we could attempt to define a microcanonical channel as an
average over all quantum channels in the ``subspace'' defined by $P_{B^nR^n}$.
Say,
$\Omega''_{B^nR^n} =
\int_{\min_\sigma \operatorname{tr}[P\mathcal{E}(\sigma^{\otimes n})]\geq1-\epsilon} dE_{B^nR^n} \,
E_{B^nR^n}$,
where the measure $dE_{B^nR^n}$ is induced by the Haar measure $dW_{E^nB^nR^n}$
on all isometries $A^n\to B^nE^n$ with $E \simeq BR$.
But it is unclear that there is a transitive unitary group action under which
the measure $dE_{B^nR^n}$ (or $dW_{E^nB^nR^n}$) is invariant, given the presence
of constraints and given the requirement that $E_{B^nR^n}$ be the Choi matrix of
a quantum channel; it is therefore unclear how to compute this average channel,
or if we can show that this channel achieves the maximal channel entropy within the
set of channels with high weight under $P$ for almost all $\sigma$.
An disadvantage of our \cref{z:v1wYNQz5} is that it makes reference
to the channel entropy.  This fact muddles an argument to claim a new
operational interpretation of the channel entropy.  Had the definition of the
microcanonical channel not made reference to the channel entropy, we could
the channel entropy would have found a new operational interpretation as
the quantity to maximize to find reduced states of the microcanonical channel
acting on arbitrary input states.
It is also natural to ask whether we could find an approximate microcanonical
channel operator that is a \emph{projector}, rather than an operator satisfying
$0\leq P_{B^nR^n} \leq \mathds{1}$, analogously to the case of the approximate
microcanonical subspace~\cite{R26,R71}.  It appears possible that we could achieve this by
using an argument similar to the proof in \R\cite{R71}.

We expect several potential improvements to our bounds.
The scaling $y^8$
that appears in these bounds are likely a product of our proof techniques
involving \cref{z:X2zwPk80,z:bQEMTISK}
(\cref{z:JVIqJyHK}); a more refined argument might yield
better bounds. 
Furthermore, the degree of the polynomial in front of the exponential decay terms
in \cref{z:JjqyeW8p} is likely prohibitive in practice for
moderate $n$; it arises from the techniques based on Schur-Weyl duality and
the postselection technique, and might be improved using an alternative analysis.
Also, it appears likely that the protocol defining $P_{B^nR^n}$
could combine the input state estimation with the constraint value estimation,
rather than discarding the samples that were used to estimate the
input state (\cref{z:STrCXktf}).

There are multiple approaches to single out the thermal state beyond Jaynes'
maximum entropy principle and the microcanonical approach
(\cref{z:YsMnGFzU}).  
\begin{figure}
  \centering
  \includegraphics{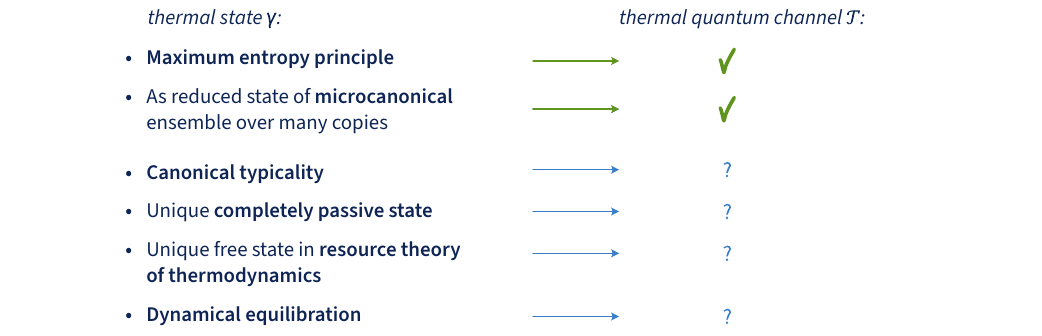}
  \caption{Extending the multiple approaches to define the thermal state to
    quantum channels.  In this work, we extend the maximum entropy principle and
    the microcanonical approach to quantum channels.  We anticipate other
    approaches can be extended to quantum channels, as well.  These approaches
    include canonical typicality~\protect\cite{R27},
    complete
    passivity~\protect\cite{R139,R29,R26}, free resources in the resource
    theory of thermodynamics~\protect\cite{R32,R34} and standard dynamical equilibration arguments
    (e.g.~\protect\cite{R19}).}
  \label{z:YsMnGFzU}
\end{figure}
We anticipate a research program of
understanding how to extend these definitions from states to channels, and to
determine whether they lead to the same thermal quantum channel.  
One such approach is to invoke dynamical equilibration arguments~\cite{R0,R1,R2,R19,R20,R21,R22}.
The thermal state is typically the state to which a many-body system
equilibrates after long times.
We anticipate such arguments could be extended to the case of channels, to prove
that the system's evolution $\mathcal{U}_t$ \emph{equilibrates} in some sense to
the thermal quantum channel.  This equilibration might happen on average,
$\int dt\,\mathcal{U}_t \approx \mathcal{T}$, or might be apparent for a set of
accessible observables $\{{ C^j_{BR} }\}$:
$\operatorname{tr}[{C^j_{BR} \mathcal{U}_t(\Phi_{A:R}) }] \to \operatorname{tr}[{ C^j_{BR}
\mathcal{T}(\Phi_{A:R}) }]$ as $t\to\infty$.  Such arguments would likely require
finer assumptions about the details of the evolution $\mathcal{U}_t$ that go
beyond a maximum channel entropy principle or a microcanonical approach.  This
type of argument would provide an appealing picture of how the evolution of a
system, seen as a full quantum process, converges to the thermal quantum
channel.
Another approach to characterize the thermal state is via the resource theory
of thermodynamics.
In a resource theory of quantum
channels~\cite{R50,R53,R140,R126,R52},
a measure of resourcefulness of a channel $\mathcal{N}$ is the channel
relative entropy with respect to the set of free channels, namely
the smallest channel relative entropy of $\mathcal{N}$ with respect to
some free channel $\mathcal{M}$~\cite{R53,R51,R57}.
The problem considered in this work
is a related problem: supposing we have a single free channel, the maximally
depolarizing channel $\mathcal{D}$, then our task is to find the channel
$\mathcal{N}$ that has the smallest channel relative entropy with respect
to $\mathcal{D}$, subject to a set of constraints. 
Our approach might therefore identify free states in a resource
theory of channels in the presence of additional, linear constraints on the
channels.  For example, if we have a global symmetry where operators are
restricted to act within charge sectors only, then the thermal channel is
a depolarizing map acting within each sector.  This channel appears suitable
for use as a free channel in such a resource theory (see \cref{z:MO8enbcL}
for a discussion of some challenges).

As also discussed in our companion overview paper (ref.~\cite{R18}),
the thermal quantum channel is the ``least informative'' channel that can model
some unknown or complex thermalizing dynamics of a many-body system.  The channel
nature of the problem enables $\mathcal{T}$ to model partial or ``local'' thermalizing
effects that keep some memory of the initial state of the system.  Such is the
case in the example of the average energy conservation constraint in
\cref{z:VuBBeOAb}.
The thermal quantum channel might therefore provide a well-founded model for
local relaxation effects that are known to occur, for example, in Gaussian
systems~\cite{R141,R142,R143}.  
We anticipate further uses of interest for the thermal quantum channel to model
settings with several thermalization mechanisms operating on different
time scales, such as in hydrodynamic regimes~\cite{R144,R145,R146}.

Finally, our work highlights an exciting opportunity to extend a vast landscape
of concepts and methods from the thermal state to the quantum quantum channel,
thereby establishing to which extent
the thermal quantum channel can enjoy a similar level of universality and broad
applicability as the thermal state.

\paragraph*{Note added:}
Our results were submitted to
\href{https://sites.google.com/view/beyondiid13}{%
  \emph{Beyond i.i.d.\@ in information theory} 2025} in April 2025
and accepted as a talk in early June 2025 (cf.\@ 
\url{https://sites.google.com/view/beyondiid13/program}).
During the final stages of completion of our manuscript,
a paper with independent related work by Siddhartha Das and Ujjwal Sen
appeared on the arXiv on July 1, 2025 [Das and Sen, arXiv:2506.24079].

\begin{acknowledgments}
The authors thank
David Jennings,
Jens Eisert,
Michael Walter,
Hamza Fawzi, Omar Fawzi,
Gereon Ko\ss{}mann,
Mark Mitchison,
John Preskill,
and Andreas Winter
for insightful discussions.
PhF is supported by the project FOR 2724 of the Deutsche Forschungsgemeinschaft
(DFG).
\end{acknowledgments}

\appendix

\section{Some general lemmas}
\label{z:0JwGPCOy}
\label{z:JVIqJyHK}

Recall that $P(\rho,\sigma) = \bigl[{ 1 - F^2(\rho,\sigma) }\bigr]^{1/2}$ is the
\emph{purified distance} of states.

\begin{lemma}[Reference state smoothing]
  \label{z:X2zwPk80}
  Let $A\simeq R$ and let $\rho_R$, $\sigma_R$ be any two quantum states on $R$.
  Then
  \begin{align}
    F\Bigl({\rho_R^{1/2} \lvert {\Phi_{A:R}}\rangle , \sigma_R^{1/2} \lvert {\Phi_{A:R}}\rangle }\Bigr)
    =
    \operatorname{tr}\bigl({\rho_R^{1/2}\sigma_R^{1/2}}\bigr)
    \geq 1 - \sqrt{ 2 P\bigl({\sigma_R, \rho_R}\bigr) }\ .
  \end{align}
\end{lemma}
\begin{proof}[**z:X2zwPk80]
  A key ingredient of this proof is a result presented in Bhatia's book on
  matrix analysis \cite[Theorem~X.I.3]{R147}.  This
  result implies that for all positive semidefinite operators $A,B$, we have
  \begin{align}
    \bigl \lVert { \sqrt{A} - \sqrt{B} }\bigr \rVert _2
    \leq \bigl \lVert { \sqrt{ \lvert { A - B }\rvert  } }\bigr \rVert _2\ .
  \end{align}
  Let $w = P(\sigma_R, \rho_R) = \sqrt{1 - F^2(\sigma_R, \rho_R)}$.  By the
  theorem in Bhatia's book,
  \begin{align}
    \bigl \lVert { \sqrt{\rho} - \sqrt{\sigma} }\bigr \rVert _2
    \leq \bigl \lVert { \sqrt{ \lvert { \rho - \sigma }\rvert  } }\bigr \rVert _2
    = \bigl[{ \operatorname{tr}\, \lvert {\rho-\sigma}\rvert  }\bigr]^{1/2} =\sqrt{ 2D(\rho,\sigma) } \leq \sqrt{2w} \ ,
  \end{align}
  writing $\rho\equiv\rho_R$ and $\sigma\equiv\sigma_R$ for short.
  We then see, using H\"older's inequality, that
  \begin{align}
    \operatorname{tr}\bigl({\rho^{1/2} \sigma^{1/2}}\bigr)
    &= \operatorname{tr}({\rho}) + \operatorname{tr}\bigl[{\rho^{1/2}\bigl({\sigma^{1/2} - \rho^{1/2}}\bigr)}\bigr]
      \geq 1 - \bigl \lVert {\rho^{1/2}\bigl({\sigma^{1/2} - \rho^{1/2}}\bigr)}\bigr \rVert _1
      \nonumber\\
    &\geq 1 - \bigl \lVert {\rho^{1/2}}\bigr \rVert _2 \,\bigl \lVert {\sigma^{1/2} - \rho^{1/2}}\bigr \rVert _2
      \geq 1 - \sqrt{2w}\ ,
  \end{align}
  using the fact that $\bigl \lVert {\rho^{1/2}}\bigr \rVert _2 = \sqrt{\operatorname{tr}({\rho})} = 1$.
  The claim follows by noting that
  \begin{align}
    F\bigl({\rho_R^{1/2} \lvert {\Phi_{A:R}}\rangle , \sigma_R^{1/2} \lvert {\Phi_{A:R}}\rangle }\bigr)
    = \bigl \lvert { \langle {\Phi_{A:R}}\mkern 1.5mu\relax \vert \mkern 1.5mu\relax { \rho_R^{1/2} \sigma_R^{1/2} }\mkern 1.5mu\relax \vert \mkern 1.5mu\relax {\Phi_{A:R}}\rangle  }\bigr \rvert 
    = \operatorname{tr}\bigl({\rho_R^{1/2}\sigma_R^{1/2}}\bigr)\ .
    \tag*\qedhere
  \end{align}
\end{proof}

The gentle measurement lemma has a widespread use across quantum information
theory and appears in multiple standard references, including textbooks such
as~\cite{R79}.  A proof of the specific version we state
here can be found, for instance, as \cite[Lemma~B.2]{R82}.
\begin{proposition}[Gentle measurement lemma]
  \noproofref
  \label{z:bQEMTISK}
  Let $\rho$ be any subnormalized quantum state and let $0 \leq R \leq \mathds{1}$.
  Let $\delta\geq 0$ such that $\operatorname{tr}({R^2\rho})\geq 1 - \delta$.  Then
  \begin{align}
    P \bigl({\rho, R \rho R}\bigr) \leq \sqrt{2\delta}\ .
  \end{align}
\end{proposition}

The following is a straightforward consequence of the data processing inequality
for the fidelity.  It is convenient to have it in this form for direct use in our proofs:
\begin{lemma}[Upper bound on fidelity through distinguishing test]
  \label{z:MduRjmZz}
  Let $\rho,\sigma$ be any subnormalized quantum states and let
  $\{{ Q, Q^\perp }\}$ be a two-outcome POVM.  Then
  \begin{align}
    F(\rho, \sigma) \leq \sqrt{\operatorname{tr}({Q\rho})} + \sqrt{\operatorname{tr}({Q^\perp\sigma})}\ .
  \end{align}
\end{lemma}
\begin{proof}[**z:MduRjmZz]
  From the data processing inequality for the fidelity,
  \begin{align}
    F(\rho,\sigma)
    &\leq F\Bigl({ \bigl[{\operatorname{tr}({Q\rho}), \operatorname{tr}({Q^\perp \rho}) }\bigr],
      \bigl[{\operatorname{tr}({Q\sigma}), \operatorname{tr}({Q^\perp \sigma}) }\bigr] }\Bigr)
      \nonumber\\
    &= \sqrt{\operatorname{tr}({Q\rho})}\sqrt{\operatorname{tr}({Q\sigma})} + \sqrt{\operatorname{tr}({Q^\perp\rho})}\sqrt{\operatorname{tr}({Q^\perp\sigma})}
      \leq \sqrt{\operatorname{tr}({Q\rho})} + \sqrt{\operatorname{tr}({Q^\perp\sigma})}\ .
      \tag*\qedhere
  \end{align}
\end{proof}

The fidelity between two classical-quantum states takes a simple form.
\begin{lemma}
  \label{z:2jsqnpDK}
  Let $\{{ p_k }\}$ be a subnormalized probability distribution and let
  $\{{ \rho_k }\}$, $\{{ \sigma_k }\}$ be two families of quantum states.  Then
  \begin{align}
    F\biggl({ \sum_k p_k \lvert {k}\rangle \mkern -1.8mu\relax \langle{k}\rvert \otimes\rho_k \,,\; \sum_k p_k \lvert {k}\rangle \mkern -1.8mu\relax \langle{k}\rvert \otimes\sigma_k }\biggr)
    = \sum p_k F\bigl({\rho_k, \sigma_k}\bigr)\ .
  \end{align}
\end{lemma}
\begin{proof}[**z:2jsqnpDK]
  Write
  \begin{align}
    F\biggl({ \sum_k p_k \lvert {k}\rangle \mkern -1.8mu\relax \langle{k}\rvert \otimes\rho_k \,,\; \sum_k p_k \lvert {k}\rangle \mkern -1.8mu\relax \langle{k}\rvert \otimes\sigma_k }\biggr)
    &= \Bigl \lVert { \sum \lvert {k}\rangle \mkern -1.8mu\relax \langle{k}\rvert  \otimes \bigl({ p_k \rho_k^{1/2} \sigma_k^{1/2} }\bigr) }\Bigr \rVert _1
    = \Bigl \lVert { \bigoplus \bigl({ p_k \rho_k^{1/2} \sigma_k^{1/2} }\bigr) }\Bigr \rVert _1
    \nonumber\\
    &= \sum \Bigl \lVert { p_k \rho_k^{1/2} \sigma_k^{1/2} }\Bigr \rVert _1
    = \sum p_k F\bigl({\rho_k, \sigma_k}\bigr)\ .
  \end{align}
\end{proof}

We also need the following generalization of the ``pinching lemma.''  This
standard lemma has appeared many times in the quantum information literature;
cf.\@ e.g.\@~\cite[Lemma~B.1]{R82} for a proof.
\begin{lemma}
  \noproofref
  \label{z:h6tAs5Us}
  Let $\{{ E_k }\}_{k=1}^M$ be a collection of $M$ operators.  Then, for any
  $A\geq 0$,
  \begin{align}
    \biggl({\sum_{k=1}^M E_k }\biggr) \, A \, \biggl({\sum_{k=1}^M E_k }\biggr)^\dagger
    \leq M\, \sum_{k=1}^M E_k A E_k^\dagger \ .
  \end{align}
\end{lemma}

In our proofs, we need a POVM that is capable, when acting on an $m$-fold
i.i.d.\@ state $\sigma^{\otimes m}$, of estimating the state $\sigma$.  While
multiple POVMs have this property (cf.\@
e.g.\@~\cite{R78}), we focus on the following
\emph{pretty good measurement}~\cite{R148,R149,R150,R79}.

\begin{proposition}
  \noproofref
  \label{z:Ehi5guwR}
  Let $R$ be a quantum system and let $m>0$.  For any $\tilde\sigma_R$, let
  \begin{align}
    R_{R^m}^{(\tilde\sigma)}
    &\equiv \bigl[{\tilde\sigma_R^{\otimes m}}\bigr]^{1/2} \, \zeta_{R^m}^{-1/2}
      = R_{R^m}^{(\tilde\sigma) \dagger}\ ;
  \end{align}
  where $\zeta_{R^m} = \int d\sigma'_{R} \, \sigma_{R}'^{\otimes m}$ is the
  de~Finetti state introduced in the main text and in
  \cref{z:Mf3kbNlp}.  Then
  \begin{align}
    \int d\tilde\sigma \, R_{R^m}^{(\tilde\sigma)\dagger} R_{R^m}^{(\tilde\sigma)} = \mathds{1}\ ,
    \label{z:1gWrSiuG}
  \end{align}
  so $\bigl\{{ R^{(\tilde\sigma)\dagger}R^{(\tilde\sigma)} }\bigr\}$
  is a POVM.  Furthermore, for any
  $x>0$,
  \begin{align}
    \int_{F^2(\tilde\sigma,\sigma)\leq e^{-x}} d\sigma\,
    \operatorname{tr}\bigl({ R^{(\tilde\sigma)\dagger}R^{(\tilde\sigma)}\,
    \sigma^{\otimes m} }\bigr)
    \leq \operatorname{poly}(m) \exp(-mx)\ .
  \end{align}
\end{proposition}
\begin{proof}[**z:Ehi5guwR]
  That $R_{R^m}^{(\tilde\sigma)\dagger} = R_{R^m}^{(\tilde\sigma)}$ follows from
  the fact that $\zeta_{R^m}$ is constant over each Schur-Weyl block (cf.\@
  e.g.\@ \cref{z:TvA2E9KA}) and therefore commutes with
  the permutation-invariant operator $\tilde\sigma_{R}^{\otimes m}$.
  \Cref{z:1gWrSiuG} holds by definition of
  $\zeta_{R^m}$.

  Now let $x>0$ and write the shorthand
  $M_{R^m}^{(\tilde\sigma)} \equiv
  R_{R^m}^{(\tilde\sigma)\dagger}R_{R^m}^{(\tilde\sigma)}$.  We make use of
  Schur-Weyl notation introduced in \cref{z:pJbOynb9}.  In
  \cite[\S\,V.A, after Eq.~(16)]{R78}, it was proven
  that for any states $\tilde\sigma_R$, $\sigma_R$,
  \begin{align}
    \operatorname{tr}\bigl[{ M_{R^m}^{(\tilde\sigma)}  \sigma_R^{\otimes m}}\bigr]
    &\leq
    \sum_{\lambda\in\Young(d_R, m)}
    \frac{ d_{\mathcal{Q}_\lambda}^2 }{
          {e}^{m{S}_{}^{}({\bar\lambda})} (d_{\mathcal{Q}_\lambda}\zeta_\lambda) }
    \bigl[{F(\sigma_R,\tilde\sigma_R)}\bigr]^{2m}
    \ ;
    &
    \zeta_\lambda
      &= \frac1{d_{\mathcal{Q}_\lambda}} \int d\sigma_R \operatorname{tr}\bigl[{q_\lambda(\sigma_R)}\bigr]\ .
  \end{align}
  The coefficients $\zeta_\lambda$ are precisely the the Schur-Weyl block
  coefficients of the de Finetti state
  $\zeta_{R^m} = \sum_\lambda \zeta_\lambda \Pi_{R^m}^\lambda$.
  \Cref{z:TvA2E9KA} provides the values of these
  coefficients,
  $\zeta_\lambda = d_{\mathcal{Q}_\lambda} /
  \bigl({d_{\mathcal{P}_\lambda}d_{\mathrm{Sym}(m,d_R^2)}}\bigr)$.  Therefore, for any
  $\tilde\sigma_R$, $\sigma_R$,
  \begin{align}
    \operatorname{tr}\bigl[{ M_{R^m}^{(\tilde\sigma)}  \sigma_R^{\otimes m}}\bigr]
    \leq \sum_{\lambda\in\Young(d_R,m)}
    \frac{d_{\mathcal{P}_\lambda}  d_{\mathrm{Sym}(m,d_R^2)}}{{e}^{m{S}_{}^{}({\bar\lambda})}}
    \bigl[{F(\sigma_R,\tilde\sigma_R)}\bigr]^{2m}
    \leq \operatorname{poly}({m}) \bigl[{F(\sigma_R,\tilde\sigma_R)}\bigr]^{2m}\ ,
  \end{align}
  using the upper bound $d_{\mathcal{P}_\lambda} \leq {e}^{m{S}_{}^{}({\bar\lambda})}$.
  This enables us to compute
  \begin{align}
    \int_{F^2(\tilde\sigma, \sigma) \leq {e}^{-x}} d\sigma
    \operatorname{tr}\bigl({ M^{(\tilde\sigma)}_{R^m} \, \sigma_R^{\otimes  m} }\bigr)
    \leq
    \operatorname{poly}(m)\,{e}^{-mx}\ ,
  \end{align}
  proving the last part of the proposition.
\end{proof}

\section{Proofs for the maximum-channel-entropy derivation of the thermal
  channel}
\label{z:N3vMol3P}
\label{z:QmA3ZHFT}

\subsection{Lemma: thermal channels with respect to any $\phi$ lie in the interior of the objective domain}

We first prove a lemma that ensures our approach to find the thermal channel
with respect to any $\phi_R$ does not miss any solutions.  Our approach involves
writing a Lagrangian of the problem including the relevant constraints, and
applying the Karush-Kuhn-Tucker conditions to find optimal
solutions~\cite{R95}.  This approach, however, might
fail to find optimal solutions that lie on the boundary of the domain of the
optimization's objective function.  The following lemma provides a technical
statement enabling us to rule out such an undesirable situation in the proofs of
\cref{z:hd7.cLe.,z:gV43EWT.}.

\begin{lemma}
  \label{z:dTDQ-ZX2}
  Consider the following optimization problem: %
  \begin{align}
    \label{z:1w8WlXSN}
    \begin{aligned}[t]
    \textup{maximize:} \quad
    & f_{\mathrm{obj}}(\mathcal{N}_{A\to B})
    \\
    \textup{over:}\quad
    & \mathcal{N}_{A\to B}\ \textup{c.p., t.p.}
    \\
    \textup{such that:}\quad
    &f_{\mathrm{cons},j}(\mathcal{N}_{A\to B}) \leq 0\ \quad\forall\,j=1,\ldots,J',
  \end{aligned}
  \end{align}
  with
  \begin{align}
    f_{\mathrm{obj}}(\mathcal{N}_{A\to B})
    = {S}_{}^{}({\mathcal{N}_{A\to B}(\phi_{AR})}) + f_{\mathrm{Q}}(\mathcal{N}_{A\to B})
  \end{align}
  where $f_{\mathrm{Q}}(\mathcal{N}_{A\to B})$ is a quadratic function of
  $\mathcal{N}_{A\to B}$, where each $f_{\mathrm{cons},j}$ is linear
  in $\mathcal{N}_{A\to B}$, and where $\lvert {\phi}\rangle _{AR}$
  is a fixed pure state of the form $\lvert {\phi}\rangle _{AR} \equiv \phi_A^{1/2}\lvert {\Phi_{A:R}}\rangle $.
  Assume that there exists some quantum channel
  $\mathcal{N}^{(\mathrm{int})}_{A\to B}$ with
  $N^{(\mathrm{int})}_{BR} \equiv \mathcal{N}^{(\mathrm{int})}_{A\to
    B}(\Phi_{A:R}) >0$ that is feasible, i.e., that satisfies all the
  problem's constraints.  Then any optimal channel $\mathcal{N}_{A\to B}$
  in~\eqref{z:1w8WlXSN}
  is such that
  $\mathcal{N}_{A\to B}({\phi_{AR}})$ has full rank within the support of
  $\mathds{1}_B\otimes\Pi^{\phi_R}_R$.
\end{lemma}

The optimization 
problem~\eqref{z:1w8WlXSN}
is meant to cover all the settings considered
in~\cref{z:u4ZpmFvZ}.  Linear equality constraints
can be written as a pair of inequality constraints, one in each direction.  The
optimization objectives ${S}_{}^{}({\mathcal{N}(\phi_{AR})}) - {S}_{}^{}({\phi_R})$,
$-{D}_{\phi}^{}({\mathcal{N}}\mathclose{}\,\Vert\,\mathopen{}{\mathcal{M}}) =
{S}_{}^{}({\mathcal{N}(\phi_{AR})})
+ \operatorname{tr}\bigl[{\mathcal{N}(\phi_{AR})\,\log\bigl({\phi_R^{1/2} M_{BR} \phi_R^{1/2}}\bigr)}\bigr]$,
and $-{D}_{}^{}({\mathcal{N}(\phi_{AR})}\mathclose{}\,\Vert\,\mathopen{}{\mathcal{M}(\phi_{AR})})
+ \sum \tilde\eta_m [{s_m - \operatorname{tr}({E_{BR}^m N_{BR}})}]^2$ all fit in the structure
of~\eqref{z:1w8WlXSN}.

\begin{proof}[**z:dTDQ-ZX2]
  Let $\mathcal{N}_{A\to B}^{(0)}$ be any channel that does not satisfy the
  desired conclusion, that is, suppose that there exists a nonzero projector
  $P_{BR}$ that lies within the support of $\mathds{1}_B\otimes\Pi^{\phi_R}_R$
  such that $\mathcal{N}_{A\to B}^{(0)}({\phi_{AR}}) \, P_{BR} = 0$.  We'll show
  that $\mathcal{N}_{A\to B}^{(0)}$ cannot be optimal
  in~\eqref{z:1w8WlXSN}.
  
  For any $\theta\in[0,1]$, let
  \begin{align}
    \mathcal{N}^{(\theta)}_{A\to B}
    &\equiv (1-\theta)\mathcal{N}_{A\to B}^{(0)} + \theta \mathcal{N}_{A\to B}^{(\mathrm{int})}\ ;
    &
      \rho_{BR}^{(\theta)}
    &\equiv \mathcal{N}^{(\theta)}_{A\to B}(\phi_{AR})\ .
  \end{align}
  The state $\rho_{BR}^{(\theta)}$ always lies within the support of
  $\mathds{1}_B\otimes\Pi^{\phi_R}_R$ by construction.  Furthermore, for any
  $\theta\in(0,1]$, the state $\rho_{BR}^{(\theta)}$ always has full rank within
  the support of $\mathds{1}_B\otimes\Pi^{\phi_R}_R$.  This can be seen because
  $\mathcal{N}^{(\mathrm{int})}_{A\to B}$, having positive definite Choi matrix,
  can be written as a convex combination of a completely depolarizing channel
  (with Choi matrix proportional to the identity) and another completely
  positive map; the completely depolarizing channel component guarantees that
  $\mathcal{N}_{A\to B}^{(\mathrm{int})}(\phi_{AR})$ has full rank within
  $\mathds{1}_B\otimes\Pi^{\phi_R}_R$.  Therefore,
  $\rho_{BR}^{(\theta)}$ has full rank within $\mathds{1}_B\otimes\Pi^{\phi_R}_R$
  for $\theta\in(0,1]$.  On the other hand, recall that
  $\rho_{BR}^{(\theta{=}0)} P_{BR} = 0$ with $P_{BR}$ a nontrivial projector acting
  within $\mathds{1}_B\otimes\Pi^{\phi_R}_R$'s support.

  The channel $\mathcal{N}^{(\theta)}_{A\to B}$ obeys all problem constraints for
  all $\theta\in[0,1]$, by convexity of the constraints.  We'll show that there
  exists $\theta\in(0,1]$ for which $\mathcal{N}^{(\theta)}_{A\to B}$ achieves a
  better objective value than $\mathcal{N}^{(0)}_{A\to B}$, and hence the latter
  cannot be optimal.  The objective value achieved by
  $\mathcal{N}^{(\theta)}_{A\to B}$ is
  \begin{align}
    f_{\mathrm{obj}}(\theta)
    &\equiv f_{\mathrm{obj}}\bigl({\mathcal{N}_{A\to B}^{(\theta)}}\bigr)
    = s(\theta) + f_{\mathrm{Q}}(\theta)\ ;
    &
    s(\theta)
    &\equiv {S}_{}^{}\bigl ({\rho_{BR}^{(\theta)}}\bigr ) \ ;
    &
    f_{\mathrm{Q}}\bigl({\theta}\bigr)
    &\equiv
    f_{\mathrm{Q}}\bigl({\mathcal{N}_{A\to B}^{(\theta)}}\bigr)\ .
  \end{align}
  For $\theta\in(0,1)$, we can compute
  \begin{align}
    \frac{d}{d\theta} s(\theta)
    =
    -\operatorname{tr}\Bigl[{ \Bigl({\log\bigl({\rho_{BR}^{(\theta)}}\bigr) + \mathds{1}}\Bigr)\,\frac{d}{d\theta} \rho_{BR}^{(\theta)} }\Bigr]
    \ ,
  \end{align}
  where
  \begin{align}
    \frac{d}{d\theta} \rho_{BR}^{(\theta)}
    = \mathcal{N}^{(\mathrm{int})}_{A\to B}(\phi_{AR}) - \mathcal{N}^{(0)}_{A\to B}(\phi_{AR})
    = \rho_{BR}^{(\theta{=}1)} - \rho_{BR}^{(\theta{=}0)}\ ,
  \end{align}
  and therefore
  \begin{align}
    \frac{d}{d\theta} s(\theta)
    =
    -\operatorname{tr}\Bigl[{
    \rho_{BR}^{(\theta{=}1)}\,\log\bigl({\rho_{BR}^{(\theta)}}\bigr)
    }\Bigr]
    + \operatorname{tr}\Bigl[{
    \rho_{BR}^{(\theta{=}0)}\,\log\bigl({\rho_{BR}^{(\theta)}}\bigr)
    }\Bigr]\ .
  \end{align}
  Using $\rho_{BR}^{(\theta)}\geq (1-\theta)\rho_{BR}^{(0)}$, the operator
  monotonicity of the logarithm, and the pinching inequality
  $\rho_{BR}^{(\theta)} \leq 2P_{BR} \rho_{BR}^{(\theta)} P_{BR} + 2
  P_{BR}^\perp \rho_{BR}^{(\theta)} P_{BR}^\perp$ with here
  $P_{BR}^\perp \equiv \Pi^{\phi_R}_{BR} - P_{BR}$, we find
  \begin{align}
    \frac{d}{d\theta} s(\theta)
    &\geq
      -\operatorname{tr}\mathopen{}\left\{{
      \rho_{BR}^{(\theta{=}1)} \,
      \begin{bmatrix}
        P_{BR}\,\log\bigl({ 2P_{BR} \rho_{BR}^{(\theta)} P_{BR} }\bigr) & 0
        \\
        0 & P_{BR}^\perp\,\log\bigl({ 2P_{BR}^\perp \rho_{BR}^{(\theta)} P_{BR}^\perp }\bigr)
      \end{bmatrix}
      }\right\}\mathclose{}
      - ({1-\theta})\,{S}_{}^{}({\rho_{BR}^{(\theta)}})\ ,
  \end{align}
  where the matrix notation separates the blocks associated with the supports of
  $P_{BR}$ and $P_{BR}^\perp$, respectively.  Further using
  $P_{BR}\rho_{BR}^{(\theta)} P_{BR} = \theta P_{BR} \rho_{BR}^{(\theta{=}1)}
  P_{BR}$ and
  $P_{BR}^\perp \rho^{(\theta)}_{BR} P_{BR}^\perp \leq
  \mathds{1}_{B}\otimes\Pi^{\phi_R}_R$, along with $0\leq \theta\leq 1$,
  ${S}_{}^{}({\rho_{BR}^{(\theta)}})\leq\log(d_Bd_R)$, we find
  \begin{align}
    \frac{d}{d\theta} s(\theta)
    &\geq
      -\operatorname{tr}\Bigl\{{
      \rho_{BR}^{(\theta{=}1)} P_{BR}\,\Bigl[{\log(2\theta) P_{BR}
      + \log\bigl({P_{BR}\rho_{BR}^{(\theta{=}1)} P_{BR}}\bigr)}\Bigr] }\Bigr\}
      - \log(d_Bd_R)
      \nonumber\\
    &= -\log({2\theta})\,\operatorname{tr}\bigl[{\rho_{BR}^{(\theta{=}1)} P_{BR}}\bigr]
      + {S}_{}^{}({P_{BR} \rho_{BR}^{(\theta{=}1)} P_{BR}})
      - \log(d_Bd_R)
      \ .
  \end{align}
  The $-\log(2\theta)$ term has some positive nonzero coefficient,
  since $\operatorname{tr}\bigl({\rho_{BR}^{(\theta{=}1)} P_{BR}}\bigr) > 0$, and the entropy term
  is some constant independent of $\theta$.
  On the other hand, the function
  $f_{\mathrm{Q}}(\mathcal{N}_{A\to B})$ is quadratic in $\mathcal{N}_{A\to B}$;
  thus, the function $f_{\mathrm{Q}}(\theta)$ is quadratic in $\theta$ and
  $(d/d\theta) f_{\mathrm{Q}}(\theta) = f_{\mathrm{Q},1}\theta + f_{\mathrm{Q},0}$
  for some $f_{\mathrm{Q},1},f_{\mathrm{Q},0}\in\mathbb{R}$.
  Therefore,
  \begin{align}
    \frac{d}{d\theta} f_{\mathrm{obj}}(\theta)
    &= \frac{d}{d\theta} \bigl({ s(\theta) + f_{\mathrm{Q}}(\theta) }\bigr)
    \to \infty\quad\text{as}\ \theta\to 0\ .
  \end{align}
  Given as $f_{\mathrm{obj}}(\theta)$ is a continuous function on $[0,1]$ and is
  differentiable
  on $(0,1)$, the fact that its derivative is strictly positive for small enough
  $\theta$ ensures that $f_{\mathrm{obj}}(\theta)$ is strictly increasing
  as $\theta$ increases away from 0, for small enough $\theta$.
  Therefore $\theta=0$ cannot be
  the maximum of $f_{\mathrm{obj}}(\theta)$, and $\mathcal{N}_{A\to B}^{(0)}$
  cannot be optimal in
  \eqref{z:1w8WlXSN}.
\end{proof}

\subsection{Structure of the generalized thermal channel: Proof of
  \cref{z:gV43EWT.}}

\begin{proof}[*z:gV43EWT.]
  As a matter of convenience, we formally replace the objective function
  in~\eqref{z:z4Hq5iYa}
  by the function
  \begin{align}
    f_{\mathrm{obj}}(N_{BR})
    &= {D}_{}^{}\bigl ({\phi_R^{1/2} N_{BR}\phi_R^{1/2}}\mathclose{}\,\big \Vert\,\mathopen{}{\phi_R^{1/2} M_{BR} \phi_R^{1/2} }\bigr )
      + \sum \tilde\eta_m \bigl[{s_m - \operatorname{tr}\bigl({E^m_{BR} N_{BR}}\bigr)}\bigr]^2
      + \bigl({1 - \operatorname{tr}[{\mathcal{N}(\phi_{AR})}]}\bigr)
      \nonumber\\
    &= \operatorname{tr}\bigl[{\phi_R^{1/2} N_{BR}\phi_R^{1/2}\,\log\bigl({\phi_R^{1/2} N_{BR}\phi_R^{1/2}}\bigr)}\bigr]
      - \operatorname{tr}\bigl[{\phi_R^{1/2} N_{BR}\phi_R^{1/2}\,\log\bigl({\phi_R^{1/2} M_{BR} \phi_R^{1/2}}\bigr)}\bigr]
      \nonumber\\
    &\quad+ \sum \tilde\eta_m \bigl[{s_m - \operatorname{tr}\bigl({E^m_{BR} N_{BR}}\bigr)}\bigr]^2
            + \bigl({1 - \operatorname{tr}[{\mathcal{N}(\phi_{AR})}]}\bigr)\ ,
      \label{z:Wjf46Lne}
  \end{align}
  where ${D}_{}^{}({X}\mathclose{}\,\Vert\,\mathopen{}{Y}) = \operatorname{tr}\bigl({X\log X}\bigr) - \operatorname{tr}\bigl({X\log Y}\bigr)$ is formally extended
  to arguments $X,Y$ that are arbitrary positive semidefinite operators.
  The additional term $\bigl({1 - \operatorname{tr}[{\mathcal{N}(\phi_{AR})}]}\bigr)$ is irrelevant
  for any choice of variable $\mathcal{N}$ that obeys the problem constraints,
  but will simplify the computation of the gradients of the objective function
  later on.  Clearly, the modified problem yields the same optimal variables as
  the original one
  in~\eqref{z:z4Hq5iYa}.
  The assumption that there exists $N_{BR} > 0$ that satisfies all the problem
  constraints enables us to invoke
  \Cref{z:dTDQ-ZX2}.  We are thus
  guaranteed that any optimal solution $N_{BR}$ to the problem~\eqref{z:z4Hq5iYa}
  must be such that $\phi_R^{1/2}\, N_{BR}\, \phi_R^{1/2}$, and therefore
  $\Pi^{\phi_R}_R\, N_{BR}\, \Pi^{\phi_R}_R$, has full rank within the support
  of $\mathds{1}_B\otimes\Pi^{\phi_R}_R$.  
  The objective function $f_{\mathrm{obj}}(N_{BR})$ is well defined and
  continuous for all $N_{BR}\geq 0$.  However, since its value only depends on
  $\phi_R^{1/2} N_{BR} \phi_R^{1/2}$, we extend this function formally as a
  function whose domain is all Hermitian matrices $N_{BR} = N_{BR}^\dagger$ that
  satisfy $\Pi^{\phi_R}_R\, N_{BR}\,\Pi^{\phi_R}_R \geq 0$.  (%
  In our optimization, we'll still require $N_{BR}\geq 0$; simply, rather than
  treating this condition through the domain of the objective function, we'll
  formally impose it as an explicit constraint.)  Let us write
  \begin{align}
    N_{BR} = \begin{bmatrix}
               N_{BR}^{00} & N_{BR}^{01} \\[5pt]
               N_{BR}^{01\,\dagger} & N_{BR}^{11}
           \end{bmatrix}\ ,
    \label{z:81SIufBH}
  \end{align}
  with $N_{BR}^{00} = N_{BR}^{00\,\dagger}$,
  $N_{BR}^{11} = N_{BR}^{11\,\dagger}$, and where the matrix blocks correspond
  to the subspaces spanned by $\Pi^{\phi_R}_R$, $\Pi^{\phi_R\perp}_R$.  The
  requirement that $\Pi^{\phi_R}_R\, N_{BR}\,\Pi^{\phi_R}_R \geq 0$ then
  translates into the condition $N_{BR}^{00} \geq 0$; the set of operators
  $N_{BR}$ we consider formally as the domain of our objective function is
  \begin{align}
    \mathfrak{S} \equiv \mathopen{}\left\{{
    N_{BR} = \begin{bmatrix}
               N_{BR}^{00} & N_{BR}^{01} \\[5pt]
               N_{BR}^{01\,\dagger} & N_{BR}^{11}
             \end{bmatrix}
    :\  N_{BR} = N_{BR}^\dagger \text{ and } N_{BR}^{00} \geq 0
    }\right\}\mathclose{}\ .
  \end{align}
  The interior of this set is
  \begin{align}
    \topoint({\mathfrak{S}}) = \mathopen{}\left\{{
    N_{BR} = \begin{bmatrix}
               N_{BR}^{00} & N_{BR}^{01} \\[5pt]
               N_{BR}^{01\,\dagger} & N_{BR}^{11}
             \end{bmatrix}
    :\  N_{BR} = N_{BR}^\dagger \text{ and } N_{BR}^{00} > 0
    }\right\}\mathclose{}\ .
  \end{align}
  As we have seen, \cref{z:dTDQ-ZX2}
  guarantees that any optimal solution
  to~\eqref{z:.wBvf2qe} must lie in
  $\topoint({\mathfrak{S}})$.
  
  Let us construct a Lagrangian for our optimization problem.  We minimize the
  function $f_{\mathrm{obj}}(N_{BR})$ in~\eqref{z:Wjf46Lne} over
  $N_{BR}\in\topoint(\mathfrak{S})$, with the following constraints:
  \begin{enumerate}[(\roman*)]
  \item $N_{BR} \geq 0$ (dual variable $S_{BR}\geq 0$),
  \item $\operatorname{tr}_B({N_{BR}}) = \mathds{1}_R$ (dual variable $F_R = F_R^\dagger$),
  \item $\operatorname{tr}({C^j_{BR} N_{BR}}) = q_j$ (dual variable $\mu_j\in\mathbb{R}$) for
    $j=1, \ldots, n_C$, and %
  \item\label{z:BRhqcIil}
    $\operatorname{tr}({D^\ell_{BR} N_{BR}}) \leq r_\ell$ (dual variable $\nu_\ell\geq 0$) for
    $\ell =1, \ldots, n_D$.
  \end{enumerate}
  The Lagrangian reads:
  \begin{align}
    \mathcal{L}_\phi[N_{BR}, S_{BR}, \mu_j, \nu_\ell, F_R]
    &= f_{\mathrm{obj}}(N_{BR})
    - \sum_{j=1}^{n_C} \mu_j \bigl[{ q_j - \operatorname{tr}\bigl({ C^j_{BR} N_{BR} }\bigr) }\bigr]
      - \sum_{\ell=1}^{n_D} \nu_\ell \bigl[{ r_\ell - \operatorname{tr}\bigl({ D^\ell_{BR} N_{BR} }\bigr) }\bigr]
      \nonumber\\
    &\qquad
    + \operatorname{tr}\bigl({ F_R \bigl[{ \mathds{1}_R - \operatorname{tr}_B({N_{BR}}) }\bigr] }\bigr)
    - \operatorname{tr}\bigl({ S_{BR} N_{BR} }\bigr)\ .
      \label{z:Hjx04jRS}
  \end{align}

  If the problem were strictly feasible, we could use Slater's condition to
  assert that strong duality holds~\cite{R95}.  It is
  unclear, however, whether the inequality
  constraints~\ref{z:BRhqcIil} can be strictly satisfied.
  Instead, we employ a weaker version of Slater's condition, which states that
  strong duality also holds if the problem is strictly feasible with respect to
  all nonaffine constraints~\cite{R95}.
  The Karush-Kuhn-Tucker (KKT) theorem~\cite{R95}
  then states that optimal (primal, dual) variable pairs are exactly the points
  that satisfy all following conditions, known as the KKT conditions:
  \begin{enumerate}[(\alph*)]
  \item the gradient of $\mathcal{L}$ with
    respect to $N_{BR}$ vanishes;
  \item all primary and dual constraints are satisfied; and
  \item the complementary slackness conditions hold, namely, $S_{BR}N_{BR}=0$
    and $\nu_\ell [{r_\ell - \operatorname{tr}(D^\ell_{BR} N_{BR})}] = 0$.
  \end{enumerate}
  We now compute the gradient of $\mathcal{L}$ by a calculus of variations.
  Henceforth, $N_{BR}^{00}$ is understood as isometrically embedded in the
  support of $\mathds{1}_B\otimes \Pi^{\phi_R}_R$ whenever necessary from context.
  Observe, for instance, that
  $\phi_R^{1/2}\,N_{BR}^{00}\,\phi_R^{1/2} \equiv
  \phi_R^{1/2}\,N_{BR}\,\phi_R^{1/2}$.  Recalling the computation of the
  entropy's derivative in the proof of
  \cref{z:hd7.cLe.}, we find:
  \begin{align}
    \delta  f_{\mathrm{obj}}(N_{BR}) 
    &=
      \operatorname{tr}\Bigl\{{ \phi_R^{1/2}
      \Bigl[{ \log\bigl({ \phi_R^{1/2}N_{BR}^{00}\phi_R^{1/2}}\bigr) + \mathds{1}_B\otimes\Pi_R^{\phi_R} }\Bigr]
      \phi_R^{1/2} \, \delta N_{BR} }\Bigr\}
      - 
      \operatorname{tr}\Bigl[{ \phi_R^{1/2} \log\bigl({ \phi_R^{1/2} M_{BR} \phi_R^{1/2} }\bigr) \phi_R^{1/2} \,
      \delta N_{BR} }\Bigr]
      \nonumber\\
    &\qquad
      + \sum_{m=1}^{n_E} 2\tilde\eta_m\bigl[{\operatorname{tr}\bigl({E_{BR}^m N_{BR}}\bigr) - s_m}\bigr]
      \operatorname{tr}\bigl({E_{BR}^m \delta N_{BR}}\bigr)
      - \operatorname{tr}\bigl[{(\mathds{1}_B\otimes\phi_R)\,\delta N_{BR}}\bigr]\ .
  \end{align}
  Let
  \begin{align}
    w_m \equiv 2\tilde\eta_m \bigl[{\operatorname{tr}\bigl({E_{BR}^m N_{BR}}\bigr) - s_m}\bigr]\ .
    \label{z:fiA-z3IP}
  \end{align}
  Then,
  \begin{align}
    \delta \mathcal{L}_\phi
    &= \delta f_{\mathrm{obj}}(N_{BR}) 
      + \sum_{j=1}^{n_C} \mu_j \operatorname{tr}\bigl[{ C^j_{BR} \,\delta N_{BR} }\bigr]
      + \sum_{\ell=1}^{n_D} \nu_\ell \operatorname{tr}\bigl[{ D^\ell_{BR} \,\delta N_{BR} }\bigr]
      - \operatorname{tr}\bigl[{ F_R \delta N_{BR} }\bigr]
      - \operatorname{tr}\bigl[{ S_{BR} \delta N_{BR} }\bigr]
      \nonumber\\
    &
      \begin{aligned}[b]
        &= \operatorname{tr}\biggl\{\biggl[
          \phi_R^{1/2} \log\bigl({\phi_R^{1/2} N_{BR}^{00}\phi_R^{1/2}}\bigr) \phi_R^{1/2}
        - \phi_R^{1/2} \log\bigl({\phi_R^{1/2} M_{BR} \phi_R^{1/2}}\bigr) \phi_R^{1/2}
        \\
        &\qquad\quad + \sum_{j=1}^{n_C} \mu_j C_{BR}^j
        + \sum_{\ell=1}^{n_D} \nu_\ell D_{BR}^\ell
        + \sum_{m=1}^{n_E} w_m E_{BR}^m
        - \mathds{1}_B\otimes F_{R} - S_{BR}
        \biggr]\,\delta N_{BR} \biggr\}\ .
      \end{aligned}
      \label{z:JjNT-6TX}
  \end{align}
  Define
  \begin{align}
    G_{BR}
    &= \sum \mu_j C^j_{BR} + \sum \nu_\ell D^\ell_{BR} + \sum w_m E^m_{BR}
      - \mathds{1}_B\otimes F_R
      - \phi_R^{1/2} \log\bigl({\phi_R^{1/2} M_{BR} \phi_R^{1/2}}\bigr) \phi_R^{1/2}
      - S_{BR}\ .
      \label{z:vViwg20w}
  \end{align}
  The gradient $\delta \mathcal{L}_\phi$ vanishes exactly when the term in
  square brackets
  in~\eqref{z:JjNT-6TX}
  is identically zero, namely:
  \begin{align}
    \phi_R^{1/2} \log\bigl({\phi_R^{1/2} N_{BR}^{00} \phi_R^{1/2}}\bigr) \phi_R^{1/2}
    = -G_{BR}\ .
    \label{z:gOkCgPf9}
  \end{align}
  Applying $\Pi_R^{\phi_R\perp}(\cdot)$, we find
  $\Pi_R^{\phi_R\perp} G_{BR} = 0$, which implies that
  $G_{BR} = \Pi_R^{\phi_R} G_{BR} \Pi_R^{\phi_R}$.  Applying
  $\exp\bigl\{{\phi_R^{-1/2}(\cdot)\phi_R^{-1/2}}\bigr\}$
  onto~\eqref{z:gOkCgPf9}, we find
  \begin{align}
    \mathcal{N}_{A\to B}(\phi_{AR})
    = \phi_R^{1/2} N_{BR} \phi_R^{1/2}
    = \Pi_R^{\phi_R} \exp\bigl\{{ - \phi_R^{-1/2} G_{BR} \phi_R^{-1/2} }\bigr\}  \Pi_R^{\phi_R}\ .
    \label{z:KbGQF8pL}
  \end{align}
  This completely determines $N_{BR}^{00}$, the upper left block
  in~\eqref{z:81SIufBH}, since
  $N_{BR}^{00} = \phi_R^{-1/2} \mathcal{N}(\phi_{AR}) \phi_R^{-1/2}$.  The other
  blocks $N_{BR}^{01}$, $N_{BR}^{01\,\dagger}$, and $N_{BR}^{11}$ are collected
  into some general Hermitian matrix $Y_{BR}$.  This proves that any optimal
  $N_{BR}$ is of the form stated
  in~\eqref{z:drLpOUuq}.
  Conversely, if $N_{BR}$ satisfies all problem constraints and is of the
  form~\eqref{z:drLpOUuq}
  with all the stated conditions, then all KKT conditions are satisfied
  (including~\eqref{z:gOkCgPf9}
  along with the complementary slackness conditions), implying that $N_{BR}$ is
  optimal.

  Any optimal $\mathcal{N} \equiv \ThChGen[A\to B][\phi]$, which necessarily has
  the above form, further satisfies the following properties.  We know that
  $\phi_R^{1/2} N_{BR} \phi_R^{1/2}$ is a normalized quantum state and therefore
  obeys $\phi_R^{1/2} N_{BR} \phi_R^{1/2}\leq \mathds{1}_{BR}$.  Plugging
  in~\eqref{z:KbGQF8pL}, we find that
  $\exp\bigl\{{ -\phi_R^{-1/2}\,G_{BR} \phi_R^{-1/2} }\bigr\} \leq \mathds{1}_{BR}$ and
  therefore $\phi_R^{-1/2} G_{BR} \phi_R^{-1/2}$ must be positive semidefinite.
  Applying $\phi_R^{-1/2}\,({\cdot})\,\phi_R^{-1/2}$ and recalling that
  $G_{BR} = \Pi_R^{\phi_R} G_{BR} \Pi_R^{\phi_R}$ enables us to conclude that
  $G_{BR}$ is positive semidefinite.  The property satisfied by the $Y_{BR}$
  operator can be found by computing
  \begin{align}
    \mathds{1}_R
    &= \operatorname{tr}_B({\mathcal{N}(\Phi_{A:R})})
      = \phi_R^{-1/2} \operatorname{tr}_B\bigl[{ {e}^{-\phi_R^{-1/2} G_{BR} \phi_R^{-1/2}} }\bigr] \phi_R^{-1/2}
      + \operatorname{tr}_B(Y_{BR})
      \nonumber\\
    &= \phi_R^{-1/2} \operatorname{tr}_B\bigl[{ \mathcal{N}(\phi_{AR}) }\bigr] \phi_R^{-1/2}
      + \operatorname{tr}_B(Y_{BR})
      \nonumber\\
    &= \Pi^{\phi_R}_R + \operatorname{tr}_B(Y_{BR})\ ,
  \end{align}
  and therefore $\operatorname{tr}_B(Y_{BR}) = \Pi_R^{\phi_R\perp}$.  The value attained for
  ${D}_{\phi}^{}\bigl ({\mathcal{N}}\mathclose{}\,\big \Vert\,\mathopen{}{\mathcal{M}}\bigr )$ for
  $\mathcal{N} \equiv \ThChGen[A\to B][\phi]$, recalling
  \cref{z:KbGQF8pL,z:vViwg20w} and
  $\Pi_R^{\phi_R \perp} G_{BR} = 0$, is
  \begin{align}
    {D}_{\phi}^{}\bigl ({\mathcal{N}}\mathclose{}\,\big \Vert\,\mathopen{}{\mathcal{M}}\bigr )
    &= -\operatorname{tr}\bigl[{ \mathcal{N}(\phi_{AR})\,\phi_R^{-1/2} G_{BR} \phi_R^{-1/2} }\bigr]
      - \operatorname{tr}\bigl[{ \mathcal{N}(\phi_{AR})\,\log\bigl({\mathcal{M}(\phi_{AR})}\bigr) }\bigr]
      \nonumber\\[1ex]
    &= -\operatorname{tr}\bigl[{ \mathcal{N}(\Phi_{A:R})\, G_{BR} }\bigr]
      - \operatorname{tr}\bigl[{ \mathcal{N}(\phi_{AR})\,\log\bigl({\mathcal{M}(\phi_{AR})}\bigr) }\bigr]
      \nonumber\\[1ex]
    &= -\sum \mu_j \operatorname{tr}\bigl({C_{BR}^j N_{BR}}\bigr) - \sum \nu_\ell \operatorname{tr}\bigl({D_{BR}^\ell N_{BR}}\bigr)
      - \sum w_m \operatorname{tr}\bigl({E_{BR}^m N_{BR}}\bigr)
      + \operatorname{tr}\bigl({N_{BR} F_R}\bigr)
        \nonumber\\
        &\qquad
      + \operatorname{tr}\bigl[{N_{BR} \, \phi_R^{1/2}\log\bigl({\phi_R^{1/2} M_{BR} \phi_R^{1/2}}\bigr) \phi_R^{1/2}}\bigr]
      + \operatorname{tr}\bigl({S_{BR} N_{BR}}\bigr)
      - \operatorname{tr}\bigl[{ \mathcal{N}(\phi_{AR})\,\log\bigl({\mathcal{M}(\phi_{AR})}\bigr) }\bigr]
      \nonumber\\[1ex]
    &= -\sum \mu_j q_j - \sum \nu_\ell r_\ell - \sum w_m \Bigl({s_m + \frac{w_m}{2\tilde\eta_m}}\Bigr)
      + \operatorname{tr}\bigl({F_R}\bigr)\ .
  \end{align}
  In the last equality, we used the equality constraints, both slackness
  conditions, \cref{z:fiA-z3IP}, and the fact that
  $\operatorname{tr}_B({N_{BR}}) = \mathds{1}_R$.
\end{proof}

\subsection{Dual problem of the channel relative entropy minimization: Proof of
  \cref{z:t5t-uuUF}}
\label{z:UmB-TaUp}

We begin by deriving the Lagrange dual problem of~\eqref{z:z4Hq5iYa}.
This dual is presented in the Lemma below.  We then use this dual problem to
prove \cref{z:t5t-uuUF}.

For the following lemma, we need a few additional definitions that
characterize how the observables $\{{ E^m }\}$ span the space orthogonal to the
support of $\phi_R$.  First, we define the superoperator projection map
\begin{align}
  \widetilde{\mathcal{P}}_{\phi}(\cdot)
  = ({\cdot}) - \Pi_R^{\phi_R}\,({\cdot})\,\Pi_R^{\phi_R}\ .
  \label{z:rkMU7UOA}
\end{align}
This map zeroes out the sub-block $\Pi_R^{\phi_R}(\cdot)\Pi_R^{\phi_R}$ of its
matrix input.  This map is not completely positive nor does it preserve the
input's trace, but it is Hermiticity-preserving.  In vectorized form, this map
is represented as
$\widetilde{\mathcal{P}}_{\phi} = \mathds{1}-
\bigl({\Pi_R^{\phi_R}\otimes\Pi_R^{\phi_R\,*}}\bigr)$.
Now, we define the linear map
$\mathsf{E}_\phi : \mathbb{R}^{n_E} \to \Herm(\mathscr{H}_{BR})$ through its action on
the canonical basis as
\begin{align}
  \lvert {m}\rangle  \mapsto \mathsf{E}_\phi \lvert {m}\rangle  = \widetilde{\mathcal{P}}_{\phi} \lvert { E^m_{BR} }\rrangle \ .
  \label{z:1JjW3rMh}
\end{align}
Equivalently,
$\mathsf{E}_\phi = \sum_m \widetilde{\mathcal{P}}_{\phi} \lvert { E^m_{BR} }\rrangle 
\langle {m}\rvert $.  Correspondingly,
$\mathsf{E}_\phi^\dagger \equiv \sum_m \lvert {m}\rangle \llangle {E^m_{BR}}\rvert 
\widetilde{\mathcal{P}}_{\phi}$.
The image of $\mathsf{E}_\phi^\dagger$, denoted by
$\OpImag\bigl[{ \mathsf{E}_\phi^\dagger }\bigr]$, describes the operators that can be
spanned by $E^m_{BR}$ if the latter are stripped of their action within
$\Pi_R^{\phi_R}$.

\begin{lemma}[Dual formulation of the minimum channel relative entropy problem]
  \label{z:MCoiQlab}
  Consider the setting of Problem~\eqref{z:z4Hq5iYa},
  and assume that there exists some quantum channel with positive definite Choi
  matrix that satisfies all problem constraints (as in
  \cref{z:gV43EWT.}).  
  Now consider the following problem:
  \begin{align}
    \textup{maximize:}
    &\quad \mathcal{G}_\phi(F_R, \mu_j, \nu_\ell, w_m, S_{BR})
      \label{z:QZA5Z07I}
      \\[1ex]
    \textup{over:}
    &\quad
      \mu_j\in\mathbb{R}\ (j=1,\ldots,n_C);
      \ \nu_\ell\geq 0\ (\ell=1,\ldots,n_D);
      \ w_m\in\mathbb{R}\ (m=1,\ldots, n_E); 
      \nonumber\\
    &\quad
      F_R=F_R^\dagger;
      \ S_{BR} \geq 0
      \nonumber\\[1ex]
    \textup{subject to:}
    &\quad
      \Pi^{\phi_R \perp}_R \, G_{BR} = 0\ ;
      \nonumber\\
    &\quad 
      \mathsf{e}
      \in \OpImag\bigl[{ \mathsf{E}_\phi^\dagger }\bigr]\ ;
      \quad \mathsf{e}_m = \frac{w_m}{2\tilde\eta_m} + s_m
      - \operatorname{tr}\Bigl({ \phi_R^{-1/2} E_{BR}^m \phi_R^{-1/2} {e}^{-\phi_R^{-1/2} G_{BR} \phi_R^{-1/2}} }\Bigr)\ ,
      \nonumber
  \end{align}
  where $\mathsf{e} \in \mathbb{R}^{n_E}$ is a vector given by its components
  $\mathsf{e}_m$, and using the shorthand expressions
  \begin{align}
    \hspace*{2em}&\hspace*{-2em}
    \mathcal{G}_\phi(F_R, \mu_j, \nu_\ell, w_m, S_{BR})
                   \nonumber\\
    &=
      \operatorname{tr}({F_R}) 
      - \sum \mu_j q_j
      - \sum \nu_\ell r_\ell
      - \sum w_m s_m
      + 1 - \operatorname{tr}\Bigl({\Pi_R^{\phi_R} {e}^{-\phi_R^{-1/2} G_{BR} \phi_R^{-1/2}}}\Bigr)
      - \sum \frac{w_m^2}{4\tilde\eta_m}\ ;
      \label{z:-9TxOWQN}
    \\
    \hspace*{2em}&\hspace*{-2em}
    G_{BR}(\mu_j, \nu_\ell, w_m, F_R, S_{BR}) \equiv G_{BR}
                   \nonumber\\
    &= \sum \mu_j C^j_{BR} + \sum \nu_\ell D^\ell_{BR} + \sum w_m E_{BR}^m
      - \mathds{1}_B\otimes F_R
      - \phi_R^{1/2} \log\bigl({\phi_R^{1/2} M_{BR} \phi_R^{1/2}}\bigr) \phi_R^{1/2}
      - S_{BR}\ .
      \label{z:di1MjM3W}
  \end{align}
  The
  problem~\eqref{z:QZA5Z07I}
  yields the same optimal value as the problem~\eqref{z:z4Hq5iYa},
  and the variables $F_R$, $\mu_j$, $\nu_\ell$, $S_{BR}$ coincide with those for
  the optimal thermal channel
  in~\cref{z:gV43EWT.}.
\end{lemma}

The optimization in~\eqref{z:QZA5Z07I}
can be extended to include a maximization over $\phi_R$, therefore solving our full original
stated problem of minimizing the channel relative entropy.  The presence of
$\phi_R^{-1/2}$, however, makes the optimization
in~\eqref{z:QZA5Z07I}
numerically less stable than the problem in
\cref{z:t5t-uuUF}.  The latter is therefore
more attractive for numerical computation, in principle.
While this optimization can be carried out numerically, we have empirically
found that the techniques of \R\cite{R107,R57}
were more reliable in our examples.

We can exploit the fact that a number of entries in the variable $S_{BR}$ are
fixed by the constraint $\Pi_R^{\phi_R} G_{BR} = 0$ to reduce the number of variables
in~\eqref{z:QZA5Z07I}.
Decompose $\mathscr{H}_{BR}$ into two orthogonal subspaces projected upon by
$(\mathds{1}_B\otimes\Pi_R^{\phi_R})$, $(\mathds{1}_B\otimes\Pi_R^{\phi_R \perp})$, and write
\begin{align}
  S_{BR} = \begin{bmatrix}
      S_{BR}^{00} & S_{BR}^{01} \\[1ex]
      S_{BR}^{01\,\dagger} & S_{BR}^{11}
  \end{bmatrix}\ ,
\end{align}
where $S_{BR}^{00} = \Pi_R^{\phi_R} S_{BR} \Pi_R^{\phi_R}$ up to an isometric embedding,
$S_{BR}^{01} = \Pi_R^{\phi_R} S_{BR} \Pi_R^{\phi_R \perp}$, etc.
The constraint $\Pi_R^{\phi_R \perp} G_{BR} = 0$
in~\eqref{z:QZA5Z07I} implies that
the blocks of $S_{BR}$ need obey
\begin{subequations}
\begin{align}
  S_{BR}^{11}
  &= \Pi_R^{\phi_R \perp} \, \Bigl({ \sum \mu_j C^j + \sum \nu_\ell D^\ell + \sum w_m E^m
     - \mathds{1}_B\otimes F_R }\Bigr) \, \Pi_R^{\phi_R \perp}\ ;
  \label{z:ekGz5m28}
  \\
  S_{BR}^{01}
  &= \Pi_R^{\phi_R} \, \Bigl({ \sum \mu_j C^j + \sum \nu_\ell D^\ell + \sum w_m E^m
    - \mathds{1}_B\otimes F_R }\Bigr) \, \Pi_R^{\phi_R \perp}\ ;
  \label{z:Mlerg4vB}
  \\
  S_{BR}^{00}
  &\geq S_{BR}^{01} \bigl({S_{BR}^{11}}\bigr)^{-1} S_{BR}^{01\,\dagger}\ ,
  \label{z:7VdBH8RU}
\end{align}
\end{subequations}
where the last equality involves no isometric embedding and follows by Schur complementarity
from the requirement that $S_{BR} \geq 0$.  Therefore, we may replace the variable
$S_{BR}$ by a potentially smaller variable $S_{BR}^{00}$ acting only in the subspace
projected onto by $\mathds{1}_B\otimes\Pi_R^{\phi_R}$; $S_{BR}^{00}$ is constrained
via~\eqref{z:7VdBH8RU}, where
$S_{BR}^{01}$ and $S_{BR}^{11}$ are determined 
from~\eqref{z:ekGz5m28}
and~\eqref{z:Mlerg4vB};
then, the constraint
$\Pi_R^{\phi_R \perp} G_{BR} = 0$ becomes unnecessary.

A further simplification can be carried out if $n_E = 0$.  For any Hermitian operators
$G_{BR}$, $G_{BR}'$ obeying $G_{BR} \leq G_{BR}'$, we have
$\operatorname{tr}\bigl[{\Pi_R^{\phi_R} \exp\bigl({-\phi_R^{-1/2} G_{BR} \phi_R^{-1/2}}\bigr)}\bigr]
\geq \operatorname{tr}\bigl[{\Pi_R^{\phi_R} \exp\bigl({-\phi_R^{-1/2} G_{BR}' \phi_R^{-1/2}}\bigr)}\bigr]$.  This
inequality follows from the Golden-Thompson inequality
$\operatorname{tr}\bigl({{e}^{X+Y}}\bigr)\leq\operatorname{tr}\bigl({{e}^{X}{e}^{Y}}\bigr)$
applied within the subspace $\mathds{1}_B\otimes\Pi_R^{\phi_R}$ with
$X = -\phi_R^{-1/2} G_{BR} \phi_R^{-1/2}$ and
$Y = -\phi_R^{-1/2} \bigl({G_{BR}' - G_{BR}}\bigr) \phi_R^{-1/2}  \leq 0$, noting that
${e}^{Y} \leq \mathds{1}$.  As a consequence, we may eliminate the variable $S_{BR}^{00}$
entirely in problem~\eqref{z:QZA5Z07I}
if $n_E=0$, since choosing $S_{BR}^{00} = S_{BR}^{01} S_{BR}^{11} S_{BR}^{01\,\dagger}$
always yields a better value for $G_{BR}$ than one that simply 
obeys~\eqref{z:7VdBH8RU}.
This argument does not apply if $n_E > 0$ because the choice of
$S_{BR}^{00}$ might be further constrained by the
constraint in~\eqref{z:QZA5Z07I}
involving the vector $\mathsf{e}$.

\begin{proof}[*z:MCoiQlab]
  In the proof of
  \cref{z:gV43EWT.} (see
  \cpageref{proof:z:gV43EWT.}),
  we derived the corresponding Lagrangian
  in~\eqref{z:Hjx04jRS}.
  The primal variable is $N_{BR} \in \topoint(\mathfrak{S})$, and the dual
  variables are $S_{BR}\geq 0$, $\mu_j\in\mathbb{R}$, $\nu_\ell\geq 0$,
  $F_R=F_R^\dagger$.
  The dual objective function is given by~\cite{R95}
  \begin{align}
    g_\phi(S_{BR}, \mu_j, \nu_\ell, F_R)
    &= \inf_{N_{BR}\in\topoint(\mathfrak{S})}
      \mathcal{L}_\phi[N_{BR}, S_{BR}, \mu_j, \nu_\ell, F_R]\ .
      \label{z:H3G-cYzq}
  \end{align}
  Observing that~\eqref{z:H3G-cYzq} can be cast in the form
  of~\eqref{z:1w8WlXSN} by
  flipping the sign of the objective, we invoke
  \cref{z:dTDQ-ZX2} to assert that
  the infimum of $\mathcal{L}_\phi$ is attained at a point in
  $\topoint(\mathfrak{S})$ where the gradient of $\mathcal{L}_\phi$ vanishes
  (since $\mathcal{L}_\phi$ is convex in $N_{BR}$ and differentiable).
  We've already computed this gradient
  in~\eqref{z:JjNT-6TX}.
  We have seen that the gradient of $\mathcal{L}_\phi$ with respect to $N_{BR}$
  vanishes exactly when there exists values $w_m \in \mathbb{R}$ such that
  \begin{subequations}
    \begin{gather}
      w_m = 2\tilde\eta_m\bigl[{\operatorname{tr}\bigl({E^m_{BR} N_{BR}}\bigr) - s_m}\bigr]\ ;
      \label{z:0uYjOQu3}
      \\
      \phi_R^{1/2} \log\bigl({\phi_R^{1/2} N_{BR} \phi_R^{1/2}}\bigr) \phi_R^{1/2}
      = - G_{BR} \ ,
      \label{z:nuA7i-Z9}
    \end{gather}
  \end{subequations}
  where $G_{BR}$ is defined
  in~\eqref{z:vViwg20w} and is here
  viewed as a shorthand expression in terms of the variables
  $\mu_j, \nu_\ell, w_m, F_R, S_{BR}$.
  Furthermore, the condition~\eqref{z:nuA7i-Z9} holds if and
  only if there exists a Hermitian $Y_{BR}$ such that all following conditions
  hold:
  \begin{align}
    N_{BR} &= \Pi_R^{\phi_R} {e}^{-\phi_R^{-1/2} G_{BR} \phi_R^{-1/2}} \Pi_R^{\phi_R} + Y_{BR}
             \ ;
    &
      \Pi_R^{\phi_R} Y_{BR} \Pi_R^{\phi_R} &= 0\ ;
    &
      \Pi_R^{\phi_R \perp} G_{BR} &= 0\ .
    \label{z:.CtTi1Qu}
  \end{align}
  The above statement can be seen from the proof of
  \cref{z:gV43EWT.} (cf.\@
  \cpageref{proof:z:gV43EWT.}).

  At this point, the infimum in~\eqref{z:H3G-cYzq} is attained
  whenever we have variables $N_{BR}$, $w_m$, $G_{BR}$, $Y_{BR}$ satisfying
  \cref{z:vViwg20w,z:0uYjOQu3,z:.CtTi1Qu}.  We now compute the value of the objective all while simplifying these
  conditions.  We can write, using these conditions,
  \begin{align}
    g_\phi
    &=
      -\operatorname{tr}({N_{BR} G_{BR}})
      - \operatorname{tr}\bigl[{ N_{BR} \, \phi_R^{1/2} \log\bigl({\phi_R^{1/2} M_{BR} \phi_R^{1/2} }\bigr) \phi_R^{1/2} }\bigr]
      + \bigl[{ 1 - \operatorname{tr}(N_{BR} \phi_R) }\bigr]
      \nonumber\\[1ex] &\qquad
      + \sum \tilde\eta_m \bigl[{ \operatorname{tr}({E^m_{BR} N_{BR}}) - s_m }\bigr]^2
      + \operatorname{tr}\Bigl[{N_{BR} \Bigl({\sum \mu_j C^j_{BR} + \sum \nu_\ell D^\ell_{BR}}\Bigr)}\Bigr]
      - \sum \mu_j q_j - \sum \nu_\ell r_\ell
      \nonumber\\[1ex] &\qquad
      + \operatorname{tr}({F_R})
      - \operatorname{tr}\bigl({F_R N_{BR}}\bigr)
      - \operatorname{tr}\bigl({S_{BR} N_{BR}}\bigr)\ .
    \label{z:979g9Cur}
  \end{align}
  The following relations are obtained thanks
  to~\eqref{z:0uYjOQu3}:
  \begin{align}
    \sum w_m \bigl[{ \operatorname{tr}({E^m_{BR} N_{BR}}) - s_m }\bigr]
    &= \sum \frac{w_m^2}{2\tilde\eta_m}\ ;
    &
      \sum \tilde\eta_m \bigl[{ \operatorname{tr}({E^m_{BR} N_{BR}}) - s_m }\bigr]^2
      = \sum \frac{w_m^2}{4\tilde\eta_m}\ ;
  \end{align}
  they lead to
  \begin{align}
    \sum \tilde\eta_m \bigl[{ \operatorname{tr}({E^m_{BR} N_{BR}}) - s_m }\bigr]^2
    = \sum w_m \bigl[{ \operatorname{tr}({E^m_{BR} N_{BR}}) - s_m }\bigr]
    - \sum \frac{w_m^2}{4\tilde\eta_m}\ .
      \label{z:K3jIIbAR}
  \end{align}
  We now plug~\eqref{z:K3jIIbAR}
  in~\eqref{z:979g9Cur}.  In the resulting expression for
  $g_\phi$, a number of terms combine to an expression for $\operatorname{tr}({N_{BR} G_{BR}})$
  which cancels out the initial term $-\operatorname{tr}({N_{BR} G_{BR}})$
  [recall~\eqref{z:vViwg20w}].  We
  find:
  \begin{align}
    g_\phi
    &= \bigl[{ 1 - \operatorname{tr}({N_{BR} \phi_R}) }\bigr]
      -\sum w_m s_m - \sum \frac{w_m^2}{4\tilde\eta_m}
      - \sum \mu_j q_j - \sum \nu_\ell r_\ell + \operatorname{tr}({F_R})
      \nonumber\\
    &= \operatorname{tr}({F_R}) 
      - \sum \mu_j q_j
      - \sum \nu_\ell r_\ell
      - \sum w_m s_m
      + 1 - \operatorname{tr}\Bigl({\Pi_R^{\phi_R} {e}^{-\phi_R^{-1/2} G_{BR} \phi_R^{-1/2}}}\Bigr)
      - \sum \frac{w_m^2}{4\tilde\eta_m}\ ,
  \end{align}
  as in the claim.  It remains to further simplify the
  conditions~\eqref{z:0uYjOQu3}
  and~\eqref{z:.CtTi1Qu} to eliminate the use of
  $Y_{BR}$ and $N_{BR}$.  Let us compute
  \begin{align}
    \operatorname{tr}\bigl({ E^m_{BR} N_{BR}}\bigr)
    &= \operatorname{tr}\bigl[{ \phi_R^{-1/2} E^m_{BR} \phi_R^{-1/2} {e}^{-\phi_R^{-1/2} G_{BR} \phi_R^{-1/2}} }\bigr]
      + \operatorname{tr}({E^m_{BR} Y_{BR}})\ .
  \end{align}
  On the other hand, \cref{z:0uYjOQu3} implies that
  \begin{align}
    \operatorname{tr}\bigl({ E^m_{BR} N_{BR}}\bigr)
    &= \frac{w_m}{2\tilde\eta_m} + s_m \ .
  \end{align}
  Combining the two above equations eliminates the use of $N_{BR}$.  Namely, the
  infimum in $g_\phi$ is reached whenever there exists $G_{BR}$, $\{{ w_m }\}$, and
  $Y_{BR}$ such
  that~\eqref{z:vViwg20w} is
  satisfied, such that $\Pi^{\phi_R \perp}_R G_{BR}=0$ and
  $\Pi^{\phi_R}_R Y_{BR} \Pi^{\phi_R}_R = 0$, as well as such that
  \begin{align}
    \frac{w_m}{2\tilde\eta_m} + s_m
    - \operatorname{tr}\bigl[{ \phi_R^{-1/2} E^m_{BR} \phi_R^{-1/2} {e}^{-\phi_R^{-1/2} G_{BR} \phi_R^{-1/2}} }\bigr]
      = \operatorname{tr}({E^m_{BR} Y_{BR}})\ .
    \label{z:obmif7Em}
  \end{align}
  [In such a case, $N_{BR}$ can be deduced from the first equation
  in~\eqref{z:.CtTi1Qu}.]
  Now, we eliminate the explicit reference to the variable $Y_{BR}$.
  Specifically, for given $G_{BR}$ and $\{{ w_m }\}$, we seek to determine whether
  there exists $Y_{BR}$ such that~\eqref{z:obmif7Em} holds and
  such that $\Pi_R^{\phi_R} Y_{BR} \Pi_R^{\phi_R} = 0$.
  The condition $\Pi_R^{\phi_R} Y_{BR} \Pi_R^{\phi_R} = 0$ is equivalent to
  $Y_{BR} = \widetilde{\mathcal{P}}_\phi({Y_{BR}})$,
  recalling~\eqref{z:rkMU7UOA}.  Also,
  recalling~\eqref{z:1JjW3rMh},
  \begin{align}
    \operatorname{tr}({E^m_{BR} Y_{BR}})
    = \llangle {E^m_{BR}}\rvert   \widetilde{\mathcal{P}}_\phi \lvert {Y_{BR}}\rrangle 
    = \langle {m}\rvert  \mathsf{E}_\phi^\dagger \lvert {Y_{BR}}\rrangle \ .
  \end{align}
  Let
  \begin{align}
    \lvert {\mathsf{e}}\rangle 
    &\equiv \sum \mathsf{e}_m \lvert {m}\rangle  \ ;
    &
    \mathsf{e}_m
    &\equiv \frac{w_m}{2\tilde\eta_m} + s_m
    - \operatorname{tr}\bigl[{ \phi_R^{-1/2} E^m_{BR} \phi_R^{-1/2} {e}^{-\phi_R^{-1/2} G_{BR} \phi_R^{-1/2}} }\bigr]\ .
  \end{align}
  Clearly, there exists a Hermitian $Y_{BR}$ with
  $\Pi_R^{\phi_R} Y_{BR} \Pi_R^{\phi_R} = 0$ that
  satisfies~\eqref{z:obmif7Em} if and only if there exists a
  Hermitian $Y_{BR}$ such that
  $\lvert {\mathsf{e}}\rangle  = \mathsf{E}_\phi^\dagger \lvert {Y_{BR}}\rrangle $.  Equivalently,
  $\lvert {\mathsf{e}}\rangle $ must lie in the image of
  $\mathsf{E}_\phi^\dagger\rvert_{\mathrm{Herm}}$, defined as the restriction of
  $\mathsf{E}_\phi^\dagger$ to the space of Hermitian operators.
\end{proof}

We are now in the position to prove
\cref{z:t5t-uuUF}, by showing that
the optimization
problem~\eqref{z:QZA5Z07I}
can be recast as the optimization~\eqref{z:MK9lgNNV}.

\begin{proof}[*z:t5t-uuUF]
  The constraint involving the shorthand vector $\mathsf{e}$
  in~\eqref{z:QZA5Z07I} can
  also be enforced by introducing a variable $Y_{BR} = Y_{BR}^\dagger$ and
  imposing the constraints
  \begin{align}
    \operatorname{tr}\bigl({ E_{BR}^m Y_{BR}}\bigr)
    &= \frac{w_m}{2\tilde\eta_m} + s_m
    - \operatorname{tr}\Bigl({ \phi_R^{-1/2} E^m_{BR} \phi_R^{-1/2} {e}^{-\phi_R^{-1/2} G_{BR} \phi_R^{-1/2} } }\Bigr)
    \ ;
    &
    Y_{BR} &= \widetilde{\mathcal{P}}_\phi( Y_{BR} )\ .
             \label{z:p1vpW7oR}
  \end{align}
  We now replace the variable $Y_{BR}$ by the variable
  $N_{BR} = N_{BR}^\dagger$, whose bijective relationship with $Y_{BR}$ is given
  as
  \begin{align}
    N_{BR} &= \phi_R^{-1/2} {e}^{-\phi_R^{-1/2} G_{BR} \phi_R^{-1/2}} \phi_R^{-1/2} + Y_{BR}\ .
             \label{z:WJwRf4vK}
  \end{align}
  From the KKT conditions (see proofs
  of~\cref{z:gV43EWT.,z:MCoiQlab}), we know that for
  optimal choices of variables $\mu_j, \nu_\ell, w_m, F_R, S_{BR}, Y_{BR}$, the
  variable $N_{BR}$ contains the Choi matrix of the optimal quantum channel in
  the original problem~\eqref{z:z4Hq5iYa}.
  Therefore, the optimization can be restricted to operators $N_{BR}$ that
  satisfy $N_{BR} \geq 0$.  The constraints~\eqref{z:p1vpW7oR}
  with~\eqref{z:WJwRf4vK} can then equivalently be expressed as
  constraints involving $N_{BR}$ directly rather than $Y_{BR}$:
  \begin{align}
    \operatorname{tr}\bigl({ E_{BR}^m N_{BR}}\bigr) &= \frac{w_m}{2\tilde\eta_m} + s_m\ ;
    &
      \Pi_R^{\phi_R} N_{BR} \Pi_R^{\phi_R}
    &= \phi_R^{-1/2}\,{e}^{-\phi_R^{-1/2} G_{BR} \phi_R^{-1/2}} \, \phi_R^{-1/2}\ ,
  \end{align}
  thereby entirely eliminating $Y_{BR}$.  Applying
  $\phi_R^{1/2}\,\log\bigl[{\phi_R^{1/2}(\cdot) \phi_R^{1/2} }\bigr]\,\phi_R^{1/2}$
  onto the latter constraint, expanding the definition of $G_{BR}$, and
  interpreting $S_{BR}\geq 0$ as a slack variable, yields the problem~\eqref{z:MK9lgNNV}.
\end{proof}

\section{Proof of the constrained channel postselection theorem
  via Schur-Weyl duality}

\subsection{Elements of Schur-Weyl duality}
\label{z:pJbOynb9}

We rely heavily on the definitions, notations, and lemmas related to Schur-Weyl
duality used in \R\cite{R77,R78,R47} (and references therein).

Let consider $n$ copies of a quantum system $S$, with total Hilbert space
$\mathscr{H}_S^{\otimes n}$.  The general linear group $\mathrm{GL}({d_S})$ (or its subgroup the
unitary group $\mathrm{U}({d_S})$) has a natural action on $\mathscr{H}_S^{\otimes n}$ by
applying the operator each copy individually, i.e.\@ by acting on
$\mathscr{H}_S^{\otimes n}$ as
$U_S^{\otimes n} \equiv U_S\otimes U_S\otimes \cdots \otimes U_S$ for
$U_S\in\mathrm{GL}({d_S})$ or $U_S\in\mathrm{U}({d_S})$.  On the other hand, the permutation
group $\mathrm{S}_{n}$ acts naturally by permuting the subsystems: For any
$\pi\in\mathrm{S}_{n}$, we define the group action $U_{S^n}(\pi)$ as
\begin{align}
  U_{S^n}(\pi) \, \lvert {\phi_1}\rangle \otimes\lvert {\phi_2}\rangle \otimes\cdots\otimes\lvert {\phi_n}\rangle 
  = \lvert {\phi_{\pi^{-1}(1)}}\rangle \otimes\lvert {\phi_{\pi^{-1}(2)}}\rangle \otimes\cdots
  \otimes\lvert {\phi_{\pi^{-1}(n)}}\rangle \ ,
  \label{z:28jQ6L1t}
\end{align}
for any $\{{\lvert {\phi_i}\rangle }\}_{i=1}^n$.

Irreducible representations of both the unitary group $\mathrm{U}({d_S})$ as well as the
symmetric group $\mathrm{S}_{n}$ are labeled by \emph{Young diagrams}.  A \emph{Young
  diagram $\lambda\in\Young(d,n)$ of size $n$ and with $d$ rows} is a collection
of $d$ integers $\lambda\equiv(\lambda_1, \ldots, \lambda_d)$ with
$\lambda_1\geq\lambda_2\geq\cdots\geq\lambda_d\geq 0$ and
$\lambda_1+\lambda_2+\cdots+\lambda_d=n$.  A Young diagram is often represented
diagrammatically as $d$ rows of boxes, with the $i$-th row containing
$\lambda_i$ boxes.

\emph{Schur-Weyl duality} states that these two actions are the commutants of
one another, and that the total Hilbert space decomposes into irreducible
representations of these representations as
\begin{align}
  \mathscr{H}_S^{\otimes n}
  \simeq
  \bigoplus_{\lambda\in\Young(d_S,n)}
  \mathcal{Q}_\lambda\otimes\mathcal{P}_\lambda\ ,
  \label{z:iWqcTgXX}
\end{align}
where $\mathcal{Q}_\lambda$ is the irreducible representation of the general
linear group $\mathrm{GL}({d_S})$ (or the unitary group $\mathrm{U}({d_S})$) labeled by $\lambda$
and where $\mathcal{P}_\lambda$ is the irreducible representation of $\mathrm{S}_{n}$
labeled by $\lambda$.
In other words, the full Hilbert space decomposes into orthogonal projectors
$\Pi_{S^n}^\lambda$ for $\lambda\in\Young(d_S,n)$, where each
$\Pi_{S^n}^\lambda$ projects onto the subspace that supports the tensor product
$\mathcal{Q}_\lambda\otimes\mathcal{P}_\lambda$ of irreducible representations
of the unitary and symmetric groups:
\begin{align}
  \mathds{1}_{S^n} &= \sum_{\lambda\in\Young(d_S, n)} \Pi^{\lambda}_{S^n}\ ;
  &
    \Pi^\lambda_{S^n} \Pi^{\lambda'}_{S^n} &= 0\quad(\lambda\neq\lambda')\ .
\end{align}
For convenience, we also define the notation $[\,(\cdot)\,]_\lambda$ as the
isometry that embeds the space $\mathcal{Q}_\lambda\otimes\mathcal{P}_\lambda$
into the appropriate subspace of $\mathscr{H}_{S}^{\otimes n}$, meaning that
for any $X\in\mathcal{Q}_\lambda$ and $Y\in\mathcal{P}_\lambda$,
\begin{align}
  [\,X\otimes Y\,]_\lambda \, \Pi_{S^n}^\lambda = [\,X\otimes Y\,]_\lambda\ .
\end{align}
The subspaces identified by a particular $\lambda\in\Young(d_S,n)$, i.e., the
support of $\Pi_{S^n}^\lambda$, are referred to as \emph{Schur-Weyl blocks}.

The dimensions of these irreducible representations are denoted by
$ d_{\mathcal{Q}_\lambda} \equiv \dim(\mathcal{Q}_\lambda)$ and
$ d_{\mathcal{P}_\lambda} \equiv \dim(\mathcal{P}_\lambda)$; they %
satisfy~\cite{R77,R78}
\begin{align}
  d_{\mathcal{Q}_\lambda}
  &\leq \operatorname{poly}(n)\ ;
  &
    \frac1{\operatorname{poly}({n})} {e}^{n{S}_{}^{}({\bar\lambda})}
    \leq
    d_{\mathcal{P}_\lambda}
    \leq
    {e}^{n{S}_{}^{}({\bar\lambda})}\ ,
\end{align}
where $\bar\lambda=\lambda/n=(\lambda_1/n, \ldots, \lambda_d/n)$ and
${S}_{}^{}({\cdot})$ is understood to act here as the Shannon entropy
${S}_{}^{}({\bar\lambda}) = -\sum \bar\lambda_i\log(\bar\lambda_i)$.
The $\Pi_{S^n}^\lambda$'s can be written as follows (cf.\@ e.g.\@
\cite[Eq.~(S.8)]{R83} or
\cite[Eq.~(2.31)]{R118}):
\begin{align}
  \Pi^\lambda_{S^n}
  =
  \frac{d_{\mathcal{P}_\lambda}}{n!} \sum_{\pi\in \mathrm{S}_{n}}
  [\chi^\lambda(\pi)]^* \, U_{S^n}(\pi)
  =
  \frac{d_{\mathcal{P}_\lambda}}{n!} \sum_{\pi\in \mathrm{S}_{n}}
  \chi^\lambda(\pi) \, U_{S^n}(\pi)
  = \bigl({ \Pi^\lambda_{S^n} }\bigr)^*
  \ ,
  \label{z:HrtNdJyg}
\end{align}
where $\chi^\lambda(\pi) = \operatorname{tr}({ U_\lambda(\pi) })$ is known as the
\emph{character} of the irreducible representation $U_\lambda(\pi)$ of $\mathrm{S}_{n}$
on the irrep space $\mathcal{P}_\lambda$.  In general,
$\chi^\lambda(\pi^{-1}) = [{\chi^\lambda(\pi)}]^*$.  The second equality
in~\eqref{z:HrtNdJyg} follows from the fact that
the characters of the symmetric group are, in fact, real.  The third equality
follows from the fact that the matrix entries of $U_{S^n}(\pi)$ are also real
[$U_{S^n}(\pi)$ simply permutes the digits of computational basis states, as
per~\eqref{z:28jQ6L1t}, and its matrix elements are 0's and 1's].
Furthermore, the formula~\eqref{z:HrtNdJyg} can
also be applied for $\lambda\in\Young(n,n)$, $\lambda\not\in\Young(d,n)$; in
this case, we find $\Pi^{\lambda}_{S^n} = 0$, which is consistent with the Young
diagram $\lambda$ not appearing in the Schur-Weyl
decomposition~\eqref{z:iWqcTgXX}.

The Schur-Weyl block with $\lambda=(n,0,0,\ldots)$ is called the \emph{symmetric
  subspace} $\mathrm{Sym}(n, d_S)$ of $\mathscr{H}_S^{\otimes n}$.  In this block,
$\mathcal{P}_\lambda$ is one-dimensional: All permutations act trivially on any
state in the symmetric subspace.  The symmetric subspace has dimension
\begin{align}
  d_{\mathrm{Sym}(n, d_S)} \equiv \binom{n + d_S - 1}{n} \leq (n+1)^{d_S - 1}\ .
\end{align}
We can also write the projector on the symmetric subspace as a sum of
permutation operators,
\begin{align}
  \Pi^{\mathrm{Sym}}_{S^n} = \frac1{n!} \sum_{\pi\in\mathrm{S}_{n}} U_{S^n}(\pi)\ .
  \label{z:bfP3uOkZ}
\end{align}
Any operator $A_{S^n}$ can be explicitly symmetrized with a symmetrization
operation $\mathcal{S}_{S^n}(\cdot)$, resulting in a permutation-invariant
operator $\mathcal{S}_{S^n}(A_{S^n})$; here
\begin{align}
  \mathcal{S}_{S^n}(\cdot) =
  \frac1{n!} \sum_{\pi\in\mathrm{S}_{n}} U_{S^n}(\pi) \, ({\cdot}) \, U_{S^n}^\dagger(\pi)\ .
\end{align}

An important consequence of Schur-Weyl duality is that any operator $X_{S^n}$
that is permutation-invariant must be block-diagonal in the Schur-Weyl blocks.
Moreover, it admits a decomposition of the form
\begin{align}
  X_{S^n} =
  \sum_{\lambda\in\Young(d_S, n)}
  [\,X^{(\lambda)}\otimes \mathds{1}_{\mathcal{P}_\lambda} ]_\lambda\ ,
  \label{z:lyKgiW07}
\end{align}
where $X^{(\lambda)}$ lives in $\mathcal{Q}_\lambda$ and can be determined by
investigating $X_{S^n} \Pi_{S^n}^{\lambda}$.  The space $\mathcal{Q}_\lambda$
actually hosts a representation $q_\lambda(X)$ of the general linear group,
meaning that any i.i.d.\@ operator $X_S^{\otimes n}$ decomposes as
\begin{align}
  X_S^{\otimes n}
  = 
  \sum_{\lambda\in\Young(d_S, n)}
  [\,q_\lambda(X)\otimes \mathds{1}_{\mathcal{P}_\lambda} ]_\lambda\ .
  \label{z:V7ZZfDaV}
\end{align}
If an operator $X_{S^n}$ is permutation-invariant \emph{and} invariant under
$U^{\otimes n}$ for any $U\in\mathrm{U}({d_S})$, then it must be uniform over each
Schur-Weyl block:
\begin{align}
  X_{S^n} =
  \sum_{\lambda\in\Young(d_S, n)}
  x_\lambda\, [\,\mathds{1}_{\mathcal{Q}_\lambda} \otimes \mathds{1}_{\mathcal{P}_\lambda} ]_\lambda\ ,
  \label{z:KMvnFbXv}
\end{align}
where $x_\lambda\in\mathbb{C}$.  If $X_{S^n}$ is Hermitian, then
$x_\lambda\in\mathbb{R}$.  The coefficient $x_\lambda$ can be determined by
computing $\operatorname{tr}[{ X_n \Pi^\lambda }]$ and normalizing by the dimensions of the
appropriate irreducible representations.

\subsection{Schur-Weyl decompositions of copies of a bipartite system}

Now consider $n$ copies of a bipartite system $(AB)$.  The global Hilbert space
$(\mathscr{H}_{A}\otimes\mathscr{H}_{B})^{\otimes n}$ admits a Schur-Weyl decomposition
according to~\eqref{z:iWqcTgXX} (taking $S\equiv AB$), with
Schur-Weyl blocks $\Pi^\lambda_{(AB)^n}$ for $\lambda\in\Young(d_Ad_B, n)$.  On
the other hand, we can ignore all the copies of $A$ and inspect the Schur-Weyl
decomposition of the $n$ copies of $B$, yielding Schur-Weyl blocks
$\Pi^{\lambda'}_{B^n}$ of $B^n$ with $\lambda'\in\Young(d_B, n)$.  An
interesting property is that these blocks are compatible, meaning that their
corresponding projectors commute:
\begin{align}
  \bigl[{ \mathds{1}_{A^n}\otimes\Pi_{B^n}^\lambda , \Pi_{(AB)^n}^{\lambda'} }\bigr] = 0
  \quad\forall\ \lambda,\lambda'\ .
\end{align}
This property follows from the fact that $\mathds{1}_{A^n}\otimes\Pi_{B^n}$ is
invariant under permutations of the copies of $(AB)$, which implies that it is
block-diagonal in the $\Pi^{\lambda'}_{(AB)^n}$ according
to~\eqref{z:lyKgiW07}.

Another important property of the Schur-Weyl decompositions of bipartite systems
concerns the symmetric subspace of $(AB)^n$.  Namely, when projected against the
symmetric subspace of $(AB)^n$, the Schur-Weyl blocks of $A^n$ coincide with
those of $B^n$.  This fact is a manifestation of a the decomposition of the
symmetric space of $(AB)^n$ into Schur-Weyl blocks for $A^n$ and $B^n$, cf.\@
e.g.~\cite[Eq.~(2.25)]{R117}.

\begin{proposition}
  \label{z:Qql-Fkhq}
  Let $A,B$ be two quantum systems.  For any
  $\lambda\in\Young\bigl({\max(d_A, d_B),n}\bigr)$,
  \begin{align}
    \Pi^\lambda_{A^n} \Pi^{\mathrm{Sym}}_{(AB)^n}
    = \Pi^\lambda_{B^n} \Pi^{\mathrm{Sym}}_{(AB)^n}\ ,
  \end{align}
  where we set $\Pi^\lambda_{S^n} = 0$ whenever the number of rows in $\lambda$
  is greater than $d_S$ (for $S=A,B$).
\end{proposition}
\begin{proof}[**z:Qql-Fkhq]
  We use the projection
  formula~\eqref{z:HrtNdJyg}, valid for any
  $\lambda\in\Young(n,n)$, to write
  \begin{align}
    \Pi^\lambda_{A^n} \Pi^{\mathrm{Sym}}_{(AB)^n}
    &=
    \frac{d_{\mathcal{P}_\lambda}}{n!}\sum_{\pi\in \mathrm{S}_{n}} \chi^\lambda(\pi)\,U_{A^n}(\pi)
    \;
    \frac1{n!} \sum_{\pi'\in\mathrm{S}_{n}} U_{A^n}(\pi')\otimes U_{B^n}(\pi')
    \nonumber\\
    &= \frac{d_{\mathcal{P}_\lambda}}{(n!)^2} \sum_{\pi,\pi'\in\mathrm{S}_{n}} \chi^\lambda(\pi) \,
    U_{A^n}(\pi \pi')\otimes U_{B^n}(\pi')
    \nonumber\\
    &= \frac{d_{\mathcal{P}_\lambda}}{(n!)^2} \sum_{\pi,\pi''\in\mathrm{S}_{n}} \chi^\lambda(\pi) \,
    U_{A^n}(\pi'')\otimes U_{B^n}(\pi^{-1} \pi'')\ .
    \label{z:a1Nd2lrY}
  \end{align}
  Operating the change of variables $\pi' \to \pi'' = \pi \pi'$, and noting that
  $[U_{B^n}(\pi)]^* = U_{B^n}(\pi)$ given as it is a matrix of real entries that
  simply permutes subsystems,
  \begin{align}
    \eqref{z:a1Nd2lrY}
    &= \frac{d_{\mathcal{P}_\lambda}}{n!}
    \sum_{\pi\in\mathrm{S}_{n}} [\chi^\lambda(\pi^{-1})]^* \, [U_{B^n}(\pi^{-1})]^*
    \;
    \frac1{n!} \sum_{\pi''\in\mathrm{S}_{n}} U_{(AB)^n}(\pi'')
    = \bigl({ \Pi^\lambda_{B^n} }\bigr)^* \, \Pi^{\mathrm{Sym}}_{(AB)^n}\ ,
  \end{align}
  where we relabeled the first sum's index $\pi \to \pi^{-1}$ and using the fact
  that $\Pi^\lambda_{B^n}=\bigl({\Pi^\lambda_{B^n}}\bigr)^*$.
\end{proof}

\subsection{The de~Finetti state and the postselection technique}
\label{z:Mf3kbNlp}

Here we establish some notation and elementary properties related to variants of
the de Finetti state.  We refer to e.g.\@
\R\cite{R72,R82}, and references
therein, for additional proofs and details.
Let $\bar R\simeq A$ and define
\begin{align}
  \bar\zeta_{A^n\bar{R}^n}
  &= \int d\psi_{A\bar{R}} \, \lvert {\psi}\rangle \mkern -1.8mu\relax \langle{\psi}\rvert _{A\bar R}^{\otimes n} \ ,
\end{align}
where $d\psi$ is the unitarily invariant measure on pure states that is induced
by the Haar measure on the unitary group, normalized to
$\int d\psi_{A\bar R} = 1$.  We also know, by Schur's lemma, that
the mixed state $\bar\zeta_{A^n\bar{R}^n}$ is a normalized version of the symmetric subspace projector,
\begin{align}
    \bar\zeta_{A^n\bar{R}^n}
  = \frac1{d_{\mathrm{Sym}(n,d_A d_{\bar R})}}\Pi_{(AR)^n}^{\mathrm{Sym}}\ .
\end{align}
The reduced state on either $A^n$ or $\bar R^n$ are equal and can be written as
\begin{align}
  \bar\zeta_{A^n} 
  &= \operatorname{tr}_{\bar R^n}\bigl[{\bar\zeta_{A^n\bar{R}^n}}\bigr]
  = \int d\sigma_A \, \sigma_A^{\otimes n}\ ;
  &
  \bar\zeta_{\bar R^n}
  &= \int d\sigma_{\bar R} \, \sigma_{\bar R}^{\otimes n}\ ,
\end{align}
where $d\sigma_A$ is the induced measure of $d\psi_{A\bar{R}}$ on $A$ via the
partial trace, and is equal to $d\sigma_{\bar R}$ which acts on $\bar R$ instead
of $A$.  (Interestingly, the reduced measure $d\sigma_A$ coincides with the
measure induced by the Hilbert-Schmidt metric on Hermitian operators, up to
normalization~\cite{R151,R152}.)

Invoking Carath\'eodory's theorem, there exists an ensemble of $\operatorname{poly}(n)$ states
$\{{ \lvert {\phi^{(j)}}\rangle _{A\bar{R}} }\}$ with a normalized probability distribution
$\{{ \kappa_j }\}$ satisfying
$\kappa_1\geq\kappa_2\geq\cdots\geq\kappa_{\operatorname{poly}(n)}$, such that for any unitary
$U_{A\bar{R}}$,
\begin{align}
  \bar\zeta_{A^n\bar{R}^n}
  = \sum_j \kappa_j\,
  U_{A\bar{R}}^{\otimes n}\,
  \lvert { \phi^{(j)} }\rangle \mkern -1.8mu\relax \langle{ \phi^{(j)} }\rvert _{A\bar{R}}\,
  U_{A\bar{R}}^{\otimes n\,\dagger}\ .
  \label{z:OvegRcve}
\end{align}
(This argument is often formulated without the unitary $U$, but it is trivial
to include this unitary in the above statement since
$U_{A\bar{R}}^{\otimes n\ \dagger} \,\bar\zeta_{A^n\bar{R}^n}\,
U_{A\bar{R}}^{\otimes n} = \bar\zeta_{A^n\bar{R}^n}$.)
This representation of $\bar\zeta_{A^n\bar R^n}$ as a sum leads us to one out of
several arguments to prove the \emph{postselection
  technique}~\cite{R72,R74,R73}: For any quantum state $\sigma_A$,
let $\lvert {\sigma}\rangle _{A\bar R} = \sigma_A^{1/2}\lvert {\Phi_{A:\bar R}}\rangle $, and pick
$U_{A\bar{R}}$ such that
$U_{A\bar R}\lvert {\phi^{(1)}}\rangle _{A\bar R} = \lvert {\sigma}\rangle _{A\bar R}$.  Then
\begin{align}
  \sigma_A^{\otimes n}
  = \operatorname{tr}_{\bar R^n}\bigl[{ U_{A\bar R}^{\otimes n} \,\lvert {\phi^{(1)}}\rangle \mkern -1.8mu\relax \langle{\phi^{(1)}}\rvert _{A\bar R}^{\otimes n}\,
      U^{\otimes n\,\dagger} }\bigr]
  \leq \kappa_1^{-1} \operatorname{tr}_{\bar R^n}\bigl[{ \bar\zeta_{A^n\bar R^n} }\bigr]
  \leq \operatorname{poly}(n) \, \bar\zeta_{A^n}\ ,
\end{align}
noting that $\kappa_1\geq 1/\operatorname{poly}(n)$ as the greatest coefficient of a
$\operatorname{poly}(n)$-sized normalized probability distribution.  In other words: any
i.i.d.\@ state can be operator-upper-bounded by the universal state
$\bar\zeta_{A^n}$, up to a polynomial factor.

Now we discuss two distinct purifications of the state $\bar\zeta_{A^n}$.
Let $R'$ be a quantum register of dimension $\operatorname{poly}(n)$.  We can purify
$\bar\zeta_{A^n\bar{R}^n}$ using this register, thanks to the
representation~\eqref{z:OvegRcve}:
\begin{align}
  \lvert {\bar\zeta}\rangle _{A^n \bar{R}^n R'}
  = \sum_j \sqrt{\kappa_j}\,\lvert {\phi^{(j)}}\rangle _{A \bar{R}}^{\otimes n} \otimes \lvert {j}\rangle _{R'}\ .
\end{align}
Alternatively, we can purify the de~Finetti state $\bar\zeta_{A^n}$ directly on
a copy $R^n$ of $A^n$, as $\bar\zeta_{A^n}^{1/2}\,\lvert {\Phi_{A^n:R^n}}\rangle $.  We
denote the resulting state by $\zeta_{A^nR^n}$:
\begin{align}
  \lvert {\zeta}\rangle _{A^n R^n}
  \equiv \bigl({\bar\zeta_{A^n}^{1/2}}\bigr) \lvert {\Phi_{A:R}}\rangle ^{\otimes n}
  \equiv \bigl({\zeta_{R^n}^{1/2}}\bigr) \lvert {\Phi_{A:R}}\rangle ^{\otimes n}\ .
\end{align}
The reduced states of both $\lvert {\bar\zeta}\rangle _{A^nR^nR'}$ and $\lvert {\zeta}\rangle _{A^nR^n}$
obey
$\bar\zeta_{A^n} = \zeta_{A^n}=\bar\zeta_{\bar R^n} = \zeta_{R^n} = \int
d\sigma_{R} \sigma_R^{\otimes n}$, where isometric mappings between $A$,
$\bar R$ and $R$ are implied.
As purifications of the same state $\zeta_{A^n}$ on $A^n$, we note that the two
states $\lvert {\zeta}\rangle _{A^nR^n}$ and $\lvert {\bar\zeta}\rangle _{A^n\bar{R}^nR'}$ are related
by a partial isometry on $R\to\bar{R} R'$.

\subsection{Integration formulas for Haar-random channels: Proofs of
  \cref{z:TvA2E9KA,z:7r5n0QDD}}

The construction of our approximate microcanonical channel operator relies on
extending the previous section's de~Finetti techniques to quantum channels.

We make use of an integration formula for computing averages over the unitary
group acting in tensor product form.  Specifically, we rely on an integration
formula stated as Theorem~S.2 in \R\cite{R83} (it appears as
Theorem~5 in that reference's arXiv preprint version).  We restate it here,
referring to \R\cite{R83} for a proof:
\begin{theorem}[Integration formula for Haar
  twirling~\unexpanded{\cite[Theorem~S.2]{R83}}]
  \noproofref
  \label{z:ewy6vxah}
  Let $S$ be a quantum system and let $n>0$.  Then for any
  operator $X_{S^n}$,
  \begin{align}
    \int dW_{S} \, W_{S^n}^{\otimes n}\, X_{S^n} \, W_{S^n}^{\otimes n\,\dagger}
    = \frac1{n!} \sum_{\pi\in \mathrm{S}_{n}}
    \operatorname{tr}_{S^n}\bigl({ X_{S^n} U_{S^n}(\pi) }\bigr) \, U_{S^n}(\pi^{-1}) \!\!\!
    \sum_{\lambda\in\Young(d_S,n)}  \!\!
    \frac{d_{\mathcal{P}_\lambda}}{d_{\mathcal{Q}_{\lambda}}}\,\Pi_{S^n}^{\lambda}\ ,
    \label{z:ICkrpeCO}
  \end{align}
  where $dW_S$ denotes the Haar measure on $\mathrm{U}({d_S})$.
\end{theorem}

\begin{lemma}
  \label{z:smOb5fVV}
  Let $S$, $R$ be any quantum systems with $d_S\geq d_R$, and let $n>0$.  Let
  $\lvert {\Psi^0}\rangle _{SR}$ be any ket such that $\operatorname{tr}_S[{ \Psi^0_{SR} }] = \mathds{1}_R$.
  Then
  \begin{align}
    \int dW_{S} \, W_{S}^{\otimes n}\,
    [{\Psi^0_{SR}}]^{\otimes n} \, W_{S}^{\otimes n\,\dagger}
    =
    \Pi^{\mathrm{Sym}}_{(SR)^n}
    \sum_{\lambda\in\Young(d_R,n)}
    \frac{d_{\mathcal{P}_\lambda}}{d_{\mathcal{Q}_\lambda}}
    \Pi_{R^n}^\lambda\ .
  \end{align}
\end{lemma}
\begin{proof}[**z:smOb5fVV]
  Any $\lvert {\Psi^0}\rangle _{SR}$ with $\operatorname{tr}_S[{ \Psi^0_{SR} }] = \mathds{1}_R$ can be written
  in the form
  \begin{align}
    \lvert {\Psi^0}\rangle _{SR} = K_{R'\to S} \, \lvert {\Phi_{R':R}}\rangle \ ,
  \end{align}
  for some isometry $K_{R'\to S}$ by making use of the Schmidt decomposition.
  Plugging in $[{\Psi_{SR}^0}]^{\otimes n}$ as the $X_{S^n}$ operator in
  \cref{z:ICkrpeCO}, we obtain
  \begin{align}
    \hspace*{4em}&\hspace*{-4em}
    \int dW_{S} \, W_{S}^{\otimes n}\, [{\Psi_{SR}^0}]^{\otimes n} \, W_{S}^{\otimes n\,\dagger}
                   \nonumber\\
    &= 
    \frac1{n!} \sum_{\pi\in\mathrm{S}_{n}}
      \operatorname{tr}_{S^n}\Bigl[{  \bigl({K\, \Phi_{R':R}\, K^\dagger}\bigr)^{\otimes n}\, U_{S^n}(\pi) }\Bigr]
      \, U_{S^n}(\pi^{-1})
    \sum_{\lambda\in\Young(d_S,n)} \frac{d_{\mathcal{P}_\lambda}}{d_{\mathcal{Q}_\lambda}}
    \Pi_{S^n}^\lambda\ .
    \label{z:ZrIRm9ie}
  \end{align}
  Observe first of all that
  $U_{S^n}(\pi) K^{\otimes n} = K^{\otimes n} U_{R'^n}(\pi)$.
  Then, note that
  $\lvert {\Phi_{R':R}}\rangle ^{\otimes n} = U_{(R'R)^n}(\pi) \lvert {\Phi_{R':R}}\rangle ^{\otimes n} =
  U_{R'^n}(\pi) \otimes U_{R^n}(\pi) \lvert {\Phi}\rangle _{R':R}^{\otimes n}$, since
  $U_{X^n}(\pi)$ simply permutes the given tensor factors, which implies that
  $U_{R'^n}(\pi) \lvert {\Phi_{R':R}}\rangle ^{\otimes n} = U_{R^n}(\pi^{-1})
  \lvert {\Phi_{R':R}}\rangle ^{\otimes n}$.  Therefore,
  \begin{align}
    \operatorname{tr}_{S^n}\Bigl[{
      K^{\otimes n} \Phi_{R':R}^{\otimes n} K^{\otimes n\,\dagger}\,U_{S^n}(\pi)
    }\Bigr]
    = 
    \operatorname{tr}_{R'^n}\Bigl[{ U_{R'^n}(\pi) \, \Phi_{R':R}^{\otimes n} }\Bigr]
    = 
    \operatorname{tr}_{R'^n}\Bigl[{ U_{R^n}(\pi^{-1}) \, \Phi_{R':R}^{\otimes n} }\Bigr]
    = U_{R^n}(\pi^{-1})\ .
  \end{align}
  Continuing from above,
  \begin{align}
    \eqref{z:ZrIRm9ie}
    &= \frac1{n!} \sum_{\pi\in\mathrm{S}_{n}}
    U_{S^n}(\pi^{-1}) \otimes U_{R^n}(\pi^{-1})
    \sum_{\lambda\in\Young(d_S,n)}
    \frac{d_{\mathcal{P}_\lambda}}{d_{\mathcal{Q}_\lambda}}
    \Pi_{S^n}^\lambda 
      \nonumber\\
    &= \frac1{n!} \sum_{\pi\in\mathrm{S}_{n}}
    U_{(SR)^n}(\pi)
    \sum_{\lambda\in\Young(d_S,n)}
    \frac{d_{\mathcal{P}_\lambda}}{d_{\mathcal{Q}_\lambda}}
    \Pi_{S^n}^\lambda
      \nonumber\\
    &= \Pi^{\mathrm{Sym}}_{(SR)^n}
    \sum_{\lambda\in\Young(d_S,n)}
    \frac{d_{\mathcal{P}_\lambda}}{d_{\mathcal{Q}_\lambda}}
    \Pi_{S^n}^\lambda\ ,
      \label{z:1cJt5QMy}
  \end{align}
  where the last equality follows from expressing the projector onto the
  symmetric subspace as a sum of permutation operators
  [\cref{z:bfP3uOkZ}]. 
  Finally, we invoke \cref{z:Qql-Fkhq} to
  move the $\Pi_{S^n}^\lambda$'s over to the $R^n$ system:
  \begin{align}
    \eqref{z:1cJt5QMy}
    &=
      \Pi^{\mathrm{Sym}}_{(SR)^n} \; 
      \sum_{\lambda\in\Young(d_R,n)}
      \frac{d_{\mathcal{P}_\lambda}}{d_{\mathcal{Q}_\lambda}}
      \Pi_{R^n}^\lambda\ ,
  \end{align}
  further noting that the product of $\Pi^{\mathrm{Sym}}_{(SR)^n}$ with any
  terms in~\eqref{z:1cJt5QMy} with $\lambda\not\in\Young(d_R,n)$
  must vanish thanks to \cref{z:Qql-Fkhq}.
\end{proof}

\begin{proof}[*z:TvA2E9KA]
  Applying \cref{z:smOb5fVV} with $S\simeq R$ and
  $\lvert {\Psi^0}\rangle _{SR} = \lvert {\Phi_{S:R}}\rangle $, we find
  \begin{align}
    \mathds{1}_{R^n}
    &= \operatorname{tr}_{S^n}\mathopen{}\left[{
      \int dW_{S}  \, W_S^{\otimes n}\,[\Phi_{S:R}]^{\otimes n}\,W_S^{\otimes n\,\dagger}
    }\right]\mathclose{}
    = \operatorname{tr}_{S^n}\bigl[{ \Pi^{\mathrm{Sym}}_{(SR)^n} }\bigr]
      \sum_{\lambda\in\Young(d_S,n)}
      \frac{d_{\mathcal{P}_\lambda}}{d_{\mathcal{Q}_\lambda}} \Pi^\lambda_{R^n}
      \ .
  \end{align}
  The claim follows by recalling that
  $\zeta_{R^n} = d_{\mathrm{Sym}(n,d_R^2)}^{-1}\,
  \operatorname{tr}_{S^n}\bigl[{\Pi^{\mathrm{Sym}}_{(SR)^n}}\bigr]$.
\end{proof}

\begin{proof}[*z:7r5n0QDD]
  Follows immediately from \cref{z:smOb5fVV,z:TvA2E9KA}.
\end{proof}

\subsection{Proof of the constrained channel postselection theorem (\cref{z:UEUfoDLC})}
\label{z:WfwOwBXs}

\begin{proof}[*z:UEUfoDLC]
  Let $E_B\simeq B$, $E_R\simeq R$ be additional quantum systems.  The system
  $E=E_B E_R$ then has a size that is suitable to serve as a Stinespring
  dilation environment of any channel $A\to B$.  Fix any pure state
  $\lvert {\phi}\rangle _{BE_B}$, and let
  \begin{align}
    \lvert {\Psi^0}\rangle _{EBR} = \lvert {\Phi_{E_R:R}}\rangle  \otimes \lvert {\phi}\rangle _{B E_B}\ .
  \end{align}
  Consider the object
  \begin{align}
    \Xi_{E^nB^nR^n}
    = \int dW_{EB}\, W_{EB}^{\otimes n} \, \lvert {\Psi^0}\rangle \mkern -1.8mu\relax \langle{\Psi^0}\rvert _{EBR}^{\otimes n}
    \, W_{EB}^{\otimes n\,\dagger}\ ,
  \end{align}
  where $dW_{EB}$ is the Haar measure on all unitaries acting on $ER$,
  normalized to $\int dW_{EB} = 1$.  We have $\Xi_{R^n} = \mathds{1}_{R^n}$ by
  construction.  The object $\Xi_{R^n}$ can be interpreted as sampling a quantum
  channel completely at random by sampling its Stinespring dilation with respect
  to the Haar measure on $EB$, and computing the average of its $n$-fold tensor
  product.

  Thanks to \cref{z:7r5n0QDD,z:TvA2E9KA}, we know that
  \begin{align}
    \Xi_{E^nB^nR^n} = \alpha^{-1} \zeta_{R^n}^{-1}\, \Pi^{\mathrm{Sym}}_{(EBR)^n}\ ,
  \end{align}
  with
  \begin{align}
  \alpha\zeta_{R^n}
    &= \sum_{\lambda\in\Young(d_R,n)}
      \frac{d_{\mathcal{Q}_\lambda}}{d_{\mathcal{P}_\lambda}} \,\Pi_{R^n}^\lambda\ ;
    &
      \alpha &\equiv d_{\mathrm{Sym}(n, d_R^2)}\ ,
  \end{align}
  and noting that $\Xi_{E^nB^nR^n}$, $\zeta_{R^n}$, and
  $\Pi^{\mathrm{Sym}}_{(EBR)^n}$ all commute pairwise.  Therefore,
  \begin{align}
    \Pi^{\mathrm{Sym}}_{(EBR)^n} = \alpha\,\zeta_{R^n}  \, \Xi_{E^nB^nR^n}\ .
  \end{align}
  Furthermore, note that $\zeta_{R^n}$, $\Xi_{E^nB^nR^n}$, and
  $\Pi^{\mathrm{Sym}}_{(EBR)^n}$ all commute with any operator $X_{R^n}$ that is
  permutation-invariant.  Indeed, $\mathds{1}_{(EB)^n}\otimes X_{R^n}$ is
  permutation-invariant so admits a decomposition along the Schur-Weyl blocks of
  $(EBR)^n$ and therefore commutes with $\Pi_{(EBR)^n}^\lambda$; then, $X_{R^n}$
  decomposes in the Schur-Weyl blocks on $R^n$ by permutation invariance and so
  it commutes with $\zeta_{R^n}$; finally, $X_{R^n}$ commutes with
  $\Xi_{E^nB^nR^n}$ since it commutes with both $\Pi_{(EBR)^n}^\lambda$ and
  $\zeta_{R^n}^{-1}$.  This argument applies to both operators $X_{R^n}$ and
  $Y_{R^n}$ of the claim as well as to their adjoints $X_{R^n}^\dagger$,
  $Y_{R^n}^\dagger$.

  Let $\lvert {E}\rangle _{E^nB^nR^n} = E_{B^nR^n}^{1/2}\lvert {\Phi_{E_BE_R:BR}}\rangle $ be a
  purification of $E_{B^nR^n}$.  Permutation invariance of $E_{B^nR^n}$ implies
  $\Pi^{\mathrm{Sym}}_{(EBR)^n} \lvert {E}\rangle _{E^nB^nR^n} = \lvert {E}\rangle _{E^nB^nR^n}$.
  Then
  \begin{align}
    X_{R^n}^\dagger Y_{R^n} E_{E^nB^nR^n} Y_{R^n}^\dagger X_{R^n}
 &= \Pi^{\mathrm{Sym}}_{(EBR)^n} \, X_{R^n}^\dagger Y_{R^n} \, E_{E^nB^nR^n}
      \, Y_{R^n}^\dagger X_{R^n} \, \Pi^{\mathrm{Sym}}_{(EBR)^n}
      \nonumber\\
    &= \alpha^2\,
      \Xi_{E^nB^nR^n}\,\zeta_{R^n}^{1/2} X_{R^n}^\dagger Y_{R^n}  \zeta_{R^n}^{1/2}
      \, E_{E^nB^nR^n} \,
      \zeta_{R^n}^{1/2} Y_{R^n}^\dagger X_{R^n} \zeta_{R^n}^{1/2} \, \Xi_{E^nB^nR^n}\ .
      \label{z:i4qKm4GW}
  \end{align}

  We now identify another expression for $\Xi_{E^nB^nR^n}$.  Using the operator
  vectorized (double-ket) notation, we have
  \begin{align}
    \bigl \lvert {\Xi}\bigr \rrangle _{E^nB^nR^n}
    &= \mathcal{W}_{E^nB^nR^n} \; \bigl \lvert { \Psi^0 }\bigr \rrangle _{EBR}^{\otimes n}\ ;
    \\
    \mathcal{W}_{E^nB^nR^n}
    &= 
      \int dW_{EB}\, \Bigl({ W_{EB}^{\otimes n} \otimes W_{E'B'}^{\otimes n\,*} }\Bigr) \,
      \Bigl({ \Pi^{\mathrm{Sym}}_{(EBR)^n}\otimes\Pi^{\mathrm{Sym}}_{(E'B'R')^n} }\Bigr) \ .
  \end{align}
  All the individual objects
  $(W^{\otimes n}\otimes W^{\otimes n\, *})
  (\Pi^{\mathrm{Sym}}\otimes\Pi^{\mathrm{Sym}} )$ (for each $W$) live in a
  Hilbert-Schmidt operator space of matrices of dimension
  $\bigl({d_{\mathrm{Sym}(n,d_Ed_Bd_R)}}\bigr)^4 \leq \operatorname{poly}(n)$.  By Carath\'eodory's
  theorem, there exists a subset of $\operatorname{poly}(n)$ of such elements, identified by a
  set $\{{ W_\ell }\}_{\ell=1}^{\operatorname{poly}(n)}$, along with a probability distribution
  $\{{ \kappa'_\ell }\}$ with $\kappa'_1 \geq \kappa'_2\geq \cdots$, such that
  \begin{align}
    \mathcal{W}_{E^nB^nR^n}
    = \sum_\ell \kappa'_\ell \,
    \Bigl({ W_{\ell}^{\otimes n} \otimes W_{\ell}^{\otimes n\,*} }\Bigr) \,
      \Bigl({ \Pi^{\mathrm{Sym}}_{(EBR)^n}\otimes\Pi^{\mathrm{Sym}}_{(E'B'R')^n} }\Bigr)\ .
  \end{align}
  Furthermore, $\mathcal{W}_{E^nB^nR^n}$ is invariant under the action of any
  tensor product unitary on $EB$, by definition and by unitary invariance of the
  measure $dW_{EB}$.  In summary, there exists $\{{ W_\ell }\}_{\ell=1}^{\operatorname{poly}(n)}$
  as above such that for any unitary $W'_{EB}$, we have
  \begin{align}
    \mathcal{W}_{E^nB^nR^n}
    = \sum_\ell \kappa'_\ell \,
    \Bigl({ (W'W_{\ell})^{\otimes n} \otimes (W'W_{\ell})^{\otimes n\,*} }\Bigr) \,
      \Bigl({ \Pi^{\mathrm{Sym}}_{(EBR)^n}\otimes\Pi^{\mathrm{Sym}}_{(E'B'R')^n} }\Bigr)\ .
  \end{align}
  So, for any unitary $W'_{EB}$, we have
  \begin{align}
    \Xi_{E^nB^nR^n}
    &= \mathcal{W}_{E^nB^nR^n}\bigl[{ \bigl({\Psi^0_{EBR}}\bigr)^{\otimes n} }\bigr]
      \nonumber\\
    &= \sum_{\ell} \kappa'_\ell \; (W'_{EB} W_{\ell;ER})^{\otimes n}
      \bigl({\Psi^0_{EBR}}\bigr)^{\otimes n}  (W'_{EB} W_{\ell;ER})^{\otimes n\,\dagger}
      \nonumber\\
    &= \sum_{\ell} \kappa'_\ell \; \Bigl({ \Psi^{(W' W_\ell)}_{EBR} }\Bigr)^{\otimes n} \ ,
  \end{align}
  where we defined
  \begin{align}
    \bigl \lvert {\Psi^{(W)}}\bigr \rangle _{EBR} \equiv W_{EB} \, \bigl \lvert {\Psi^0}\bigr \rangle _{EBR}\ .
  \end{align}

  We return to~\eqref{z:i4qKm4GW} with the intent of plugging
  in the above expression for $\Xi_{E^nB^nR^n}$.  Define the shorthand notation
  \begin{align}
    \widetilde{E}_{E^nB^nR^n} \equiv
    \zeta_{R^n}^{1/2} \, X^\dagger\,Y\, \zeta_{R^n}^{1/2}
    \; E \;
    \zeta_{R^n}^{1/2} \, Y^\dagger\,X\, \zeta_{R^n}^{1/2}\ ,
  \end{align}
  omitting some indices for readability. We then find that, for all unitaries
  $W'$,
  \begin{align}
    \text{\eqref{z:i4qKm4GW}}
    &=
      \alpha^2 \,
      \Xi_{E^nB^nR^n}\, \widetilde{E}_{E^nB^nR^n}\,\Xi_{E^nB^nR^n}
      \nonumber\\
    &=
      \alpha^2 \sum_{\ell,\ell'} 
      \kappa'_\ell \kappa'_{\ell'}\,
      \Bigl({ \Psi_{EBR}^{(W'W_\ell)}}\Bigr)^{\otimes n}
      \widetilde{E}_{E^nB^nR^n}\,
      \Bigl({ \Psi_{EBR}^{(W'W_{\ell'})}}\Bigr)^{\otimes n}\ .
  \end{align}
  The equality being true for all unitaries $W'$ (recall the $\{{ W_\ell }\}$ do
  not depend on $W'$), we may as well average over $W'$:
  \begin{align}
    \text{\eqref{z:i4qKm4GW}}
    &=
      \alpha^2
      \int dW'_{EB}\,
      \sum_{\ell,\ell'} 
      \kappa'_\ell \kappa'_{\ell'}\,
      \Bigl({ \Psi_{EBR}^{(W'W_\ell)}}\Bigr)^{\otimes n}
      \widetilde{E}_{E^nB^nR^n}\,
      \Bigl({ \Psi_{EBR}^{(W'W_{\ell'})}}\Bigr)^{\otimes n}\ .
      \label{z:q8-HKT8X}
  \end{align}
  By an operator pinching-type inequality (cf.\@ \cref{z:h6tAs5Us}), we
  have
  \begin{align}
    \text{\eqref{z:q8-HKT8X}}
    &\leq
      \operatorname{poly}(n) \int dW' \sum_\ell (\kappa'_\ell)^2 \,
      \Bigl({ \Psi_{EBR}^{(W'W_\ell)}}\Bigr)^{\otimes n}
      \,\widetilde{E}_{E^nB^nR^n}\,
      \Bigl({ \Psi_{EBR}^{(W'W_{\ell})}}\Bigr)^{\otimes n}
      \nonumber\\
    &\leq
      \operatorname{poly}(n) \sum_\ell \kappa'_\ell \int dW' \,
      \Bigl({ \Psi_{EBR}^{(W'W_\ell)}}\Bigr)^{\otimes n}
      \,\widetilde{E}_{E^nB^nR^n}\,
      \Bigl({ \Psi_{EBR}^{(W'W_{\ell})}}\Bigr)^{\otimes n}
      \nonumber\\
    &\leq
      \operatorname{poly}(n) \int dW'' \,
      \Bigl({ \Psi_{EBR}^{(W'')}}\Bigr)^{\otimes n}
      \,\widetilde{E}_{E^nB^nR^n}\,
      \Bigl({ \Psi_{EBR}^{(W'')}}\Bigr)^{\otimes n}
  \end{align}
  where we used $\kappa'_\ell\leq 1$, carried out the change of variables
  $W'\to W''=W'W_\ell$, and used $\sum \kappa'_\ell = 1$.  Writing out the full
  inequality, and rearranging some terms:
  \begin{align}
    \hspace*{4em}&\hspace*{-4em}
    X_{R^n}^\dagger Y_{R^n}\,E_{E^nB^nR^n}\,Y_{R^n}^\dagger X_{R^n}
    \nonumber\\
    &\leq \operatorname{poly}(n)
      \int dW\, \bigl \lvert {\Psi^{(W)}}\big \rangle \mkern -1.8mu\relax \big \langle{\Psi^{(W)}}\bigr \rvert _{EBR}^{\otimes n}
      \,
      \Bigl \lvert {
      \bigl \langle {\Psi^{(W)}}\bigr \rvert _{EBR}^{\otimes n}
      \;
      \zeta_{R^n}^{1/2} \, X^\dagger\,Y\, \zeta_{R^n}^{1/2}
      \; \bigl \lvert {E}\bigr \rangle _{E^nB^nR^n}
      }\Bigr \rvert ^2 \ .
      \label{z:gjvmIR0q}
  \end{align}

  \Cref{z:gjvmIR0q} can be viewed as the root
  form of our channel postselection theorem.  The expression in the claim is
  more natural to parse but might technically be slightly weaker
  than~\eqref{z:gjvmIR0q}.

  For a given $W$, let $M_{BR} = \operatorname{tr}_E\bigl\{{ \bigl \lvert {\Psi^{(W)}}\big \rangle \mkern -1.8mu\relax \big \langle{\Psi^{(W)}}\bigr \rvert _{EBR} }\bigr\}$,
  noting that $M_R = \mathds{1}_R$ by construction.  Now,
  $X\,\zeta_{R^n}^{1/2} \bigl({\bigl \lvert {\Psi^{(W)}}\bigr \rangle _{EBR}}\bigr)^{\otimes n}$ is a
  purification of the operator
  $X\, \zeta_{R^n}^{1/2}\,M_{BR}^{\otimes n}\,\zeta_{R^n}^{1/2}\,X^\dagger$.  Also,
  $Y\,\zeta_{R^n}^{1/2}\,\bigl \lvert {E}\bigr \rangle _{E^nB^nR^n}$ is a purification of
  $Y\,\zeta_{R^n}^{1/2}\,E_{B^nR^n}\,\zeta_{R^n}^{1/2}\,Y^\dagger$.  From
  Uhlmann's theorem, %
  \begin{align}
    \Bigl \lvert { \bigl \langle {\Psi^{(W)}}\bigr \rvert _{EBR}^{\otimes n} \; \zeta_{R^n}^{1/2} \,
    X^\dagger\,Y\, \zeta_{R^n}^{1/2} \; \bigl \lvert {E}\bigr \rangle _{E^nB^nR^n} }\Bigr \rvert ^2
    \leq F^2\Bigl({
    X\, \zeta_{R^n}^{1/2}\,M_{BR}^{\otimes n}\,\zeta_{R^n}^{1/2}\,X^\dagger\ ,\ 
    Y\,\zeta_{R^n}^{1/2}\,E_{B^nR^n}\,\zeta_{R^n}^{1/2}\,Y^\dagger
    }\Bigr)\ .
  \end{align}

  Now we define the measure $dM_{BR}$ on Choi matrices of quantum channels
  simply as the measure obtained by partial trace of $\Psi^{(W)}_{EBR}$ starting
  from the Haar measure $dW_{EB}$.  We find:
  \begin{align}
    \eqref{z:gjvmIR0q}
    &\leq
      \operatorname{poly}(n)\, \int dM_{BR}\, M_{BR}^{\otimes n}\,
      F^2\Bigl({
      \mathcal{M}^{\otimes n}\bigl({ X\, \zeta_{A^nR^n}\,X^\dagger}\bigr)
      \ ,\ 
      \mathcal{E}\bigl({ Y\,\zeta_{A^nR^n}\,Y^\dagger }\bigr)
      }\Bigr)\ ,
  \end{align}
  where we rewrote the arguments of the fidelity using channels instead of the
  corresponding Choi matrices and with the notation
  $\lvert {\zeta}\rangle _{A^nR^n} \equiv \zeta_{R^n}^{1/2}\,\lvert {\Phi_{A^n:R^n}}\rangle $.
  Phew, we're done!
\end{proof}

\subsection{Proof of the constrained channel postselection theorem for i.i.d.\@
  input states (\cref{z:pnP5Wiy4})}
\label{z:6QVwEcrI}

\begin{proof}[*z:pnP5Wiy4]
  Consider the subnormalized state
  \begin{align}
    \widehat\sigma_{R^n} \equiv
    \int_{F^2(\sigma,\tau)\geq e^{-w}} d\tau \,
      \tau_R^{\otimes n}\ .
  \end{align}
  Now let
  \begin{align}
    L_{R^n}
    &= \widehat\sigma_{R^n}^{1/2}\, \zeta_{R^n}^{-1/2}\ .
  \end{align}
  Observe that $L_{R^n}^{\dagger} = L_{R^n} \geq 0$ because
  $\zeta_{R^n}$ commutes with $\tau_R^{\otimes n}$ for all $\tau_R$.  Also
  note that $L_{R^n} \leq \mathds{1}$ since
  \begin{align}
    L_{R^n}^{\dagger} L_{R^n}
    = \zeta_{R^n}^{-1/2}
    \,\biggl[{
    \int_{F^2(\sigma, \tau)\geq e^{-w}}
    d\tau\,\tau_R^{\otimes n} }\biggr]
    \, \zeta_{R^n}^{-1/2}
    \leq \mathds{1}\ ,
  \end{align}
  since the integral in the brackets is operator-upper-bounded by
  $\int d\tau \,\tau_R^{\otimes n} = \zeta_{R^n}$.  We also have, thanks to
  \cref{z:Ehi5guwR},
  \begin{align}
    \operatorname{tr}\bigl[{ L_{R^n}^\dagger L_{R^n} \sigma_{AR}^{\otimes n} }\bigr]
    = \int_{F^2(\sigma,\tau)\geq{e}^{-w}} d\tau 
    \operatorname{tr}\bigl[{ R_{R^n}^{(\tau)\dagger} R_{R^n}^{(\tau)} \, \sigma_{R}^{\otimes n} }\bigr]
    \geq 1- \operatorname{poly}(n)\exp({-nw})\ .
  \end{align}
  Thanks to the gentle measurement lemma (use \cref{z:bQEMTISK}),
  \begin{align}
    P\bigl({\sigma_{AR}^{\otimes n}  \,,\;
    L_{R^n} \, \sigma_{AR}^{\otimes n} L_{R^n} }\bigr)
    \leq \operatorname{poly}({n}) \exp\Bigl({ -\frac{n w}{2} }\Bigr)\ .
  \end{align}
  In turn, this implies that
  $P\bigl({\sigma_{AR}^{\otimes n} \,,\; L_{R^n} \, \sigma_{AR}^{\otimes n}
  L_{R^n} }\bigr) \leq \operatorname{poly}({n}) \exp\Bigl({ -\frac{n w}{2} }\Bigr)$ and that there exists
  $\Delta'_{A^nR^n}\geq 0$ with
  $\operatorname{tr}\bigl({\Delta'_{A^nR^n}}\bigr) \leq \operatorname{poly}({n}) {e}^{-nw/2}$ such that
  \begin{align}
    \sigma_{AR}^{\otimes n} \leq
    L_{R^n} \, \sigma_{AR}^{\otimes n} L_{R^n} + \Delta'_{A^nR^n}\ .
  \end{align}
  Iteratively applying this relation, we find
  \begin{align}
    \sigma_{AR}^{\otimes n}
    \leq L^2_{R^n} \, \sigma_{AR}^{\otimes n} L^2_{R^n}
    + L_{R^n} \, \Delta'_{A^n R^n} L_{R^n}
    + \Delta'_{A^nR^n}\ .
  \end{align}
  Defining
  $\Delta_{B^nR^n} = \mathcal{E}_{A^n\to B^n}\bigl[{L_{R^n} \, \Delta'_{A^n R^n}
  L_{R^n} + \Delta'_{A^nR^n}}\bigr] \geq 0$, we find that
  $\operatorname{tr}\bigl({\Delta_{B^nR^n}}\bigr) \leq \operatorname{tr}\bigl({L_{R^n}^2\Delta'_{A^nR^n}}\bigr) +
  \operatorname{tr}\bigl({\Delta'_{A^nR^n}}\bigr) \leq \operatorname{poly}({n}){e}^{-nw/2}$ and
  \begin{align}
    \mathcal{E}_{A^n\to B^n}\bigl({\sigma_{AR}^{\otimes n}}\bigr)
    \leq L_{R^n}^2\, \mathcal{E}_{A^n\to B^n}\bigl({\sigma_{AR}^{\otimes n}}\bigr) \, L_{R^n}^2
    + \Delta_{B^nR^n}\ .
    \label{z:huKfsPAj}
  \end{align}
  Let
  \begin{align}
    X_{R^n} &= L_{R^n}\ ;
    &
    Y_{R^n} &= L_{R^n}\ .
  \end{align}
  Using our constrained channel postselection theorem
  (\cref{z:UEUfoDLC}), we find
  \begin{subequations}
    \begin{align}
      L_{R^n}^2 \,
      \mathcal{E}_n\bigl({\sigma_{AR}^{\otimes n}}\bigr) \,
      L_{R^n}^2
      &\leq \operatorname{poly}({n})
        \int dM_{BR}\,
        \mathcal{M}^{\otimes n}\bigl({\sigma_{AR}^{\otimes n}}\bigr) \,
        \mathfrak{F}^2[\mathcal{M}]\ ;
        \label{z:HzmSnetP}
      \\
      \mathfrak{F}^2[\mathcal{M}]
      &\equiv
        F^2\Bigl({
        \mathcal{M}^{\otimes n}\bigl({ L_{R^n}\,\zeta_{A^nR^n} L_{R^n} }\bigr) \,,\;
        \mathcal{E}\bigl({ L_{R^n}\,\zeta_{A^nR^n} L_{R^n} }\bigr)
        }\Bigr)\ .
    \end{align}
  \end{subequations}
  Observe that
  \begin{align}
    L_{R^n}\,\lvert {\zeta_{A^nR^n}}\rangle 
    = 
    L_{R^n}\,\zeta_{R^n}^{1/2}\,\lvert {\Phi_{A^nR^n}}\rangle 
    = 
    \widehat\sigma_{R^n}^{1/2} \lvert {\Phi_{A^n:R^n}}\rangle 
    \equiv
    \lvert {\widehat\sigma_{A^nR^n}}\rangle \ .
  \end{align}
  Now define
  \begin{align}
    \bar\omega_{A^n\bar{R}^n}
    = \int_{F^2(\sigma_{\bar R},\tau_{\bar R})\geq e^{-w}}
    d\tau_{\bar R} \, \lvert {\tau}\rangle \mkern -1.8mu\relax \langle{\tau}\rvert _{A\bar R}^{\otimes n}
    \leq \Pi^{\mathrm{Sym}}_{A^n\bar R^n}\ ,
  \end{align}
  where we write
  $\lvert {\tau}\rangle _{A\bar R} \equiv \tau_{\bar R}^{1/2}\,\lvert {\Phi_{A:\bar R}}\rangle $.
  By construction,
  \begin{align}
    \bar\omega_{A^n}
    = \operatorname{tr}_{\bar R^n}\bigl[{ \bar\omega_{A^n\bar R^n } }\bigr]
    = \widehat\sigma_{A^n}\ .
  \end{align}
  Since $\bar\omega_{A^n\bar R^n}$ has support on the symmetric subspace of
  $(A\bar R)^n$, and by Carath\'eodory's theorem, there exists a collection
  $\bigl\{{ \tau_{\bar R}^{(\ell)} }\bigr\}_{\ell=1}^{\operatorname{poly}(n)}$ of at most $\operatorname{poly}(n)$
  states $\lvert {\tau^{(\ell)}}\rangle _{A\bar R} \equiv
  \bigl({\tau_{\bar R}^{(\ell)}}\bigr)^{1/2}\,\lvert {\Phi_{A:\bar R}}\rangle $ with
  $F\bigl({\sigma_{\bar R},\tau_{\bar R}^{(\ell)}}\bigr) \geq e^{-w}$, along with a probability
  distribution $\{{ \bar\kappa_\ell }\}$ with $\bar\kappa_1\geq\bar\kappa_2\geq\cdots$,
  such that
  \begin{align}
    \bar\omega_{A^n\bar R^n} = \sum_{\ell=1}^{\operatorname{poly}(n)} \bar\kappa_\ell \,
    \bigl({\tau_{A\bar R}^{(\ell)}}\bigr)^{\otimes n}\ .
  \end{align}
  This state can be purified using an additional system $R'$ with
  $d_{R'} \leq \operatorname{poly}(n)$:
  \begin{align}
    \bigl \lvert {\bar\omega}\bigr \rangle _{A^n\bar R^n R'}
    = \sum_{\ell=1}^{\operatorname{poly}(n)}
    \sqrt{\bar\kappa_\ell} \,
    \bigl \lvert {\tau^{(\ell)}}\bigr \rangle _{A\bar R}^{\otimes n} \otimes \bigl \lvert {\ell}\bigr \rangle _{R'}\ .
  \end{align}
  Because $\lvert {\widehat\sigma}\rangle _{A^nB^n}$ and
  $\bigl \lvert {\bar\omega}\bigr \rangle _{A^n\bar R^n R'}$ are both two
  purifications of the same state $\widehat\sigma_{A^n}$, they are related by some
  isometry acting on $R^n \to \bar R^n R'$.  The fidelity is invariant under the
  application of an isometry, so
  \begin{align}
    \mathfrak{F}(\mathcal{M})
    &= F\Bigl({
      \mathcal{M}^{\otimes n}\bigl[{ \widehat\sigma_{A^nR^n} }\bigr] \,,\;
      \mathcal{E}\bigl[{ \widehat\sigma_{A^nR^n} }\bigr]
      }\Bigr)
      = F\Bigl({
      \mathcal{M}^{\otimes n}\bigl[{ \bar\omega_{A^n\bar R^n R'} }\bigr] \,,\;
      \mathcal{E}\bigl[{ \bar\omega_{A^n\bar R^n R'} }\bigr]
      }\Bigr)\ .
    \label{z:aGVz.x5p}
  \end{align}
  By the data processing inequality of the fidelity, the fidelity can only
  increase if we decohere $R'$ in its computational basis.  Further
  invoking \cref{z:2jsqnpDK}, we find
  \begin{align}
    \mathfrak{F}(\mathcal{M})
    &\leq \sum \bar\kappa_\ell  \, F\Bigl({
      \mathcal{M}^{\otimes n}\bigl[{ \tau^{(\ell) \otimes n}_{A^n\bar R^n} }\bigr] \,,\;
      \mathcal{E}\bigl[{ \tau^{(\ell) \otimes n}_{A^n\bar R^n} }\bigr]
      }\Bigr)
    \nonumber\\
    &\leq \max_\ell  F\Bigl({
      \mathcal{M}^{\otimes n}\bigl[{ \tau^{(\ell) \otimes n}_{A^n\bar R^n} }\bigr] \,,\;
      \mathcal{E}\bigl[{ \tau^{(\ell) \otimes n}_{A^n\bar R^n} }\bigr]
      }\Bigr)
      \nonumber\\
    &\leq \max_{\substack{\tau_R:\\ F^2(\tau_R,\sigma_R)\geq e^{-w}}}
      F\Bigl({
      \mathcal{M}^{\otimes n}\bigl[{ \tau_{A^n\bar R^n}^{\otimes n} }\bigr] \,,\;
      \mathcal{E}\bigl[{ \tau_{A^n\bar R^n}^{\otimes n} }\bigr]
      }\Bigr)\ ,
    \label{z:ds24LeNs}
  \end{align}
  writing $\lvert {\tau}\rangle _{AR} \equiv \tau_R^{1/2}\,\lvert {\Phi_{A:R}}\rangle $.
  Combining~\eqref{z:huKfsPAj}
  and~\eqref{z:HzmSnetP} with~\eqref{z:ds24LeNs}
  proves the claim.
\end{proof}

\section{Proofs: Construction of the approximate microcanonical channel operator}

\subsection{General test operator to discriminate i.i.d.\@ channels: Proof of
  \cref{z:HY7bJv7-,z:lmCUsTOF}}
\label{z:iBCoMrzo}
\label{z:zCc78QJk}

\begin{proof}[*z:HY7bJv7-]
  Let $0<y'<1/d_R$, $h>0$, and let $\mathcal{M}_{A\to B}$ be any quantum
  channel.  Let $\sigma_R$ be any state with $\sigma_R\geq y'\mathds{1}$.  Since
  $[{\mathcal{M}_{A\to B}(\sigma_{AR})}]^{\otimes n}$ is manifestly
  permutation-invariant, we can ignore the symmetrization operation
  $\mathcal{S}_{(BR)^n}$ in~\eqref{z:qVquPoYP}.
  We can write
  \begin{align}
    \operatorname{tr}\Bigl[{ P^{\chi_{j;M;>h}}_{B^nR^n} \mathcal{M}^{\otimes n}\bigl({\sigma_{AR}^{\otimes n}}\bigr) }\Bigr]
    &= 
      \int d\tilde\sigma\, \operatorname{tr}\bigl[{ R^{(\tilde\sigma) \dagger} R^{(\tilde\sigma)} \tilde\sigma^{\otimes m} }\bigr]
      \, \Pr[{ \chi_{j;M;>h} \mvert \sigma, \tilde\sigma }]\ ,
  \end{align}
  where
  \begin{align}
    \Pr[{ \chi \mvert \sigma, \tilde\sigma }]
    \equiv \frac1{J^{\bar n}} \sum_{\boldsymbol j} \int d\boldsymbol{z} \,
    \chi(\tilde\sigma, \boldsymbol j, \boldsymbol z)\, 
    \prod_{i=1}^{\bar n}
    \operatorname{tr}\Bigl({ \bigl\{{ \tilde\sigma_R^{-1/2} C^{j_i}_{BR} \tilde\sigma_R^{-1/2} = z_i }\bigr\}
      \, \mathcal{M}(\sigma_{AR}) }\Bigr)\ .
  \end{align}
  Consider any $x>0$ with $x < y'^2$.  We then have
  \begin{multline}
    \operatorname{tr}\Bigl[{ P^{\chi_{j;M;>h}}_{B^nR^n} \mathcal{M}^{\otimes n}\bigl({\sigma_{AR}^{\otimes n}}\bigr) }\Bigr]
    \\
    \leq 
    \underbrace{
            \int_{F^2(\tilde\sigma,\sigma)< e^{-x}}\,d\tilde\sigma
            \operatorname{tr}\bigl[{ R^{(\sigma) \dagger} R^{(\sigma)} \,\tilde\sigma_R^{\otimes m} }\bigr]
        }_{\text{(I)}}
    + \underbrace{
            \int_{F^2(\tilde\sigma,\sigma) \geq e^{-x}}\,d\tilde\sigma
            \Pr[{ \chi_{j;M;>h} \mvert \sigma, \tilde\sigma }]
        }_{\text{(II)}}\ .
  \end{multline}
  The first term is  taken care of by \cref{z:Ehi5guwR}:
  \begin{align}
    \text{(I)} \leq \operatorname{poly}({m}) \exp({-mx})\ .
  \end{align}

  We now focus on the term~(II).  In the term~(II), we have
  $F^2(\tilde\sigma_R,\sigma_R) \geq e^{-x} \geq 1 - x$, which implies
  \begin{align}
    D(\tilde\sigma_R, \sigma_R)
    \leq P(\tilde\sigma_R, \sigma_R) = \sqrt{ 1 - F^2(\tilde\sigma_R, \sigma_R) }
    \leq \sqrt{x}\ .
  \end{align}
  Furthermore,
  \begin{align}
    D(\tilde\sigma_{AR}, \sigma_{AR})
    \leq P(\tilde\sigma_{AR}, \sigma_{AR})
    = \sqrt{ 1 - F^2(\sigma_{AR}, \tilde\sigma_{AR})} \ ,
    \label{z:nxzSPzPp}
  \end{align}
  for which we can invoke \cref{z:X2zwPk80} to find
  \begin{align}
    \text{\eqref{z:nxzSPzPp}}
    \leq \sqrt{ 1 - \Bigl({1 - \sqrt{2P(\sigma_{R}, \tilde\sigma_{R})}}\Bigr)^2} 
    \leq \bigl[{ 8 P(\sigma_{R}, \tilde\sigma_{R}) }\bigr]^{1/4}
    \leq 2 x^{1/8}\ .
  \end{align}
  Furthermore, 
  \begin{align}
    \tilde\sigma_R \geq \sigma_R - D(\tilde\sigma_{R}, \sigma_{R})\,\mathds{1}_R
    \geq (y' - \sqrt{x})\,\mathds{1}_R = y''\mathds{1}\ ,
  \end{align}
  defining $y'' = y' - \sqrt{x}$ with $y''>0$ thanks to our assumption on the
  possible range of values of $x$.  
  
  For each $j=1,\ldots, J$, the variables $\tilde{z}_i^j$ (for
  $i=1, \ldots, \bar{n}$) are i.i.d., with mean
  \begin{align}
    \tilde{q}_{\sigma,M,j} &\equiv \langle { \tilde{z}_i^j }\rangle _{j_i, z_i}
    = \frac1J\sum_{j_i=1}^J \int dz_i \Pr[z_i|j_i]\,\tilde{z}_i^j
      \nonumber\\
    &= \frac1J\sum_{j'=1}^J \Bigl({
      \delta_{j,j_i} J
      \operatorname{tr}\bigl[{\tilde\sigma_R^{-1/2} C^{j}_{BR} \tilde\sigma_R^{-1/2}
      \sigma_R^{1/2} M_{BR} \sigma_R^{1/2} }\bigr]
      }\Bigr)
      \nonumber\\
    &= \operatorname{tr}\bigl[{ \tilde\sigma_R^{-1/2} C^{j}_{BR} \tilde\sigma_R^{-1/2}
      \sigma_R^{1/2} M_{BR} \sigma_R^{1/2} }\bigr]
      \ .
      \label{z:NU3.LOrY}
  \end{align}

  Now, measuring the observable
  $\tilde\sigma_R^{-1/2} C^j_{BR} \tilde\sigma_R^{-1/2}$ on
  $\mathcal{M}(\sigma_{AR})$ has an expected value of
  \begin{align}
    \tilde q_{j,\sigma,\tilde\sigma} \equiv 
    \operatorname{tr}\bigl({\tilde\sigma_R^{-1/2} C^j_{BR} \tilde\sigma_R^{-1/2}
    \sigma_R^{1/2} M_{BR} \sigma_R^{1/2}}\bigr) \ .
  \end{align}
  We have the bound
  \begin{align}
    \bigl \lvert { \tilde q_{j,\sigma,\tilde\sigma} - \operatorname{tr}(C^j M) }\bigr \rvert 
    &= \bigl \lvert { \operatorname{tr}\bigl[{
      \tilde\sigma_R^{-1/2} C^j_{BR} \tilde\sigma_R^{-1/2}\bigl({
      \sigma_R^{1/2} M_{BR} \sigma_R^{1/2} -
      \tilde\sigma_R^{1/2} M_{BR} \tilde\sigma_R^{1/2}
      }\bigr)}\bigr] }\bigr \rvert 
      \nonumber\\
    &=  \bigl \lvert { \operatorname{tr}\bigl[{
      \tilde\sigma_R^{-1/2} C^j_{BR} \tilde\sigma_R^{-1/2} \, \mathcal{M}_{A\to B}\bigl({
      \sigma_{AR} - \tilde\sigma_{AR}
      }\bigr)}\bigr] }\bigr \rvert 
      \nonumber\\
    &\leq \bigl \lVert { \tilde\sigma_{R}^{-1/2} C^j_{BR} \tilde\sigma_R^{-1/2} }\bigr \rVert 
      \,\lVert {\sigma_{AR} - \tilde\sigma_{AR}}\rVert _1
      \nonumber\\
    &\leq  \frac{4x^{1/8}}{ y'' } \bigl \lVert { C^j_{BR} }\bigr \rVert  \ .
  \end{align}
  Furthermore,
  \begin{align}
    \bigl \lvert { \nu_j(\boldsymbol j, \boldsymbol z) - \operatorname{tr}(C^j M) }\bigr \rvert  
    &\leq 
    \bigl \lvert { \nu_j(\boldsymbol j, \boldsymbol z) - \tilde q_{j,\sigma,\tilde\sigma} }\bigr \rvert 
    + \bigl \lvert { \tilde q_{j,\sigma,\tilde\sigma} - \operatorname{tr}(C^j M) }\bigr \rvert 
    \nonumber\\
    &\leq
      \bigl \lvert { \nu_j(\boldsymbol j, \boldsymbol z) - \tilde q_{j,\sigma,\tilde\sigma} }\bigr \rvert 
      + \frac{ 4 x^{1/8} }{ y'' }\bigl \lVert { C^j_{BR} }\bigr \rVert \ .
  \end{align}
  Observe that
  \begin{align}
    \Pr[ \chi_{j;M;>h} \mvert \sigma, \tilde\sigma]
    &= \Pr\bigl[{
      \lvert { \nu_j(\boldsymbol j, \boldsymbol z) - \operatorname{tr}(C^j M) }\rvert  > h
      \mathrel{\big\mvert} \sigma, \tilde\sigma
      }\bigr]
      \ .
  \end{align}
  Now, the event $\lvert { \nu_j(\boldsymbol j, \boldsymbol z) - \operatorname{tr}(C^j M) }\rvert  > h$
  implies the event
  $\bigl \lvert { \nu_j(\boldsymbol j, \boldsymbol z) - \tilde
    q_{j,\sigma,\tilde\sigma} }\bigr \rvert  > h - \frac{ 4 x^{1/8} }{ y'' }\bigl \lVert {
    C^j_{BR} }\bigr \rVert $, meaning that
  \begin{align}
    \Pr[ \chi_{j;M;>h} \mvert \sigma, \tilde\sigma]
    \leq \Pr\Bigl[{
    \bigl \lvert { \nu_j(\boldsymbol j, \boldsymbol z) - \tilde
    q_{j,\sigma,\tilde\sigma} }\bigr \rvert  > h - \frac{ 4 x^{1/8} }{ y'' }\bigl \lVert {
    C^j_{BR} }\bigr \rVert 
    \mathrel{\Big\mvert} \sigma, \tilde\sigma
    }\Bigr]\ .
    \label{z:ZT6l.fx5}
  \end{align}
  By Hoeffding's bound, we find that
  \begin{align}
    \text{\eqref{z:ZT6l.fx5}}
    &\leq 2\exp\biggl\{{
    - \bar{n}\,\frac{ 2 \bigl({h - 4x^{1/8}y''^{-1}\lVert {C^j_{BR}}\rVert }\bigr)^2 }{
        \bigl({ 2 \lVert {\tilde\sigma^{-1/2} C^j \tilde\sigma^{-1/2}}\rVert  }\bigr)^2 }
    }\biggr\}
    \leq 2\exp\biggl\{{
    -  \bar{n}\, \frac{\bigl({h - 4x^{1/8}y''^{-1}\lVert {C^j_{BR}}\rVert }\bigr)^2 }{
        2 y''^{-2}\,\lVert {C^j_{BR}}\rVert ^{2}
      }
    }\biggr\}
      \nonumber\\
    &\leq 
      2\exp\biggl\{{
      -\frac{\bar{n}}{2}\, \bigl({h y'' \lVert {C^j}\rVert ^{-1} - 4x^{1/8} }\bigr)^2
      }\biggr\}
      \ .
  \end{align}
  This gives us a bound on the term (II) we had earlier.  Along with the first
  term, the bound on the probability of $P^\chi$ passing reads
  \begin{align}
    \hspace*{3em}&\hspace*{-3em}
    \operatorname{tr}\Bigl[{
    P^{\chi_{j;M;>h}}_{B^nR^n} \, \mathcal{M}^{\otimes n}\bigl({\sigma_{AR}^{\otimes n}}\bigr)
    }\Bigr]
                   \nonumber\\
    &\leq \operatorname{poly}(m)\exp(-mx) + 
    2\exp\biggl\{{ -\frac{\bar n}{2} \bigl({h y'' \lVert {C^j}\rVert ^{-1} - 4x^{1/8} }\bigr)^2 }\biggr\}
      \nonumber\\
    &\leq \operatorname{poly}(n) \exp\biggl\{{
      - \min\Bigl({ mx \,,\; \frac{\bar n}{2}\bigl({\lVert {C^j}\rVert ^{-1} h ({y' - \sqrt{x}}) - 4x^{1/8} }\bigr)^2 }\Bigr)
      }\biggr\} \ .
  \end{align}

  To get a more specific bound, we choose a value for $x$:
  \begin{align}
    x^{1/8} &= \frac{h y'}{ 5\lVert {C^j_{BR}}\rVert } \ ;
    &
    x^{1/2} &= \frac{h^4 y'^4}{5^4 \lVert {C^j_{BR}}\rVert ^{4} } \ ;
    &
    x &= \frac{h^8 y'^8}{5^8 \lVert {C^j_{BR}}\rVert ^{8}} \ .
  \end{align}
  This value indeed satisfies $0< x < y'^2$ since $h/\lVert {C^j}\rVert  \leq 1$ and
  $y' < 1$.  Then
  \begin{align}
    mx = m\,\frac{h^8 y'^8}{5^8 \lVert {C^j_{BR}}\rVert ^{8} }
    \ .
  \end{align}
  We also have
  \begin{align}
    \frac{h}{\lVert {C^j}\rVert } (y' - \sqrt{x}) - 4x^{1/8}
    &= \frac{hy'}{\lVert {C^j}\rVert }
      - \frac{h}{\lVert {C^j}\rVert }\frac{h^4 y'^4}{5^4 \lVert {C^j}\rVert ^4}
      - \frac45\frac{hy'}{\lVert {C^j}\rVert }
      \nonumber\\
      &= \frac{hy'}{5\lVert {C^j}\rVert } - \frac{h^5 y'^5}{5^4\lVert {C^j}\rVert ^5}
      \geq \Bigl({\frac15 - \frac1{5^4}}\Bigr)\,\frac{hy'}{\lVert {C^j}\rVert }
        \geq \frac{\sqrt{2}}{5^4}\,\frac{hy'}{\lVert {C^j}\rVert }\ ,
        \label{z:cArpKm30}
  \end{align}
  recalling $h/\lVert {C^j}\rVert \leq1$ and $y' < 1/d_R \leq 1$, and noting
  that~\eqref{z:cArpKm30} is strictly positive.  Thus
  \begin{align}
    \frac{\bar n}{2}\biggl({\frac{h}{\lVert {C^j}\rVert } (y' - \sqrt{x}) - 4x^{1/8} }\biggr)^2 
    \geq \bar n\, \frac{h^2 y'^2}{5^8 \lVert {C^j}\rVert ^2}
    \geq \bar n\, \frac{h^8 y'^8}{5^8 \lVert {C^j}\rVert ^8}\ .
  \end{align}
  Therefore,
  \begin{align}
    \min\Bigl({ mx \,,\; \frac{\bar n}{2}\bigl({\lVert {C^j}\rVert ^{-1} h ({y' - \sqrt{x}}) - 4x^{1/8} }\bigr)^2 }\Bigr)
    \geq 
    \min({m, \bar n}) \,\frac{h^8 y'^8}{ 5^8 \lVert {C^j}\rVert ^{8} }\ ,
  \end{align}
  which completes the proof.
\end{proof}

\begin{proof}[*z:lmCUsTOF]
  Let's prove~\ref{z:wdi9pumP}.  With
  $M_{BR} \equiv \mathcal{M}(\Phi_{A:R})$ and for any $j=1, \ldots J$, we have
  \begin{align}
    \lvert { \nu_j(\boldsymbol j, \boldsymbol z) - q_j }\rvert 
    \leq 
    \lvert { \nu_j(\boldsymbol j, \boldsymbol z) - \operatorname{tr}({C^j_{BR} M_{BR}}) }\rvert  + a\ ,
  \end{align}
  which means that for any $j=1, \ldots, J$,
  \begin{align}
    \lvert { \nu_j(\boldsymbol j, \boldsymbol z) - q_j }\rvert  > h'
    \quad\Rightarrow\quad
    \lvert { \nu_j(\boldsymbol j, \boldsymbol z) - \operatorname{tr}({C^j_{BR} M_{BR}}) }\rvert 
    > h' - a\ .
  \end{align}
  In turn,
  \begin{align}
    \chi_{\boldsymbol q;\not \leq h'}(\tilde\sigma, \boldsymbol j, \boldsymbol z)
    &=
    \chi\Bigl\{{
    \exists\ j \in\{{1, \ldots J}\} : \ 
    \lvert { \nu_j(\boldsymbol j, \boldsymbol z) - q_j }\rvert  > h'
    }\Bigr\}
    \nonumber\\
    &\leq
    \sum_{j=1}^J \chi\Bigl\{{
      \lvert { \nu_j(\boldsymbol j, \boldsymbol z) - \operatorname{tr}({C^j_{BR} M_{BR}}) }\rvert 
      > h' - a
    }\Bigr\}
      \nonumber\\
    &=
    \sum_{j=1}^J \chi_{j;M; >(h'-a)}(\tilde\sigma, \boldsymbol j, \boldsymbol z)\ ,
  \end{align}
  with $\chi_{j;M;>h}$ defined
  in~\eqref{z:ASsLiTCZ}, and where
  the middle inequality holds because whenever the condition on the left hand
  side is true, there is at least one term on the right hand side that is equal
  to one.  Thanks to \cref{z:HY7bJv7-}, we find
  \begin{align}
    \operatorname{tr}\bigl[{
    P^{\chi_{\boldsymbol q;\not \leq h'}}_{B^nR^n}
    \mathcal{M}^{\otimes n}(\sigma_{AR}^{\otimes n})
    }\bigr]
    &\leq
    \operatorname{tr}\bigl[{
    P^{\chi_{j;M;>(h'-a)}}_{B^nR^n}
    \mathcal{M}^{\otimes n}(\sigma_{AR}^{\otimes n})
    }\bigr]
      \nonumber\\
    &\leq
    \operatorname{poly}({n}) \sum_{j=1}^J \exp\biggl\{{
    -n\,\min\Bigl({\frac{m}{n}, \frac{\bar n}{n}}\Bigr)\,\frac{(h' - a)^8\,y'^8}{5^8\,\lVert {C^j_{BR}}\rVert ^8}
    }\biggr\}
      \nonumber\\
    &\leq
    \operatorname{poly}({n}) \exp\biggl\{{
    -n\,\min\Bigl({\frac{m}{n}, \frac{\bar n}{n}}\Bigr)\,
      \frac{(h' - a)^8\,y'^8}{5^8\,\max_j \lVert {C^j_{BR}}\rVert ^8}
    }\biggr\}\ ,
  \end{align}
  proving~\ref{z:wdi9pumP}.

  Now we prove~\ref{z:9tjp9JsO}.  Thanks to
  our
  assumption~\eqref{z:.24EDF0r},
  \begin{align}
    b < \bigl \lvert { \operatorname{tr}\bigl({C^{j_0}_{BR} M_{BR}}\bigr) - q_{j_0} }\bigr \rvert 
    \leq
    \bigl \lvert { q_{j_0}  - \nu_{j_0}(\boldsymbol j, \boldsymbol z) }\bigr \rvert 
    + \bigl \lvert { \nu_{j_0}(\boldsymbol j, \boldsymbol z) - \operatorname{tr}\bigl({C^{j_0}_{BR} M_{BR}}\bigr) }\bigr \rvert \ .
  \end{align}
  Then
  \begin{align}
    \chi_{\boldsymbol q; \leq h'}(\tilde\sigma, \boldsymbol j, \boldsymbol z)
    &= \chi\biggl\{{
      \bigl \lvert { \nu_j(\boldsymbol j, \boldsymbol z) - q_j }\bigr \rvert  \leq h'\ \ \forall\ j=1, \ldots J
    }\biggr\}
      \nonumber\\
    &\leq \chi\biggl\{{
      \bigl \lvert { \nu_{j_0}(\boldsymbol j, \boldsymbol z) - \operatorname{tr}\bigl({C^{j_0}_{BR} M_{BR}}\bigr) }\bigr \rvert 
      > b - h'
    }\biggr\}
      \nonumber\\
    &= \chi_{j;M;>(b-h')}\ ,
  \end{align}
  where the middle inequality holds because the event on the left hand side
  implies the one on the right.  Invoking
  \cref{z:HY7bJv7-}, we find
  \begin{align}
    \operatorname{tr}\bigl[{
    P^{\chi_{\boldsymbol q;\leq h'}}_{B^nR^n}
    \mathcal{M}^{\otimes n}(\sigma_{AR}^{\otimes n})
    }\bigr]
    &\leq
    \operatorname{tr}\bigl[{
    P^{\chi_{j;M;>(b-h')}}_{B^nR^n}
    \mathcal{M}^{\otimes n}(\sigma_{AR}^{\otimes n})
    }\bigr]
      \nonumber\\
    &\leq
    \operatorname{poly}({n}) \exp\biggl\{{
    -n\,\min\Bigl({\frac{m}{n},
      \frac{8\bar n}{n}}\Bigr)\,\frac{(b-h')^4\,y'^4}{625\,\lVert {C^{j_0}_{BR}}\rVert ^4}
    }\biggr\}\ ,
  \end{align}
  proving~\ref{z:9tjp9JsO}.
\end{proof}

\subsection{Construction of the approximate microcanonical channel operator: Proof of
  \cref{z:JjqyeW8p}}
\label{z:p2qVGb1M}

(The following proof was established before discovering
\cref{z:pnP5Wiy4}; with apologies to the
reader, we have not yet simplified it to make direct reference to
\cref{z:pnP5Wiy4}.)

\begin{proof}[*z:JjqyeW8p]
  First let's prove~\ref{z:W3JG7egy}.
  Without loss of generality, we can assume $\mathcal{E}_{A^n\to B^n}$ to be
  permutation-invariant, since both $P^\perp_{B^nR^n}$ and the concentration
  test operators are permutation-invariant.  Consider any $\nu>1$ for now; we'll
  only use the additional assumption on $\nu$ to simplify the final bound.
  Consider any $\sigma_R\geq \nu y\mathds{1}$, and write as a shorthand
  \begin{align}
    Q_{j,\sigma} \equiv \bigl\{{ \overline{H^{j,\sigma}}_{B^nR^n} \notin [q_j \pm \eta] }\bigr\}\ .
    \label{z:zBlezNsg}
  \end{align}
  Let $w>0$ to be fixed later and consider the subnormalized state
  \begin{align}
    \widehat\sigma_{R^n} \equiv
    \int_{F(\sigma,\tau)\geq e^{-w}} d\tau \,
      \tau_R^{\otimes n}\ .
  \end{align}
  Now let
  \begin{align}
    L_{R^n}
    &= \widehat\sigma_{R^n}^{1/2}\, \zeta_{R^n}^{-1/2}\ .
  \end{align}
  Observe that $L_{R^n}^{\dagger} = L_{R^n} \geq 0$ because
  $\zeta_{R^n}$ commutes with $\tau_R^{\otimes n}$ for all $\tau_R$.  Also
  note that $L_{R^n} \leq \mathds{1}$ since
  \begin{align}
    L_{R^n}^{\dagger} L_{R^n}
    = \zeta_{R^n}^{-1/2}
    \,\biggl[{
    \int_{F(\sigma, \tau)\geq e^{-w}}
    d\tau\,\tau_R^{\otimes n} }\biggr]
    \, \zeta_{R^n}^{-1/2}
    \leq \mathds{1}\ ,
  \end{align}
  since the integral in the brackets is operator-upper-bounded by
  $\int d\tau \,\tau_R^{\otimes n} = \zeta_{R^n}$.  We also have, thanks to
  \cref{z:Ehi5guwR},
  \begin{align}
    \operatorname{tr}\bigl[{ L_{R^n}^\dagger L_{R^n} \sigma_{AR}^{\otimes n} }\bigr]
    = \int_{F(\sigma,\tau)\geq{e}^{-w}} d\tau 
    \operatorname{tr}\bigl[{ R_{R^n}^{(\tau)\dagger} R_{R^n}^{(\tau)} \, \sigma_{R}^{\otimes n} }\bigr]
    \geq 1- \operatorname{poly}(n)\exp({-nw})\ .
  \end{align}
  Thanks to the gentle measurement lemma (use \cref{z:bQEMTISK}),
  \begin{align}
    P\bigl({\sigma_{AR}^{\otimes n}  \,,\;
    L_{R^n} \, \sigma_{AR}^{\otimes n} L_{R^n} }\bigr)
    \leq \operatorname{poly}({n}) \exp\Bigl({ -\frac{n w}{2} }\Bigr)\ .
  \end{align}
  Let
  \begin{align}
    X_{R^n} &= L_{R^n}\ ;
    &
    Y_{R^n} &= L_{R^n}\ .
  \end{align}
  Using our constrained channel postselection theorem
  (\cref{z:UEUfoDLC}), we find
  \begin{subequations}
    \begin{align}
      \operatorname{tr}\Bigl[{
      Q_{j,\sigma}    \,
      \mathcal{E}_n\bigl({\sigma_{AR}^{\otimes n}}\bigr)
      }\Bigr]
      &\leq \operatorname{poly}({n}) \biggl\{{
        \int dM_{BR}\,
        \mathfrak{A}\bigl({M_{BR}}\bigr) \, \mathfrak{B}\bigl({M_{BR}}\bigr)
        \ + \ {e}^{-n w/2}
        }\biggr\}\ ;
        \label{z:REWabX3M}
        \\
      \begin{split}
      \mathfrak{A}\bigl({M_{BR}}\bigr)
      &\equiv \operatorname{tr}\Bigl[{
        Q_{j,\sigma}  \,
        \mathcal{M}^{\otimes n}\bigl({\sigma_{AR}^{\otimes n}}\bigr)
        }\Bigr]\ ;
        \\
      \mathfrak{B}\bigl({M_{BR}}\bigr)
      &\equiv F^2\Bigl({
        \mathcal{M}^{\otimes n}\bigl({ L_{R^n}\,\zeta_{A^nR^n} L_{R^n} }\bigr) \,,\;
        \mathcal{E}\bigl({ L_{R^n}\,\zeta_{A^nR^n} L_{R^n} }\bigr)
        }\Bigr)\ .
      \end{split}
    \end{align}
  \end{subequations}
  We split the integral into two parts: one integral ranging over channels
  $\mathcal{M}$ whose expectation values with $C^j$ are close to the prescribed
  $q_j$'s, and one integral over the complementary region.  Let
  $0 < \theta < \eta - \bar\eta$ to be fixed later.  We can write
  \begin{align}
    \hspace*{3em}&\hspace*{-3em}
    \int dM_{BR}\,
    \mathfrak{A}\bigl({M_{BR}}\bigr) \, \mathfrak{B}\bigl({M_{BR}}\bigr)
                   \nonumber\\
    \begin{split}
    &= 
      \int_{\forall j':\;\lvert {\operatorname{tr}[C^{j'}_{BR} M_{BR}] - q_{j'}}\rvert  < \eta - \theta} dM_{BR}\,
      \mathfrak{A}\bigl({M_{BR}}\bigr) \, \mathfrak{B}\bigl({M_{BR}}\bigr)
      \\
    &\quad+
      \int_{\exists j':\;\lvert {\operatorname{tr}[C^{j'}_{BR} M_{BR}] - q_{j'}}\rvert  \geq \eta - \theta} dM_{BR}\,
      \mathfrak{A}\bigl({M_{BR}}\bigr) \, \mathfrak{B}\bigl({M_{BR}}\bigr)\ .
    \end{split}
    \label{z:8ddKr8zS}
  \end{align}
  Consider the first integral in~\eqref{z:8ddKr8zS} and
  suppose that $\bigl \lvert {\operatorname{tr}[C^{j'}_{BR} M_{BR}] - q_{j'}}\bigr \rvert  < \eta - \theta$ for all
  $j' = 1, \ldots , J$.  By Hoeffding's inequality,
  \begin{align}
    \mathfrak{A}\bigl({M_{BR}}\bigr)
    &=
    \operatorname{tr}\bigl[{
        Q_{j,\sigma}  \,
        \mathcal{M}^{\otimes n}\bigl({\sigma_{AR}^{\otimes n}}\bigr)
    }\bigr]
      \nonumber\\
    &\leq
    \operatorname{tr}\bigl[{
        \bigl\{{ \overline{H^{j,\sigma}} \notin [{ \operatorname{tr}(C^j_{BR} M_{BR}) \pm \theta }] }\bigr\}
      \,
        \mathcal{M}^{\otimes n}\bigl({\sigma_{AR}^{\otimes n}}\bigr)
    }\bigr]
      \nonumber\\
    &\leq 2\exp\mathopen{}\left\{{
      -\frac{2\theta^2 n}{4 \bigl \lVert { \sigma_R^{-1/2} C^j_{BR} \sigma_R^{-1/2} }\bigr \rVert ^2}
    }\right\}\mathclose{}
      \leq 2\exp\mathopen{}\left\{{
      -\frac{\theta^2 y^2 n}{2 \lVert { C^j_{BR} }\rVert ^2}
    }\right\}\mathclose{}\ .
  \end{align}
  The first integral in~\eqref{z:8ddKr8zS} hence vanishes
  exponentially in $n$.  We now consider the second integral; suppose that there
  exists $j_0\in\{{1, \ldots J}\}$ such that
  $\bigl \lvert {\operatorname{tr}[C^{j_0}_{BR} M_{BR}] - q_{j_0}}\bigr \rvert  \geq \eta - \theta$.  Observe
  that
  \begin{align}
    L_{R^n}\,\lvert {\zeta_{A^nR^n}}\rangle 
    = 
    L_{R^n}\,\zeta_{R^n}^{1/2}\,\lvert {\Phi_{A^nR^n}}\rangle 
    = 
    \widehat\sigma_{R^n}^{1/2} \lvert {\Phi_{A^n:R^n}}\rangle 
    \equiv
    \lvert {\widehat\sigma_{A^nR^n}}\rangle \ .
  \end{align}
  Now define
  \begin{align}
    \bar\omega_{A^n\bar{R}^n}
    = \int_{F(\sigma_{\bar R},\tau_{\bar R})\geq e^{-w}}
    d\tau_{\bar R} \, \lvert {\tau}\rangle \mkern -1.8mu\relax \langle{\tau}\rvert _{A\bar R}^{\otimes n}
    \leq \Pi^{\mathrm{Sym}}_{A^n\bar R^n}\ ,
  \end{align}
  where we write
  $\lvert {\tau}\rangle _{A\bar R} \equiv \tau_{\bar R}^{1/2}\,\lvert {\Phi_{A:\bar R}}\rangle $.
  By construction,
  \begin{align}
    \bar\omega_{A^n}
    = \operatorname{tr}_{\bar R^n}\bigl[{ \bar\omega_{A^n\bar R^n } }\bigr]
    = \widehat\sigma_{A^n}\ .
  \end{align}
  Since $\bar\omega_{A^n\bar R^n}$ has support on the symmetric subspace of
  $(A\bar R)^n$, and by Carath\'eodory's theorem, there exists a collection
  $\bigl\{{ \tau_{\bar R}^{(\ell)} }\bigr\}_{\ell=1}^{\operatorname{poly}(n)}$ of at most $\operatorname{poly}(n)$
  states $\lvert {\tau^{(\ell)}}\rangle _{A\bar R} \equiv
  \bigl({\tau_{\bar R}^{(\ell)}}\bigr)^{1/2}\,\lvert {\Phi_{A:\bar R}}\rangle $ with
  $F\bigl({\sigma_{\bar R},\tau_{\bar R}^{(\ell)}}\bigr) \geq e^{-w}$, along with a probability
  distribution $\{{ \bar\kappa_\ell }\}$ with $\bar\kappa_1\geq\bar\kappa_2\geq\cdots$,
  such that
  \begin{align}
    \bar\omega_{A^n\bar R^n} = \sum_{\ell=1}^{\operatorname{poly}(n)} \bar\kappa_\ell \,
    \bigl({\tau_{A\bar R}^{(\ell)}}\bigr)^{\otimes n}\ .
  \end{align}
  This state can be purified using an additional system $R'$ with
  $d_{R'} \leq \operatorname{poly}(n)$:
  \begin{align}
    \bigl \lvert {\bar\omega}\bigr \rangle _{A^n\bar R^n R'}
    = \sum_{\ell=1}^{\operatorname{poly}(n)}
    \sqrt{\bar\kappa_\ell} \,
    \bigl \lvert {\tau^{(\ell)}}\bigr \rangle _{A\bar R}^{\otimes n} \otimes \bigl \lvert {\ell}\bigr \rangle _{R'}\ .
  \end{align}
  Because $\lvert {\widehat\sigma}\rangle _{A^nB^n}$ and
  $\bigl \lvert {\bar\omega}\bigr \rangle _{A^n\bar R^n R'}$ are both two
  purifications of the same state $\widehat\sigma_{A^n}$, they are related by some
  isometry acting on $R^n \to \bar R^n R'$.  The fidelity is invariant under the
  application of an isometry, so
  \begin{align}
    \mathfrak{B}(M_{BR})
    &= F^2\Bigl({
      \mathcal{M}^{\otimes n}\bigl[{ \widehat\sigma_{A^nR^n} }\bigr] \,,\;
      \mathcal{E}\bigl[{ \widehat\sigma_{A^nR^n} }\bigr]
      }\Bigr)
      \nonumber\\
    &= F^2\Bigl({
      \mathcal{M}^{\otimes n}\bigl[{ \bar\omega_{A^n\bar R^n R'} }\bigr] \,,\;
      \mathcal{E}\bigl[{ \bar\omega_{A^n\bar R^n R'} }\bigr]
      }\Bigr)\ .
    \label{z:E4j5LmXO}
  \end{align}
  Now consider the two-outcome POVM $\{{ P_{B^n\bar R^n}\otimes\mathds{1}_{R'} , %
    P_{B^n\bar R^n}^\perp\otimes\mathds{1}_{R'} }\}$.  By the data processing
  inequality of the fidelity (in the form of \cref{z:MduRjmZz}), we find
  \begin{align}
    \text{\eqref{z:E4j5LmXO}}
    &\leq \biggl({
      \sqrt{\operatorname{tr}\bigl[{ P_{B^n\bar R^n}  \,
          \mathcal{M}^{\otimes n}\bigl({ \bar\omega_{B^n\bar R^n R'} }\bigr) }\bigr]}
      + \sqrt{\operatorname{tr}\bigl[{ P_{B^n\bar R^n}^\perp  \,
          \mathcal{E}\bigl({ \bar\omega_{B^n\bar R^n R'} }\bigr) }\bigr]}
      }\biggr)^2\ .
  \end{align}
  Recall that $F\bigl({\tau_R^{(\ell)}, \sigma_R}\bigr) \geq e^{-w}$, so
  $D\bigl({\tau_R^{(\ell)}, \sigma_R}\bigr) \leq P\bigl({\tau_R^{(\ell)}, \sigma_R}\bigr)
  \leq \sqrt{1 - e^{-2w}} \leq \sqrt{2w}$, since $e^{-2w} \geq 1 - 2w$.  Then
  $
  \lambda_{\mathrm{min}}\bigl({\tau_R^{(\ell)}}\bigr)
  \geq \lambda_{\mathrm{min}}(\sigma_R) - \sqrt{2w}
  \geq \nu y - \sqrt{2w}
  $.
  At this point, we choose
  \begin{align}
    w = \frac12 (\nu - 1)^2 y^2\ ,
    \label{z:E42ukPTk}
  \end{align}
  which ensures that
  \begin{align}
    \lambda_{\mathrm{min}}\bigl({\tau_R^{(\ell)}}\bigr)
    \geq y\ .
  \end{align}
  We then find, thanks to \cref{z:lmCUsTOF},
  \begin{align}
    \operatorname{tr}\bigl[{ P_{B^n\bar R^n} \, \mathcal{M}^{\otimes n}\bigl({ \bar\omega_{A^n\bar R^n R'} }\bigr) }\bigr]
    &= \sum_{\ell =1}^{\operatorname{poly}(n)}
      \bar\kappa_\ell \operatorname{tr}\Bigl({ P_{B^n\bar R^n}\,
      \bigl[{\mathcal{M}\bigl({ \tau_{AR}^{(\ell)} }\bigr) }\bigr]^{\otimes n} }\Bigr)
      \nonumber\\
    &\leq
      \operatorname{poly}({n}) \exp\Biggl\{{
          -cn \frac{ (\eta-\theta-\bar\eta)^8 y^8 }{ 5^8 \max_j\lVert {C^{j}_{BR}}\rVert ^8 }
      }\Biggr\}\ .
  \end{align}
  recalling we chose $c=1/2$.
  On the other hand, using our initial assumption we have
  \begin{align}
    \operatorname{tr}\bigl[{ P_{B^n\bar R^n}^\perp \, \mathcal{E}\bigl({ \bar\omega_{A^n\bar R^n R'} }\bigr) }\bigr]
    &= \sum_{\ell=1}^{\operatorname{poly}(n)}
      \bar\kappa_\ell \operatorname{tr}\bigl[{ P_{B^n\bar R^n}^\perp \,
      \mathcal{E}\bigl({ \tau_{AR}^{(\ell) \otimes n} }\bigr) }\bigr]
      \leq \epsilon\ .
  \end{align}
  In summary, and using the fact that
  $(\sqrt{x_1} + \sqrt{x_2})^2 \leq [2\max(\sqrt{x_1},\sqrt{x_2})]^2 \leq
  4\max(x_1, x_2)$ for any $x_1,x_2\geq 0$, we find:
  \begin{align}
    \mathfrak{B}\bigl({M_{BR}}\bigr)
    \leq \operatorname{poly}({n})  \max\Biggl({
        \epsilon
        \,,\;
        {e}^{-cn \frac{({\eta-\theta-\bar\eta})^8 y^8}{5^8\,\max_j\lVert {C^{j}_{BR}}\rVert ^8}}
    }\Biggr)\ .
  \end{align}
  The same then bound applies to the second integral
  in~\eqref{z:8ddKr8zS}.  Combining the above inequalities,
  we find a bound on the original quantity~\eqref{z:REWabX3M} we
  were interested in:
  \begin{align}
    \hspace*{1em}&\hspace*{-1em}
    \operatorname{tr}\bigl[{
    Q_{j,\sigma}\,\mathcal{E}\bigl({\sigma_{AR}^{\otimes n}}\bigr)
    }\bigr]
    \nonumber\\
    &\leq \operatorname{poly}({n}) \Biggl\{{
      {e}^{-n\frac{w}{2}}
      + {e}^{-n\frac{\theta^2 y^2}{2\lVert {C^j_{BR}}\rVert ^2}}
      + \max\Biggl({\epsilon, 
          {e}^{-cn\frac{({\eta-\theta-\bar\eta})^8 y^8}{5^8\,\max_j\lVert {C^{j}_{BR}}\rVert ^8}}
      }\Biggr)
      }\Biggr\}
      \nonumber\\
    &\leq \operatorname{poly}({n}) \exp\Biggl\{{
      -n\min\Biggl({
        \frac{(\nu-1)^2y^2}{4}
        \,,\,
        \frac{\theta^2 y^2}{2 \max_j \lVert {C^j_{BR}}\rVert ^2}
        \,,\,
        -\frac{\log(\epsilon)}{n}
        \,,\,
        \frac{c({\eta-\theta-\bar\eta})^8 y^8}{5^8 \max_j \lVert {C^j_{BR}}\rVert ^8}
    }\Biggr)
    }\Biggr\}\ ,
  \end{align}
  recalling the value of $w$
  from~\eqref{z:E42ukPTk}, and for any
  $0<\theta<\eta-\bar\eta$.  Now choose $\theta = (\eta-\bar\eta)/2$, such that
  $\theta = \eta-\bar\eta-\theta=(\eta-\eta')/4$.  At this point, we also assume
  that $(\nu-1)/2 \geq (\eta-\eta')/(8\max_j \lVert {C^j_{BR}}\rVert )$, as in the
  theorem statement.  Then the first argument of the `$\min(\cdot)$' is always
  greater than or equal to its second argument.  Using
  $\theta/\max_j \lVert {C^j_{BR}}\rVert  \leq 1$ from our assumptions on $\eta,\eta'$,
  along with $c<1$ and $y<1$, we find that the second argument of the
  `$\min(\cdot)$' is always greater than the fourth.  The bound therefore
  simplifies to
  \begin{align}
    \operatorname{tr}\bigl[{
    Q_{j,\sigma}\,\mathcal{E}\bigl({\sigma_{AR}^{\otimes n}}\bigr)
    }\bigr]
    \leq
    \operatorname{poly}({n}) \exp\Biggl\{{
    -n y^4 \min\biggl({
        -\frac{\log(\epsilon)}{n y^4}
        \,,\;
        \frac{c({\eta-\eta'})^8}{5^8 \max_j \lVert {C^j_{BR}}\rVert ^8}
    }\biggr)
    }\Biggr\}\ ,
  \end{align}
  recalling that we chose $c=1/2$ in the theorem statement.

  Now let's prove~\ref{z:wEDQuyLi}.  The
  structure of this proof is very similar to the previous proof.  Without loss
  of generality, we can assume $\mathcal{E}_{A^n\to B^n}$ to be
  permutation-invariant, since both $P^\perp_{B^nR^n}$ and the concentration
  test operators are permutation-invariant.  We consider any $\nu'>1$ in this
  proof, and will use our additional assumption on $\nu'$ only to simplify the
  final bound.  Consider any $\sigma_R\geq \nu' y'\mathds{1}$.  For $w'>0$ to be
  fixed later, consider the subnormalized state
  \begin{align}
    \widehat\sigma'_{R^n} \equiv
    \int_{F(\sigma,\tau)\geq e^{-w'}} d\tau \,
      \tau_R^{\otimes n}\ ,
  \end{align}
  and define
  \begin{align}
    L_{R^n}'
    &= \widehat\sigma_{R^n}'^{1/2}\, \zeta_{R^n}^{-1/2}\ .
  \end{align}
  As we saw earlier in the proof
  of~\ref{z:W3JG7egy}, $L_{R^n}'$ is
  Hermitian, satisfies $0 \leq L_{R^n}' \leq \mathds{1}$, and is such that
  \begin{align}
    P\bigl({\sigma_{AR}^{\otimes n}  \,,\;
    L_{R^n}' \, \sigma_{AR}^{\otimes n} \, L_{R^n}' }\bigr)
    \leq \operatorname{poly}({n}) \exp\bigl({ -n w'/2 }\bigr)\ .
  \end{align}
  Let
  \begin{align}
    X'_{R^n} &= L_{R^n}'\ ;
    &
    Y'_{R^n} &= L_{R^n}'\ .
  \end{align}
  We write
  \begin{align}
      \operatorname{tr}\bigl[{ P^\perp_{B^nR^n} \mathcal{E}\bigl({\sigma_{AR}^{\otimes n}}\bigr) }\bigr]
    &\leq \operatorname{tr}\bigl[{ P^\perp_{B^nR^n}
      X'_{R^n}
      Y'_{R^n}
      \, \mathcal{E}\bigl({\sigma_{AR}^{\otimes n}}\bigr) \,
      Y_{R^n}'^{\dagger}
      X_{R^n}'^\dagger
      }\bigr]
      + \operatorname{poly}({n}) \, {e}^{-n w'/2}\ .
        \nonumber\\
    &\leq \operatorname{tr}\bigl[{ P^\perp_{B^nR^n}
      \bigl({\sigma_{R}^{\otimes n}}\bigr)^{1/2}
      X'_{R^n}
      Y'_{R^n}
      \, \mathcal{E}\bigl({\Phi_{A^n:R^n}}\bigr) \,
      Y_{R^n}'^{\dagger}
      X_{R^n}'^\dagger
      \bigl({\sigma_{R}^{\otimes n}}\bigr)^{1/2}
      }\bigr]
      \nonumber\\
      &\quad
      + \operatorname{poly}({n}) \, {e}^{-n w'/2}\ .
  \end{align}
  Our constrained channel postselection theorem
  (\cref{z:UEUfoDLC}) then implies that:
  \begin{subequations}
    \begin{align}
      \operatorname{tr}\bigl[{ P^\perp_{B^nR^n} \mathcal{E}\bigl({\sigma_{AR}^{\otimes n}}\bigr) }\bigr]
      &\leq \operatorname{poly}({n})  \, \Biggl\{ \int dM_{BR}\;
        \mathfrak{A}'(M_{BR})
        \;
        \mathfrak{B}'(M_{BR})
        \ +\  {e}^{-n w'/2}
        \Biggr\}\ ;
        \label{z:leertGa-}\\
      \begin{split}
        \mathfrak{A}'(M_{BR})
        &\equiv
          \operatorname{tr}\Bigl[{
          P^\perp_{B^nR^n} 
          \mathcal{M}^{\otimes n}\Bigl({  \sigma_{AR}^{\otimes n} }\Bigr)
          }\Bigr]\ ;
        \\
        \mathfrak{B}'(M_{BR})
        &\equiv
          F^2\Bigl({
          \mathcal{M}^{\otimes n}\Bigl[{L_{R^n}' \zeta_{A^nR^n} L_{R^n}'}\Bigr]\ ,\ 
          \mathcal{E}_{A^n\to B^n}\Bigl[{L_{R^n}' \zeta_{A^nR^n} L_{R^n}'}\Bigr]
          }\Bigr)\ .
      \end{split}
    \end{align}
  \end{subequations}
  We split the integral into two parts: one integral ranging over channels
  $\mathcal{M}$ whose expectation values with $C^j$ are close to the prescribed
  $q_j$'s, and one integral over the complementary region.  Let
  $0 < \theta' < \bar\eta-\eta'$ to be fixed later.  First, suppose that
  $\bigl \lvert {\operatorname{tr}[C^j_{BR} M_{BR}] - q_j}\bigr \rvert  < \bar\eta - \theta'$ for all
  $j = 1, \ldots , J$.  By \cref{z:HY7bJv7-}, we know in
  this case that
  \begin{align}
    \operatorname{tr}\bigl[{ P_{B^nR^n}^\perp \bigl({\mathcal{M}(\sigma_{AR})}\bigr)^{\otimes n} }\bigr]
    \leq \operatorname{poly}({n}) \exp\biggl\{{
        -cn\,\frac{\theta'^8 (\nu' y')^8}{5^8 \max_j \lVert {C^j_{BR}}\rVert ^8}
    }\biggr\}\ ,
  \end{align}
  so the integrand in~\eqref{z:leertGa-} vanishes exponentially
  in $n$ for channels $M_{BR}$ obeying
  $\bigl \lvert {\operatorname{tr}[C^j_{BR} M_{BR}] - q_j}\bigr \rvert  < \bar\eta - \theta'$ for all $j$.
  Now, suppose instead that there exists $j_0\in\{{1, \ldots J}\}$ such that
  $\bigl \lvert {\operatorname{tr}[C^{j_0}_{BR} M_{BR}] - q_{j_0}}\bigr \rvert  \geq \bar\eta - \theta'$.  Our
  strategy to upper bound the integrand in~\eqref{z:leertGa-}
  for such channels is to upper bound the $\mathfrak{B}'( M_{BR} )$ term.
  As earlier,
  \begin{align}
    L_{R^n}'\,\lvert {\zeta_{A^nR^n}}\rangle 
    = 
    \widehat\sigma_{R^n}'^{1/2} \lvert {\Phi_{A^n:R^n}}\rangle 
    \equiv
    \lvert {\tilde\sigma_{A^nR^n}'}\rangle \ .
  \end{align}
  Now define
  \begin{align}
    \bar\omega_{A^n\bar{R}^n}'
    = \int_{F(\sigma_{\bar R},\tau_{\bar R})\geq e^{-w'}}
    d\tau_{\bar R} \, \lvert {\tau}\rangle \mkern -1.8mu\relax \langle{\tau}\rvert _{A\bar R}^{\otimes n}
    \leq \Pi^{\mathrm{Sym}}_{A^n\bar R^n}\ ,
  \end{align}
  where we write
  $\lvert {\tau}\rangle _{A\bar R} \equiv \tau_{\bar R}^{1/2}\,\lvert {\Phi_{A:\bar R}}\rangle $.
  By construction,
  \begin{align}
    \bar\omega_{A^n}'
    = \operatorname{tr}_{\bar R^n}\bigl[{ \bar\omega_{A^n\bar R^n }' }\bigr]
    = \tilde\sigma_{A^n}'\ .
  \end{align}
  Since $\bar\omega_{A^n\bar R^n}'$ has support on the symmetric subspace of
  $(A\bar R)^n$, and by Carath\'eodory's theorem, there exists a collection
  $\bigl\{{ \tau_{\bar R}'^{(\ell)} }\bigr\}_{\ell=1}^{\operatorname{poly}(n)}$ of at most $\operatorname{poly}(n)$
  states
  $\lvert {\tau'^{(\ell)}}\rangle _{A\bar R} \equiv
  \bigl[{\tau_{\bar R}'^{(\ell)}}\bigr]^{1/2}\,\lvert {\Phi_{A:R}}\rangle $ with
  $F(\sigma_{\bar R},\tau'^{(\ell)}_{\bar R}) \geq e^{-w'}$, along with a
  probability distribution $\{{ \bar\kappa'_\ell }\}$ with
  $\bar\kappa'_1\geq\bar\kappa'_2\geq\cdots$, such that
  \begin{align}
    \bar\omega_{A^n\bar R^n}' = \sum_{\ell=1}^{\operatorname{poly}(n)} \bar\kappa'_\ell \,
    \bigl({\tau_{A\bar R}'^{(\ell)}}\bigr)^{\otimes n}\ .
  \end{align}
  This state can be purified using an additional system $R'$ with
  $d_{R'} \leq \operatorname{poly}(n)$:
  \begin{align}
    \bigl \lvert {\bar\omega'}\bigr \rangle _{A^n\bar R^n R'}
    = \sum_{\ell=1}^{\operatorname{poly}(n)}
    \sqrt{\bar\kappa'_\ell} \,
    \bigl \lvert {\tau'^{(\ell)}}\bigr \rangle _{A\bar R}^{\otimes n} \otimes \bigl \lvert {\ell}\bigr \rangle _{R'}\ .
  \end{align}
  Because $\lvert {\widehat\sigma'}\rangle _{A^nB^n}$ and
  $\bigl \lvert {\bar\omega'}\bigr \rangle _{A^n\bar R^n R'}$ are both two
  purifications of the same state $\widehat\sigma_{A^n}'$, they are related by some
  isometry acting on $R^n \to \bar R^n R'$.  The fidelity is invariant under the
  application of an isometry, so
  \begin{align}
    \mathfrak{B}'(M_{BR})
    &= F^2\Bigl({
      \mathcal{M}^{\otimes n}\bigl[{ \widehat\sigma_{A^nR^n}' }\bigr] \,,\;
      \mathcal{E}\bigl[{ \widehat\sigma_{A^nR^n}' }\bigr]
      }\Bigr)
      \nonumber\\
    &= F^2\Bigl({
      \mathcal{M}^{\otimes n}\bigl[{ \bar\omega_{A^n\bar R^n R'}' }\bigr] \,,\;
      \mathcal{E}\bigl[{ \bar\omega_{A^n\bar R^n R'}' }\bigr]
      }\Bigr)\ .
    \label{z:mT1tkH8O}
  \end{align}
  At this point, we define the following two-outcome POVM:
  \begin{align}
    \begin{split}
    \mathfrak{Q}_{B^n\bar R^n R'}
    &= \sum_{\ell=1}^{\operatorname{poly}(n)}
    \bigl\{{ \overline{H^{j_0,\tau'^{(\ell)}}}_{B^n\bar R^n} \in [ q_{j_0} \pm \eta' ] }\bigr\}
    \otimes \lvert {\ell}\rangle \mkern -1.8mu\relax \langle{\ell}\rvert _{R'}\ ;
      \\
    \mathfrak{Q}_{B^n\bar R^n R'}^\perp
    &= \sum_{\ell=1}^{\operatorname{poly}(n)}
    \bigl\{{ \overline{H^{j_0,\tau'^{(\ell)}}}_{B^n\bar R^n} \notin [ q_{j_0} \pm \eta' ] }\bigr\}
    \otimes \lvert {\ell}\rangle \mkern -1.8mu\relax \langle{\ell}\rvert _{R'}\ ,
    \end{split}
  \end{align}
  noting that
  $\mathfrak{Q}_{B^n\bar R^n R'} + \mathfrak{Q}_{B^n\bar R^n R'}^\perp =
  \mathds{1}_{B^n\bar R^n R'}$.  This measurement can be realized by first measuring
  the $R'$ register to obtain an outcome $\ell$, then testing whether or not the
  resulting state is within a subspace of eigenvalues of
  $\overline{H^{j_0,\tau'^{(\ell)}}}_{B^n\bar R^n}$ with $\eta$ of $q_{j_0}$.  By
  the data processing inequality of the fidelity (in the form of
  \cref{z:MduRjmZz}), we find
  \begin{align}
    \text{\eqref{z:mT1tkH8O}}
    &\leq \biggl({
      \sqrt{\operatorname{tr}\bigl[{ \mathfrak{Q}  \,
          \mathcal{M}^{\otimes n}\bigl({ \bar\omega_{B^n\bar R^n R'}' }\bigr) }\bigr]}
      + \sqrt{\operatorname{tr}\bigl[{ \mathfrak{Q}^\perp  \,
          \mathcal{E}^{\otimes n}\bigl({ \bar\omega_{B^n\bar R^n R'}' }\bigr) }\bigr]}
      }\biggr)^2\ .
  \end{align}
  We find
  \begin{subequations}
    \begin{align}
      \operatorname{tr}\bigl[{ \mathfrak{Q}  \,
      \mathcal{M}^{\otimes n}\bigl({ \bar\omega_{B^n\bar R^n R'}' }\bigr) }\bigr]
      &= \sum_\ell \bar\kappa'_\ell
        \operatorname{tr}\Bigl[{ 
        \bigl\{{ \overline{H^{j_0,\tau'^{(\ell)}}}_{B^n\bar R^n} \in [ q_{j_0} \pm \eta' ] }\bigr\} \;
        \mathcal{M}^{\otimes n}\bigl({\tau'^{(\ell) \otimes n}_{A\bar R} }\bigr)
        }\Bigr]\ ;
        \label{z:410UdnZj}\\
      \operatorname{tr}\bigl[{ \mathfrak{Q}^\perp  \,
      \mathcal{E}\bigl({ \bar\omega_{B^n\bar R^n R'}' }\bigr) }\bigr]
      &= \sum_\ell \bar\kappa'_\ell \operatorname{tr}\Bigl[{ 
        \bigl\{{ \overline{H^{j_0,\tau'^{(\ell)}}}_{B^n\bar R^n} \notin [ q_{j_0} \pm \eta' ] }\bigr\} \;
        \mathcal{E}\bigl({\tau'^{(\ell) \otimes n}_{A\bar R} }\bigr)
        }\Bigr]\ .
        \label{z:oadtlJuI}
    \end{align}
  \end{subequations}
  As before, we know that $F\bigl({\tau_R'^{(\ell)}, \sigma_R}\bigr) \geq e^{-w'}$, which implies
  $D\bigl({\tau_R'^{(\ell)}, \sigma_R}\bigr)
  \leq \sqrt{2w}$
  and
  $
  \lambda_{\mathrm{min}}\bigl({\tau_R'^{(\ell)}}\bigr)
  \geq \nu' y' - \sqrt{2w'}
  $.
  A suitable choice of $w'$ ensures that
  $\lambda_{\mathrm{min}}\bigl({\tau_R'^{(\ell)}}\bigr) \geq y'$, namely
  \begin{align}
    w' = \frac12 (\nu' - 1)^2 y'^2\ .
    \label{z:yWe3-F-H}
  \end{align}
  Then, by assumption, we have
  \begin{align}
    \operatorname{tr}\bigl[{ \mathfrak{Q}^\perp  \,
    \mathcal{E}\bigl({ \bar\omega_{B^n\bar R^n R'}' }\bigr) }\bigr]
    =
    \text{\eqref{z:oadtlJuI}}
    \leq \delta' \ .
  \end{align}
  On the other hand, the measurement of
  $\overline{H^{j_0,\tau'^{(\ell)}}}_{B^n\bar R^n}$ on the i.i.d.\@ state
  $\bigl[{\mathcal{M}\bigl({\tau_R'^{(\ell)}}\bigr)}\bigr]^{\otimes n}$ concentrates around
  the measurement average $\operatorname{tr}({C^{j_0}_{BR} M_{BR}})$.  By Hoeffding's
  inequality, and since
  $\bigl \lvert { \operatorname{tr}\bigl[{ C^{j_0}_{BR} \, M_{BR} }\bigr] - q_{j_0} }\bigr \rvert  \geq \bar\eta -
  \theta'$,
  \begin{align}
    \text{\eqref{z:410UdnZj}}
    &\leq
    \sum_\ell \bar\kappa'_\ell \operatorname{tr}\Bigl[{
        \bigl\{{ \overline{H^{j_0,\tau'^{(\ell)}}}_{B^n\bar R^n}
           \not\in [ \operatorname{tr}({C^{j_0}_{BR} M_{BR}}) \pm (\bar\eta -\theta'-\eta') ] }\bigr\} \;
        \mathcal{M}^{\otimes n}\bigl({\tau'^{(\ell) \otimes n}_{A\bar R} }\bigr)
        }\Bigr]
    \nonumber\\
    &\leq
      \sum_\ell 2 \bar\kappa'_\ell \exp\biggl\{{
      -\frac{ 2 (\bar\eta - \theta' - \eta)^2 n}{4
          \bigl \lVert {[\tau_R'^{(\ell)}]^{-1/2} \, C^{j_0}_{BR}\, [\tau_R'^{(\ell)}]^{-1/2} }\bigr \rVert ^2
      }
      }\biggr\}
      \nonumber\\
    &\leq
      2\exp\biggl\{{
      -\frac{ (\bar\eta - \theta' - \eta)^2 y'^2 n}{2 \bigl \lVert {C^{j_0}_{BR}}\bigr \rVert ^2 }
      }\biggr\}\ .
  \end{align}
  Then we find
  \begin{align}
    \mathfrak{B}'\bigl({M_{BR}}\bigr)
    &\leq \operatorname{poly}({n})\,\Biggl({
      \sqrt{\delta'}
      + {e}^{ -\frac{(\bar\eta-\theta'-\eta)^2 y'^2 n}{4\lVert {C^{j_0}_{BR}}\rVert ^2}  }
      }\Biggr)^2
      \leq \operatorname{poly}({n})\, \max\Biggl({
      \delta'  \,,\;
      {e}^{-\frac{(\bar\eta-\theta'-\eta)^2 y'^2 n}{2\lVert {C^{j_0}_{BR}}\rVert ^2}  }
      }\Biggr)\ .
  \end{align}

  Combining the above inequalities, we finally find that for any
  $\sigma_R\geq \nu' y'\mathds{1}_R$, we have
  \begin{align}
    \operatorname{tr}\bigl({ P_{B^nR^n}^\perp  \mathcal{E}(\sigma_{AR}^{\otimes n}) }\bigr)
    \leq
    \operatorname{poly}({n})\,\Biggl\{{
    {e}^{-\frac{n w'}{2}}
    + {e}^{-cn \frac{\theta'^8 (\nu' y')^8}{5^8\max_j\lVert {C^j_{BR}}\rVert ^8}}
    + \max\Biggl({
      \delta'  \,,\;
      {e}^{-\frac{(\bar\eta-\theta'-\eta)^2 y'^2 n}{2\lVert {C^{j_0}_{BR}}\rVert ^2}  }
      }\Biggr)
    }\Biggr\}
    \nonumber\\
    \leq \operatorname{poly}({n})
    \exp\Biggl\{{
        -n\min\Biggl({
            -\frac{\log(\delta')}{n}   ,
            \frac{(\nu'-1)^2y^2}{4} ,
            \frac{c\,\theta'^8\nu'^8y^8}{5^8 \max_j\lVert {C^j_{BR}}\rVert ^8}   ,
            \frac{(\bar\eta - \theta' - \eta)^2 y'^2}{2\max_j\lVert {C^{j}_{BR}}\rVert ^2}
        }\Biggr)
    }\Biggr\}\ ,
  \end{align}
  recalling the value of $w'$ in~\eqref{z:yWe3-F-H} and for any
  $0<\theta'<\bar\eta -\eta'$.  Now, we choose $\theta' = (\bar\eta - \eta')/2$
  such that $\theta'=\bar\eta-\theta'-\eta = (\eta-\eta')/4$.  Additionally, we
  now assume that $(\nu'-1)/2 \geq (\eta - \eta')/(8 \max_j\lVert {C^j_{BR}}\rVert )$, as
  per the theorem statement; in consequence, the second argument of the minimum
  is always lower bounded by the fourth.  Using $y'^2\geq y'^8$ and
  $\theta'/\lVert {C^j_{BR}}\rVert  \leq 1$, we can further simplify the bound to
  \begin{align}
    \operatorname{tr}\bigl({ P_{B^nR^n}^\perp  \mathcal{E}(\sigma_{AR}^{\otimes n}) }\bigr)
    \leq
    \operatorname{poly}({n}) \exp\Biggl\{{
    -n y'^8 \min\Biggl({
        \frac{-\log\bigl({\delta'}\bigr)}{ n y'^8 } ,
        \frac{c\,(\eta-\eta')^8}{5^8\,\max_j \lVert {C^j_{BR}}\rVert ^8 }
    }\Biggr)
    }\Biggr\}\ ,
  \end{align}
  recalling we chose $c=1/2$.
\end{proof}

\end{document}